\documentclass[acmsmall,nonacm]{acmart}

\newif\ifshort
\newif\iflong
\longtrue

\ifshort\documentclass[acmsmall,screen]{acmart}\fi

\usepackage{mathpartir} 
\usepackage{tikz-cd}    
\usepackage{environ}    

\usepackage{xcolor}

\setcopyright{rightsretained}
\acmJournal{PACMPL}
\acmYear{2019} \acmVolume{3} \acmNumber{POPL} \acmArticle{15} \acmMonth{1} \acmPrice{}\acmDOI{10.1145/3290328}

\setcopyright{none}             

\setcitestyle{nosort}

\title{Gradual Type Theory\iflong (Extended Version)\fi}

\author{Max S. New}
\affiliation{
  \institution{Northeastern University}            
}
\email{maxnew@ccs.neu.edu}          

\author{Daniel R. Licata}
\affiliation{
  \institution{Wesleyan University}            
}
\email{dlicata@wesleyan.edu}          

\author{Amal Ahmed}
\affiliation{
  \institution{Northeastern University and Inria Paris}            
}
\email{amal@ccs.neu.edu}          

\renewcommand{\u}{\underline}

\NewEnviron{longproof}{%
  \iflong\begin{proof}\BODY \end{proof}\fi%
}
\NewEnviron{longfigure}{%
  \iflong\begin{figure}\BODY \end{figure}\fi%
}
\NewEnviron{longonly}{%
  \iflong\BODY \fi%
}
\NewEnviron{shortonly}{%
  \ifshort\BODY \fi%
}

\newtheorem*{nonnum-theorem}{Theorem}

\newcommand{\cbpv}{CBPV}
\newcommand{\cbpvstar}{CBPV*}
\newcommand{\cbpvtxt}{\cbpv}
\newcommand{\sem}[1]{\llbracket#1\rrbracket}
\newcommand{\sdncast}[2]{\sem{\dncast{#1}{#2}}}
\newcommand{\supcast}[2]{\sem{\upcast{#1}{#2}}}
\newcommand{\srho}[1]{\sem{#1}_\rho}

\newcommand{\vtype}{\,\,\text{val type}}
\newcommand{\ctype}{\,\,\text{comp type}}

\newcommand{\pipe}{\,\,|\,\,}

\newcommand{\floor}[1]{\lfloor#1\rfloor}
\newcommand{\ltdyn}{\sqsubseteq}
\newcommand{\gtdyn}{\sqsupseteq}

\newcommand{\precltdyn}{\mathrel{\preceq\ltdyn}}

\newcommand{\ltdynsucc}{\mathrel{\ltdyn\succeq}}
\newcommand{\equidyn}{\mathrel{\gtdyn\ltdyn}}

\newcommand{\pole}{\Bot}

\newcommand{\logty}[2]{\mathrel{\lesssim^{#1}_{#2,\pole}}}

\newcommand{\apreorder}{\trianglelefteq}

\newcommand{\ctxize}[1]{\mathrel{{#1}^{\text{ctx}}}}
\newcommand{\ix}[2]{\mathrel{#1^{#2}}}
\newcommand{\simsub}[1]{\mathrel{\sim_{#1}}}

\newcommand{\itylrof}[3]{\ilrof{#1}{#3,#2}}
\newcommand{\ilrof}[2]{\mathrel{{#1}^{\text{log}}_{#2}}}
\newcommand{\itylr}[2]{\itylrof{\apreorder}{#1}{#2}}
\newcommand{\ilr}[1]{\ilrof{\apreorder}{#1}}

\newcommand{\simp}[1]{{#1}^{\dag}}
\newcommand{\simpp}[1]{\simp{({#1})}}

\newcommand{\step}{\mapsto}
\newcommand{\stepsin}[1]{\mathrel{\mapsto^{#1}}}
\newcommand{\stepzero}{\stepsin 0}
\newcommand{\stepone}{\stepsin 1}

\newcommand{\bigstepsin}[1]{\mathrel{\Mapsto^{#1}}}
\newcommand{\bigstepzero}{\bigstepsin{0}}
\newcommand{\bigstepany}{\bigstepsin{}}

\newcommand{\pair}[2]{\{ \pi \mapsto {#1} \pipe \pi' \mapsto {#2}\}}
\newcommand{\pairone}[1]{\{ \pi \mapsto {#1}}
\newcommand{\pairtwo}[1]{\pipe \pi' \mapsto {#1}\}}
\newcommand{\emptypair}[0]{\{\}}

\newcommand{\tru}{\texttt{true}}
\newcommand{\fls}{\texttt{false}}
\newcommand{\ifXthenYelseZ}[3]{\kw{if}#1\,\kw{then}#2\,\kw{else}#3}
\newcommand{\bool}{\mathbb{B}}
\newcommand{\inl}{\kw{inl}}
\newcommand{\inr}{\kw{inr}}

\newcommand{\els}{\kw {else}}
\newcommand{\seq}{\texttt{seq}}

\newcommand{\dyn}{{?}}
\newcommand{\dynv}{{?}}
\newcommand{\dync}{\u {\text{?`}}}
\newcommand{\uarrow}{\mathrel{\rotatebox[origin=c]{-30}{$\leftarrowtail$}}}
\newcommand{\darrow}{\mathrel{\rotatebox[origin=c]{30}{$\twoheadleftarrow$}}}
\newcommand{\upcast}[2]{\langle{#2}\uarrow{#1}\rangle}
\newcommand{\dncast}[2]{\langle{#1}\darrow{#2}\rangle}
\newcommand{\obcast}[2]{\langle{#1}\Leftarrow{#2}\rangle}
\newcommand{\err}{\mho}
\newcommand{\diverge}{\Omega}
\newcommand{\print}{\kw{print}}
\newcommand{\roll}{\kw{roll}}
\newcommand{\rollty}[1]{\texttt{roll}_{#1}\,\,}
\newcommand{\unroll}{\kw{unroll}}
\newcommand{\unrollty}[1]{\texttt{unroll}_{#1}\,\,}
\newcommand{\defupcast}[2]{\langle\kern-1ex~\langle{#2}\uarrow{#1}\rangle\kern-1ex~\rangle}
\newcommand{\defdncast}[2]{\langle\kern-1ex~\langle{#1}\darrow{#2}\rangle\kern-1ex~\rangle}

\newcommand{\result}{\text{result}}

\newcommand{\lett}{\kw{let}}
\newcommand{\letXbeYinZ}[2]{\lett#2 = #1;}
\newcommand{\bindXtoYinZ}[2]{\kw{bind}#2 \leftarrow #1;}

\newcommand{\case}{\kw{case}}
\newcommand{\kw}[1]{\texttt{#1}\,\,}
\newcommand{\absurd}{\kw{absurd}}
\newcommand{\caseofXthenYelseZ}[3]{\case #1 \{ #2 \pipe #3 \}}
\newcommand{\caseofX}[1]{\case #1}
\newcommand{\thenY}{\{}
\newcommand{\elseZ}[1]{\pipe #1 \}}
\newcommand{\pmpairWtoXYinZ}[4]{\kw{split} #1\,\kw{to} (#2,#3). #4}
\newcommand{\pmpairWtoinZ}[2]{\kw{split} #1\,\kw{to} (). #2}
\newcommand{\pmmuXtoYinZ}[3]{\kw{unroll} #1 \,\kw{to} \roll #2. #3}
\newcommand{\ret}{\kw{ret}}
\newcommand{\thunk}{\kw{thunk}}
\newcommand{\force}{\kw{force}}
\newcommand{\abort}{\kw {abort}}
\newcommand{\with}{\mathbin{\&}}
\DeclareMathOperator*{\With}{\&}

\newcommand{\dyncaseofXthenOnePairSumU}[5]{\kw{tycase} #1\,\{ #2 \pipe #3 \pipe #4 \pipe #5 \}}
\newcommand{\dyncaseofXthenYelseZ}[3]{\kw{tycase} #1\,\{ #2 \pipe #3 \}}
\newcommand{\dyncocaseWithFunF}[3]{\{\with \mapsto #1 \pipe (\to) \mapsto #2 \pipe \u F \mapsto #3 \}}
\newcommand{\dyncaseofXthenBoolUPair}[4]{\kw{tycase} #1\,\{ #2 \pipe #3 \pipe #4 \}}
\newcommand{\dyncocaseFunF}[2]{\{(\to) \mapsto #1 \pipe \u F \mapsto #2 \}}

\newcommand{\bnfalt}{\mathrel{\bf \,\mid\,}}
\newcommand{\alt}{\bnfalt}
\newcommand{\bnfdef}{\mathrel{\bf ::=}}
\newcommand{\bnfadd}{\mathrel{\bf +::=}}
\newcommand{\bnfsub}{\mathrel{\bf -::=}}

\newcommand\dynvctx[0]{\mathsf{dyn\mathord{-}vctx}}
\newcommand\dyncctx[0]{\mathsf{dyn\mathord{-}cctx}}

\newcommand\leastdynv[0]{\bot_{\mathsf{v}}}
\newcommand\leastdync[0]{\bot_{\mathsf{c}}}

\newcommand\errordivergeleft[0]{\preceq\ltdyn}
\newcommand\errordivergeright[0]{\ltdyn\succeq}
\newcommand\errordivergerightop[0]{\preceq\gtdyn}

\begin{document}

\begin{abstract}
Gradually typed languages are designed to support both dynamically typed
and statically typed programming styles while preserving the benefits of
each.  While existing gradual type soundness theorems for these
languages aim to show that type-based reasoning is preserved when
moving from the fully static setting to a gradual one, these theorems do
not imply that correctness of type-based refactorings and optimizations
is preserved.  Establishing correctness of program transformations is
technically difficult, because it requires reasoning about program
equivalence, and is often neglected in the metatheory of gradual
languages.  

In this paper, we propose an \emph{axiomatic} account of
program equivalence in a gradual cast calculus, which we formalize in a
logic we call \emph{gradual type theory} (GTT). Based on Levy's
call-by-push-value, GTT gives an axiomatic account of both call-by-value
and call-by-name gradual languages.  Based on our axiomatic account we
prove many theorems that justify optimizations and refactorings in
gradually typed languages. For example, 
\emph{uniqueness principles} for gradual type connectives show that if the
$\beta\eta$ laws hold for a connective, then casts between that
connective must be equivalent to the so-called ``lazy'' cast semantics.
Contrapositively, this shows that ``eager'' cast semantics violates the
extensionality of function types.  As another example, we show that
gradual upcasts are pure functions and, dually, gradual downcasts are
strict functions.  We show the consistency and applicability of our
axiomatic theory by proving that a contract-based implementation using
the lazy cast semantics gives a logical relations model of our type
theory, where equivalence in GTT implies contextual equivalence of the
programs.  Since GTT also axiomatizes the dynamic gradual guarantee, our
model also establishes this central theorem of gradual typing.  The
model is parametrized by the implementation of the dynamic types, and so
gives a family of implementations that validate type-based optimization
and the gradual guarantee.

\end{abstract}

\ccsdesc[500]{Theory of computation~Axiomatic semantics}
\ccsdesc[500]{Software and its engineering~Functional languages}

\keywords{gradual typing, graduality, call-by-push-value}

\maketitle

\section{Introduction}

Gradually typed languages are designed to support a mix of dynamically
typed and statically typed programming styles and preserve the
benefits of each.
Dynamically typed code can be written without conforming to a
syntactic type discipline, so the programmer can always run their
program interactively with minimal work.
On the other hand, statically typed code provides mathematically
sound reasoning principles that justify type-based refactorings,
enable compiler optimizations, and underlie formal software verification.
%
The difficulty is accommodating both of these styles and their benefits simultaneously:
allowing the dynamic and static code to interact without forcing the
dynamic code to be statically checked or violating the correctness of
type-based reasoning.

The linchpin to the design of a gradually typed language is the
semantics of \emph{runtime type casts}.  These are runtime checks that ensure
that typed reasoning principles are valid by checking types of dynamically typed
code at the boundary between static and dynamic typing.
For instance, when a statically typed function $f : \texttt{Num} \to
\texttt{Num}$ is applied to a dynamically typed argument $x : \dyn$,
the language runtime must check if $x$ is a number, and otherwise
raise a dynamic type error.
A programmer familiar with dynamically typed programming might object
that this is overly strong: for instance if $f$ is just a constant
function $f = \lambda x:\texttt{Num}. 0$ then why bother checking if
$x$ is a number since the body of the program does not seem to depend
on it?
The reason the value is rejected is because the annotation $x :
\texttt{Num}$ should introduce an assumption that that the programmer,
compiler and automated tools can rely on for behavioral reasoning in the
body of the function.  
%
For instance, if the variable $x$ is guaranteed to only be
instantiated with numbers, then the programmer is free to replace $0$
with $x - x$ or vice-versa.
However, if $x$ can be instantiated with a closure, then $x - x$ will
raise a runtime type error while $0$ will succeed, violating the
programmers intuition about the correctness of refactorings.
We can formalize such relationships by \emph{observational equivalence} of
programs: the two closures $\lambda x:\texttt{Num}. 0$ and $\lambda
x:\texttt{Num}. x - x$ are indistinguishable to any other program in
the language.
This is precisely the difference between gradual typing and so-called
\emph{optional} typing: in an optionally typed language (Hack,
TypeScript, Flow), annotations are checked for consistency but are unreliable
to the user, so provide no leverage for reasoning.  
%
In a gradually typed language, type annotations should relieve the
programmer of the 
burden of reasoning about incorrect inputs, as long as we are willing to accept that the program
as a whole may crash, which is already a possibility in many \emph{effectful}
statically typed languages.
%


However, the dichotomy between gradual and optional typing is not as
firm as one might like.
There have been many different proposed semantics of run-time type
checking: ``transient'' cast semantics~\citep{vitousekswordssiek:2017}
only checks the head connective of a type (number, function, list,
\ldots), ``eager'' cast semantics~\citep{herman2010spaceefficient} checks
run-time type information on closures, whereas ``lazy'' cast
semantics~\citep{findler-felleisen02} will always delay a type-check on
a function until it is called (and there are other possibilities, see
e.g.  \cite{siek+09designspace,greenberg15spaceefficient}).
The extent to which these different semantics have been shown to
validate type-based reasoning has been limited to syntactic type
soundness and blame soundness theorems.
In their strongest form, these theorems say ``If $t$ is a closed
program of type $A$ then it diverges, or reduces to a runtime error
blaming dynamically typed code, or reduces to a value that satisfies $A$ to a
certain extent.''
However, the theorem at this level of generality is quite weak, and
justifies almost no program equivalences without more information.
Saying that a resulting value satisfies type $A$ might be a strong
statement, but in transient semantics constrains only the head
connective.
The blame soundness theorem might also be quite strong, but depends on
the definition of blame, which is part of the operational semantics of
the language being defined.
%
We argue that these type soundness theorems are only indirectly
expressing the actual desired properties of the gradual language,
which are \emph{program equivalences in the typed portion of the code} that are
not valid in the dynamically typed portion. 

Such program equivalences typically include $\beta$-like principles,
which arise from computation steps, as well as \emph{$\eta$ equalities},
which express the uniqueness or universality of certain constructions.
%
The $\eta$ law of the untyped $\lambda$-calculus, which
states that any $\lambda$-term $M \equiv \lambda x. M x$, is
restricted in a typed language to only hold for terms of function type $M
: A \to B$ ($\lambda$ is the unique/universal way of making an element
of the function type).  
This famously ``fails'' to hold in call-by-value languages in the
presence of effects: if $M$ is a program that prints \texttt{"hello"}
before returning a function, then $M$ will print \emph{now}, whereas
$\lambda x. M x$ will only print when given an argument. But this can be
accommodated with one further modification: the $\eta$ law is valid in
simple call-by-value languages\footnote{This does not hold in languages
  with some intensional feature of functions such as reference
  equality. We discuss the applicability of our main results more generally in Section \ref{sec:related}.} (e.g. SML) if we have a ``value
restriction'' $V \equiv \lambda x. V x$.
This illustrates that $\eta$/extensionality rules must be stated for
each type connective, and be sensitive to the effects/evaluation order
of the terms involved.
For instance, the $\eta$ principle for the boolean type $\texttt{Bool}$
\emph{in call-by-value} is that for any term $M$ with a free variable $x :
\texttt{Bool}$, $M$ is equivalent to a term that performs an if
statement on $x$: $M \equiv \kw{if} x (M[\texttt{true}/x])
(M[\texttt{false}/x])$.
If we have an \texttt{if} form that is strongly typed (i.e., errors on
non-booleans) then this tells us that it is \emph{safe} to run an if
statement on any input of boolean type (in CBN, by contrast an if
statement forces a thunk and so is not necessarily safe).
In addition, even if our \texttt{if} statement does some kind of
coercion, this tells us that the term $M$ only cares about whether $x$
is ``truthy'' or ``falsy'' and so a client is free to change e.g. one
truthy value to a different one without changing behavior.
This $\eta$ principle justifies a number of program optimizations,
such as dead-code and common subexpression elimination, and 
hoisting an if
statement outside of the body of a function if it is well-scoped
($\lambda x. \kw{if} y \, M \, N \equiv \kw {if} y \, (\lambda x.M) \, (\lambda x.N)$).
Any eager datatype, one whose elimination form is given by pattern
matching such as $0, +, 1, \times, \mathtt{list}$, has a similar $\eta$
principle which enables similar reasoning, such as proofs by induction.
The $\eta$ principles for lazy types \emph{in call-by-name} support dual
behavioral reasoning about lazy functions, records, and streams.

\textbf{An Axiomatic Approach to Gradual Typing.}
In this paper, we systematically study questions of program equivalence
for a class of gradually typed languages by working in an
\emph{axiomatic theory} of gradual program equivalence, a language and
logic we call \emph{gradual type theory} (GTT). 
Gradual type theory is the combination of a language of terms and
gradual types with a simple logic for proving program equivalence and
\emph{error approximation} (equivalence up to one program erroring when
the other does not) results.
The logic axiomatizes the equational properties gradual
programs should satisfy, and offers a high-level syntax for proving
theorems about many languages at once:
if a language models gradual type theory, then it satisfies all
provable equivalences/approximations.
Due to its type-theoretic design, different axioms of program
equivalence are easily added or removed.
Gradual type theory can be used both to explore language design questions and
to verify behavioral properties of specific programs, such as correctness of
optimizations and refactorings.

To get off the ground, we take two properties of the gradual language
for granted.
First, we assume a compositionality property: that any cast from $A$
to $B$ can be factored through the dynamic type $\dyn$, i.e., the cast
$\obcast{B}{A}{t}$ is equivalent to first casting up from $A$ to
$\dyn$ and then down to $B$: $\obcast{B}{\dyn}\obcast{\dyn}{A} t$.
These casts often have quite different performance characteristics,
but should have the same extensional behavior: of the cast semantics
presented in \citet{siek+09designspace}, only the partially eager
detection strategy violates this principle, and this strategy is not
common.
The second property we take for granted is that the language satisfies
the \emph{dynamic gradual guarantee}~\cite{refined} (``graduality'')---a
strong correctness theorem of gradual typing--- which constrains how
changing type annotations changes behavior.  Graduality says that if we
change the types in a program to be ``more precise''---e.g., by changing
from the dynamic type to a more precise type such as integers or
functions---the program will either produce the same behavior as the
original or raise a dynamic type error.  Conversely, if a program does
not error and some types are made ``less precise'' then behavior does
not change.

We then study what program equivalences are provable in GTT under
various assumptions.
Our central application is to study when the $\beta, \eta$ equalities
are satisfied in a gradually typed language.
%
%
We approach this problem by a surprising tack: rather than defining the
behavior of dynamic type casts and then verifying or invalidating the
$\beta$ and $\eta$ equalities, we \emph{assume} the language satisfies
$\beta$ and $\eta$ equality and then show that certain reductions of
casts are in fact program equivalence \emph{theorems} deducible from the
axioms of GTT.

The cast reductions that we show satisfy all three constraints are
those given by the ``lazy cast semantics''~\cite{findler-felleisen02,siek+09designspace}.
As a contrapositive, any gradually typed language for which these
reductions are not program equivalences is \emph{not} a model of the
axioms of gradual type theory.
This mean the language violates either compositionality, the gradual
guarantee, or one of the $\beta, \eta$ axioms---and in practice, it is
usually $\eta$.

For instance, a transient semantics, where only the top-level
connectives are checked, violates $\eta$ for strict pairs
\begin{small}
  \[ {x : A_1 \times A_2} \vdash (\letXbeYinZ x {(x_1,x_2)} 0) \neq 0 \]
\end{small}%
because the top-level connectives of $A_1$ and $A_2$ are only checked
when the pattern match is introduced. As a concrete counterexample to
contextual equivalence, let $A_1, A_2$ all be \texttt{String}.  Because
only the top-level connective is checked, $(0,1)$ is a valid value of
type $\texttt{String} \times \texttt{String}$, but pattern matching on
the pair ensures that the two components are checked to be strings, so
the left-hand side $\letXbeYinZ {(0,1)} {(x_1,x_2)} 0 \mapsto \err$
(raises a type error). On the right-hand side, with no pattern, match a
value (0) is returned. This means simple program changes that are valid
in a typed language, such as changing a function of two arguments to
take a single pair of those arguments, are invalidated by the transient
semantics.
In summary, transient semantics is ``lazier'' than the types dictate,
catching errors only when the term is inspected.

As a subtler example, in call-by-value ``eager cast semantics'' the
$\beta\eta$ principles for all of the eager datatypes ($0, +, 1,
\times$, lists, etc.) will be satisfied, but the $\eta$ principle for
the function type $\to$ is violated: there are values $V : A \to A'$ for
which $V \neq \lambda x:A. V x $.
For instance, take an arbitrary function value $V : A \to
\texttt{String}$ for some type $A$, and let $V' = \obcast{A \to
  \dyn}{A \to \texttt{String}}{V}$ be the result of casting it to have a
dynamically typed output.
Then in eager semantics, the following programs are not equivalent:
\begin{small}
\[   \lambda x:A. V' x \neq V' : A \to \dyn\]
\end{small}
We cannot observe any difference between these two programs by
applying them to arguments, however, they are distinguished from each
other by their behavior when \emph{cast}.
Specifically, if we cast both sides to $A \to \texttt{Number}$, then
$\obcast{A\to \texttt{Number}}{A\to\dyn}(\lambda x:A.V' x)$ is a
value, but $\obcast{A \to \texttt{Number}}{A\to \dyn}V'$ reduces to an
error because $\texttt{Number}$ is incompatible with
$\texttt{String}$.
However this type error might not correspond to any actual typing
violation of the program involved.
For one thing, the resulting function might never be executed.
Furthermore, in the presence of effects, it may be that the original
function $V : A \to \texttt{String}$ never returns a string (because
it diverges, raises an exception or invokes a continuation), and so
that same value casted to $A \to \texttt{Number}$ might be a perfectly
valid inhabitant of that type.
In summary the ``eager'' cast semantics is in fact overly eager: in
its effort to find bugs faster than ``lazy'' semantics it disables the
very type-based reasoning that gradual typing should provide.

While criticisms of transient semantics on the basis of type soundness
have been made before \citep{greenmanfelleisen:2018}, our development
shows that the $\eta$ principles of types are enough to uniquely
determine a cast semantics, and helps clarify the trade-off between
eager and lazy semantics of function casts.

\textbf{Technical Overview of GTT.}  The gradual type theory developed
in this paper unifies our previous work on
operational (logical relations) reasoning for gradual typing in a
call-by-value setting~\citep{newahmed18} (which did not consider a proof theory), and on an
axiomatic proof theory for gradual typing~\citep{newlicata2018-fscd} in
a call-by-name setting (which considered only function and product
types, and denotational but not operational models).

In this paper, we develop an axiomatic gradual type theory GTT for a unified
language that includes \emph{both} call-by-value/eager types and
call-by-name/lazy types (Sections~\ref{sec:gtt}, \ref{sec:theorems-in-gtt}), and
show that it is sound for contextual equivalence via a logical relations model
(Sections~\ref{sec:contract}, \ref{sec:complex}, \ref{sec:operational}).
Because the $\eta$ principles for types play a key role in our approach, it is
necessary to work in a setting where we can have $\eta$ principles for both
eager and lazy types.  We use Levy's
Call-by-Push-Value~\citep{levy03cbpvbook} (CBPV), which fully and faithfully
embeds both call-by-value and call-by-name evaluation with both eager and lazy
datatypes,\footnote{The distinction between ``lazy'' vs ``eager'' casts above is
  different than lazy vs. eager datatypes.}  and underlies much recent work on
reasoning about effectful programs~\cite{bauerpretnar13eff,lindley+17frank}.
GTT can prove results in and about existing call-by-value gradually typed
languages, and also suggests a design for call-by-name and full
call-by-push-value gradually typed languages.

In the prior work \cite{newlicata2018-fscd,newahmed18}, gradual type
casts are decomposed into upcasts and downcasts, as suggested above.
A \emph{type dynamism}
relation (corresponding to type precision~\cite{refined} and na\"ive
subtyping~\cite{wadler-findler09}) controls which casts exist: a type
dynamism $A \ltdyn A'$ induces an upcast from $A$ to $A'$ and a downcast
from $A'$ to $A$.  Then, a \emph{term dynamism} judgement is used for
equational/approximational reasoning about programs.  Term dynamism
relates two terms whose types are related by type dynamism, and the
upcasts and downcasts are each \emph{specified} by certain term
dynamism judgements holding.
This specification axiomatizes only the properties of casts needed to
ensure the graduality theorem, and not their precise behavior, so cast
reductions can be \emph{proved from it}, rather than stipulated in
advance.  The specification defines the casts ``uniquely up to
equivalence'', which means that any two implementations satisfying it
are behaviorally equivalent.

We generalize this axiomatic approach to call-by-push-value
(Section~\ref{sec:gtt}), where there are both eager/value types and
lazy/computation types. This is both a subtler question than it might at
first seem, and has a surprisingly nice answer: we find that upcasts are
naturally associated with eager/value types and downcasts with
lazy/computation types, and that the modalities relating values and
computations induce the downcasts for eager/value types and upcasts for
lazy/computation types.  Moreover, this analysis articulates an
important behavioral property of casts that was proved operationally for
call-by-value in \citep{newahmed18} but missed for call-by-name in
\citep{newlicata2018-fscd}: upcasts for eager types and downcasts for
lazy types are both ``pure'' in a suitable sense, which enables more
refactorings and program optimizations.  In particular, we show that
these casts can be taken to be (and are essentially forced to be)
``complex values'' and ``complex stacks'' (respectively) in
call-by-push-value, which corresponds to a behavioral property of
\emph{thunkability} and
\emph{linearity}~\cite{munchmaccagnoni14nonassociative}.  We argue in
Section~\ref{sec:related} that this property is related to blame
soundness.  Our gradual type theory naturally has two dynamic types, a
dynamic eager/value type and a dynamic lazy/computation type, where the
former can be thought of as a sum of all possible values, and the latter
as a product of all possible behaviors.  At the language design level,
gradual type theory can be used to prove that, for a variety of
eager/value and lazy/computation types, the ``lazy'' semantics of casts
is the unique implementation satisfying $\beta,\eta$ and graduality
(Section~\ref{sec:theorems-in-gtt}). These behavioral equivalences can
then be used in reasoning about optimizations, refactorings, and
correctness of specific programs.

\textbf{Contract-Based Models.}
To show the consistency of GTT as a theory, and to give a concrete
operational interpretation of its axioms and rules, we provide a
concrete model based on an operational semantics.
The model is a \emph{contract} interpretation of GTT in that the
``built-in'' casts of GTT are translated to ordinary functions in a
CBPV language that perform the necessary checks.  
%

To keep the proofs high-level, we break the proof into two steps.
First (Sections~\ref{sec:contract}, \ref{sec:complex}), we translate the
axiomatic theory of GTT into an axiomatic theory of CBPV extended with
recursive types and an uncatchable error, implementing casts by CBPV
code that does contract checking.
Then
(Section~\ref{sec:operational}) we give an operational semantics
for the extended CBPV and define a step-indexed biorthogonal logical
relation that interprets the ordering relation on terms as contextual
error approximation, which underlies the definition of graduality as
presented in \citep{newahmed18}.
Combining these theorems gives an implementation of the term
language of GTT in which $\beta, \eta$ are observational equivalences
and the dynamic gradual guarantee is satisfied.

Due to the uniqueness theorems of GTT, the only part of this translation that is not
predetermined is the definition of the dynamic types themselves and the
casts between ``ground'' types and the dynamic types.
We use CBPV to explore the design space of possible implementations of
the dynamic types, and give one that faithfully distinguishes all types
of GTT, and another more Scheme-like implementation that implements sums
and lazy pairs by tag bits.
Both can be restricted to the CBV or CBN subsets of CBPV, but the
unrestricted variant is actually more faithful to Scheme-like
dynamically typed programming, because it accounts for variable-argument
functions.
Our modular proof architecture allows us to easily prove correctness
of $\beta, \eta$ and graduality for all of these interpretations.

\textbf{Contributions.}
The main contributions of the paper are as follows.
\begin{enumerate}
\item We present Gradual Type Theory in Section \ref{sec:gtt}, a simple
  axiomatic theory of gradual typing. The theory axiomatizes three
  simple assumptions about a gradual language: compositionality,
  graduality, and type-based reasoning in the form of $\eta$
  equivalences.
\item We prove many theorems in the formal logic of Gradual Type
  Theory in Section \ref{sec:theorems-in-gtt}. These include the
  unique implementation theorems for casts, which show that for
  each type connective of GTT, the $\eta$ principle for the type
  ensures that the casts must implement the lazy contract
  semantics. Furthermore, we show that upcasts are always pure
  functions and dually that downcasts are always strict functions, as
  long as the base type casts are pure/strict.
\item To substantiate that GTT is a reasonable axiomatic theory for
  gradual typing, we construct \emph{models} of GTT in Sections
  \ref{sec:contract}, \ref{sec:complex} and \ref{sec:lr}.  This proceeds
  in two stages. First (Section \ref{sec:contract}), we use
  call-by-push-value as a typed metalanguage to construct several models
  of GTT using different recursive types to implement the dynamic types
  of GTT and interpret the casts as embedding-projection pairs. This
  extends standard translations of dynamic typing into static typing
  using type tags: the dynamic value type is constructed as a recursive
  sum of basic value types, but dually the dynamic computation type is
  constructed as a recursive \emph{product} of basic computation
  types. This dynamic computation type naturally models stack-based
  implementations of variable-arity functions as used in the Scheme
  language.
\item We then give an operational model of the term dynamism ordering
  as contextual error approximation in Sections \ref{sec:complex} and
  \ref{sec:lr}. To construct this model, we extend previous work on
  logical relations for error approximation from call-by-value to
  call-by-push-value \cite{newahmed18}, simplifying the presentation
  in the process.
\end{enumerate}

\begin{shortonly}
\textbf{Extended version:} An extended version of the paper, which
includes the omitted cases of definitions, lemmas, and proofs is
available in \citet{newlicataahmed19:extended}.
\end{shortonly}

\section{Axiomatic Gradual Type Theory}
\label{sec:gtt}

In this section we introduce the syntax of Gradual Type Theory, an
extension of Call-by-push-value~\citep{levy03cbpvbook} to support the constructions of
gradual typing.
First we introduce call-by-push-value and then describe in turn the
gradual typing features: dynamic types, casts, and the dynamism
orderings on types and terms.

\begin{figure}
  \begin{small}
  \[
  \begin{array}{l}
  \begin{array}{rl|rl}
    A ::= & \colorbox{lightgray}{$\dynv$} \mid U \u B \mid 0 \mid A_1 + A_2 \mid 1 \mid A_1 \times A_2 & 
    \u B ::= & \colorbox{lightgray}{$\dync$} \mid \u F A \mid \top \mid \u B_1 \with \u B_2 \mid A \to \u B\\

    V ::= & \begin{array}{l}
            \colorbox{lightgray}{$\upcast A {A'} V$} \mid x \mid \abort{V} \\
            \mid \inl{V} \mid \inr{V} \\
            \mid \caseofXthenYelseZ V {x_1. V_1}{x_2.V_2} \\
            \mid () \mid \pmpairWtoinZ V V' \\
            \mid (V_1,V_2) \mid \pmpairWtoXYinZ V x y V' \\
            \mid \thunk{M}
            \end{array} & 

    M,S ::= & \begin{array}{l}
            \colorbox{lightgray}{$\dncast{\u B} {\u B'} M$} \mid \bullet \mid \err_{\u B} \\
            \mid \abort{V} \mid \caseofXthenYelseZ V {x_1. M_1}{x_2.M_2}\\
            \mid \pmpairWtoinZ V M \mid\pmpairWtoXYinZ V x y M \\
            \mid  \force{V} \mid \ret{V} \mid \bindXtoYinZ{M}{x}{N}\\
            \mid \lambda x:A.M \mid M\,V\\
            \mid \emptypair \mid \pair{M_1}{M_2} \\
            \mid \pi M \mid \pi' M
            \end{array}\\

    \Gamma ::= & \cdot \mid \Gamma, x : A & 
    \Delta ::= & \cdot \mid \bullet : \u B \\
    \colorbox{lightgray}{$\Phi$} ::= & \colorbox{lightgray}{$\cdot \mid \Phi, x \ltdyn x': A \ltdyn A'$} &
    \colorbox{lightgray}{$\Psi$} ::= & \colorbox{lightgray}{$\cdot \mid \bullet \ltdyn \bullet : \u B \ltdyn \u B'$} \\  
    \end{array}\\\\
\iflong
    \begin{array}{c}
    \hspace{2.5in} T ::= A \mid \u B \\
    \hspace{2.5in} E ::= V \mid M  \\
  \end{array}\\\\
\fi
  \begin{array}{c}
    \framebox{$\Gamma \vdash V : A$ and $\Gamma \mid \Delta \vdash M : \u B$} \qquad
    \colorbox{lightgray}{
    $\inferrule*[lab=UpCast]
    {\Gamma \vdash V : A \and A \ltdyn A'}
    {\Gamma \vdash \upcast A {A'} V : A'}$
    \qquad
    $\inferrule*[lab=DnCast]
    {\Gamma\pipe \Delta \vdash M : \u B' \and \u B \ltdyn \u B'}
    {\Gamma\pipe \Delta \vdash \dncast{\u B}{\u B'} M : \u B}$
    }
    \\\\
    \inferrule*[lab=Var]
    { }
    {\Gamma,x : A,\Gamma' \vdash x : A}
    \qquad
    \inferrule*[lab=Hole]
    { }
    {\Gamma\pipe \bullet : \u B \vdash \bullet : \u B}
    \qquad
    \inferrule*[lab=Err]
    { }
    {\Gamma \mid \cdot \vdash \err_{\u B} : \u B}
    \\
\iflong
    \\
    \inferrule*[lab=$0$E]
    {\Gamma \vdash V : 0}
    {\Gamma \mid \Delta \vdash \abort V : T}
    \qquad
    \inferrule*[lab=$+$Il]
    {\Gamma \vdash V : A_1}
    {\Gamma \vdash \inl V : A_1 + A_2}
    \qquad
    \inferrule*[lab=$+$Ir]
    {\Gamma \vdash V : A_2}
    {\Gamma \vdash \inr V  : A_1 + A_2}
    \qquad
    \inferrule*[lab=$+$E]
        {
          \Gamma \vdash V : A_1 + A_2 \\\\
          \Gamma, x_1 : A_1 \mid \Delta \vdash E_1 : T \\\\
          \Gamma, x_2 : A_2 \mid \Delta \vdash E_2 : T
        }
    {\Gamma \mid \Delta \vdash \caseofXthenYelseZ V {x_1. E_1}{x_2.E_2} : T}
    \\\\
    \fi
    \inferrule*[lab=$1$I]
    { }
    {\Gamma \vdash (): 1}
    \,\,\,
    \inferrule*[lab=$1$E]
    {\Gamma \vdash V : 1 \and
      \Gamma \mid \Delta \vdash E : T
    }
    {\Gamma \mid \Delta \vdash \pmpairWtoinZ V E : T}
    \,\,\,
    \inferrule*[lab=$\times$I]
    {\Gamma \vdash V_1 : A_1\and
      \Gamma\vdash V_2 : A_2}
    {\Gamma \vdash (V_1,V_2) : A_1 \times A_2}
    \,\,\,
    \inferrule*[lab=$\times$E]
    {\Gamma \vdash V : A_1 \times A_2 \\\\
      \Gamma, x : A_1,y : A_2 \mid \Delta \vdash E : T
    }
    {\Gamma \mid \Delta \vdash \pmpairWtoXYinZ V x y E : T}
    \\\\
    \inferrule*[lab=$U$I]
    {\Gamma \mid \cdot \vdash M : \u B}
    {\Gamma \vdash \thunk M : U \u B}
    \,\,\,
    \inferrule*[lab=$U$E]
    {\Gamma \vdash V : U \u B}
    {\Gamma \pipe \cdot \vdash \force V : \u B}
    \,\,\,
    \inferrule*[lab=$F$I]
    {\Gamma \vdash V : A}
    {\Gamma \pipe \cdot \vdash \ret V : \u F A}
    \,\,\,
    \inferrule*[lab=$F$E]
    {\Gamma \pipe \Delta \vdash M : \u F A \\
      \Gamma, x: A \pipe \cdot \vdash N : \u B}
    {\Gamma \pipe \Delta \vdash \bindXtoYinZ M x N : \u B}
    \\\\
    \inferrule*[lab=$\to$I]
    {\Gamma, x: A \pipe \Delta \vdash M : \u B}
    {\Gamma \pipe \Delta \vdash \lambda x : A . M : A \to \u B}
    \quad
    \inferrule*[lab=$\to$E]
    {\Gamma \pipe \Delta \vdash M : A \to \u B\and
      \Gamma \vdash V : A}
    {\Gamma \pipe \Delta \vdash M\,V : \u B }
\iflong
      \\\\
    \inferrule*[lab=$\top$I]{ }{\Gamma \mid \Delta \vdash \emptypair : \top}
    \quad
    \inferrule*[lab=$\with$I]
    {\Gamma \mid \Delta \vdash M_1 : \u B_1\and
      \Gamma \mid \Delta \vdash M_2 : \u B_2}
    {\Gamma \mid \Delta \vdash \pair {M_1} {M_2} : \u B_1 \with \u B_2}
    \quad
    \inferrule*[lab=$\with$E]
    {\Gamma \mid \Delta \vdash M : \u B_1 \with \u B_2}
    {\Gamma \mid \Delta \vdash \pi M : \u B_1}
    \quad
    \inferrule*[lab=$\with$E']
    {\Gamma \mid \Delta \vdash M : \u B_1 \with \u B_2}
    {\Gamma \mid \Delta \vdash \pi' M : \u B_2}
\fi
  \end{array}
  \end{array}
  \]
\end{small}
  \vspace{-0.1in}
  \caption{GTT Syntax and Term Typing }
  \label{fig:gtt-syntax-and-terms}
\end{figure}

\subsection{Background: Call-by-Push-Value}

GTT is an extension of CBPV, so we first present CBPV as the unshaded rules in
Figure~\ref{fig:gtt-syntax-and-terms}.  CBPV makes a distinction between
\emph{value types} $A$ and \emph{computation types} $\u B$, where value
types classify \emph{values} $\Gamma \vdash V : A$ and computation types
classify \emph{computations} $\Gamma \vdash M : \u B$.  Effects are
computations: for example, we might have an error computation $\err_{\u
  B} : \u B$ of every computation type, or printing $\print V;M : \u B$
if $V : \kw{string}$ and $M : \u B$, which prints $V$ and then behaves as
$M$.

\emph{Value types and complex values.}
The value types include \emph{eager} products $1$ and $A_1 \times A_2$
and sums $0$ and $A_1 + A_2$, which behave as in a call-by-value/eager
language (e.g. a pair is only a value when its components are).  The
notion of value $V$ is more permissive than one might expect, and
expressions $\Gamma \vdash V : A$ are sometimes called \emph{complex
  values} to emphasize this point: complex values include not only
closed runtime values, but also open values that have free value
variables (e.g. $x : A_1 , x_2 : A_2 \vdash (x_1,x_2) : A_1 \times
A_2$), and expressions that pattern-match on values (e.g. $p : A_1
\times A_2 \vdash \pmpairWtoXYinZ{p}{x_1}{x_2}{(x_2,x_1)} : A_2 \times
A_1$).  Thus, the complex values $x : A \vdash V : A'$ are a syntactic
class of ``pure functions'' from $A$ to $A'$ (though there is no pure
function \emph{type} internalizing this judgement), which can be treated
like values by a compiler because they have no effects (e.g. they can be
dead-code-eliminated, common-subexpression-eliminated, and so on).
\begin{longonly}
In focusing~\cite{andreoli92focus} terminology, complex
values consist of left inversion and right focus rules.
\end{longonly}
For each pattern-matching construct (e.g. case analysis on a sum,
splitting a pair), we have both an elimination rule whose branches are
values (e.g. $\pmpairWtoXYinZ{p}{x_1}{x_2}{V}$) and one whose branches
are computations (e.g. $\pmpairWtoXYinZ{p}{x_1}{x_2}{M}$).  To
abbreviate the typing rules for both in
Figure~\ref{fig:gtt-syntax-and-terms}, we use the following convention:
we write $E ::= V \mid M$ for either a complex value or a computation,
and $T ::= A \mid \u B$ for either a value type $A$ or a computation
type $\u B$, and a judgement $\Gamma \mid \Delta \vdash E : T$ for
either $\Gamma \vdash V : A$ or $\Gamma \mid \Delta \vdash M : \u B$
(this is a bit of an abuse of notation because $\Delta$ is not present
in the former).  Complex values can be translated away without loss of
expressiveness by moving all pattern-matching into computations (see
Section~\ref{sec:complex}), at the expense of using a behavioral
condition of \emph{thunkability}~\cite{munchmaccagnoni14nonassociative} to capture the properties
complex values have (for example, an analogue of
$\letXbeYinZ{V}{x}{\letXbeYinZ{V'}{x'}{M}} \equiv
\letXbeYinZ{V'}{x'}{\letXbeYinZ{V}{x}{M}}$ --- complex values can be
reordered, while arbitrary computations cannot).  

\emph{Shifts.}
A key notion in CBPV is the \emph{shift} types $\u F A$ and $U \u B$,
which mediate between value and computation types: $\u F A$ is the
computation type of potentially effectful programs that return a value
of type $A$, while $U \u B$ is the value type of thunked computations of
type $\u B$.  The introduction rule for $\u F A$ is returning a value of
type $A$ (\ret{V}), while the elimination rule is sequencing a
computation $M : \u F A$ with a computation $x : A \vdash N : \u B$ to
produce a computation of a $\u B$ ($\bindXtoYinZ{M}{x}{N}$).  While any
closed complex value $V$ is equivalent to an actual value, a computation
of type $\u F A$ might perform effects (e.g. printing) before returning
a value, or might error or non-terminate and not return a value at all.
The introduction and elimination rules for $U$ are written $\thunk{M}$
and $\force{V}$, and say that computations of type $\u B$ are bijective
with values of type $U \u B$.  As an example of the action of the
shifts,
\begin{longonly}
  $0$ is the empty value type, so $\u F 0$ classifies effectful
computations that never return, but may perform effects (and then, must
e.g. non-terminate or error), while $U \u F 0$ is the value type where
such computations are thunked/delayed and considered as values.
\end{longonly}
$1$ is the trivial value type, so $\u F 1$ is the type of computations
that can perform effects with the possibility of terminating
  successfully by returning $()$, and $U \u F 1$ is the value type where
  such computations are delayed values.
\begin{longonly}  
  $U \u F$ is a monad on value
  types~\citep{moggi91}, while $\u F U$ is a comonad on computation types.
\end{longonly}

\emph{Computation types.}
The computation type constructors in CBPV include lazy unit/products
$\top$ and $\u B_1 \with \u B_2$, which behave as in a call-by-name/lazy
language (e.g. a component of a lazy pair is evaluated only when it is
projected).  Functions $A \to \u B$ have a value type as input and a
computation type as a result.  The equational theory of effects in CBPV
computations may be surprising to those familiar only with
call-by-value, because at higher computation types effects have a
call-by-name-like equational theory.  For example, at computation type
$A \to \u B$, we have an equality $\print c; \lambda x. M = \lambda
x.\print c; M$.  Intuitively, the reason is that $A \to \u B$ is not
treated as an \emph{observable} type (one where computations are run):
the states of the operational semantics are only those computations of
type $\u F A$ for some value type $A$.  Thus, ``running'' a function
computation means supplying it with an argument, and applying both of
the above to an argument $V$ is defined to result in $\print c;M[V/x]$.
This does \emph{not} imply that the corresponding equations holds for
the call-by-value function type, which we discuss below.
\begin{longonly}
As another
example, \emph{all} computations are considered equal at type $\top$,
even computations that perform different effects ($\print c$ vs. $\{\}$
vs. $\err$), because there is by definition \emph{no} way to extract an
observable of type $\u F A$ from a computation of type $\top$.
Consequently, $U \top$ is isomorphic to $1$.
\end{longonly}

\emph{Complex stacks.} Just as the complex values $V$ are a syntactic
class terms that have no effects, CBPV includes a judgement for
``stacks'' $S$, a syntactic class of terms that reflect \emph{all}
effects of their input.  A \emph{stack} $\Gamma \mid \bullet : \u B
\vdash S : \u B'$ can be thought of as a linear/strict function from $\u
B$ to $\u B'$, which \emph{must} use its input hole $\bullet$
\emph{exactly} once at the head redex position.  Consequently, effects
can be hoisted out of stacks, because we know the stack will run them
exactly once and first.  For example, there will be contextual
equivalences $S[\err/\bullet] = \err$ and $S[\print V;M] = \print
V;S[M/\bullet]$.  Just as complex values include pattern-matching,
\emph{complex stacks} include pattern-matching on values and
introduction forms for the stack's output type.  For example, $\bullet :
\u B_1 \with \u B_2 \vdash \pair{\pi' \bullet}{\pi \bullet} : \u B_2
\with \u B_1$ is a complex stack, even though it mentions $\bullet$ more
than once, because running it requires choosing a projection to get to
an observable of type $\u F A$, so \emph{each time it is run} it uses
$\bullet$ exactly once.
\begin{longonly}
In
focusing terms, complex stacks include both left and right inversion,
and left focus rules.
\end{longonly}
In the equational theory of CBPV, $\u F$ and $U$
are \emph{adjoint}, in the sense that stacks $\bullet : \u F A \vdash S
: \u B$ are bijective with values $x : A \vdash V : U \u B$, as both are
bijective with computations $x : A \vdash M : \u B$.

To compress the presentation in Figure~\ref{fig:gtt-syntax-and-terms},
we use a typing judgement $\Gamma \mid \Delta \vdash M : \u B$ with a
``stoup'', a typing context $\Delta$ that is either
empty or contains exactly one assumption $\bullet : \u B$, so $\Gamma
\mid \cdot \vdash M : \u B$ is a computation, while $\Gamma \mid \bullet
: \u B \vdash M : \u B'$ is a stack.  The typing
rules for $\top$ and $\with$ treat the stoup additively
(it is arbitrary in the conclusion and the same in all premises); for a
function application to be a stack, the stack input must occur in the
function position.  The elimination form for $U \u B$, $\force{V}$, is
the prototypical non-stack computation ($\Delta$ is required to be
empty), because forcing a thunk does not use the stack's input.

\emph{Embedding call-by-value and call-by-name.}  To translate
call-by-value (CBV) into CBPV, a judgement $x_1 : A_1, \ldots, x_n : A_n
\vdash e : A$ is interpreted as a computation $x_1 : A_1^v, \ldots, x_n
: A_n^v \vdash e^v : \u F A^v$, where call-by-value products and sums
are interpreted as $\times$ and $+$, and the call-by-value function type
$A \to A'$ as $U(A^v \to \u F A'^v)$.  Thus, a call-by-value term $e : A
\to A'$, which should mean an effectful computation of a function value,
is translated to a computation $e^v : \u F U (A^v \to \u F A'^v)$. Here,
the comonad $\u F U$ offers an opportunity to perform effects
\emph{before} returning a function value---so under translation the CBV
terms $\print c; \lambda x. e$ and $\lambda x.\print c; e$ will not be
contextually equivalent.  To translate call-by-name (CBN) to CBPV, a
judgement $x_1 : \u B_1, \ldots, x_m : \u B_m \vdash e : \u B$ is
translated to $x_1 : U \u {B_1}^n, \ldots, x_m : U \u {B_m}^n \vdash e^n
: \u B^n$, representing the fact that call-by-name terms are passed
thunked arguments.  Product types are translated to $\top$ and $\times$,
while a CBN function $B \to B'$ is translated to $U \u B^n \to \u B'^n$
with a thunked argument.  Sums $B_1 + B_2$ are translated to $\u F (U \u
{B_1}^n + U \u {B_2}^n)$, making the ``lifting'' in lazy sums explicit.
Call-by-push-value \emph{subsumes} call-by-value and call-by-name in
that these embeddings are \emph{full and faithful}: two CBV or CBN programs are
equivalent if and only if their embeddings into CBPV are equivalent, and
every CBPV program with a CBV or CBN type can be back-translated.






\emph{Extensionality/$\eta$ Principles.}  The main advantage of CBPV for
our purposes is that it accounts for the $\eta$/extensionality
principles of both eager/value and lazy/computation types, because
value types have $\eta$ principles relating them to the value
assumptions in the context $\Gamma$, while computation types have $\eta$
principles relating them to the result type of a computation $\u B$.  For
example, the $\eta$ principle for sums says that any complex
value or computation $x : A_1 + A_2 \vdash E : T$ is equivalent to
$\caseofXthenYelseZ{x}{x_1.E[\inl{x_1}/x]}{x_2.E[\inr{x_2}/x]}$, i.e. a
case on a value can be moved to any point in a program (where all
variables are in scope) in an optimization.  Given this, the above
translations of CBV and CBN into CBPV explain why $\eta$ for
sums holds in CBV but not CBN: in CBV, $x : A_1 + A_2 \vdash E : T$ is
translated to a term with $x : A_1 + A_2$ free, but in CBN, $x : B_1 +
B_2 \vdash E : T$ is translated to a term with $x : U \u F(U \u B_1 + U
\u B_2)$ free, and the type $U \u F(U \u B_1 + U \u B_2)$ of monadic
computations that return a sum does not satisfy the $\eta$ principle for
sums in CBPV.  Dually, the $\eta$ principle for functions in CBPV is
that any computation $M : A \to \u B$ is equal to $\lambda x.M \, x$.  A
CBN term $e : B \to B'$ is translated to a CBPV computation of type $U
\u B \to \u B'$, to which CBPV function extensionality applies, while a
CBV term $e : A \to A'$ is translated to a computation of type $\u F U(A
\to \u F A')$, which does not satisfy the $\eta$ rule for functions.  We
discuss a formal statement of these $\eta$ principles with term
dynamism below.


\subsection{The Dynamic Type(s)}

Next, we discuss the additions that make CBPV into our gradual type
theory GTT.  A dynamic type plays a key role in gradual typing, and
since GTT has two different kinds of types, we have a new question of
whether the dynamic type should be a value type, or a computation type,
or whether we should have \emph{both} a dynamic value type and a dynamic
computation type.
Our modular, type-theoretic presentation of gradual typing allows us to
easily explore these options, though we find that having
both a dynamic value $\dynv$ and a dynamic computation type $\dync$
gives the most natural implementation (see
Section~\ref{sec:dynamic-type-interp}).  Thus, we add both $\dynv$ and
$\dync$ to the grammar of types in
Figure~\ref{fig:gtt-syntax-and-terms}.  We do \emph{not} give
introduction and elimination rules for the dynamic types, because we
would like constructions in GTT to imply results for many different
possible implementations of them.  Instead, the terms for the dynamic
types will arise from type dynamism and casts.

\subsection{Type Dynamism}

The \emph{type dynamism} relation of gradual type theory is written $A
\ltdyn A'$ and read as ``$A$ is less dynamic than $A'$''; intuitively,
this means that $A'$ supports more behaviors than $A$.
Our previous work~\citep{newahmed18,newlicata2018-fscd} analyzes this as the existence of an \emph{upcast}
from $A$ to $A'$ and a downcast from $A'$ to $A$ which form an
embedding-projection pair (\emph{ep pair}) for term error approximation
(an ordering where runtime errors are minimal): the upcast followed by the
downcast is a no-op, while the downcast followed by the upcast might
error more than the original term, because it imposes a run-time type
check.  Syntactically, type dynamism is defined (1) to be reflexive and
transitive (a preorder), (2) where every type constructor is monotone in
all positions, and (3) where the dynamic type is greatest in the type
dynamism ordering.  This last condition, \emph{the
  dynamic type is the most dynamic type}, implies the existence of an
upcast $\upcast{A}{\dynv}$ and a downcast $\dncast{A}{\dynv}$ for every
type $A$: any type can be embedded
into it and projected from it.  However, this by design does not
characterize $\dynv$ uniquely---instead, it is open-ended exactly
which types exist (so that we can always add more), and some properties
of the casts are undetermined; we exploit this freedom in
Section~\ref{sec:dynamic-type-interp}.

This extends in a straightforward way to CBPV's distinction between
value and computation types in Figure~\ref{fig:gtt-type-dynamism}: there
is a type dynamism relation for value types $A \ltdyn A'$ and for
computation types $\u B \ltdyn \u B'$, which (1) each are preorders
(\textsc{VTyRefl}, \textsc{VTyTrans}, \textsc{CTyRefl}, \textsc{CTyTrans}),
(2) every type constructor is monotone
(\textsc{$+$Mon}, \textsc{$\times$Mon}, \textsc{$\with$Mon} ,\textsc{$\to$Mon})
where the shifts $\u F$ and $U$ switch which relation is being
considered (\textsc{$U$Mon}, \textsc{$F$Mon}), and (3) the dynamic types
$\dynv$ and $\dync$ are the most dynamic value and computation types
respectively (\textsc{VTyTop}, \textsc{CTyTop}).  For example, we have
$U(A \to \u F A') \ltdyn U(\dynv \to \u F \dynv)$, which is the analogue
of $A \to A' \ltdyn \dynv \to \dynv$ in call-by-value: because $\to$
preserves embedding-retraction pairs, it is monotone, not contravariant,
in the domain~\citep{newahmed18,newlicata2018-fscd}.

\begin{figure}
\begin{small}
    
  \begin{mathpar}
    \framebox{$A \ltdyn A'$ and $\u B \ltdyn \u B'$}
    
    \inferrule*[lab=VTyRefl]{ }{A \ltdyn A}

    \inferrule*[lab=VTyTrans]{A \ltdyn A' \and A' \ltdyn A''}
              {A \ltdyn A''}

    \inferrule*[lab=CTyRefl]{ }{\u B \ltdyn \u B'}

    \inferrule*[lab=CTyTrans]{\u B \ltdyn \u B' \and \u B' \ltdyn \u B''}
              {\u B \ltdyn \u B''}

    \inferrule*[lab=VTyTop]{ }{A \ltdyn \dynv}

    \inferrule*[lab=$U$Mon]{\u B \ltdyn \u B'}
              {U B \ltdyn U B'}

    \inferrule*[lab=$+$Mon]{A_1 \ltdyn A_1' \and A_2 \ltdyn A_2' }
              {A_1 + A_2 \ltdyn A_1' + A_2'}

    \inferrule*[lab=$\times$Mon]{A_1 \ltdyn A_1' \and A_2 \ltdyn A_2' }
              {A_1 \times A_2 \ltdyn A_1' \times A_2'}
\\
    \inferrule*[lab=CTyTop]{ }{\u B \ltdyn \dync}

\inferrule*[lab=$F$Mon]{A \ltdyn A' }{ \u F A \ltdyn \u F A'}

\inferrule*[lab=$\with$Mon]{\u B_1 \ltdyn \u B_1' \and \u B_2 \ltdyn \u B_2'}
              {\u B_1 \with \u B_2 \ltdyn \u B_1' \with \u B_2'}

\inferrule*[lab=$\to$Mon]{A \ltdyn A' \and \u B \ltdyn \u B'}
          {A \to \u B \ltdyn A' \to \u B'}
\begin{longonly}
\\
\framebox{Dynamism contexts} 
\quad
\inferrule{ }{\cdot \, \dynvctx}
\quad
\inferrule{\Phi \, \dynvctx \and
            A \ltdyn A'}
          {\Phi, x \ltdyn x' : A \ltdyn A' \, \dynvctx}
\quad
\inferrule{ }{\cdot \, \dyncctx}
\quad         
\inferrule{\u B \ltdyn \u B'}
          {(\bullet \ltdyn \bullet : \u B \ltdyn \u B') \, \dyncctx}
\end{longonly}
  \end{mathpar}
  \vspace{-0.2in}
\caption{GTT Type Dynamism \iflong and Dynamism Contexts \fi}
\label{fig:gtt-type-dynamism}
\end{small}
\end{figure}

\subsection{Casts}
\label{sec:gtt-casts}

It is not immediately obvious how to add type casts to CPBV, because
CBPV exposes finer judgemental distinctions than previous work
considered.  However, we can arrive at a first proposal by considering
how previous work would be embedded into CBPV.
In the previous work on both CBV and
CBN~\citep{newahmed18,newlicata2018-fscd} every type dynamism judgement
$A \ltdyn A'$ induces both an upcast from $A$ to $A'$ and a downcast
from $A'$ to $A$.
Because CBV types are associated to CBPV value types and CBN types are
associated to CBPV computation types, this suggests that each value type
dynamism $A \ltdyn A'$ should induce an upcast and a downcast, and each
computation type dynamism $\u B \ltdyn \u B'$ should also induce an
upcast and a downcast.
In CBV, a cast from $A$ to $A'$ typically can be represented by a CBV
function $A \to A'$, whose analogue in CBPV is $U(A \to \u F A')$, and
values of this type are bijective with computations $x : A \vdash M : \u
F A'$, and further with stacks $\bullet : \u F A \vdash
S : \u F A'$. This suggests that a
\emph{value} type dynamism $A \ltdyn A'$ should induce an
embedding-projection pair of \emph{stacks} $\bullet : \u F A \vdash S_u
: \u F A'$ and $\bullet : \u F A' \vdash S_d : \u F A$, which allow both
the upcast and downcast to a priori be effectful computations.
Dually, a CBN cast typically can be represented by a CBN function of
type $B \to B'$, whose CBPV analogue is a computation of type $U \u B
\to \u B'$, which is equivalent with a computation $x : U \u B \vdash M : \u B'$,
and with a value $x : U \u B \vdash V : U \u B'$. This suggests that a
\emph{computation} type dynamism $\u B \ltdyn \u B'$ should induce an
embedding-projection pair of \emph{values} $x : U \u B \vdash V_u : U \u
B'$ and $x : U \u B' \vdash V_d : U \u B$, where both the upcast and the
downcast again may a priori be (co)effectful, in the sense that they may
not reflect all effects of their input.

However, this analysis ignores an important property of CBV casts in practice:
\emph{upcasts} always terminate without performing any effects, and in
some systems upcasts are even defined to be values, while only the
\emph{downcasts} are effectful (introduce errors).  For example, for many types $A$, the
upcast from $A$ to $\dynv$ is an injection into a sum/recursive type,
which is a value constructor.  Our previous work on a logical
relation for call-by-value gradual typing~\cite{newahmed18} proved that all
upcasts were pure in this sense as a consequence of the embedding-projection pair properties (but their proof depended on the only effects being
divergence and type error).
In GTT, we can make this property explicit
in the syntax of the casts, by making the upcast $\upcast{A}{A'}$
induced by a value type dynamism $A \ltdyn A'$ itself a complex value,
rather than computation.  On the other hand, many downcasts between value
types are implemented as a case-analysis looking for a specific
tag and erroring otherwise, and so are not complex values.

We can also make a dual observation about CBN casts.  The
\emph{downcast} arising from $\u B \ltdyn \u B'$ has a stronger property
than being a computation $x : U \u B' \vdash M : \u B$ as suggested
above: it can be taken to be a stack $\bullet : \u B' \vdash \dncast{\u
  B}{\u B'}{\bullet} : \u B$, because a downcasted computation
evaluates the computation it is ``wrapping'' exactly once.  One
intuitive justification for this point of view, which we make precise
in Section \ref{sec:contract}, is to think of the dynamic computation type $\dync$ as a
recursive \emph{product} of all possible behaviors that a computation
might have, and the downcast as a recursive type unrolling and product
projection, which is a stack.  From this point of view, an \emph{upcast}
can introduce errors, because the upcast of an object supporting some
``methods'' to one with all possible methods will error dynamically on
the unimplemented ones.

These observations are expressed in the (shaded) \textsc{UpCast} and
\textsc{DnCasts} rules for casts in
Figure~\ref{fig:gtt-syntax-and-terms}: the upcast for a value type
dynamism is a complex value, while the downcast for a computation type
dynamism is a stack (if its argument is).  Indeed, this description of
casts is simpler than the intuition we began the section with: rather
than putting in both upcasts and downcasts for all value and computation
type dynamisms, it suffices to put in only \emph{upcasts} for
\emph{value} type dynamisms and \emph{downcasts} for \emph{computation}
type dynamisms, because of monotonicity of type dynamism for $U$/$\u F$
types.  The \emph{downcast} for a \emph{value} type dynamism $A \ltdyn
A'$, as a stack $\bullet : \u F A' \vdash \dncast{\u F A}{\u F
  A'}{\bullet} : \u F A$ as described above, is obtained from $\u F A
\ltdyn \u F A'$ as computation types.  The upcast for a computation type
dynamism $\u B \ltdyn \u B'$ as a value $x : U \u B \vdash \upcast{U \u
  B}{U \u B'}{x} : U \u B'$ is obtained from $U \u B \ltdyn U \u B'$ as
value types.  Moreover, we will show below that the value upcast
$\upcast{A}{A'}$ induces a stack $\bullet : \u F A \vdash \ldots : \u F
A'$ that behaves like an upcast, and dually for the downcast, so this
formulation implies the original formulation above.

We justify this design in two ways in the remainder of the paper.  In
Section~\ref{sec:contract}, we show how to implement casts by a
contract translation to CBPV where upcasts are complex values and
downcasts are complex stacks.
However, one goal of
GTT is to be able to prove things about many gradually typed languages
at once, by giving different models, so one might wonder whether this
design rules out useful models of gradual typing where casts can have more general effects.  In
Theorem~\ref{thm:upcasts-values-downcasts-stacks}, we show instead that
our design choice is forced for all casts, as long as the casts between ground types and the dynamic types are values/stacks.





\subsection{Term Dynamism: Judgements and Structural Rules}

\begin{figure}
  \begin{small}
  \[
  \begin{array}{c}
    \framebox{$\Phi \vdash V \ltdyn V' : A \ltdyn A'$ and $\Phi \mid \Psi \vdash M \ltdyn M' : \u B \ltdyn \u B'$}
    \\\\
    
    \inferrule*[lab=TmDynRefl]{ }{\Gamma \ltdyn \Gamma \mid \Delta \ltdyn \Delta \vdash E \ltdyn E : T \ltdyn T}
    \qquad
    \inferrule*[lab=TmDynVar]
    { }
    {\Phi,x \ltdyn x' : A \ltdyn A',\Phi' \vdash x \ltdyn x' : A \ltdyn A'}
    \\\\
    \inferrule*[lab=TmDynTrans]{\Gamma \ltdyn \Gamma' \mid \Delta \ltdyn \Delta' \vdash E \ltdyn E' : T \ltdyn T' \\\\
      \Gamma' \ltdyn \Gamma'' \mid \Delta' \ltdyn \Delta'' \vdash E' \ltdyn E'' : T' \ltdyn T''
    }
    {\Gamma \ltdyn \Gamma'' \mid \Delta \ltdyn \Delta'' \vdash E \ltdyn E'' : T \ltdyn T''}
    \qquad
    \inferrule*[lab=TmDynValSubst]
    {\Phi \vdash V \ltdyn V' : A \ltdyn A' \\\\
      \Phi, x \ltdyn x' : A \ltdyn A',\Phi' \pipe \Psi \vdash E \ltdyn E' : T \ltdyn T'
    }
    {\Phi \mid \Psi \vdash E[V/x] \ltdyn E'[V'/x'] : T \ltdyn T'}
    \\\\
    \inferrule*[lab=TmDynHole]
    { }
    {\Phi \pipe \bullet \ltdyn \bullet : \u B \ltdyn \u B' \vdash \bullet \ltdyn \bullet : \u B \ltdyn \u B'}
    \qquad
    \inferrule*[lab=TmDynStkSubst]
    {\Phi \pipe \Psi \vdash M_1 \ltdyn M_1' : \u B_1 \ltdyn \u B_1' \\\\
      \Phi \pipe \bullet \ltdyn \bullet : \u B_1 \ltdyn \u B_1' \vdash M_2 \ltdyn M_2' : \u B_2 \ltdyn \u B_2'
    }
    {\Phi \mid \Psi \vdash M_2[M_1/\bullet] \ltdyn M_2'[M_1'/\bullet] : \u B_2 \ltdyn \u B_2'}
    \\\\
  \end{array}
  \]
  \vspace{-0.25in}
  \caption{GTT Term Dynamism (Structural Rules) }
  \label{fig:gtt-term-dynamism-structural}
\end{small}
\end{figure}

\iflong
\begin{figure}
  \begin{small}
  \[
  \begin{array}{c}
    \inferrule*[lab=$+$IlCong]
    {\Phi \vdash V \ltdyn V' : A_1 \ltdyn A_1'}
    {\Phi \vdash \inl V \ltdyn \inl V' : A_1 + A_2 \ltdyn A_1' + A_2'}
    \qquad
    \inferrule*[lab=$+$IrCong]
    {\Phi \vdash V \ltdyn V' : A_2 \ltdyn A_2'}
    {\Phi \vdash \inr V \ltdyn \inr V' : A_1 + A_2 \ltdyn A_1' + A_2'}
    \\\\
    \inferrule*[lab=$+$ECong]
        {
          \Phi \vdash V \ltdyn V' : A_1 + A_2 \ltdyn A_1' + A_2' \\\\
          \Phi, x_1 \ltdyn x_1' : A_1 \ltdyn A_1' \mid \Psi \vdash E_1 \ltdyn E_1' : T \ltdyn T' \\\\
          \Phi, x_2 \ltdyn x_2' : A_2 \ltdyn A_2' \mid \Psi \vdash E_2 \ltdyn E_2' : T \ltdyn T'
        }
    {\Phi \mid \Psi \vdash \caseofXthenYelseZ V {x_1. E_1}{x_2.E_2} \ltdyn \caseofXthenYelseZ V {x_1'. E_1'}{x_2'.E_2'} : T'}
    \qquad
    \inferrule*[lab=$0$ECong]
    {\Phi \vdash V \ltdyn V' : 0 \ltdyn 0}
    {\Phi \mid \Psi \vdash \abort V \ltdyn \abort V' : T \ltdyn T'}
    \\\\
    \inferrule*[lab=$1$ICong]{ }{\Phi \vdash () \ltdyn () : 1 \ltdyn 1}
    \qquad
    \inferrule*[lab=$1$ECong]
    {\Phi \vdash V \ltdyn V' : 1 \ltdyn 1 \\\\
      \Phi \mid \Psi \vdash E \ltdyn E' : T \ltdyn T'
    }
    {\Phi \mid \Psi \vdash \pmpairWtoinZ V E \ltdyn \pmpairWtoinZ V' E' : T \ltdyn T'}
    \\\\
    \inferrule*[lab=$\times$ICong]
    {\Phi \vdash V_1 \ltdyn V_1' : A_1 \ltdyn A_1'\\\\
      \Phi\vdash V_2 \ltdyn V_2' : A_2 \ltdyn A_2'}
    {\Phi \vdash (V_1,V_2) \ltdyn (V_1',V_2') : A_1 \times A_2 \ltdyn A_1' \times A_2'}
    \quad
    \inferrule*[lab=$\to$ICong]
    {\Phi, x \ltdyn x' : A \ltdyn A' \pipe \Psi \vdash M \ltdyn M' : \u B \ltdyn \u B'}
    {\Phi \pipe \Psi \vdash \lambda x : A . M \ltdyn \lambda x' : A' . M' : A \to \u B \ltdyn A' \to \u B'}
    
    \\\\
    \inferrule*[lab=$\times$ECong]
    {\Phi \vdash V \ltdyn V' : A_1 \times A_2 \ltdyn A_1' \times A_2' \\\\
      \Phi, x \ltdyn x' : A_1 \ltdyn A_1', y \ltdyn y' : A_2 \ltdyn A_2' \mid \Psi \vdash E \ltdyn E'  : T \ltdyn T'
    }
    {\Phi \mid \Psi \vdash \pmpairWtoXYinZ V x y E \ltdyn \pmpairWtoXYinZ {V'} {x'} {y'} {E'} : T \ltdyn T'}
    \,\,
    \inferrule*[lab=$\to$ECong]
    {\Phi \pipe \Psi \vdash M \ltdyn M' : A \to \u B \ltdyn A' \to \u B' \\\\
      \Phi \vdash V \ltdyn V' : A \ltdyn A'}
    {\Phi \pipe \Psi \vdash M\,V \ltdyn M'\,V' : \u B \ltdyn \u B' }
    \\\\
    \inferrule*[lab=$U$ICong]
    {\Phi \mid \cdot \vdash M \ltdyn M' : \u B \ltdyn \u B'}
    {\Phi \vdash \thunk M \ltdyn \thunk M' : U \u B \ltdyn U \u B'}
    \qquad
    \inferrule*[lab=$U$ECong]
    {\Phi \vdash V \ltdyn V' : U \u B \ltdyn U \u B'}
    {\Phi \pipe \cdot \vdash \force V \ltdyn \force V' : \u B \ltdyn \u B'}
    \\\\
    \inferrule*[lab=$F$ICong]
    {\Phi \vdash V \ltdyn V' : A \ltdyn A'}
    {\Phi \pipe \cdot \vdash \ret V \ltdyn \ret V' : \u F A \ltdyn \u F A'}
    \qquad
    \inferrule*[lab=$F$ECong]
    {\Phi \pipe \Psi \vdash M \ltdyn M' : \u F A \ltdyn \u F A' \\\\
      \Phi, x \ltdyn x' : A \ltdyn A' \pipe \cdot \vdash N \ltdyn N' : \u B \ltdyn \u B'} 
    {\Phi \pipe \Psi \vdash \bindXtoYinZ M x N \ltdyn \bindXtoYinZ {M'} {x'} {N'} : \u B \ltdyn \u B'} 
    \\\\
    \inferrule*[lab=$\top$ICong]{ }{\Phi \mid \Psi \vdash \{\} \ltdyn \{\} : \top \ltdyn \top}
    \qquad
    \inferrule*[lab=$\with$ICong]
    {\Phi \mid \Psi \vdash M_1 \ltdyn M_1' : \u B_1 \ltdyn \u B_1'\and
      \Phi \mid \Psi \vdash M_2 \ltdyn M_2' : \u B_2 \ltdyn \u B_2'}
    {\Phi \mid \Psi \vdash \pair {M_1} {M_2} \ltdyn \pair {M_1'} {M_2'} : \u B_1 \with \u B_2 \ltdyn \u B_1' \with \u B_2'}
    \\\\ 
    \inferrule*[lab=$\with$ECong]
    {\Phi \mid \Psi \vdash M \ltdyn M' : \u B_1 \with \u B_2 \ltdyn \u B_1' \with \u B_2'}
    {\Phi \mid \Psi \vdash \pi M \ltdyn \pi M' : \u B_1 \ltdyn \u B_1'}
    \qquad
    \inferrule*[lab=$\with$E'Cong]
    {\Phi \mid \Psi \vdash M \ltdyn M' : \u B_1 \with \u B_2 \ltdyn \u B_1' \with \u B_2'}
    {\Phi \mid \Psi \vdash \pi' M \ltdyn \pi' M' : \u B_2 \ltdyn \u B_2'}
  \end{array}
  \]
  \caption{GTT Term Dynamism (Congruence Rules)}
  \label{fig:gtt-term-dynamism-ext-congruence}
\end{small}
\end{figure}
\fi

The final piece of GTT is the \emph{term dynamism} relation, a syntactic
judgement that is used for reasoning about the behavioral properties of
terms in GTT.  To a first approximation, term dynamism can be thought of
as syntactic rules for reasoning about \emph{contextual approximation}
relative to errors (not divergence), where $E \ltdyn E'$ means that
either $E$ errors or $E$ and $E'$ have the same result.  However, a key
idea in GTT is to consider a \emph{heterogeneous} term dynamism
judgement $E \ltdyn E' : T \ltdyn T'$ between terms $E : T$ and $E' :
T'$ where $T \ltdyn T'$---i.e. relating two terms at two different
types, where the type on the right is more dynamic than the type on the
right.  This judgement structure allows simple axioms characterizing the
behavior of casts~\cite{newlicata2018-fscd} and axiomatizes the
graduality property~\cite{refined}.
Here, we break this judgement up into
value dynamism $V \ltdyn V' : A \ltdyn A'$ and computation dynamism $M
\ltdyn M' : \u B \ltdyn \u B'$.  To support reasoning about open terms,
the full form of the judgements are
\begin{itemize}
\item $\Gamma \ltdyn \Gamma' \vdash V \ltdyn V' : A \ltdyn A'$ where
  $\Gamma \vdash V : A$ and $\Gamma' \vdash V' : A'$ and $\Gamma \ltdyn
  \Gamma'$ and $A \ltdyn A'$.
\item 
$\Gamma \ltdyn \Gamma' \mid \Delta \ltdyn \Delta' \vdash M \ltdyn M' :
  \u B \ltdyn \u B'$ where $\Gamma \mid \Delta \vdash M : \u B$ and
  $\Gamma' \mid \Delta' \vdash M' : \u B'$.
\end{itemize}
where $\Gamma \ltdyn \Gamma'$ is the pointwise lifting of value type
dynamism, and $\Delta \ltdyn \Delta'$ is the optional lifting of
computation type dynamism.  We write $\Phi : \Gamma \ltdyn \Gamma'$ and
$\Psi : \Delta \ltdyn \Delta'$ as syntax for ``zipped'' pairs of
contexts that are pointwise related by type dynamism, $x_1 \ltdyn x_1' : A_1 \ltdyn A_1', \ldots, x_n \ltdyn x_n' :
A_n \ltdyn A_n'$, which correctly suggests that one can substitute related
terms for related variables.  We will implicitly zip/unzip pairs of
contexts, and sometimes write e.g. $\Gamma \ltdyn \Gamma$ to mean $x
\ltdyn x : A \ltdyn A$ for all $x : A$ in $\Gamma$.

The main point of our rules for term dynamism is that \emph{there are no
  type-specific axioms in the definition} beyond the $\beta\eta$-axioms
that the type satisfies in a non-gradual language.  Thus, adding a new
type to gradual type theory does not require any a priori consideration
of its gradual behavior in the language definition; instead, this is
deduced as a theorem in the type theory.  The basic structural rules of
term dynamism in Figure~\ref{fig:gtt-term-dynamism-structural}\iflong\ and Figure~\ref{fig:gtt-term-dynamism-ext-congruence}\fi\ say that
it is reflexive and transitive (\textsc{TmDynRefl},
\textsc{TmDynTrans}), that assumptions can be used and substituted for
(\textsc{TmDynVar}, \textsc{TmDynValSubst}, \textsc{TmDynHole},
\textsc{TmDynStkSubst}), and that every term constructor is monotone
(the \textsc{Cong} rules).
\begin{longonly}
While we could add congruence rules for errors and casts,
these follow from the axioms characterizing their behavior below.  
\end{longonly}

We will often abbreviate a ``homogeneous'' term dynamism (where the type
or context dynamism is given by reflexivity) by writing e.g. $\Gamma
\vdash V \ltdyn V' : A \ltdyn A'$ for $\Gamma \ltdyn \Gamma \vdash V
\ltdyn V' : A \ltdyn A'$, or $\Phi \vdash V \ltdyn V' : A$ for $\Phi
\vdash V \ltdyn V' : A \ltdyn A$, and similarly for computations.  The
entirely homogeneous judgements $\Gamma \vdash V \ltdyn V' : A$ and
$\Gamma \mid \Delta \vdash M \ltdyn M' : \u B$ can be thought of as a
syntax for contextual error approximation (as we prove below).  We write
$V \equidyn V'$ (``equidynamism'') to mean term dynamism relations in
both directions (which requires that the types are also equidynamic
$\Gamma \equidyn \Gamma'$ and $A \ltdyn A'$), which is a syntactic
judgement for contextual equivalence.

\subsection{Term Dynamism: Axioms}

Finally, we assert some term dynamism axioms that describe the behavior
of programs.  The cast universal properties at the top of
Figure~\ref{fig:gtt-term-dyn-axioms}, following~\citet{newlicata2018-fscd}, say that
the defining property of an upcast from $A$ to $A'$ is that it is the
least dynamic term of type $A'$ that is more dynamic that $x$, a ``least
upper bound''.  That is, $\upcast{A}{A'}{x}$ is a term of type $A'$ that
is more dynamic that $x$ (the ``bound'' rule), and for any other term
$x'$ of type $A'$ that is more dynamic than $x$, $\upcast{A}{A'}{x}$ is
less dynamic than $x'$ (the ``best'' rule). Dually, the downcast
$\dncast{\u B}{\u B'}{\bullet}$ is the most dynamic term of type $\u B$
that is less dynamic than $\bullet$, a ``greatest lower bound''.  These
defining properties are entirely independent of the types involved in
the casts, and do not change as we add or remove types from the system.

We will show that these defining properties already imply that the shift
of the upcast $\upcast{A}{A'}$ forms a Galois connection/adjunction with
the downcast $\dncast{\u F A}{\u F A'}$, and dually for computation
types (see Theorem~\ref{thm:cast-adjunction}).  They do not
automatically form a Galois insertion/coreflection/embedding-projection
pair, but we can add this by the retract axioms in
Figure~\ref{fig:gtt-term-dyn-axioms}.  Together with other theorems of
GTT, these axioms imply that any upcast followed by its corresponding
downcast is the identity (see Theorem~\ref{thm:retract-general}).  This
specification of casts leaves some behavior undefined: for example, we
cannot prove in the theory that $\dncast{\u F 1+1}{\u
  F\dynv}\upcast{1}{\dynv}$ reduces to an error.  We choose this design
because there are valid models in which it is not an error, for instance
if the unique value of $1$ is represented as the boolean \texttt{true}. In
Section~\ref{sec:dynamic-type-interp}, we show additional axioms that
fully characterize the behavior of the dynamic type.

The type universal properties in the middle of the figure, which are
taken directly from CBPV, assert the $\beta\eta$ rules for each type as
(homogeneous) term equidynamisms---these should be understood as having,
as implicit premises, the typing conditions that make both sides type
check, in equidynamic contexts.

The final axioms assert properties of the run-time error term $\err$: it
is the least dynamic term (has the fewest behaviors) of every
computation type, and all complex stacks are strict in errors, because
stacks force their evaluation position.  We state the first axiom in a
heterogeneous way, which includes congruence $\Gamma \ltdyn \Gamma'
\vdash \err_{\u B} \ltdyn \err_{\u B'} : \u B \ltdyn \u B'$.

\begin{figure}
  \begin{small}

    \framebox{Cast Universal Properties}
    \medskip
    
    \begin{tabular}{c|c|c}
      & Bound & Best \\
      \hline
      Up &
      ${x : A \vdash x \ltdyn \upcast A {A'} x : A \ltdyn A'}$ & 
      ${x \ltdyn x' : A \ltdyn A' \vdash \upcast A {A'} x \ltdyn x' : A' }$\\

      \hline 
      Down & 
      ${\bullet : \u B' \vdash \dncast{\u B}{\u B'} \bullet \ltdyn \bullet : \u B \ltdyn \u B'}$
      & 
      ${\bullet \ltdyn \bullet : \u B \ltdyn \u B' \vdash \bullet \ltdyn \dncast{\u B}{\u B'} \bullet : \u B}$\\
    \end{tabular}

    \[
    \framebox{Retract Axiom}
    \quad
    \begin{array}{c}
      x : A \vdash \dncast{\u F A}{\u F \, \dynv}{(\ret{(\upcast{A}{\dynv}{x})})} \ltdyn \ret{x} : \u F A \\
      x : U \u B \vdash \dncast{\u B}{\dync}{(\force{(\upcast{U \u B}{U \dync}{x})})}  \ltdyn \force{x} : \u B \\
    \end{array}
    \]

    \bigskip
    
    \framebox{Type Universal Properties}
    \medskip

    \begin{tabular}{c|c|c}
      Type & $\beta$ & $\eta$\\

      \hline
      + & 
      $\begin{array}{l}
        {\caseofXthenYelseZ{\inl V}{x_1. E_1}{\ldots} \equidyn E_1[V/x_1]}\\
        {\caseofXthenYelseZ{\inr V}{\ldots}{x_2. E_2} \equidyn
          E_2[V/x_2]}
      \end{array}$
      &
      $\begin{array}{l}
        E \equidyn \caseofXthenYelseZ x {x_1. E[\inl x_1/x]}{x_2. E[\inr x_2/x]}\\
        \text{where } x:A_1+A_2 \vdash E : T
      \end{array}$
      \\
\iflong
      \hline
        $0$
      & $-$
      & $\begin{array}{l}
         E \equidyn \abort x\\
         \text{where } x:0 \vdash E : T
         \end{array}$ \\

      \hline
      $\times$ &
      ${\pmpairWtoXYinZ{(V_1,V_2)}{x_1}{x_2}{E} \equidyn E[V_1/x_1,V_2/x_2]}$
      &
      $\begin{array}{l}
        E \equidyn \pmpairWtoXYinZ x {x_1}{x_2} E[(x_1,x_2)/x] \\
        \text{where } {x : A_1 \times A_2 \vdash E : T}
      \end{array}$\\

      \hline
      $1$
      & $\pmpairWtoinZ{()}{E} \equidyn E$
      &
      $\begin{array}{l}
        {x : 1 \vdash E \equidyn \pmpairWtoinZ{x}{E[()/x]} : T}\\
        \text{where } {x : 1 \vdash E : T}
      \end{array}$\\
\fi
      \hline
      $U$
      & ${\force\thunk M \equidyn M}$
      & ${x : U \u B \vdash x \equidyn \thunk\force x : U \u B}$\\

      \hline
      $F$
      &
      ${\bindXtoYinZ {\ret V} x M \equidyn M[V/x]}$
      &
      ${\bullet : \u F A \vdash M \equidyn \bindXtoYinZ \bullet x M[\ret x/\bullet] : \u B}$\\

      \hline
      $\to$
      &
      ${(\lambda x:A. M)\,V \equidyn M[V/x]}$
      &
      ${\bullet : A \to \u B \vdash \bullet \equidyn \lambda x:A. \bullet\,x : A \to \u B}$\\

      \hline
      $\with$
      &
      $\begin{array}{l}
        {\pi \pair{M}{M'} \equidyn M}\\
        {\pi' \pair{M}{M'} \equidyn M'}
      \end{array}$
      & ${\bullet : \u B_1 \with \u B_2 \vdash \bullet \equidyn\pair{\pi \bullet}{\pi' \bullet} : \u B_1 \with \u B_2}$ \\
\iflong
      \hline
      $\top$
      & - 
      &
      ${\bullet : \top \vdash \bullet \equidyn \{\} : \top}$\\
\fi
    \end{tabular}

    \smallskip
    
    \begin{mathpar}
    \framebox{Error Properties}
    \qquad
    \inferrule*[lab=ErrBot]{ \Gamma' \mid \cdot \vdash M' : \u B' }
              { \Gamma \ltdyn \Gamma' \mid \cdot \vdash \err \ltdyn M' : \u B \ltdyn \u B'}
    \qquad
    \inferrule*[lab=StkStrict] { \Gamma \mid x : \u B \vdash S : \u B'}
               {\Gamma \mid \cdot \vdash S[\err_{\u B}] \ltdyn \err_{\u{B'}} : \u B'}
    \end{mathpar}
  \end{small}
  \caption{GTT Term Dynamism Axioms }
  \label{fig:gtt-term-dyn-axioms}
\end{figure}

\section{Theorems in Gradual Type Theory}
\label{sec:theorems-in-gtt}

In this section, we show that the axiomatics of gradual type theory
determine most properties of casts, which shows that these behaviors of
casts are forced in any implementation of gradual typing satisfying
graduality and $\beta,\eta$.
\begin{shortonly}
  For proofs, see the extended version of the paper.
\end{shortonly}
  
\begin{longonly}
\subsection{Properties inherited from CBPV}

Because the GTT term equidynamism relation $\equidyn$ includes the
congruence and $\beta\eta$ axioms of the CBPV equational theory, types
inherit the universal properties they have there~\cite{levy03cbpvbook}.  We recall
some relevant definitions and facts.

\begin{definition}[Isomorphism] ~
  \begin{enumerate}
  \item We write $A \cong_v A'$ for a \emph{value isomorphism between
    $A$ and $A'$}, which consists of two complex values $x : A \vdash V'
    : A'$ and $x' : A' \vdash V : A$ such that $x : A \vdash V[V'/x']
    \equidyn x : A$ and $x' : A' \vdash V'[V/x] \equidyn x' : A'$.
  \item We write $\u B \cong_c \u B'$ for a \emph{computation
    isomorphism between $\u B$ and $\u B'$}, which consists of two
    complex stacks $\bullet : \u B \vdash S' : \u B'$ and $\bullet' : \u
    B' \vdash S : \u B$ such that $\bullet : \u B \vdash S[S'/x']
    \equidyn \bullet : \u B$ and $\bullet' : \u B' \vdash S'[S/\bullet]
    \equidyn \bullet' : \u B'$.
  \end{enumerate}
\end{definition}
Note that a value isomorphism is a strong condition, and an isomorphism
in call-by-value between types $A$ and $A'$ corresponds to a computation
isomorphism $\u F A \cong \u F A'$, and dually~\cite{levy17popl}.

\smallskip

\begin{lemma}[Initial objects] ~ \label{lem:initial}
  \begin{enumerate}
  \item For all (value or computation) types $T$, there exists a unique
    expression $x : 0 \vdash E : T$.
  \item For all $\u B$, there exists a unique stack $\bullet : \u F 0
    \vdash S : \u B$.
  \item
    0 is strictly initial: Suppose there is a type $A$ with a complex
    value $x : A \vdash V : 0$.  Then $V$ is an isomorphism $A \cong_v
    0$.

  \item $\u F 0$ is not provably \emph{strictly} initial among computation types.
  \end{enumerate}
\end{lemma}
\begin{proof}~
    \begin{enumerate}
    \item Take $E$ to be $x : 0 \vdash \abort{x} : T$.  Given any $E'$,
      we have $E \equidyn E'$ by the $\eta$ principle for $0$.

    \item Take $S$ to be $\bullet : \u F 0 \vdash
      \bindXtoYinZ{\bullet}{x}{\abort{x}} : \u B$.  Given another $S'$,
      by the $\eta$ principle for $F$ types, $S' \equidyn
      \bindXtoYinZ{\bullet}{x}{S'[\ret x]}$.  By congruence, to show $S
      \equidyn S'$, it suffices to show $x : 0 \vdash \abort{x} \equidyn
      S[\ret{x}] : \u B$, which is an instance of the previous part.
      
    \item 
      We have $y : 0 \vdash \abort{y} : A$.  The composite $y : 0 \vdash
      V[\abort{y}/x] : 0$ is equidynamic with $y$ by the $\eta$
      principle for $0$, which says that any two complex values with
      domain $0$ are equal.
  
      The composite $x : A \vdash \abort{V} : A$ is equidynamic
      with $x$, because 
  \[
  x : A, y : A, z : 0 \vdash x \equidyn \abort{z} \equidyn y : A
  \]
  where the first is by $\eta$ with $x : A, y : A, z : 0 \vdash E[z] :=
  x : A$ and the second with $x : 0, y : 0 \vdash E[z] := y : A$ (this
  depends on the fact that $0$ is ``distributive'', i.e. $\Gamma,x:0$
  has the universal property of $0$).  Substituting $\abort{V}$ for $y$
  and $V$ for $z$, we have $\abort{V} \equidyn x$.

\item $\u F 0$ is not \emph{strictly} initial among computation types,
  though.  Proof sketch: a domain model along the lines of
  \citep{newlicata2018-fscd} with only non-termination and type errors shows this,
  because there $\u F 0$ and $\top$ are isomorphic (the same object is
  both initial and terminal), so if $\u F 0$ were strictly initial (any
  type $\u B$ with a stack $\bullet : B \vdash S : \u F 0$ is isomorphic
  to $\u F 0$), then because every type $\u B$ has a stack to $\top$
  (terminal) and therefore $\u F 0$, every type would be isomorphic to
  $\top$/$\u F 0$---i.e. the stack category would be trivial.  But there
  are non-trivial computation types in this model.  
    \end{enumerate}
\end{proof}
  
  \begin{lemma}[Terminal objects] ~ \label{lem:terminal}
    \begin{enumerate}
    \item For any computation type $\u B$, there exists a unique stack
      $\bullet : \u B \vdash S : \top$.
    \item (In any context $\Gamma$,) there exists a unique complex value
      $V : U \top$.
    \item (In any context $\Gamma$,) there exists a unique complex value
      $V : 1$.
    \item $U \top \cong_v 1$
    \item $\top$ is not a strict terminal object.  
    \end{enumerate}
  \end{lemma}
  \begin{proof} ~
    \begin{enumerate}
    \item Take $S = \{\}$.  The $\eta$ rule for $\top$, $\bullet : \top
      \vdash \bullet \equidyn \{\} : \top$, under the substitution of
      $\bullet : \u B \vdash S : \top$, gives $S \equidyn
      \{\}[S/\bullet] = \{\}$.  

    \item Take $V = \thunk{\{\}}$.  We have $x : U \top \vdash x
      \equidyn \thunk{\force{x}} \equidyn \thunk{\{\}} : U \top$ by the
      $\eta$ rules for $U$ and $\top$.

    \item Take $V = ()$.  By $\eta$ for $1$ with $x : 1 \vdash E[x] :=
      () : 1$, we have $x : 1 \vdash () \equidyn \pmmuXtoYinZ{x}{()} :
      1$.  By $\eta$ fro $1$ with $x : 1 \vdash E[x] := x : 1$, we have
      $x : 1 \vdash x \equidyn \pmmuXtoYinZ{x}{()}$.  Therefore $x : 1
      \vdash x \equidyn () : 1$.

    \item We have maps $x : U \top \vdash () : 1$ and $x : 1 \vdash
      \thunk{\{\}} : U \top$.  The composite on $1$ is the identity by
      the previous part.  The composite on $\top$ is the identity by
      part (2).

    \item Proof sketch: As above, there is a domain model with
      $\top \cong \u F 0$, so if $\top$ were a strict terminal object,
      then $\u F 0$ would be too.  But $\u F 0$ is also initial, so it
      has a map to every type, and therefore every type would be
      isomorphic to $\u F 0$ and $\top$.  But there are non-trivial
      computation types in the model.  
    \end{enumerate}
  \end{proof}
\end{longonly}

\begin{longonly}
\subsection{Derived Cast Rules}

As noted above, monotonicity of type dynamism for $U$ and $\u F$ means
that we have the following as instances of the general cast rules:
\begin{lemma}[Shifted Casts]
  The following are derivable:
\begin{small}
  \begin{mathpar}
      \inferrule
    {\Gamma \pipe \Delta \vdash M : \u F A' \and A \ltdyn A'}
    {\Gamma \pipe \Delta \vdash \dncast {\u F A} {\u F A'} M : \u F A}

    \inferrule
    {\Gamma \vdash V : U \u B \and \u B \ltdyn \u B'}
    {\Gamma \vdash \upcast {U \u B} {U \u B'} V : U \u B'}
  \end{mathpar}
\end{small}
\end{lemma}
\begin{longproof}
  They are instances of the general upcast and downcast rules, using the
  fact that $U$ and $\u F$ are congruences for type dynamism, so in the
  first rule $\u F A \ltdyn \u F A'$, and in the second, $U \u B \ltdyn
  U \u B'$.
\end{longproof}

The cast universal properties in Figure~\ref{fig:gtt-term-dyn-axioms}
imply the following seemingly more general rules for reasoning about
casts:
\begin{lemma}[Upcast and downcast left and right rules] \label{lem:cast-left-right}
  The following are derivable:
\begin{small}
  \begin{mathpar}
    \inferrule*[Right=UpR]
    {A \ltdyn A' \and \Phi \vdash V \ltdyn V' : A \ltdyn A'}
    {\Phi \vdash V \ltdyn \upcast {A'} {A''} {V'} : A \ltdyn A''}

    \inferrule*[Right=UpL]
    {\Phi \vdash V \ltdyn V'' : A \ltdyn A''}
    {\Phi \vdash \upcast A {A'} V \ltdyn V'' : A' \ltdyn A'' }

    \inferrule*[Right=DnL]
    { \u B' \ltdyn \u B'' \and \Phi \mid \Psi \vdash M' \ltdyn M'' : \u B' \ltdyn \u B''}
    { \Phi \mid \Psi \vdash \dncast{\u B}{\u B'} M' \ltdyn M'' : \u B \ltdyn \u B''}

    \inferrule*[Right=DnR]
    { \Phi \mid \Psi \vdash M \ltdyn M'' : B \ltdyn B'' }
    { \Phi \mid \Psi  \vdash M \ltdyn \dncast{\u B'}{\u B''} M'' : \u B \ltdyn \u B''} 
  \end{mathpar}
\end{small}
\end{lemma}
  In sequent calculus terminology, an upcast is left-invertible, while a
downcast is right-invertible, in the sense that any time we have a
conclusion with a upcast on the left/downcast on the right, we can
without loss of generality apply these rules (this comes from upcasts
and downcasts forming a Galois connection).  We write the $A \ltdyn A'$
and $\u B' \ltdyn \u B''$ premises on the non-invertible rules to
emphasize that the premise is not necessarily well-formed given that the
conclusion is.

\begin{longproof}
For upcast left, substitute $V'$ into the axiom $x \ltdyn
\upcast{A'}{A''}{x} : A' \ltdyn A''$ to get $V' \ltdyn
\upcast{A'}{A''}{V'}$, and then use transitivity with the premise.

For upcast right, by transitivity of
\[
x \ltdyn x' : A \ltdyn A' \vdash \upcast{A}{A'}{x} \ltdyn x' : A' \ltdyn A' \qquad
x' \ltdyn x'' : A' \ltdyn A'' \vdash x' \ltdyn x'' : A' \ltdyn A''
\]
we have
\[
x \ltdyn x'' : A \ltdyn A'' \vdash \upcast{A}{A'}{x} \ltdyn x'' : A' \ltdyn A''
\]
Substituting the premise into this gives the conclusion.  

For downcast left, substituting $M'$ into the axiom $\dncast{\u B}{\u
  B'}{\bullet} \ltdyn \bullet : \u B \ltdyn \u B'$ gives $\dncast{\u
  B}{\u B'}{M} \ltdyn M$, and then transitivity with the premise gives
the result.

For downcast right, transitivity of
\[
\bullet \ltdyn \bullet' : \u B \ltdyn \u B' \vdash \bullet \ltdyn \bullet' : \u B \ltdyn \u B' \quad
\bullet' \ltdyn \bullet'' : \u B' \ltdyn \u B'' \vdash \bullet' \ltdyn \dncast{\u B'}{\u B''}{\bullet''}
\]
gives $\bullet \ltdyn \bullet'' : \u B \ltdyn \u B'' \vdash \bullet \ltdyn \dncast{\u B'}{\u B''}{\bullet''}$,
and then substitution of the premise into this gives the conclusion.
\end{longproof}

Though we did not include congruence rules for casts in
Figure~\ifshort\else\ref{fig:gtt-term-dynamism-ext-congruence}\fi, it is derivable:
\begin{lemma}[Cast congruence rules] \label{lem:cast-congruence}
  The following congruence rules for casts are derivable:
\begin{small}
  \begin{mathpar}
    \inferrule
        { A \ltdyn A' \and A' \ltdyn A''}
        { x \ltdyn x' : A \ltdyn A' \vdash \upcast{A}{A''}{x} \ltdyn \upcast{A'}{A''}{x'} : A''}
    \and
    \inferrule
        { A \ltdyn A' \and A' \ltdyn A''}
        { x : A \vdash \upcast{A}{A'}{x} \ltdyn \upcast{A}{A''}{x} : A' \ltdyn A''}

    \inferrule
        { \u B \ltdyn \u B' \and \u B' \ltdyn \u B''}
        { \bullet' \ltdyn \bullet'' : \u B' \ltdyn \u B'' \vdash \dncast{\u B}{\u B'}{\bullet'} \ltdyn \dncast{\u B}{\u B''}{\bullet''} : \u B}
    \and
    \inferrule
        { \u B \ltdyn \u B' \and \u B' \ltdyn \u B''}
        { \bullet'' : \u B'' \vdash \dncast{\u B}{\u B''}{\bullet''}\ltdyn \dncast{\u B'}{\u B''}{\bullet''} : \u B \ltdyn \u B'}
  \end{mathpar}
\end{small}
\end{lemma}
\begin{longproof}
In all cases, uses the invertible and then non-invertible rule for the
cast.  For the first rule, by upcast left, it suffices to show $x \ltdyn
x' : A \ltdyn A' \vdash {x} \ltdyn \upcast{A'}{A''}{x'} : A \ltdyn A''$
which is true by upcast right, using $x \ltdyn x'$ in the premise.

For the second, by upcast left, it suffices to show 
$x : A \vdash {x} \ltdyn \upcast{A}{A''}{x} : A \ltdyn A''$,
which is true by upcast right.

For the third, by downcast right, it suffices to show 
$\bullet' \ltdyn \bullet'' : \u B' \ltdyn \u B'' \vdash \dncast{\u B}{\u B'}{\bullet'} \ltdyn {\bullet''} : \u B \ltdyn \u B''$,
which is true by downcast left, using $\bullet' \ltdyn \bullet''$ in the premise.

For the fourth, by downcast right, it suffices show
$\dncast{\u B}{\u B''}{\bullet''}\ltdyn {\bullet''} : \u B \ltdyn \u B''$,
which is true by downcast left.
\end{longproof}
\end{longonly}

\subsection{Type-generic Properties of Casts}

The universal property axioms for upcasts and downcasts in
Figure~\ref{fig:gtt-term-dyn-axioms} define them \emph{uniquely} up to
equidynamism ($\equidyn$): anything with the same property 
is behaviorally equivalent to a cast.

\begin{theorem}[Specification for Casts is a Universal Property]
  ~ \label{thm:casts-unique}
  \begin{enumerate}
  \item 
  If $A \ltdyn A'$ and $x : A \vdash V : A'$ is a complex value such that
  ${x : A \vdash x \ltdyn V : A \ltdyn A'}$
  and
  ${x \ltdyn x' : A \ltdyn A' \vdash V \ltdyn x' : A'}$
  then $x : A \vdash V \equidyn \upcast{A}{A'}{x} : A'$.

  \item 
  If $\u B \ltdyn \u B'$ and $\bullet' : \u B' \vdash S :
  \u B$ is a complex stack such that
  ${\bullet' : \u B' \vdash S \ltdyn \bullet' : \u B \ltdyn \u B'}$ and
  ${\bullet \ltdyn \bullet' : \u B \ltdyn \u B' \vdash \bullet \ltdyn S : \u B}$
  then $\bullet' : \u B' \vdash S \equidyn \dncast{\u B}{\u B'}\bullet' : \u B$
  \end{enumerate}
\end{theorem}
\begin{longproof}
  For the first part, to show $\upcast{A}{A'}{x} \ltdyn V$, by upcast
  left, it suffices to show $x \ltdyn V : A \ltdyn A'$, which is one
  assumption.  To show $V \ltdyn \upcast{A}{A'}{x}$, we substitute into
  the second assumption with $x \ltdyn \upcast{A}{A'}{x} : A \ltdyn A'$,
  which is true by upcast right.

  For the second part, to show $S \ltdyn \dncast{\u B}{\u
    B'}{\bullet'}$, by downcast right, it suffices to show $S \ltdyn
  \bullet' : \u B \ltdyn \u B'$, which is one of the assumptions.  To
  show $\dncast{\u B}{\u B'}{\bullet'} \ltdyn S$, we substitute into the
  second assumption with $\dncast{\u B}{\u B'}{\bullet'} \ltdyn
  \bullet'$, which is true by downcast left.
\end{longproof}

 Casts satisfy an identity and composition law:
\begin{theorem}[Casts (de)composition] \label{thm:decomposition}
  For any $A \ltdyn A' \ltdyn A''$ and $\u B \ltdyn \u B' \ltdyn \u B''$:
  \begin{enumerate}
  \item $x : A \vdash \upcast A A x \equidyn x : A$
  \item $x : A \vdash \upcast A {A''}x \equidyn \upcast{A'}{A''}\upcast A{A'} x : A''$
  \item $\bullet : \u B \vdash \dncast {\u B}{\u B} \bullet \equidyn \bullet : \u B$
  \item $\bullet : \u B'' \vdash \dncast {\u B}{\u B''} \bullet \equidyn
    \dncast{\u B}{\u B'}{(\dncast{\u B'}{\u B''} \bullet)} : \u B \ltdyn
    \u B$
  \end{enumerate}
\end{theorem}

\begin{longproof} ~
We use Theorem~\ref{thm:casts-unique} in all cases, and show that the
right-hand side has the universal property of the left.   
\begin{enumerate}
\item Both parts expand to showing 
  $x \ltdyn x : A \ltdyn A \vdash x \ltdyn x : A \ltdyn A$,
  which is true by assumption.
  
\item
  First, we need to show $x \ltdyn \upcast{A'}{A''}{(\upcast A{A'} x)} :
  A \ltdyn A''$.  By upcast right, it suffices to show $x \ltdyn
  \upcast{A}{A'}{x} : A \ltdyn A'$, which is also true by upcast right.

  For $x \ltdyn x'' : A \ltdyn A'' \vdash \upcast{A'}{A''}{(\upcast
    A{A'} x)} \ltdyn x''$, by upcast left twice, it suffices to show $x
  \ltdyn x'' : A \ltdyn A''$, which is true by assumption.
  
\item Both parts expand to showing $\bullet : \u B \vdash \bullet \ltdyn
  \bullet : \u B$, which is true by assumption.

\item
  To show $\bullet \ltdyn \bullet'' : \u B \ltdyn \u B'' \vdash \bullet
  \ltdyn \dncast{\u B}{\u B'}{(\dncast{\u B'}{\u B''} \bullet)}$, by
  downcast right (twice), it suffices to show $\bullet : \u B \ltdyn
  \bullet'' : \u B'' \vdash {\bullet} \ltdyn \bullet'' : \u B \ltdyn \u
  B''$, which is true by assumption.  Next, we have to show $\dncast{\u
    B}{\u B'}{(\dncast{\u B'}{\u B''} \bullet)} \ltdyn \bullet : \u B
  \ltdyn \u B''$, and by downcast left, it suffices to show $\dncast{\u
    B'}{\u B''}{\bullet} \ltdyn \bullet : \u B' \ltdyn \u B''$, which is
  also true by downcast left.
\end{enumerate}
\end{longproof}

\noindent In particular, this composition property implies that the casts into and
out of the dynamic type are coherent, for example if $A \ltdyn A'$
then
$\upcast{A}{\dynv}{x} \equidyn \upcast{A'}{\dynv}{\upcast{A}{A'}{x}}$.  

 The following theorem says essentially that $x \ltdyn
 \dncast{T}{T'}{\upcast{T}{T'}{x}}$ (upcast then downcast might error
 less but but otherwise does not change the behavior) and
 $\upcast{T}{T'}{\dncast{T}{T'}{x}} \ltdyn x$ (downcast then upcast
 might error more but otherwise does not change the behavior).  However,
 since a value type dynamism $A \ltdyn A'$ induces a value upcast $x :
 A \vdash \upcast{A}{A'}{x} : A'$ but a stack downcast $\bullet : \u F
 A' \vdash \dncast{\u F A}{\u F A'}{\bullet} : \u F A$ (and dually for
 computations), the statement of the theorem wraps one cast with 
 the constructors for $U$ and $\u F$ types (functoriality of $\u F/U$).
\begin{theorem}[Casts are a Galois Connection] \label{thm:cast-adjunction} ~~~
  \begin{enumerate}
  \item $\bullet' : \u F A' \vdash \bindXtoYinZ{\dncast{\u F A}{\u F A'}{\bullet'}}{x}{\ret{(\upcast{A}{A'}{x})}} \ltdyn \bullet' : \u F A'$
  \item $\bullet : \u F A \vdash \bullet \ltdyn \bindXtoYinZ{\bullet}{x}{\dncast{\u F A}{\u F A'}{(\ret{(\upcast{A}{A'}{x})})}}  : \u F A$
  \item $x : U \u B' \vdash {\upcast{U \u B}{U \u B'}{(\thunk{({\dncast{\u B}{\u B'}{\force x}})})}} \ltdyn x : U \u B'$
  \item $x : U \u B \vdash x \ltdyn \thunk{(\dncast{B}{B'}{(\force{(\upcast{U \u B}{U \u B'}{x})})})} : U \u B$
  \end{enumerate}
\end{theorem}
\begin{longproof} ~
  \begin{enumerate}
  \item By $\eta$ for $F$ types, $\bullet' : \u F A' \vdash \bullet'
    \equidyn \bindXtoYinZ{\bullet'}{x'}{\ret{x'}} : \u F A'$, so it
    suffices to show
    \[
    \bindXtoYinZ{\dncast{\u F A}{\u F A'}{\bullet'}}{x}{\ret{(\upcast{A}{A'}{x})}} \ltdyn \bindXtoYinZ{\bullet'}{x':A'}{\ret{x'}}
    \]
    By congruence, it suffices to show ${\dncast{\u F A}{\u F
        A'}{\bullet'}} \ltdyn \bullet' : \u F A \ltdyn \u F A'$, which
    is true by downcast left, and 
    $x \ltdyn x' : A \ltdyn A' \vdash {\ret{(\upcast{A}{A'}{x})}} \ltdyn
    {\ret{x'}} : A'$,
    which is true by congruence for $\mathsf{ret}$, upcast left, and the assumption.

  \item By $\eta$ for $F$ types, it suffices to show
    \[
    \bullet : \u F A \vdash \bindXtoYinZ{x}{\bullet}{\ret{x}} \ltdyn \bindXtoYinZ{\bullet}{x}{\dncast{\u F A}{\u F A'}{(\ret{(\upcast{A}{A'}{x})})}}  : \u F A
    \]
    so by congruence,
    \[
    x : A \vdash \ret{x} \ltdyn {\dncast{\u F A}{\u F A'}{(\ret{(\upcast{A}{A'}{x})})}}
    \]
    By downcast right, it suffices to show
    \[
    x : A \vdash \ret{x} \ltdyn (\ret{(\upcast{A}{A'}{x})}) : \u F A \ltdyn \u F A'
    \]
    and by congruence 
    \[
    x : A \vdash x \ltdyn ({(\upcast{A}{A'}{x})}) : A \ltdyn A'
    \]
    which is true by upcast right.  
    
  \item By $\eta$ for $U$ types, it suffices to show
    \[
    x : U \u B' \vdash {\upcast{U \u B}{U \u B'}{(\thunk{({\dncast{\u B}{\u B'}{\force x}})})}} \ltdyn \thunk{(\force{x})} : U \u B'
    \]
    By upcast left, it suffices to show
    \[
    x : U \u B' \vdash {(\thunk{({\dncast{\u B}{\u B'}{\force x}})})} \ltdyn \thunk{(\force{x})} : U \u B \ltdyn U \u B'
    \]
    and by congruence 
    \[
    x : U \u B' \vdash {\dncast{\u B}{\u B'}{\force x}} \ltdyn \force{x} : \u B \ltdyn \u B'
    \]
    which is true by downcast left.
    
  \item By $\eta$ for $U$ types, it suffices to show
    \[
    x : U \u B \vdash \thunk{(\force x)} \ltdyn \thunk{(\dncast{B}{B'}{(\force{(\upcast{U \u B}{U \u B'}{x})})})} : U \u B
    \]
    and by congruence
    \[
    x : U \u B \vdash {(\force x)} \ltdyn {(\dncast{B}{B'}{(\force{(\upcast{U \u B}{U \u B'}{x})})})} : \u B
    \]
    By downcast right, it suffices to show
    \[
    x : U \u B \vdash {(\force x)} \ltdyn {(\force{(\upcast{U \u B}{U \u B'}{x})})} : \u B \ltdyn \u B'
    \]
    and by congruence 
    \[
    x : U \u B \vdash {x} \ltdyn {(\upcast{U \u B}{U \u B'}{x})} : \u B \ltdyn \u B'
    \]
    which is true by upcast right.
  \end{enumerate}
\end{longproof}

The retract property says roughly that $x \equidyn
\dncast{T'}{T}{\upcast{T}{T'}{x}}$ (upcast then downcast does not change
the behavior), strengthening the $\ltdyn$ of
Theorem~\ref{thm:cast-adjunction}.  In
Figure~\ref{fig:gtt-term-dyn-axioms}, we asserted the retract axiom for
casts with the dynamic type.  This and the composition property implies
the retraction property for general casts:
\begin{theorem}[Retract Property for General Casts] ~~~ \label{thm:retract-general}
  \begin{enumerate}
  \item
    $\bullet : \u F A \vdash \bindXtoYinZ{\bullet}{x}{\dncast{\u F A}{\u F A'}{(\ret{(\upcast{A}{A'}{x})})}} \equidyn \bullet  : \u F A$
  \item
    $x : U \u B \vdash \thunk{(\dncast{\u B}{\u B'}{(\force{(\upcast{U \u B}{U \u B'}{x})})})} \equidyn x : U \u B$
  \end{enumerate}
\end{theorem}
\begin{longproof}
  We need only to show the $\ltdyn$ direction, because the converse is
  Theorem~\ref{thm:cast-adjunction}.

  \begin{enumerate}
  \item
  Substituting $\ret{(\upcast{A}{A'}{x})}$ into
  Theorem~\ref{thm:cast-adjunction}'s 
  \[\bullet : \u F A \vdash \bullet \ltdyn \bindXtoYinZ{\bullet}{x}{\dncast{\u F A}{\u F A'}{(\ret{(\upcast{A}{A'}{x})})}}  : \u F A
  \]
  and $\beta$-reducing gives
  \[
  x : A \vdash {\ret{(\upcast{A}{A'}{x})}} \ltdyn {\dncast{\u F A}{\u F \dynv}{(\ret{(\upcast{A'}{\dynv}{\upcast{A}{A'}{x}})})}}
  \]
  Using this, after $\eta$-expanding $\bullet : \u F A$ on the right and using congruence for
  $\mathsf{bind}$, it suffices to derive as follows:
  \[
  \begin{array}{lll}
    \dncast{\u F A}{\u F A'}{(\ret{(\upcast{A}{A'}{x})})} & \ltdyn & \text{ congruence }\\
    \dncast{\u F A}{\u F A'}{\dncast{\u F A'}{\u F \dynv}{(\ret{(\upcast{A'}{\dynv}{\upcast{A}{A'}{x}})})}} & \ltdyn & \text{ composition }\\
    \dncast{\u F A}{\u F \dynv}{(\ret{{(\upcast{A}{\dynv}{x})}})} & \ltdyn & \text{ retract axiom for $\upcast{A}{\dynv}$ }\\
    \ret{x}\\
  \end{array}
  \]

\item After using $\eta$ for $U$ and congruence, it suffices to show
  \[
  x : U \u B \vdash \dncast{\u B}{\u B'}{(\force{(\upcast{U \u B}{U \u B'}{x})})} \ltdyn \force{x}  : \u B
  \]
  Substituting $x : U \u B \vdash {\upcast{U \u B}{U \u B'}{x}} : U \u B'$
  into Theorem~\ref{thm:cast-adjunction}'s
  \[
  x : U \u B' \vdash x \ltdyn \thunk{(\dncast{B'}{\dync}{(\force{(\upcast{U \u B'}{U \dync}{x})})})} : U \u B'
  \]
  gives
  \[
  x : U \u B \vdash {\upcast{U \u B}{U \u B'}{x}}  \ltdyn \thunk{(\dncast{B'}{\dync}{(\force{(\upcast{U \u B'}{U \dync}{{\upcast{U \u B}{U \u B'}{x}}})})})} : U \u B'
  \]
  So we have
  \[
  \begin{array}{lll}
    \dncast{B}{B'}{(\force{{\upcast{U \u B}{U \u B'}{x}}})} & \ltdyn \\
    \dncast{B}{B'}{\force{(\thunk{(\dncast{B'}{\dync}{(\force{(\upcast{U \u B'}{U \dync}{{\upcast{U \u B}{U \u B'}{x}}})})})})}} & \ltdyn & \beta\\
    \dncast{B}{B'}{(\dncast{B'}{\dync}{(\force{(\upcast{U \u B'}{U \dync}{{\upcast{U \u B}{U \u B'}{x}}})})})} & \ltdyn & \text{composition}\\
    \dncast{B}{\dync}{(\force{(\upcast{U \u B}{U \dync}{x})})} & \ltdyn & \text{retract axiom for $\dncast{\u B}{\dync}$}\\
    \ret{x} & \ltdyn & \text{composition}\\
  \end{array}
  \]
\end{enumerate}
\end{longproof}

\subsection{Unique Implementations of Casts}

\begin{longonly}
\begin{definition}
  Let a \emph{type constructor} $C$ be a (value or computation) type that
  well-formed according to the grammar in Figure~\ref{fig:gtt-syntax-and-terms} with
  additional hypotheses $X \vtype$ and $\u Y \ctype$ standing for value
  or computation types, respectively.  We write $C[A/X]$ and $C[\u B/\u
    Y]$ for the substitution of a type for a variable.  
\end{definition}
For example,
\[
\begin{array}{l}
  X_1 \vtype, X_2 \vtype \vdash X_1 + X_2 \vtype \\
  \u Y \ctype \vdash U \u Y \vtype \\
  X_1 \vtype, X_2 \vtype \vdash \u F(X_1 + X_2) \ctype
\end{array}
\]
are type constructors.

It is admissible that all type constructors are monotone in type
dynamism, because we included a congruence rule for every type
constructor in Figure~\ref{fig:gtt-type-dynamism}:  

\begin{lemma}[Monotonicity of Type Constructors]
  For any type constructor $X \vtype \vdash C$, if $A \ltdyn A'$ then
  $C[A/X] \ltdyn C[A'/x]$.  For any type constructor $\u Y \ctype \vdash
  C$, if $\u B \ltdyn \u B'$ then $C[\u B/\u Y] \ltdyn C[\u B'/\u Y]$.
\end{lemma}
\begin{proof}
Induction on $C$.  In the case for a variable $X$ or $\u Y$, $A \ltdyn
A'$ or $\u B \ltdyn \u B'$ by assumption.  In all other cases, the
result follows from the inductive hypotheses and the congruence rule for
type dynamism for the type constructor
(Figure~\ref{fig:gtt-type-dynamism}).  For example, in the case for $+$,
$A_1[A/x] \ltdyn A_1[A'/x]$ and $A_2[A/x] \ltdyn A_2[A'/x]$, so
$A_1[A/x] + A_2[A/x] \ltdyn A_1[A'/x] + A_2[A'/x]$.
\end{proof}

The following lemma helps show that a complex value
$\defupcast{C[A_i/X_i,\u B_i/\u Y_i]}{C[A_i'/X_i,\u B_i'/\u Y_i]}$ is an
upcast from $C[A_i/X_i,\u B_i/\u Y_i]$ to $C[A_i'/X_i,\u B_i'/\u Y_i]$.
\begin{lemma}[Upcast Lemma] \label{lem:upcast}
  Let $X_1 \vtype, \ldots X_n \vtype, \u Y_1 \ctype, \ldots \u Y_n
  \ctype \vdash C \vtype$ be a value type constructor.  We abbreviate
  the instantiation \\ $C[A_1/X_1,\ldots,A_n/X_n,\u B_1/\u Y_i,\ldots,\u
    B_m/\u Y_m]$ by $C[A_i,\u B_i]$.

  Suppose $\defupcast{C[A_i,\u B_i]}{C[A_i',\u B_i']}{-}$ is a complex
  value (depending on $C$ and each $A_i,A_i',\u B_i,\u B_i'$) such that
  \begin{enumerate}
  \item
    For all value types $A_1,\ldots,A_n$ and $A_1',\ldots,A_n'$ with
    $A_i \ltdyn A_i'$, and all computation types $\u B_1,\ldots,\u B_m$
    and $\u B_1',\ldots,\u B_n'$ with $\u B_i \ltdyn \u B_i'$,
    \[
    x : C[A_i,\u B_i] \vdash \defupcast{C[A_i,\u B_i]}{C[A_i',\u B_i']}{x} : C[A_i',\u B_i']
    \]
  \item 
    For all value types $A_i \ltdyn A_i'$ and computation types $\u B_i
    \ltdyn \u B_i'$,
    \begin{small}
      \[
    \begin{array}{c}
      x : C[A_i,\u B_i] \vdash \defupcast{C[A_i,\u B_i]}{C[A_i,\u B_i]}{x} \ltdyn \defupcast{C[A_i,\u B_i]}{C[A_i',\u B_i']}{x} : C[A_i,\u B_i] \ltdyn C[A_i',\u B_i']\\
      x \ltdyn x' : C[A_i,\u B_i] \ltdyn C[A_i',\u B_i'] \vdash
      \defupcast{C[A_i,\u B_i]}{C[A_i',\u B_i']}{x} \ltdyn \defupcast{C[A_i',\u B_i']}{C[A_i',\u B_i']}{x'} : C[A_i',\u B_i'] 
    \end{array}
    \]
    \end{small}

  \item For all value types $A_1,\ldots,A_n$ and all computation types
    $\u B_1,\ldots,\u B_m$,
    \[
    x : C[A_i,\u B_i] \vdash \defupcast{C[A_i,\u B_i]}{C[A_i,\u B_i]}{x} \equidyn x : C[A_i,\u B_i]
    \]
  \end{enumerate}
  Then $\defupcast{C[A_i,\u B_i]}{C[A_i',\u B_i']}$ satisfies the
  universal property of an upcast, so by Theorem~\ref{thm:casts-unique}
  \[
  x : C[A_i,\u B_i] \vdash \defupcast{C[A_i,\u B_i]}{C[A_i',\u B_i']}{x} \equidyn \upcast{C[A_i,\u B_i]}{C[A_i',\u B_i']}{x} : C[A_i',\u B_i']
  \]
  Moreover, the left-to-right direction uses only the left-to-right
  direction of assumption (3), and the right-to-left uses only the
  right-to-left direction of assumption (3).
\end{lemma}

\begin{proof}
  First, we show that $\defupcast{C[A_i,\u B_i]}{C[A_i',\u B_i']}$
  satisfies the universal property of an upcast.  
  
  To show 
  \[
  x \ltdyn x' : {C[A_i,\u B_i]} \ltdyn {C[A_i',\u B_i']} \vdash \defupcast{C[A_i,\u B_i]}{C[A_i',\u B_i']}{x} \ltdyn x' : {C[A_i',\u B_i']}
  \]
  assumption (2) part 2 gives
  \[
  \defupcast{C[A_i,\u B_i]}{C[A_i',\u B_i']}{x} \ltdyn \defupcast{C[A_i',\u B_i']}{C[A_i',\u B_i']}{x'} : C[A_i',\u B_i'] 
  \]
  Then transitivity with the left-to-right direction of assumption (3) 
  \[
  \defupcast{C[A_i',\u B_i']}{C[A_i',\u B_i']}{\upcast{C[A_i,\u B_i]}{C[A_i',\u B_i']}{x}}
  \ltdyn {\upcast{C[A_i,\u B_i]}{C[A_i',\u B_i']}{x}}
  \]
  gives the result.

  To show 
  \[
    {x} \ltdyn \defupcast{C[A_i,\u B_i]}{C[A_i',\u B_i']}{x} : {C[A_i,\u B_i]} \ltdyn {C[A_i',\u B_i']}
  \]
  By assumption (2) part 1, we have 
  \[
  \defupcast{C[A_i,\u B_i]}{C[A_i,\u B_i]}{x} \ltdyn \defupcast{C[A_i,\u B_i]}{C[A_i',\u B_i']} : C[A_i,\u B_i] \ltdyn C[A_i',\u B_i']
  \]
  so transitivity with the right-to-left direction of assumption (3)
  gives the result:
  \[
  x \ltdyn \defupcast{C[A_i,\u B_i]}{C[A_i,\u B_i]}{x} 
  \]

  Then Theorem~\ref{thm:casts-unique} implies that
  $\defupcast{C[A_i,\u B_i]}{C[A_i',\u B_i']}$ is equivalent to
  $\upcast{C[A_i,\u B_i]}{C[A_i',\u B_i']}$.  
\end{proof}

Dually, we have
\begin{lemma}[Downcast Lemma] \label{lem:downcast}
  Let $X_1 \vtype, \ldots X_n \vtype, \u Y_1 \ctype, \ldots \u Y_n
  \ctype \vdash C \ctype$ be a computation type constructor.  We
  abbreviate the instantiation \\
  $C[A_1/X_1,\ldots,A_n/X_n,\u B_1/\u Y_i,\ldots,\u B_m/\u Y_m]$ by $C[A_i,\u B_i]$.

  Suppose $\defdncast{C[A_i,\u B_i]}{C[A_i',\u B_i']}{-}$ is a complex
  stack (depending on $C$ and each $A_i,A_i',\u B_i,\u B_i'$) such that
  \begin{enumerate}
  \item
    For all value types $A_1,\ldots,A_n$ and $A_1',\ldots,A_n'$ with
    $A_i \ltdyn A_i'$, and all computation types $\u B_1,\ldots,\u B_m$
    and $\u B_1',\ldots,\u B_n'$ with $\u B_i \ltdyn \u B_i'$,
    \[
    \bullet : C[A_i',\u B_i'] \vdash \defdncast{C[A_i,\u B_i]}{C[A_i',\u B_i']}{\bullet} : C[A_i,\u B_i]
    \]
  \item 
    For all value types $A_i \ltdyn A_i'$ and computation types $\u B_i
    \ltdyn \u B_i'$,
    \begin{small}
      \[
    \begin{array}{c}
      \bullet : C[A_i',\u B_i'] \vdash \defdncast{C[A_i,\u B_i]}{C[A_i',\u B_i']}{\bullet} \ltdyn \defdncast{C[A_i',\u B_i']}{C[A_i',\u B_i']}{\bullet} : C[A_i,\u B_i] \ltdyn C[A_i',\u B_i']\\
      \bullet \ltdyn \bullet : C[A_i,\u B_i] \ltdyn C[A_i',\u B_i'] \vdash
      \defdncast{C[A_i,\u B_i]}{C[A_i,\u B_i]}{x} \ltdyn \defdncast{C[A_i,\u B_i]}{C[A_i',\u B_i']}{x'} : C[A_i,\u B_i] 
    \end{array}
    \]
    \end{small}

  \item For all value types $A_1,\ldots,A_n$ and all computation types
    $\u B_1,\ldots,\u B_m$,
    \[
    \bullet : C[A_i,\u B_i] \vdash \defdncast{C[A_i,\u B_i]}{C[A_i,\u B_i]}{\bullet} \equidyn \bullet : C[A_i,\u B_i]
    \]
  \end{enumerate}
  Then $\defdncast{C[A_i,\u B_i]}{C[A_i',\u B_i']}$ satisfies the
  universal property of a downcast, so by Theorem~\ref{thm:casts-unique}
  \[
  \bullet : C[A_i',\u B_i'] \vdash \defdncast{C[A_i,\u B_i]}{C[A_i',\u B_i']}{\bullet} \equidyn \dncast{C[A_i,\u B_i]}{C[A_i',\u B_i']}{\bullet} : C[A_i,\u B_i]
  \]
  Moreover, the left-to-right direction uses only the left-to-right
  direction of assumption (3), and the right-to-left uses only the
  right-to-left direction of assumption (3).
\end{lemma}

\begin{proof}

  First, we show that $\dncast{C[A_i,\u B_i]}{C[A_i',\u B_i']}$
  satisfies the universal property of a downcast, and then apply
  Theorem~\ref{thm:casts-unique}.
  To show
  \[
  \bullet \ltdyn \bullet' : C[A_i,\u B_i] \ltdyn C[A_i',\u B_i'] \vdash \bullet \ltdyn \defdncast{C[A_i,\u B_i]}{C[A_i',\u B_i']}{\bullet'} : C[A_i,\u B_i] 
  \]
  assumption (2) part 2 gives
  \[
  \defdncast{C[A_i,\u B_i]}{C[A_i,\u B_i]}{\bullet} \ltdyn \defdncast{C[A_i,\u B_i]}{C[A_i',\u B_i']}{\bullet'} 
  \]
  Then transitivity with the right-to-left direction of assumption (3)
  \[
    \bullet
    \ltdyn
    \defdncast{C[A_i,\u B_i]}{C[A_i,\u B_i]}{\bullet}
  \]
  gives the result.

  To show 
  \[
    \defdncast{C[A_i,\u B_i]}{C[A_i',\u B_i']}{\bullet} \ltdyn \bullet : {C[A_i,\u B_i]} \ltdyn {C[A_i',\u B_i']} 
    \]
    by assumption (2) part 1, we have 
    \[
    \defdncast{C[A_i,\u B_i]}{C[A_i',\u B_i']}{\bullet} \ltdyn \defdncast{C[A_i',\u B_i']}{C[A_i',\u B_i']}{\bullet} : C[A_i,\u B_i] \ltdyn C[A_i',\u B_i']
    \]
    so transitivity with the left-to-right direction of assumption (3)
    \[
    \defdncast{C[A_i',\u B_i']}{C[A_i',\u B_i']}{\bullet} \ltdyn \bullet
    \]
    gives the result.
\end{proof}
\end{longonly}

\begin{longonly}
\subsubsection{Functions, Products, and Sums}
\end{longonly}

Together, the universal property for casts and the $\eta$ principles for
each type imply that the casts must behave as in lazy cast semantics:
\begin{theorem}[Cast Unique Implementation Theorem for $+,\times,\to,\with$] \label{thm:functorial-casts}
The casts' behavior is uniquely determined as follows: 
\begin{small}
\[
   \begin{array}{c}
\iflong  
    \upcast{A_1 + A_2}{A_1' + A_2'}{s} \equidyn \caseofXthenYelseZ{s}{x_1.\inl{(\upcast{A_1}{A_1'}{x_1})}}{x_2.\inr{(\upcast{A_2}{A_2'}{x_2})}}\\\\
    \begin{array}{rcl}
    \dncast{\u F (A_1' + A_2')}{\u F (A_1 + A_2)}{\bullet} & \equidyn 
    & \bindXtoYinZ{\bullet}{(s : (A_1' + A_2'))}\caseofX{s}\\
    & &  \{{x_1'.\bindXtoYinZ{(\dncast{\u F A_1}{\u F A_1'}{(\ret{x_1'})})}{x_1}{\ret{(\inl {x_1})}}} \\
    & &  \mid {x_2'.\bindXtoYinZ{(\dncast{\u F A_2}{\u F A_2'}{(\ret{x_2'})})}{x_2}{\ret{(\inr {x_2})}}} \}\\
    \end{array}
    \\\\
\fi
    \upcast{A_1 \times A_2}{A_1' \times A_2'}{p} \equidyn \pmpairWtoXYinZ{p}{x_1}{x_2}{(\upcast{A_1}{A_1'}{x_1},\upcast{A_2}{A_2'}{x_2})} \\\\
    \begin{array}{rcl}
    \dncast{\u F (A_1' \times A_2')}{\u F (A_1 \times A_2)}{\bullet} &
      \equidyn & 
      \bindXtoYinZ{\bullet}{p'} {\pmpairWtoXYinZ{p'}{x_1'}{x_2'}{}}\\
      & & \bindXtoYinZ{\dncast{\u F A_1}{\u F A_1'}{\ret x_1'}}{x_1}\\
      & & \bindXtoYinZ{\dncast{\u F A_2}{\u F A_2'}{\ret x_2'}}{x_2} {\ret (x_1,x_2) }\\
\iflong
      & \equidyn &  
      \bindXtoYinZ{\bullet}{p'} \pmpairWtoXYinZ{p'}{x_1'}{x_2'}{} \\
      & & \bindXtoYinZ{\dncast{\u F A_2}{\u F A_2'}{\ret x_2'}}{x_2} \\
      & & \bindXtoYinZ{\dncast{\u F A_1}{\u F A_1'}{\ret x_1'}}{x_1} {\ret (x_1,x_2) }
\fi
    \end{array}\\\\
\iflong  
    \dncast{\u B_1 \with \u B_2}{\u B_1' \with \u B_2'}{\bullet} \equidyn 
    \pair{\dncast{\u B_1}{\u B_1'}{\pi \bullet}}{\dncast{\u B_2}{\u B_2'}{\pi' \bullet}}\\\\
    \begin{array}{rcll}
      \upcast{U (\u B_1 \with \u B_2)}{U (\u B_1' \with \u B_2')}{p} & \equidyn &
      \thunk{} & {\{\pi \mapsto {\force{(\upcast{U \u B_1}{U \u B_1'}{(\thunk{\pi (\force{p})})})}}}\\
      &&& \pi' \mapsto {\force{(\upcast{U \u B_2}{U \u B_2'}{(\thunk{\pi' (\force{p})})})}} \}
    \end{array}\\\\
\fi
    
    \dncast{A \to \u B}{A' \to \u B'}{\bullet} \equidyn
    \lambda{x}.{\dncast{\u B}{\u B'}{(\bullet \, (\upcast{A}{A'}{x}))}} \\\\

    \begin{array}{rcll}
      \upcast{U (A \to \u B)}{U (A' \to \u B')}{f} & \equidyn &
      \thunk (\lambda x'. & \bindXtoYinZ{\dncast{\u F A}{\u F A'}{(\ret x')}}{x}{} \\
      & & & { \force{(\upcast{U \u B}{U \u B'}{(\thunk{(\force{(f)}\,x)})})}} )
    \end{array}
        \\
  \end{array}
  \]
  \end{small}
\end{theorem}
In the case for an eager product $\times$, we can actually also show
that reversing the order and running ${\dncast{\u F A_2}{\u F A_2'}{\ret
    x_2'}}$ and then ${\dncast{\u F A_1}{\u F A_1'}{\ret x_1'}}$ is also
an implementation of this cast, and therefore equal to the above.
Intuitively, this is sensible because the only effect a downcast
introduces is a run-time error, and if either downcast errors, both
possible implementations will.

\begin{longproof}~\\

  \begin{enumerate}
  \item Sums upcast.  We use Lemma~\ref{lem:upcast} with the type
    constructor $X_1 \vtype, X_2 \vtype \vdash X_1 + X_2 \vtype$.
    Suppose $A_1 \ltdyn A_1'$ and $A_2 \ltdyn A_2'$ and let
    \[s : A_1 + A_2 \vdash \defupcast{A_1 + A_2}{A_1' + A_2'}{s} : A_1' + A_2'
    \]
    stand for
    \[
    \caseofXthenYelseZ{s}{x_1.\inl{(\upcast{A_1}{A_1'}{x_1})}}{x_2.\inr{(\upcast{A_2}{A_2'}{x_2})}}
    \]
    which has the type required for the lemma's assumption (1).  
    
    Assumption (2) requires two condition, both of which are proved by
    the congruence rules for $\mathsf{case}$, $\mathsf{inl}$,
    $\mathsf{inr}$, and upcasts.  The first, 
    \[
    s : A_1 + A_2 \vdash \defupcast{A_1 + A_2}{A_1 + A_2}{s} \ltdyn \defupcast{A_1 + A_2}{A_1' + A_2'}{s} : A_1 + A_2 \ltdyn A_1' + A_2'\\
    \]
    expands to
    \[
    \begin{array}{c}
      \caseofXthenYelseZ{s}{x_1.\inl{(\upcast{A_1}{A_1}{x_1})}}{x_2.\inr{(\upcast{A_2}{A_2}{x_2})}} \\
      \ltdyn \\
      \caseofXthenYelseZ{s}{x_1.\inl{(\upcast{A_1}{A_1'}{x_1})}}{x_2.\inr{(\upcast{A_2}{A_2'}{x_2})}}
    \end{array}
    \]
    The second,
    \[
    s \ltdyn s' : A_1 + A_2 \ltdyn A_1' + A_2' \vdash
    \defupcast{A_1 + A_2}{A_1' + A_2'}{s} \ltdyn \defupcast{A_1' + A_2'}{A_1' + A_2'}{s'} : A_1' + A_2'
    \]
    expands to
    \[
    \begin{array}{c}
      \caseofXthenYelseZ{s}{x_1.\inl{(\upcast{A_1}{A_1'}{x_1})}}{x_2.\inr{(\upcast{A_2}{A_2'}{x_2})}} \\
      \ltdyn \\
      \caseofXthenYelseZ{s'}{x_1.\inl{(\upcast{A_1'}{A_1'}{x_1'})}}{x_2.\inr{(\upcast{A_2'}{A_2'}{x_2'})}}
    \end{array}
    \]

    Finally, for assumption (3), we need to show
    \[
    \caseofXthenYelseZ{s}{x_1.\inl{(\upcast{A_1}{A_1}{x_1})}}{x_2.\inr{(\upcast{A_2}{A_2}{x_2})}}
    \equidyn s
    \]
    which is true because $\upcast{A_1}{A_1}$ and $\upcast{A_2}{A_2}$
    are the identity, and using ``weak $\eta$'' for sums,
    $\caseofXthenYelseZ{s}{x_1.\inl{x_1}}{x_2.\inr{x_2}} \equidyn x$,
    which is the special case of the $\eta$ rule in
    Figure~\ref{fig:gtt-term-dyn-axioms} for the identity complex
    value:
    \[
    \begin{array}{rcl}
      \caseofXthenYelseZ{s}{x_1.\inl{(\upcast{A_1}{A_1}{x_1})}}{x_2.\inr{(\upcast{A_2}{A_2}{x_2})}} & \equidyn &\\
      \caseofXthenYelseZ{s}{x_1.\inl{({x_1})}}{x_2.\inr{({x_2})}} & \equidyn &\\
      s 
    \end{array}
    \]
    
  \item Sums downcast.  We use the downcast lemma with $X_1 \vtype, X_2
    \vtype \vdash \u F(X_1 + X_2) \ctype$.  Let 
    \[
    \bullet' : \u F (A_1' + A_2') \vdash \defdncast{\u F (A_1 + A_2)}{\u F (A_1' + A_2')}{\bullet'} : \u F (A_1 + \u A_2)
    \]
    stand for
    \[
    \bindXtoYinZ{\bullet}{(s : (A_1' +
      A_2'))}{}
                {\caseofXthenYelseZ{s}
           {x_1'.\bindXtoYinZ{(\dncast{\u F A_1}{\u F A_1'}{(\ret{x_1'})})}{x_1}{\ret{(\inl {x_1})}}}
           {\ldots}}\\
     \]
     (where, as in the theorem statement, $\mathsf{inr}$ branch is
     analogous), which has the correct type for the lemma's assumption
     (1).

     For assumption (2), we first need to show
     \begin{small}
     \[
     \bullet : {\u F (A_1' + A_2')} \vdash
     \defdncast{\u F (A_1 + A_2)}{\u F (A_1' + A_2')}{\bullet'}
     \ltdyn
     \defdncast{\u F (A_1' + A_2')}{\u F (A_1' + A_2')}{\bullet'}
     : {\u F (A_1 + A_2)} \ltdyn {\u F (A_1' + A_2')}
     \]
     \end{small}
     i.e.
     \begin{small}
     \[
     \begin{array}{c}
     \bindXtoYinZ{\bullet}{(s' : (A_1' + A_2'))}{}
                {\caseofXthenYelseZ{s'}
           {x_1'.\bindXtoYinZ{(\dncast{\u F A_1}{\u F A_1'}{(\ret{x_1'})})}{x_1}{\ret{(\inl {x_1})}}}
           {\ldots}}\\
     \ltdyn\\
     \bindXtoYinZ{\bullet}{(s' : (A_1' + A_2'))}{}
                {\caseofXthenYelseZ{s'}
                  {x_1'.\bindXtoYinZ{(\dncast{\u F A_1'}{\u F A_1'}{(\ret{x_1'})})}{x_1'}{\ret{(\inl {x_1'})}}}
                  {\ldots}}
     \end{array}
     \]
     \end{small}
     which is true by the congruence rules for $\mathsf{bind}$,
     $\mathsf{case}$, downcasts, $\mathsf{ret}$, and $\mathsf{inl}/\mathsf{inr}$.

     Next, we need to show
     \begin{small}
     \[
     \bullet \ltdyn \bullet' : {\u F (A_1 + A_2)} \ltdyn {\u F (A_1' + A_2')} \vdash
     \defdncast{\u F (A_1 + A_2)}{\u F (A_1 + A_2)}{\bullet}
     \ltdyn
     \defdncast{\u F (A_1 + A_2)}{\u F (A_1' + A_2')}{\bullet'}
     : {\u F (A_1 + A_2)}
     \]
     \end{small}
     i.e.
     \begin{small}
     \[
     \begin{array}{c}
     \bindXtoYinZ{\bullet}{(s : (A_1 + A_2))}{}
                {\caseofXthenYelseZ{s}
                  {x_1.\bindXtoYinZ{(\dncast{\u F A_1}{\u F A_1}{(\ret{x_1})})}{x_1}{\ret{(\inl {x_1})}}}
                  {\ldots}}\\
     \ltdyn\\
     \bindXtoYinZ{\bullet}{(s' : (A_1' + A_2'))}{}
                {\caseofXthenYelseZ{s'}
           {x_1'.\bindXtoYinZ{(\dncast{\u F A_1}{\u F A_1'}{(\ret{x_1'})})}{x_1}{\ret{(\inl {x_1})}}}
           {\ldots}}\\
     \end{array}
     \]
     \end{small}
     which is also true by congruence.

     Finally, for assumption (3), we show
     \begin{small}
       \[
       \begin{array}{lll}
         \bindXtoYinZ{\bullet}{(s : (A_1 + A_2))}{}
                     {\caseofXthenYelseZ{s}
                       {x_1.\bindXtoYinZ{(\dncast{\u F A_1}{\u F A_1}{(\ret{x_1})})}{x_1}{\ret{(\inl {x_1})}}}
                       {\ldots}} & \equidyn & \\
         \bindXtoYinZ{\bullet}{(s : (A_1 + A_2))}{}
                     {\caseofXthenYelseZ{s}
                       {x_1.\bindXtoYinZ{({(\ret{x_1})})}{x_1}{\ret{(\inl {x_1})}}}
                       {\ldots}} & \equidyn & \\
         \bindXtoYinZ{\bullet}{(s : (A_1 + A_2))}{}
                     {\caseofXthenYelseZ{s}
                       {x_1.{\ret{(\inl {x_1})}}}
                       {x_2.{\ret{(\inr {x_2})}}}} & \equidyn & \\
         \bindXtoYinZ{\bullet}{(s : (A_1 + A_2))}{}
                     {\ret{s}} & \equidyn &\\
         \bullet
       \end{array}
       \]
     \end{small}
     using the downcast identity, $\beta$ for $\u F$ types, $\eta$ for
     sums, and $\eta$ for $\u F$ types.  
     
  \item Eager product upcast. We use Lemma~\ref{lem:upcast} with the type
    constructor $X_1 \vtype, X_2 \vtype \vdash X_1 \times X_2 \vtype$.
    Let
    \[p : A_1 \times A_2 \vdash \defupcast{A_1 \times A_2}{A_1' \times A_2'}{s} : A_1' \times A_2'
    \]
    stand for
    \[
    \pmpairWtoXYinZ{p}{x_1}{x_2}{(\upcast{A_1}{A_1'}{x_1},\upcast{A_2}{A_2'}{x_2})} 
    \]
    which has the type required for the lemma's assumption (1).  
    
    Assumption (2) requires two condition, both of which are proved by
    the congruence rules for $\mathsf{split}$, pairing, and upcasts.
    The first,
    \[
    p : A_1 \times A_2 \vdash \defupcast{A_1 \times A_2}{A_1 \times A_2}{s} \ltdyn \defupcast{A_1 \times A_2}{A_1' \times A_2'}{s} : A_1 \times A_2 \ltdyn A_1' \times A_2'\\
    \]
    expands to
    \[
    \begin{array}{c}
      \pmpairWtoXYinZ{p}{x_1}{x_2}{(\upcast{A_1}{A_1}{x_1},\upcast{A_2}{A_2}{x_2})}\\
      \ltdyn \\
      \pmpairWtoXYinZ{p}{x_1}{x_2}{(\upcast{A_1}{A_1'}{x_1},\upcast{A_2}{A_2'}{x_2})}\\
    \end{array}
    \]
    The second,
    \[
    p \ltdyn p' : A_1 \times A_2 \ltdyn A_1' \times A_2' \vdash
    \defupcast{A_1 \times A_2}{A_1' \times A_2'}{s} \ltdyn \defupcast{A_1' \times A_2'}{A_1' \times A_2'}{s'} : A_1' \times A_2'
    \]
    expands to
    \[
    \begin{array}{c}
      \pmpairWtoXYinZ{p}{x_1}{x_2}{(\upcast{A_1}{A_1'}{x_1},\upcast{A_2}{A_2'}{x_2})}\\
      \ltdyn \\
      \pmpairWtoXYinZ{p'}{x_1'}{x_2'}{(\upcast{A_1'}{A_1'}{x_1'},\upcast{A_2'}{A_2'}{x_2'})}\\
    \end{array}
    \]

    Finally, for assumption (3), using $\eta$ for products and
    the fact that $\upcast{A}{A}{}$ is the identity, we have
    \[
      \pmpairWtoXYinZ{p}{x_1}{x_2}{(\upcast{A_1}{A_1}{x_1},\upcast{A_2}{A_2}{x_2})} \equidyn
      \pmpairWtoXYinZ{p}{x_1}{x_2}{({x_1},{x_2})} \equidyn
      p
    \]

  \item Eager product downcast.

      We use the downcast lemma with $X_1 \vtype, X_2 \vtype \vdash \u
      F(X_1 \times X_2) \ctype$.  Let
    \[
    \bullet' : \u F (A_1' \times A_2') \vdash \defdncast{\u F (A_1 \times A_2)}{\u F (A_1' \times A_2')}{\bullet'} : \u F (A_1 \times \u A_2)
    \]
    stand for
    \begin{small}
    \[
      \bindXtoYinZ{\bullet}{p'}{\pmpairWtoXYinZ{p'}{x_1'}{x_2'}{
          \bindXtoYinZ{\dncast{\u F A_1}{\u F A_1'}{\ret x_1'}}{x_1}{
            \bindXtoYinZ{\dncast{\u F A_2}{\u F A_2'}{\ret x_2'}}{x_2} {\ret (x_1,x_2) }}}}
      \]
    \end{small}
     which has the correct type for the lemma's assumption (1).

     For assumption (2), we first need to show
     \begin{small}
     \[
     \bullet : {\u F (A_1' \times A_2')} \vdash
     \defdncast{\u F (A_1 \times A_2)}{\u F (A_1' \times A_2')}{\bullet'}
     \ltdyn
     \defdncast{\u F (A_1' \times A_2')}{\u F (A_1' \times A_2')}{\bullet'}
     : {\u F (A_1 \times A_2)} \ltdyn {\u F (A_1' \times A_2')}
     \]
     \end{small}
     i.e.
     \begin{small}
     \[
     \begin{array}{c}
       \bindXtoYinZ{\bullet}{p'}{\pmpairWtoXYinZ{p'}{x_1'}{x_2'}{
           \bindXtoYinZ{\dncast{\u F A_1}{\u F A_1'}{\ret x_1'}}{x_1}{
            \bindXtoYinZ{\dncast{\u F A_2}{\u F A_2'}{\ret x_2'}}{x_2} {\ret (x_1,x_2) }}}}\\
     \ltdyn\\
       \bindXtoYinZ{\bullet}{p'}{\pmpairWtoXYinZ{p'}{x_1'}{x_2'}{
           \bindXtoYinZ{\dncast{\u F A_1'}{\u F A_1'}{\ret x_1'}}{x_1'}{
            \bindXtoYinZ{\dncast{\u F A_2'}{\u F A_2'}{\ret x_2'}}{x_2'} {\ret (x_1',x_2') }}}}
     \end{array}
     \]
     \end{small}
     which is true by the congruence rules for $\mathsf{bind}$,
     $\mathsf{split}$, downcasts, $\mathsf{ret}$, and pairing.

     Next, we need to show
     \begin{small}
     \[
     \bullet \ltdyn \bullet' : {\u F (A_1 \times A_2)} \ltdyn {\u F (A_1' \times A_2')} \vdash
     \defdncast{\u F (A_1 \times A_2)}{\u F (A_1 \times A_2)}{\bullet}
     \ltdyn
     \defdncast{\u F (A_1 \times A_2)}{\u F (A_1' \times A_2')}{\bullet'}
     : {\u F (A_1 + A_2)}
     \]
     \end{small}
     i.e.
     \begin{small}
     \[
     \begin{array}{c}
       \bindXtoYinZ{\bullet}{p}{\pmpairWtoXYinZ{p}{x_1}{x_2}{
           \bindXtoYinZ{\dncast{\u F A_1}{\u F A_1}{\ret x_1}}{x_1}{
            \bindXtoYinZ{\dncast{\u F A_2}{\u F A_2'}{\ret x_2}}{x_2} {\ret (x_1,x_2) }}}}\\
     \ltdyn\\
       \bindXtoYinZ{\bullet}{p'}{\pmpairWtoXYinZ{p'}{x_1'}{x_2'}{
           \bindXtoYinZ{\dncast{\u F A_1}{\u F A_1'}{\ret x_1'}}{x_1}{
            \bindXtoYinZ{\dncast{\u F A_2}{\u F A_2'}{\ret x_2'}}{x_2} {\ret (x_1,x_2) }}}}\\
     \end{array}
     \]
     \end{small}
     which is also true by congruence.

     Finally, for assumption (3), we show
     \begin{small}
       \[
       \begin{array}{lll}
       \bindXtoYinZ{\bullet}{p}{\pmpairWtoXYinZ{p}{x_1}{x_2}{
           \bindXtoYinZ{\dncast{\u F A_1}{\u F A_1}{\ret x_1}}{x_1}{
            \bindXtoYinZ{\dncast{\u F A_2}{\u F A_2'}{\ret x_2}}{x_2} {\ret (x_1,x_2) }}}} & \equidyn & \\
       \bindXtoYinZ{\bullet}{p}{\pmpairWtoXYinZ{p}{x_1}{x_2}{
           \bindXtoYinZ{{\ret x_1}}{x_1}{
            \bindXtoYinZ{{\ret x_2}}{x_2} {\ret (x_1,x_2) }}}} & \equidyn & \\
       \bindXtoYinZ{\bullet}{p}{\pmpairWtoXYinZ{p}{x_1}{x_2}{{\ret (x_1,x_2) }}} & \equidyn & \\
       \bindXtoYinZ{\bullet}{p}{\ret p} & \equidyn & \\
       \bullet \\
       \end{array}
       \]
     \end{small}
     using the downcast identity, $\beta$ for $\u F$ types, $\eta$ for
     eager products, and $\eta$ for $\u F$ types.  
     
     An analogous argument works if we sequence the downcasts of the
     components in the opposite order:
     \begin{small}
     \[
     \bindXtoYinZ{\bullet}{p'}{\pmpairWtoXYinZ{p'}{x_1'}{x_2'}{\bindXtoYinZ{\dncast{\u F A_2}{\u F A_2'}{\ret x_2'}}{x_2} {\bindXtoYinZ{\dncast{\u F A_1}{\u F A_1'}{\ret x_1'}}{x_1} {\ret (x_1,x_2) }}}}
     \]
     \end{small}
     (the only facts about downcasts used above are congruence and the
     downcast identity), which shows that these two implementations of
     the downcast are themselves equidynamic.  

  \item Lazy product downcast. 
    We use Lemma~\ref{lem:downcast} with 
    the type constructor $\u Y_1 \ctype, \u Y_2 \ctype \vdash \u Y_1 \with \u Y_2 \vtype$.
    Let
    \[\bullet' : \u B_1' \with \u B_2' \vdash \defdncast{\u B_1 \with \u B_2}{\u B_1 \with \u B_2}{\bullet'} : \u B_1 \with \u B_2
    \]
    stand for
    \begin{small}
    \[
    \pair{\dncast{\u B_1}{\u B_1'}{\pi \bullet'}}{\dncast{\u B_2}{\u B_2'}{\pi' \bullet'}}\\
    \]
    \end{small}
    which has the type required for the lemma's assumption (1).

    Assumption (2) requires two conditions, both of which are proved by
    the congruence rules for pairing, projection, and downcasts.  The first,
    \begin{small}
    \[\bullet' : \u B_1' \with \u B_2' \vdash \defdncast{\u B_1 \with \u B_2}{\u B_1' \with \u B_2'}{\bullet'} \ltdyn 
        \defdncast{\u B_1' \with \u B_2'}{\u B_1' \with \u B_2'}{\bullet'} : \u B_1 \with \u B_2 \ltdyn \u B_1' \with \u B_2'
    \]
    \end{small}
    expands to
    \begin{small}
    \[
    \begin{array}{c}
    \pair{\dncast{\u B_1}{\u B_1'}{\pi \bullet'}}{\dncast{\u B_2}{\u B_2'}{\pi' \bullet'}} \\
    \ltdyn \\
    \pair{\dncast{\u B_1'}{\u B_1'}{\pi \bullet'}}{\dncast{\u B_2'}{\u B_2'}{\pi' \bullet'}} \\
    \end{array}
    \]
    \end{small}
    The second, 
    \begin{small}
    \[
    \bullet \ltdyn \bullet' : \u B_1 \with \u B_2 \ltdyn \u B_1' \with \u B_2' \vdash
    \defdncast{\u B_1 \with \u B_2}{\u B_1 \with \u B_2}{\bullet} \ltdyn
    \defdncast{\u B_1 \with \u B_2}{\u B_1' \with \u B_2'}{\bullet'} : \u B_1 \with \u B_2
    \]
    \end{small}
    expands to
    \[
    \begin{array}{c}
    \pair{\dncast{\u B_1}{\u B_1}{\pi \bullet}}{\dncast{\u B_2}{\u B_2}{\pi' \bullet}} \\
    \ltdyn \\
    \pair{\dncast{\u B_1}{\u B_1'}{\pi \bullet'}}{\dncast{\u B_2}{\u B_2'}{\pi' \bullet'}} \\
    \end{array}
    \]

    For assumption (3), we have, using $\dncast{\u B}{\u B}$ is the
    identity and $\eta$ for $\with$,
    \[
    \pair{\dncast{\u B_1}{\u B_1}{\pi \bullet}}{\dncast{\u B_2}{\u B_2}{\pi' \bullet}}
    \equidyn
    \pair{{\pi \bullet}}{{\pi' \bullet}}
    \equidyn
    \bullet
    \]
    
  \item Lazy product upcast.
    
    We use Lemma~\ref{lem:upcast} with the type
    constructor $\u Y_1 \ctype, \u Y_2 \ctype \vdash U (\u Y_1 \with \u Y_2) \vtype$.
    Let
    \[p : U (\u B_1 \with \u B_2) \vdash \defupcast{U (\u B_1 \with \u B_2)}{U (\u B_1 \with \u B_2)}{p} : U (\u B_1' \with \u B_2')
    \]
    stand for
    \begin{small}
    \[
    \thunk{\pair{\force{(\upcast{U \u B_1}{U \u B_1'}{(\thunk{\pi (\force{p})})})}}{\force{(\upcast{U \u B_2}{U \u B_2'}{(\thunk{\pi' (\force{p})})})}}}
    \]
    \end{small}
    which has the type required for the lemma's assumption (1).  
    
    Assumption (2) requires two conditions, both of which are proved by
    the congruence rules for $\mathsf{thunk}$, $\mathsf{force}$,
    pairing, projections, and upcasts.  The first,
    \begin{small}
    \[p : U (\u B_1 \with \u B_2) \vdash \defupcast{U (\u B_1 \with \u B_2)}{U (\u B_1 \with \u B_2)}{p} \ltdyn \defupcast{U (\u B_1 \with \u B_2)}{U (\u B_1' \with \u B_2')}{p} : U (\u B_1 \with \u B_2) \ltdyn U (\u B_1' \with \u B_2')
    \]
    \end{small}
    expands to
    \begin{small}
    \[
    \begin{array}{c}
      \thunk{\pair{\force{(\upcast{U \u B_1}{U \u B_1}{(\thunk{\pi (\force{p})})})}}{\force{(\upcast{U \u B_2}{U \u B_2}{(\thunk{\pi' (\force{p})})})}}}\\
      \ltdyn \\
      \thunk{\pair{\force{(\upcast{U \u B_1}{U \u B_1'}{(\thunk{\pi (\force{p})})})}}{\force{(\upcast{U \u B_2}{U \u B_2'}{(\thunk{\pi' (\force{p})})})}}}
    \end{array}
    \]
    \end{small}
    The second,
    \begin{small}
    \[
    p \ltdyn p' : U (\u B_1 \with \u B_2) \ltdyn U (\u B_1' \with \u B_2') \vdash
    \defupcast{U (\u B_1 \with \u B_2)}{U (\u B_1' \with \u B_2')}{p} \ltdyn \defupcast{U (\u B_1' \with \u B_2')}{U (\u B_1' \with \u B_2')}{p} : U (\u B_1' \with \u B_2')
    \]
    \end{small}
    expands to
    \begin{small}
    \[
    \begin{array}{c}
      \thunk{\pair{\force{(\upcast{U \u B_1}{U \u B_1'}{(\thunk{\pi (\force{p})})})}}{\force{(\upcast{U \u B_2}{U \u B_2'}{(\thunk{\pi' (\force{p})})})}}}\\
      \ltdyn \\
      \thunk{\pair{\force{(\upcast{U \u B_1'}{U \u B_1'}{(\thunk{\pi (\force{p'})})})}}{\force{(\upcast{U \u B_2'}{U \u B_2'}{(\thunk{\pi' (\force{p'})})})}}}
    \end{array}
    \]
    \end{small}

    Finally, for assumption (3), using $\eta$ for $times$, $\beta$ and
    $\eta$ for $U$ types, and the fact that $\upcast{A}{A}{}$ is the
    identity, we have
    \begin{small}
    \[
    \begin{array}{rl}
      \thunk{\pair{\force{(\upcast{U \u B_1}{U \u B_1}{(\thunk{\pi (\force{p})})})}}{\force{(\upcast{U \u B_2}{U \u B_2}{(\thunk{\pi' (\force{p})})})}}} & \equidyn \\
      \thunk{\pair{\force{(\thunk{\pi (\force{p})})}}{\force{(\thunk{\pi' (\force{p})})}}} & \equidyn \\
      \thunk{\pair{\pi (\force{p})}{\pi' (\force{p})}} & \equidyn \\
      \thunk{(\force{p})} & \equidyn \\
      p
    \end{array}
    \]
    \end{small}

  \item Function downcast. 

    We use Lemma~\ref{lem:downcast} with 
    the type constructor $X \vtype, \u Y \ctype \vdash X \to \u Y \ctype$.
    Let
    \[\bullet' : A' \to \u B' \vdash \defdncast{A \to \u B}{A' \to \u B'}{\bullet'} : A \to \u B
    \]
    stand for
    \[
    \lambda{x}.{\dncast{\u B}{\u B'}{(\bullet \, (\upcast{A}{A'}{x}))}} \\
    \]
    which has the type required for the lemma's assumption (1).

    Assumption (2) requires two conditions, both of which are proved by
    the congruence rules for $\lambda$, application, upcasts, and
    downcasts.  The first,
    \begin{small}
    \[\bullet' : A' \to \u B' \vdash \defdncast{A \to \u B}{A' \to \u B'}{\bullet'} \ltdyn 
        \defdncast{A' \to \u B'}{A' \to \u B'}{\bullet'} : A \to \u B \ltdyn \u A' \to \u B'
    \]
    \end{small}
    expands to
    \[
    \begin{array}{c}
    \lambda{x}.{\dncast{\u B}{\u B'}{(\bullet \, (\upcast{A}{A'}{x}))}} \\
      \ltdyn \\
    \lambda{x'}.{\dncast{\u B'}{\u B'}{(\bullet \, (\upcast{A'}{A'}{x'}))}} \\
    \end{array}
    \]
    The second, 
    \begin{small}
    \[
    \bullet \ltdyn \bullet' : \u A \to \u B \ltdyn A' \to \u B' \vdash
    \defdncast{A \to \u B}{A \to \u B}{\bullet} \ltdyn
    \defdncast{A \to \u B}{A' \to \u B'}{\bullet'} : A \to \u B
    \]
    \end{small}
    expands to
    \[
    \begin{array}{c}
    \lambda{x}.{\dncast{\u B}{\u B}{(\bullet \, (\upcast{A}{A}{x}))}} \\
    \ltdyn \\
    \lambda{x}.{\dncast{\u B}{\u B'}{(\bullet' \, (\upcast{A}{A'}{x}))}} \\
    \end{array}
    \]

    For assumption (3), we have, using $\upcast{A}{A}$ and $\dncast{\u
      B}{\u B}$ are the identity and $\eta$ for $\to$,
    \[
    \lambda{x}.{\dncast{\u B}{\u B}{(\bullet \, (\upcast{A}{A}{x}))}} \\
    \equidyn
    \lambda{x}.{{(\bullet \, ({x}))}} \\
    \equidyn
    \bullet
    \]
    
  \item Function upcast.
    
    We use Lemma~\ref{lem:upcast} with the type
    constructor $\u X \vtype, \u Y \ctype \vdash U (\u X \to \u Y) \vtype$.
    Suppose $A \ltdyn A'$ as value types and $\u B \ltdyn \u B'$ as
    computation types and let
    \[p : U (A \to \u B) \vdash \defupcast{U (A \to \u B)}{U (A \to \u B)}{p} : U (A' \to \u B')
    \]
    stand for
    \begin{small}
    \[
      \thunk{(\lambda x'.\bindXtoYinZ{\dncast{\u F A}{\u F A'}{(\ret
            x')}}{x}{ \force{(\upcast{U \u B}{U \u B'}{(\thunk{(\force{(f)}\,x)})})}})}
    \]
    \end{small}
    which has the type required for the lemma's assumption (1).  
    
    Assumption (2) requires two conditions, both of which are proved by
    the congruence rules for $\mathsf{thunk}$, $\mathsf{force}$,
    functions, application, upcasts, and downcasts.  The first,
    \begin{small}
    \[
      f : U (A \to \u B) \vdash \defupcast{U (A \to \u B)}{U (A \to \u B)}{f} \ltdyn \defupcast{U (A \to \u B)}{U (A' \to \u B')}{f} : U (A \to \u B) \ltdyn U (A' \to \u B')
    \]
    \end{small}
    expands to
    \begin{small}
    \[
    \begin{array}{c}
      \thunk{(\lambda x.\bindXtoYinZ{\dncast{\u F A}{\u F A}{(\ret x)}}{x}{ \force{(\upcast{U \u B}{U \u B}{(\thunk{(\force{(f)}\,x)})})}})}\\
      \ltdyn \\
      \thunk{(\lambda x'.\bindXtoYinZ{\dncast{\u F A}{\u F A'}{(\ret x')}}{x}{ \force{(\upcast{U \u B}{U \u B'}{(\thunk{(\force{(f)}\,x)})})}})}
    \end{array}
    \]
    \end{small}
    The second,
    \begin{small}
    \[
    f \ltdyn f' : U (A \to \u B) \ltdyn U (A' \to \u B') \vdash
    \defupcast{U (A \to \u B)}{U (A' \to \u B')}{f} \ltdyn 
    \defupcast{U (A' \to \u B')}{U (A' \to \u B')}{f'} : U (A' \to \u B')
    \]
    \end{small}
    expands to
    \begin{small}
    \[
    \begin{array}{c}
      \thunk{(\lambda x'.\bindXtoYinZ{\dncast{\u F A}{\u F A'}{(\ret x')}}{x}{ \force{(\upcast{U \u B}{U \u B'}{(\thunk{(\force{(f)}\,x)})})}})}\\
      \ltdyn \\
      \thunk{(\lambda x'.\bindXtoYinZ{\dncast{\u F A'}{\u F A'}{(\ret x')}}{x'}{ \force{(\upcast{U \u B'}{U \u B'}{(\thunk{(\force{(f')}\,x')})})}})}
    \end{array}
    \]
    \end{small}

    Finally, for assumption (3), using $\eta$ for $\to$, $\beta$ for $F$
    types and 
    $\beta/\eta$ for $U$ types, and the fact that $\upcast{\u B}{\u B}{}$ and
    $\dncast{A}{A}$ are the identity, we have
    \begin{small}
    \[
    \begin{array}{rl}
      \thunk{(\lambda x.\bindXtoYinZ{\dncast{\u F A}{\u F A}{(\ret x)}}{x}{ \force{(\upcast{U \u B}{U \u B}{(\thunk{(\force{(f)}\,x)})})}})} & \equidyn \\
      \thunk{(\lambda x.\bindXtoYinZ{{(\ret x)}}{x}{\force{{(\thunk{(\force{(f)}\,x)})}}})} & \equidyn \\
      \thunk{(\lambda x.\force{{(\thunk{(\force{(f)}\,x)})}})} & \equidyn \\
      \thunk{(\lambda x.(\force{(f)}\,x))} & \equidyn \\
      \thunk{(\force{(f)})} & \equidyn \\
      f
    \end{array}
    \]
    \end{small}

  \item $z : 0 \vdash \upcast{0}{A}z \equidyn \absurd z : A$ is
    immediate by $\eta$ for 0 on the map $z : 0 \vdash \upcast{0}{A}z :
    A$.

  \end{enumerate}
\end{longproof}

\begin{longonly}
\subsubsection{Shifts}
\end{longonly}

In GTT, we assert the existence of value upcasts and computation
downcasts for derivable type dynamism relations.  While we do not assert
the existence of all \emph{value} downcasts and \emph{computation}
upcasts, we can define the universal property that identifies a term as
such:

\begin{definition}[Stack upcasts/value downcasts] \label{def:value-down-computation-up} ~
\begin{enumerate}
\item
  If $\u B \ltdyn \u B'$, a \emph{stack upcast from $B$ to $B'$}
  is a stack $\bullet : \u B \vdash \defupcast{\u B}{\u B'} \bullet : \u
  B'$ that satisfies the computation dynamism rules of an upcast
  ${\bullet : \u B \vdash \bullet \ltdyn \defupcast{\u B}{\u B'} \bullet
    : \u B \ltdyn \u B'}$ and
  ${\bullet \ltdyn \bullet' : \u B \ltdyn \u B' \vdash \defupcast{\u B}{\u B'} \bullet \ltdyn \bullet' : \u B'}$.
    
\item If $A \ltdyn A'$, a \emph{value downcast from $A'$ to $A$} is a
  complex value $x : A' \vdash \defdncast{A}{A'} x : A$ that satisfies
  the value dynamism rules of a downcast
  ${x : A' \vdash \defdncast{A}{A'}{x} \ltdyn x : A \ltdyn A'}$
  and
  ${x \ltdyn x' : A \ltdyn A' \vdash x \ltdyn \defdncast{A}{A'} x' : A}$.
\end{enumerate}
\end{definition}
\begin{longonly}
Because the proofs of Lemma~\ref{lem:cast-left-right},
Lemma~\ref{lem:cast-congruence}, Theorem~\ref{thm:decomposition},
Theorem~\ref{thm:casts-unique} rely only on the axioms for
upcasts/downcasts, the analogues of these theorems hold for stack
upcasts and value downcasts as well.
\end{longonly}
Some value downcasts and computation upcasts do exist, leading to a
characterization of the casts for the monad $U \u F A$ and comonad $\u F
U \u B$ of $F \dashv U$:
\begin{theorem}[Cast Unique Implementation Theorem for $U \u F, \u F U$] \label{thm:monadic-comonadic-casts}
  Let $A \ltdyn A'$ and $\u B \ltdyn \u B'$.  
  \begin{enumerate}
  \item $\bullet : \u F A \vdash
    \bindXtoYinZ{\bullet}{x:A}{\ret{(\upcast{A}{A'}{x})}} : \u F A'$
    is a stack upcast.
    
  \item If $\defupcast{\u B}{\u B'}$ is a stack upcast, then\\
    $x : \u U B \vdash \upcast{U \u B}{\u U B'}{x} \equidyn \thunk{(\defupcast{\u B}{\u B'}{(\force x)})} : U \u B'$
    
  \item $x : \u U B' \vdash \thunk{(\dncast{\u B}{\u B'}{(\force x)})} : U
    \u B$ is a value downcast.
    
  \item If $\defdncast{A}{A'}$ is a value downcast, then\\
    $\bullet : \u F A' \vdash \dncast{\u F A}{\u F A'}{\bullet} \equidyn \bindXtoYinZ{\bullet}{x':A'}{\ret{(\dncast{A}{A'}{x})}}$

  \item 
    $\begin{array}{c}
    x : U \u F A \vdash \upcast{U \u F A}{U \u F A'}{x} \equidyn \thunk{ (\bindXtoYinZ{{\force x}}{x:A}{\ret{(\upcast{A}{A'}{x})}})}\\
    \bullet : \u F U \u B' \vdash \dncast{\u F U \u B}{\u F U \u B'}{\bullet} \equidyn
    \bindXtoYinZ{\bullet}{x':U \u B'}{\ret{(\thunk{(\dncast{\u B}{\u B'}{(\force x)})})}}
    \end{array}$
  \end{enumerate}
\end{theorem}

\begin{longproof}
  \begin{enumerate}
  \item
    To show
    \[
    \bullet : \u F A \vdash \bullet \ltdyn
    \bindXtoYinZ{\bullet}{x:A}{\ret{(\upcast{A}{A'}{x})}} : \u F A
    \ltdyn \u F A'
    \]
    we can $\eta$-expand $\bullet \equidyn
    \bindXtoYinZ{\bullet}{x}{\ret{x}}$ on the left, at which point by
    congruence it suffices to show $x \ltdyn \upcast{A}{A'}{x}$, which
    is true up upcast right.  To show
    \[
    \bullet \ltdyn \bullet' : \u F A \ltdyn \u F A' \vdash 
    \bindXtoYinZ{\bullet}{x:A}{\ret{(\upcast{A}{A'}{x})}}
    \ltdyn
    \bullet'
    : \u F A'
    \]
    we can $\eta$-expand $\bullet' \equidyn
    \bindXtoYinZ{\bullet'}{x'}{\ret{x'}}$ on the right,
    and then apply congruence, the assumption that $\bullet \ltdyn
    \bullet'$, and upcast left.

  \item We apply the upcast lemma with the type constructor $\u Y \ctype
    \vdash U \u Y \vtype$.  The term $\thunk{(\defupcast{\u B}{\u
        B'}{(\force x)})}$ has the correct type for assumption (1).  For
    assumption (2), we show
    \[
    x : U \u B \vdash \thunk{(\defupcast{\u B}{\u B}{(\force x)})} \ltdyn
    \thunk{(\defupcast{\u B}{\u B'}{(\force x)})} : U \u B \ltdyn U \u B'
    \]
    by congruence for $\mathsf{thunk}$, $\defupcast{\u B}{\u B}$ (proved
    analogously to Lemma~\ref{lem:cast-congruence}), and $\mathsf{force}$.
    We show
    \[
    x \ltdyn x' : U \u B \ltdyn U \u B' \vdash
    \thunk{(\defupcast{\u B}{\u B'}{(\force x)})}
    \thunk{(\defupcast{\u B'}{\u B'}{(\force x')})} 
    : U \u B'
    \]
    by congruence as well.
    Finally, for assumption (3), we have
    \[
    \begin{array}{cc}
    \thunk{(\defupcast{\u B}{\u B}{(\force x)})} & \equidyn \\
    \thunk{({(\force x)})} & \equidyn \\
    x
    \end{array}
    \]
    using $\eta$ for $U$ types and the identity principle for 
    $\defupcast{\u B}{\u B}$ (proved analogously to
    Theorem~\ref{thm:decomposition}).  
    
  \item To show 
    \[
    x' : U \u B' \vdash \thunk{(\dncast{\u B}{\u B'}{(\force x')})} \ltdyn x' : U \u B \ltdyn U \u B'
    \]
    we can $\eta$-expand $x'$ to $\thunk{\force{x'}}$, and then by
    congruence it suffices to show $\dncast{\u B}{\u B'}{(\force x')}
    \ltdyn \force{x'} : \u B \ltdyn \u B'$, which is downcast left.
    Conversely, for 
    \[
    x \ltdyn x' : U \u B \ltdyn U \u B' \vdash x \ltdyn \thunk{(\dncast{\u B}{\u B'}{(\force x')})} : U \u B
    \]
    we $\eta$-expand $x$ to $\thunk{(\force{x})}$, and then it suffices
    to show $\dncast{\u B}{\u B'}{(\force{x})} \ltdyn \force{x'}$, which
    is true by downcast right and congruence of $\mathsf{force}$ on the
    assumption $x \ltdyn x'$.  
    
  \item We use the downcast lemma with $X \vtype \vdash \u F X \ctype$,
    where $\bindXtoYinZ{\bullet}{x':A'}{\ret{(\defdncast{A}{A'}{x})}}$
    has the correct type for assumption (1).  For assumption (2), we
    show
    \[
    \bullet : \u F A' \vdash
     \bindXtoYinZ{\bullet}{x':A'}{\ret{(\defdncast{A}{A'}{x})}}
    \ltdyn 
    \bindXtoYinZ{\bullet}{x':A'}{\ret{(\defdncast{A'}{A'}{\bullet})}}
    \]
    by congruence for $\mathsf{bind}$, $\mathsf{ret}$, and
    $\defdncast{A'}{A'}$ (which is proved analogously to
    Lemma~\ref{lem:cast-congruence}).
    We also show
    \[
    \bullet \ltdyn \bullet' : \u F A \ltdyn \u F A' \vdash
    \bindXtoYinZ{\bullet}{x:A}{\ret{(\defdncast{A}{A}{x})}}
    \ltdyn 
    \bindXtoYinZ{\bullet}{x':A'}{\ret{(\defdncast{A}{A'}{\bullet'})}}
    : \u F A
    \]
    by congruence.
    Finally, for assumption (3), we have
    \[
    \begin{array}{rc}
      \bindXtoYinZ{\bullet}{x:A}{\ret{(\defdncast{A}{A}{x})}} & \equidyn \\
      \bindXtoYinZ{\bullet}{x:A}{\ret{({x})}} & \equidyn \\
      \bullet
    \end{array}
    \]
    using the identity principle for $\defdncast{A}{A}$ (proved
    analogously to Theorem~\ref{thm:decomposition}) and $\eta$ for $F$
    types.
    
  \item Combining parts (1) and (2) gives the first equation, while
    combining parts (3) and (4) gives the second equation.
  \end{enumerate}
\end{longproof}

\begin{longonly}
  \subsubsection{Derived Rules for Call-by-value Function Types}
\end{longonly}

Recall that for value types $A_1$ and $A_2$, the CBV function type is
$U(A_1 \to \u F A_2)$.  As a corollary of
Theorems~\ref{thm:functorial-casts} and
\ref{thm:monadic-comonadic-casts}, we have
\begin{corollary}[Cast Unique Implementation for CBV Functions]
  \[
  \begin{small}
  \begin{array}{l}
    \begin{array}{rcll}
      \upcast{U(A_1 \to \u F A_2)}{U(A_1' \to \u F A_2')}{f} & \equidyn &
      \thunk (\lambda x'. & \bindXtoYinZ{\dncast{\u F A_1}{\u F A_1'}{(\ret x')}}{x}{} \\
                          & & & \bindXtoYinZ{(\force{(f)}\,x)}{y}\\ & & & {\ret{(\upcast{A_2}{A_2'}{y})}})\\
    \end{array}
    \\
    \begin{array}{rcl}
    \dncast{\u F U(A_1 \to \u F A_2)}{\u F U(A_1' \to \u F A_2')}{\bullet} & \equidyn &
    \bindXtoYinZ{\bullet}{f}\\
    & & {\ret{\lambda{x}.{\dncast{\u F A_2}{\u F A_2'}{(\force{(f)} \, (\upcast{A_1}{A_1'}{x}))}}}}
    \end{array}
  \end{array}
  \end{small}
  \]
\end{corollary}

\begin{longproof}
  For the upcast, by Theorem~\ref{thm:functorial-casts}, it's equal to
  \[
      \thunk (\lambda x'. \bindXtoYinZ{\dncast{\u F A_1}{\u F A_1'}{(\ret x')}}{x}{} 
         { \force{(\upcast{U \u F A_2}{U \u F A_2'}{(\thunk{(\force{(f)}\,x)})})}} )
         \]
  By Theorem~\ref{thm:monadic-comonadic-casts}, $\upcast{U \u F A_2}{U \u F A_2'}$ is equal to
  \[
  \thunk{ (\bindXtoYinZ{{\force -}}{x}{\ret{(\upcast{A_2}{A_2'}{x})}})}
  \]
  so $\beta$-reducing $\force$ and $\thunk$ twice gives the result.

  For the downcast, by Theorem~\ref{thm:monadic-comonadic-casts}, it's
  equal to
  \[
  \bindXtoYinZ{\bullet}{x}{\ret{(\thunk{(\dncast{\u (A_1 \to \u F A_2)}{\u (A_1 \to \u F A_2)}{(\force x)})})}}
  \]
  and by Theorem~\ref{thm:functorial-casts} $\dncast{\u (A_1 \to \u F A_2)}{\u (A_1 \to \u F A_2)}{-}$ is equal to
  \[
   \lambda{x}.{\dncast{\u F A_2}{\u F A_2'}{(- \, (\upcast{A_1}{A_1'}{x}))}}
  \]
\end{longproof}
These are equivalent to the CBPV translations of the standard CBV wrapping
implementations; for example, the CBV upcast term
$\lambda x'.\letXbeYinZ{\dncast{A_1}{A_1'}{x'}}{x}{\upcast{A_2}{A_2'}{(f x')}}$
has its evaluation order made explicit, and the fact that its upcast is
a (complex) value exposed.  In the downcast, the GTT term is free to
let-bind $(\upcast{A_1}{A_1'}{x})$ to avoid duplicating it, but because
it is a (complex) value, it can also be substituted directly, which
might expose reductions that can be optimized.

\begin{longonly}
\subsection{Least Dynamic Types}

\begin{longonly}
\begin{theorem}[Least Dynamic Value Type]
  If $\leastdynv$ is a type such that $\leastdynv \ltdyn A$ for all $A$,
  then in GTT with a strict initial object $0$, $\leastdynv \cong_{v}
  0$.
\end{theorem}
\begin{proof}
We have the upcast $x : \leastdynv \vdash \upcast{\leastdynv}{0}{x} :
0$, so Lemma~\ref{lem:initial} gives the result.
\end{proof}
The fact that $\leastdynv$ is strictly initial seems to depend on the
fact that we have a strictly initial object: In GTT without a $0$ type,
it seems that we cannot prove that $x : \leastdynv \vdash
\upcast{\leastdynv}{A}{x} : A$ is the unique such map.  

\begin{theorem}[Least Dynamic Computation Type]
If $\leastdync$ is a type such that $\leastdync \ltdyn \u B$ for all $\u
B$, and we have a terminal computation type $\top$, then $U \leastdync
\cong_{v} U \top$.
\end{theorem}
\begin{proof}
We have stacks $\bullet : \top \dncast{\leastdync}{\top}{\bullet} :
\leastdync$ and $\bullet : \leastdync \vdash \{\} : \top$.  The
composite at $\top$ is the identity by Lemma~\ref{lem:terminal}.  However,
because $\top$ is not a strict terminal object, the dual of the above
argument does not give a stack isomorphism $\leastdync \cong_c \top$.

However, using the retract axiom, we have
\[
\begin{array}{c}
  x : U \leastdync \vdash \upcast{U \leastdync}{U \top}{x} : U \top\\
  y : U \top \vdash \thunk{(\dncast{\leastdync}{\top}{(\force{x})})} : U \leastdync\\
  x : U \leastdync \vdash \thunk{(\dncast{\leastdync}{\top}{(\force{(\upcast{U \leastdync}{U \top}{x})})})} \equidyn x : U \leastdync
\end{array}
\]
and the composite 
\[
y : U \top \vdash \upcast{U \leastdync}{U \top}{(\thunk{(\dncast{\leastdync}{\top}{(\force{x})})})} : U \top
\]
is the identity by uniqueness for $U \top$ (Lemma~\ref{lem:terminal}).  
\end{proof}

This suggests taking $\bot_v := 0$ and $\bot_c := \top$.
\end{longonly}

\begin{theorem}
The casts determined by $0 \ltdyn A$ are
\[
\upcast{0}{A}z \equidyn \absurd z \qquad \dncast{\u F 0}{\u F A}{\bullet} \equidyn \bindXtoYinZ{\bullet}{\_}{\err}
\]
Dually, the casts determined by $\top \ltdyn \u B$ are
\[
\dncast{\top}{\u B}{\bullet} \equidyn \{\} \qquad \upcast{U \top}{U \u B}{u} \equidyn \thunk \err
\]
\end{theorem}

\begin{longproof}
  \begin{enumerate}
  \item $x : 0 \vdash \upcast{0}{A}{x} \equidyn \abort{x} : A$ is
    immediate by $\eta$ for $0$.

  \item First, to show
    $\bullet : \u F A \vdash \bindXtoYinZ{\bullet}{\_}{\err} \ltdyn \dncast{\u F 0}{\u F A}{\bullet}$,
    we can $\eta$-expand the right-hand side into
    $\bindXtoYinZ{\bullet}{x:A}{\dncast{\u F 0}{\u F A}{\ret{x}}}$,
    at which point the result follows by congruence
    and the fact that type error is minimal, so
    $\err \ltdyn {\dncast{\u F 0}{\u F A}{\ret{x}}}$.  

    Second, to show
    $\bullet : \u F A \vdash \dncast{\u F 0}{\u F A}{\bullet} \ltdyn \bindXtoYinZ{\bullet}{\_}{\err}$,
    we can $\eta$-expand the left-hand side to
    $\bullet : \u F A \vdash \bindXtoYinZ{\dncast{\u F 0}{\u F A}{\bullet}}{y}{\ret y}$,
    so we need to show
    \[
    \bullet: \u F A \vdash \bindXtoYinZ{\dncast{\u F 0}{\u F A}{\bullet}}{y:0}{\ret y} \ltdyn \bindXtoYinZ{\bullet}{y':A}{\err} : \u F 0
    \]
    We apply congruence, with $\bullet : \u F A \vdash {\dncast{\u F
        0}{\u F A}{\bullet}} \ltdyn \bullet : 0 \ltdyn A$ by the
    universal property of downcasts in the first premise, so it suffices
    to show
    \[
    y \ltdyn y' : 0 \ltdyn A \vdash \ret{y} \ltdyn \err_{\u F 0} : \u F 0
    \]
    By transitivity with $y \ltdyn y' : 0 \ltdyn A \vdash \err_{\u F 0}
    \ltdyn \err_{\u F 0} : \u F 0 \ltdyn \u F 0$, it suffices to show
    \[
    y \ltdyn y : 0 \ltdyn 0 \vdash \ret{y} \ltdyn \err_{\u F 0} : \u F 0
    \]
    But now both sides are maps out of $0$, and therefore equal by
    Lemma~\ref{lem:initial}.

  \item The downcast is immediate by $\eta$ for $\top$,
    Lemma~\ref{lem:terminal}.  

  \item First,
    \[
    u : U \top \vdash \thunk \err \ltdyn \thunk{(\force{(\upcast{U \top}{U \u B}{u})})} \equidyn {\upcast{U \top}{U \u B}{u}} : U \u B
    \]
    by congruence, $\eta$ for $U$, and the fact that error is minimal.
    Conversely, to show
    \[
    u : U \top \vdash {\upcast{U \top}{U \u B}{u}} \ltdyn \thunk \err  : U \u B
    \]
    it suffices to show
    \[
    u : U \top \vdash u \ltdyn \thunk \err_{\u B}  : U \top \ltdyn U \u B
    \]
    by the universal property of an upcast.  By Lemma~\ref{lem:terminal},
    any two elements of $U \top$ are equidynamic, so in particular $u
    \equidyn \thunk{\err_{\top}}$, at which point congruence for
    $\mathsf{thunk}$ and $\err_\top \ltdyn \err_{\u B } : \top \ltdyn \u
    B$ gives the result.
  \end{enumerate}
\end{longproof}
\end{longonly}

\subsection{Upcasts are Values, Downcasts are Stacks}
\label{sec:upcasts-necessarily-values}

Since GTT is an axiomatic theory, we can consider different fragments
than the one presented in Section~\ref{sec:gtt}.  Here, we use this
flexibility to show that taking upcasts to be complex values and
downcasts to be complex stacks is forced if this property holds for
casts between \emph{ground} types and $\dynv$/$\dync$.  For this section, we define a \emph{ground
  type}\footnote{In gradual
  typing, ``ground'' is used to mean a one-level unrolling of a dynamic type, not first-order data.} to be generated by the following grammar:
  \[
    G ::= 1 \mid \dynv \times \dynv \mid 0 \mid \dynv + \dynv \mid U \dync
    \qquad
    \u G ::= \dynv \to \dync \mid \top \mid \dync \with \dync \mid \u F \dynv
  \]

\begin{longonly}
\begin{definition}[Ground type dynamism]
  Let $A \ltdyn' A'$ and $\u B \ltdyn' \u B'$ be the relations defined
  by the rules in Figure~\ref{fig:gtt-type-dynamism}
  with the axioms $A \ltdyn \dynv$ and $\u B \ltdyn \dync$ restricted to
  ground types---i.e., replaced by $G \ltdyn \dynv$ and $\u G \ltdyn \dync$.
\end{definition}

\begin{lemma} \label{lem:find-ground-type}
  For any type $A$, $A \ltdyn' \dynv$.
  For any type $\u B$, $\u B \ltdyn' \dync$.
\end{lemma}
\begin{proof}
By induction on the type.  For example, in the case for $A_1 + A_2$, we
have by the inductive hypothesis $A_1 \ltdyn' \dynv$ and $A_2 \ltdyn'
\dynv$, so $A_1 + A_2 \ltdyn' \dynv + \dynv \ltdyn \dynv$ by congruence
and transitivity, because $\dynv + \dynv$ is ground.  In the case for
$\u F A$, we have $A \ltdyn \dynv$ by the inductive hypothesis, so $\u F
A \ltdyn \u F \dynv \ltdyn \dync$.
\end{proof}

\begin{lemma}[$\ltdyn$ and $\ltdyn'$ agree]
  $A \ltdyn A'$ iff $A \ltdyn' A'$ and $\u B \ltdyn \u B'$ iff $\u B
  \ltdyn' \u B'$
\end{lemma}

\begin{proof}
The ``if'' direction is immediate by induction because every rule of
$\ltdyn'$ is a rule of $\ltdyn$.  To show $\ltdyn$ is contained in
$\ltdyn'$, we do induction on the derivation of $\ltdyn$, where every
rule is true for $\ltdyn'$, except $A \ltdyn \dynv$ and $\u B \ltdyn
\dync$, and for these, we use Lemma~\ref{lem:find-ground-type}.
\end{proof}
\end{longonly}

Let GTT$_G$ be the fragment of GTT where the only primitive casts are
those between ground types and the dynamic types, i.e. the cast terms
are restricted to the substitution closures of
\[
\begin{small}
\begin{array}{llll}
x : G \vdash \upcast{G}{\dynv}{x} : \dynv &
\bullet : \u F \dynv \vdash \dncast{\u F G}{\u F \dynv}{\bullet} : \u F \dynv &
\bullet : \dync \vdash \dncast{\u G}{\dync}{\bullet} : \dync &
x : U \u G \vdash \upcast{U \u G}{U \dync}{x} : U \dync 
\end{array}
\end{small}
\]

\begin{lemma}[Casts are Admissible] \label{lem:casts-admissible}
In GTT$_G$ it is admissible that
\begin{enumerate}
\item for all $A \ltdyn A'$
  there is a complex value $\defupcast{A}{A'}$
  satisfying the universal property of an upcast
  and a complex stack $\defdncast{\u F A}{\u F A'}$
  satisfying the universal property of a downcast
\item for all $\u B \ltdyn \u B'$ there is a complex
  stack $\defdncast{\u B}{\u B'}$ satisfying the universal property of a
  downcast and a complex value $\defupcast{U \u B}{U \u B'}$ satisfying
  the universal property of an upcast.
\end{enumerate}
\end{lemma}

\begin{proof}
  To streamline the exposition above, we stated
  Theorems~\ref{thm:decomposition}, Theorem~\ref{thm:functorial-casts}
  Theorem~\ref{thm:monadic-comonadic-casts} as showing that the
  ``definitions'' of each cast are equidynamic with the cast that is a
  priori postulated to exist (e.g. $\upcast{A}{A''} \equidyn
  \upcast{A'}{A''}{\upcast{A}{A'}}$).  However, the proofs
  \begin{longonly}
    factor
  through Theorem~\ref{thm:casts-unique} and Lemma~\ref{lem:upcast} and
  Lemma~\ref{lem:downcast}, which
  \end{longonly}
  show directly that the right-hand sides have the desired universal
  property---i.e. the stipulation that some cast with the correct
  universal property exists is not used in the proof that the
  implementation has the desired universal property.  Moreover, the
  proofs given do not rely on any axioms of GTT besides the universal
  properties of the ``smaller'' casts used in the definition and the
  $\beta\eta$ rules for the relevant types.  So these proofs can be used
  as the inductive steps here, in GTT$_G$.
  \begin{shortonly}
  In the extended version we define an alternative type dynamism
  relation where casts into dynamic types are factored through ground
  types, and use that to drive the induction here.  
  \end{shortonly}
  \begin{longonly}
  By induction on type dynamism $A \ltdyn' A'$ and $\u B \ltdyn' \u B'$.

  (We chose not to make this more explicit above, because we believe the
    equational description in a language with all casts is a clearer
    description of the results, because it avoids needing to hypothesize
    terms that behave as the smaller casts in each case.)

  We show a few representative cases:

  In the cases for $G \ltdyn \dynv$ or $\u G \ltdyn \dync$, we have
  assumed appropriate casts $\upcast{G}{\dynv}$ and 
  $\dncast{\u F G}{\u F \dynv}$ and 
  $\dncast{\u G}{\dync}$ and
  $\upcast{U \u G}{U \dync}$.
  
  In the case for identity $A \ltdyn A$, we need to show that there is
  an upcast $\defupcast{A}{A}$ and a downcast $\defdncast{\u F A}{\u F A}$
  The proof of Theorem~\ref{thm:decomposition} shows that the identity
  value and stack have the correct universal property.  

  In the case where type dynamism was concluded by
  transitivity between $A \ltdyn A'$ and $A' \ltdyn A''$, by the
  inductive hypotheses we get upcasts $\defupcast{A}{A'}$ and
  $\defupcast{A'}{A''}$, and the proof of
   Theorem~\ref{thm:decomposition} shows that defining
  $\defupcast{A}{A''}$ to be $\defupcast{A'}{A''}{\defupcast{A}{A'}}$
  has the correct universal property.  For the downcast, we get
  $\defdncast{\u F A}{\u F A'}$ and
  $\defdncast{\u F A'}{\u F A''}$ by the inductive hypotheses, and the
  proof of Theorem~\ref{thm:decomposition} shows that their composition
  has the correct universal property.
  
  In the case where type dynamism was concluded by the congruence rule
  for $A_1 + A_2 \ltdyn A_1' + A_2'$ from $A_i \ltdyn A_i'$, we have
  upcasts $\defupcast{A_i}{A_i'}$ and downcasts $\defdncast{\u F A_i}{\u
    F A_i'}$ by the inductive hypothesis, and the proof of
  Theorem~\ref{thm:decomposition} shows that the definitions given there
  have the desired universal property.
  
  In the case where type dynamism was concluded by the congruence rule
  for $\u F A \ltdyn \u F A'$ from $A \ltdyn A'$, we obtain by induction
  an \emph{upcast} $A \ltdyn A'$ and a downcast $\defdncast{\u F A}{\u F A'}$.
  We need a 
  \emph{downcast} $\defdncast{\u F A}{F A'}$, which we have,
  and an \emph{upcast} $\defdncast{U \u F A}{U \u F A'}$, which is
  constructed as in Theorem~\ref{thm:monadic-comonadic-casts}.
  \end{longonly}
\end{proof}

As discussed in Section~\ref{sec:gtt-casts}, rather than an upcast being
a complex value $x : A \vdash \upcast{A}{A'}{x} : A'$, an a priori more
general type would be a stack $\bullet : \u F A \vdash \upcast{\u F
  A}{\u F A'}{\bullet} : \u F A'$, which allows the upcast to perform
effects; dually, an a priori more general type for a downcast $\bullet :
\u B' \vdash \dncast{\u B}{\u B'}{\bullet} : \u B$ would be a value $x :
U \u B' \vdash \dncast{U \u B}{U \u B'}{x} : U \u B$, which allows the
downcast to ignore its argument.  The following shows that in GTT$_G$,
if we postulate such stack upcasts/value downcasts as originally
suggested in Section~\ref{sec:gtt-casts}, then in fact these casts
\emph{must} be equal to the action of $U$/$\u F$ on some
value upcasts/stack downcasts, so the potential
for (co)effectfulness affords no additional flexibility.

\begin{theorem}[Upcasts are Necessarily Values, Downcasts are Necessarily Stacks]
  \label{thm:upcasts-values-downcasts-stacks}
  Suppose we extend GTT$_G$ with the following postulated stack upcasts
  and value downcasts (in the sense of
  Definition~\ref{def:value-down-computation-up}): For every type
  precision $A \ltdyn A'$, there is a stack upcast $\bullet : \u F A
  \vdash \upcast{\u F A}{\u F A'}{\bullet} : \u F A'$, and for every $\u
  B \ltdyn \u B'$, there is a complex value downcast $x : U \u B' \vdash
  \dncast{U \u B}{U \u B'}{x} : U \u B$.

  Then there exists a value upcast $\defupcast{A}{A'}$ and a stack
  downcast $\defdncast{\u B}{\u B'}$ such that
  \[
  \begin{array}{c}
  \bullet : \u F A \vdash \upcast{\u F A}{\u F A'}{\bullet} \equidyn { (\bindXtoYinZ{{\bullet}}{x:A}{\ret{(\defupcast{A}{A'}{x})}})}\\
  x : U \u B' \vdash \dncast{U \u B}{U \u B'}{x} \equidyn {(\thunk{(\defdncast{\u B}{\u B'}{(\force x)})})}
  \end{array}
  \]
\end{theorem}

\begin{proof}
Lemma~\ref{lem:casts-admissible} constructs $\defupcast{A}{A'}$ and
$\defdncast{\u B}{\u B'}$, so the proof of
Theorem~\ref{thm:monadic-comonadic-casts} (which really works for any
$\defupcast{A}{A'}$ and $\defdncast{\u B}{\u B'}$ with the correct
universal properties, not only the postulated casts) implies that the
right-hand sides of the above equations are stack upcasts and value
downcasts of the appropriate type.  Since stack upcasts/value downcasts
are unique by an argument analogous to Theorem~\ref{thm:casts-unique},
the postulated casts must be equal to these.
\end{proof}

\begin{longonly}
Indeed, the following a priori even more general assumption provides no
more flexibility:
\begin{theorem}[Upcasts are Necessarily Values, Downcasts are Necessarily Stacks II]
  Suppose we extend GTT$_G$ only with postulated monadic upcasts $x : U
  \u F A \vdash \upcast{U \u F A}{U \u F A'}{x} : U \u F A'$ for every
  $A \ltdyn A'$ and comonadic downcasts $\bullet : \u F U \u B' \vdash
  \dncast{\u F U \u B}{\u F U \u B'}{\bullet} : \u F U \u B$ for every
  $\u B \ltdyn \u B'$.

  Then there exists a value upcast $\defupcast{A}{A'}$ and a stack
downcast $\defdncast{\u B}{\u B'}$ such that
   \[
   \begin{array}{c}
   x : U \u F A \vdash \upcast{U \u F A}{U \u F A'}{x} \equidyn \thunk{ (\bindXtoYinZ{{\force x}}{x:A}{\ret{(\defupcast{A}{A'}{x})}})}\\
   \bullet : \u F U \u B' \vdash \dncast{\u F U \u B}{\u F U \u B'}{\bullet} \equidyn
   \bindXtoYinZ{\bullet}{x':U \u B'}{\ret{(\thunk{(\defdncast{\u B}{\u B'}{(\force x)})})}}
   \end{array}
   \]
\end{theorem}
In CBV terms, the monadic upcast is like an upcast from $A$ to $A'$
taking having type $(1 \to A) \to A'$, i.e. it takes a thunked
effectful computation of an $A$ as input and produces an effectful
computation of an $A'$.  

\begin{proof}
Again, Lemma~\ref{lem:casts-admissible} constructs $\defupcast{A}{A'}$
and $\defdncast{\u B}{\u B'}$, so the proof of part (5) of
Theorem~\ref{thm:monadic-comonadic-casts} gives the result.
\end{proof}
\end{longonly}

\begin{longonly}
\subsection{Equidynamic Types are Isomorphic}

\begin{theorem}[Equidynamism implies Isomorphism] 
  \begin{enumerate}
  \item
    If $A \ltdyn A'$ and $A' \ltdyn A$ then $A \cong_v A'$.
  \item
    If $\u B \ltdyn \u B'$ and $\u B' \ltdyn \u B$ then $\u B \cong_c \u B'$.
  \end{enumerate}
\end{theorem}
\begin{proof}
  \begin{enumerate}
  \item We have upcasts $x : A \vdash \upcast{A}{A'}{x} : A'$ and $x' : A' \vdash \upcast{A'}{A}{x'} : A$.
    For the composites, to show
    $x : A \vdash \upcast{A'}{A}{\upcast{A}{A'}{x}} \ltdyn x$
    we apply upcast left twice, and conclude $x \ltdyn x$ by assumption.
    To show, 
    $x : A \vdash x \ltdyn \upcast{A'}{A}{\upcast{A}{A'}{x}}$,
    we have $x : A \vdash x \ltdyn {\upcast{A}{A'}{x}} : A \ltdyn A'$
    by upcast right, and therefore
    $x : A \vdash x \ltdyn \upcast{A'}{A}{\upcast{A}{A'}{x}} : A \ltdyn A$
    again by upcast right.
    The other composite is the same proof with $A$ and $A'$ swapped.
    
  \item We have downcasts $\bullet : \u B \vdash \dncast{\u B}{\u B'}{\bullet} : \u B'$ and
    $\bullet : \u B' \vdash \dncast{\u B'}{\u B}{\bullet} : \u B$.

      For the composites, to show $\bullet : \u B' \vdash \bullet \ltdyn
      \dncast{\u B'}{\u B}{\dncast{\u B}{\u B'}}{\bullet}$, we apply
      downcast right twice, and conclude $\bullet \ltdyn \bullet$.  For
      $\dncast{\u B'}{\u B}{\dncast{\u B}{\u B'}}{\bullet} \ltdyn
      \bullet$, we first have $\dncast{\u B}{\u B'}{\bullet} \ltdyn
      \bullet : \u B \ltdyn \u B'$ by downcast left, and then the result
      by another application of downcast left.
      The other composite is the same proof with $\u B$ and $\u B'$ swapped.
  \end{enumerate}
\end{proof}
\end{longonly}




\section{Contract Models of GTT}
\label{sec:contract}

To show the soundness of our theory, and demonstrate its relationship to
operational definitions of observational equivalence and the gradual
guarantee, we develop \emph{models} of GTT using observational error
approximation of a \emph{non-gradual} CBPV.
We call this the \emph{contract translation} because it translates the
built-in casts of the gradual language into ordinary terms implemented
in a non-gradual language.
While contracts are typically implemented in a dynamically typed
language, our target is typed, retaining type information similarly to
manifest contracts \cite{greenberg-manifest}.
We give implementations of the dynamic value type in the usual way as
a recursive sum of basic value types, i.e., using type tags, and we
give implementations of the dynamic computation type as the dual: a
recursive product of basic computation types.

Writing $\sem{M}$ for any of the contract translations, the remaining
sections of the paper establish:
\begin{theorem}[Equi-dynamism implies Observational Equivalence]
  If $\Gamma \vdash M_1 \equidyn M_2 : \u B$, then for any closing
  GTT context $C : (\Gamma \vdash \u B) \Rightarrow (\cdot \vdash \u F
  (1+1))$, $\sem{C[M_1]}$ and $\sem{C[M_2]}$ have the same behavior: both diverge,
  both run to an error, or both run to $\tru$ or both run to $\fls$.
\end{theorem}
\begin{theorem}[Graduality]
  If $\Gamma_1 \ltdyn \Gamma_2 \vdash M_1 \ltdyn M_2 : B_1 \ltdyn B_2$,
  then for any GTT context $C : (\Gamma_1 \vdash B_1) \Rightarrow (\cdot
  \vdash \u F (1+1))$, and any valid interpretation of the dynamic
  types, either
  \begin{enumerate}
  \item $\sem{C[M_1]} \Downarrow \err$, or
  \item $\sem{C[M_1]} \Uparrow$ and $\sem{C[\dncast{B_1}{B_2}M_2[\upcast{\Gamma_1}{\Gamma_2}{\Gamma_1}]]} \Uparrow$, or
  \item $\sem{C[M_1]} \Downarrow \ret V$,~~
    $\sem{C[\dncast{B_1}{B_2}M_2[\upcast{\Gamma_1}{\Gamma_2}{\Gamma_1}]]} 
    \Downarrow \ret V$, and $V = \tru$ or $V = \fls$.
  \end{enumerate}
\end{theorem}

As a consequence we will also get consistency of our logic of
dynamism:
\begin{corollary}[Consistency \iflong of GTT \fi]
  $\cdot \vdash \ret \kw{true} \ltdyn \ret \kw{false} : \u F(1+1)$ is not
  provable in GTT.
\end{corollary}
\begin{longproof}
  They are distinguished by the identity context.
\end{longproof}

We break down this proof into 3 major steps.
\begin{enumerate}
\item (This section) We translate GTT into a statically typed \cbpvstar\/
  language where the casts of GTT are translated to ``contracts'' in
  GTT: i.e., CBPV terms that implement the runtime type checking. We
  translate the term dynamism of GTT to an inequational theory for CBPV.
  Our translation is parameterized by the implementation of the dynamic
  types, and we demonstrate two valid implementations, one more direct
  and one more Scheme-like.
\item (Section \ref{sec:complex}) Next, we eliminate all uses of complex
  values and stacks from the CBPV language. We translate the complex
  values and stacks to terms with a proof that they are ``pure''
  (thunkable or linear~\cite{munchmaccagnoni14nonassociative}). This part has little to do with GTT
  specifically, except that it shows the behavioral property that
  corresponds to upcasts being complex values and downcasts being
  complex stacks.
\item (Section \ref{sec:lr}) Finally, with complex values and stacks
  eliminated, we give a standard operational semantics for CBPV and
  define a \emph{logical relation} that is sound and complete with
  respect to observational error approximation. Using the logical
  relation, we show that the inequational theory of CBPV is sound for
  observational error approximation.
\end{enumerate}

By composing these, we get a model of GTT where equidynamism is sound
for observational equivalence and an operational semantics that
satisfies the graduality theorem.

\subsection{Call-by-push-value}
\label{sec:cbpvstar}

Next, we define the call-by-push-value language \cbpvstar\ that will be
the target for our contract translations of GTT.
\cbpvstar\ is the axiomatic version of call-by-push-value \emph{with}
complex values and stacks, while \cbpv\ (Section~\ref{sec:complex}) will
designate the operational version of call-by-push-value with only
operational values and stacks.
\cbpvstar\ is almost a subset of GTT obtained as follows: We remove the
casts and the dynamic types $\dynv, \dync$ (the shaded pieces) from the
syntax and typing rules in Figure~\ref{fig:gtt-syntax-and-terms}.  There
is no type dynamism, and the inequational theory of \cbpv* is the
homogeneous fragment of term dynamism in
Figure~\ref{fig:gtt-term-dynamism-structural}\iflong\ and Figure~\ref{fig:gtt-term-dynamism-ext-congruence}\fi\ (judgements $\Gamma \vdash
E \ltdyn E' : T$ where $\Gamma \vdash E,E' : T$, with all the same rules
in that figure thus restricted).  The inequational axioms are the
Type Universal Properties ($\beta\eta$ rules)
and Error Properties (with \textsc{ErrBot} made homogeneous) from 
Figure~\ref{fig:gtt-term-dyn-axioms}.
To implement the casts and dynamic types, we \emph{add} general
\emph{recursive} value types ($\mu X.A$, the fixed point of $X \vtype
\vdash A \vtype$) and \emph{corecursive} computation types ($\nu \u Y.\u
B$, the fixed point of $\u Y \ctype \vdash \u B \ctype$).
The recursive type $\mu X.A$ is a value type with constructor
$\texttt{roll}$, whose eliminator is pattern matching, whereas the
corecursive type $\nu \u Y.\u B$ is a computation type defined by its
eliminator (\texttt{unroll}), with an introduction form that we also write
as \texttt{roll}.
We extend the inequational theory with monotonicity of each term constructor of
the recursive types, and with their $\beta\eta$ rules.
\begin{shortonly}
The rules for recursive types are in the extended version.
\end{shortonly}

\begin{longonly}
  In the following figure, we write $\bnfadd$ and $\bnfsub$ to indicate
  the diff from the grammar in Figure~\ref{fig:gtt-syntax-and-terms}.

\begin{figure}[h]
\begin{small}
  \[
  \begin{array}{lrcl}
      \text{Value Types} & A & \bnfadd & \mu X. A \alt X\\
      &   & \bnfsub & \dynv \\
      \text{Computation Types} & \u B & \bnfadd & \nu \u Y. \u B \alt \u Y\\
      &      & \bnfsub & \dync\\
      \text{Values} & V & \bnfadd & \rollty{\mu X.A} V\\
      &   & \bnfsub & \upcast{A}{A} V\\
      \text{Terms} & M & \bnfadd & \rollty{\nu \u Y. \u B} M \alt \unroll M\\
      & M & \bnfsub & \dncast{\u B}{\u B}M\\
      \text{Both}  & E & \bnfadd & \pmmuXtoYinZ V x E
  \end{array}
  \]
  \begin{mathpar}
  \inferrule*[right=$\mu$I]
    {\Gamma \vdash V : A[\mu X. A/X]}
    {\Gamma \vdash \rollty{\mu X. A} V : \mu X.A}
    \qquad
    \inferrule*[right=$\mu$E]
    { \Gamma \vdash V : \mu X. A \\\\
      \Gamma, x : A[\mu X.A/X] \pipe \Delta \vdash E : T
    }
    {\Gamma\pipe\Delta \vdash \pmmuXtoYinZ V x E : T}

    \inferrule*[right=$\nu$I]
    {\Gamma \mid \Delta \vdash M : \u B[\nu \u Y. \u B]}
    {\Gamma \mid \Delta \vdash \rollty{\nu \u Y. \u B} M : \nu \u Y. \u B}\\
    \qquad
    \inferrule*[right=$\nu$E]
    {\Gamma \mid \Delta \vdash M : \nu \u Y. \u B}
    {\Gamma \mid \Delta \vdash \unroll M : \u B[\nu \u Y. \u B]}

    \inferrule*[right=$\mu$ICong]
    {\Gamma \vdash V \ltdyn V' : A[\mu X. A/X]}
    {\Gamma \vdash \roll V \ltdyn \roll V' : \mu X. A}

    \inferrule*[right=$\mu$ECong]
    {\Gamma \vdash V \ltdyn V' : \mu X.A\and
    \Gamma, x : A[\mu X. A/X] \pipe \Delta \vdash E \ltdyn E' : T}
    {\Gamma \pipe \Delta \vdash \pmmuXtoYinZ V x E \ltdyn\pmmuXtoYinZ {V'} x {E'} : T}

    \inferrule*[right=$\nu$ICong]
    {\Gamma\pipe \Delta \vdash M \ltdyn M' : \u B[\nu \u Y. \u B/\u Y]}
    {\Gamma\pipe \Delta \vdash \roll M \ltdyn \roll M' : \nu \u Y. \u B}

    \inferrule*[right=$\nu$ECong]
    {\Gamma\pipe \Delta \vdash M \ltdyn M' : \nu \u Y. \u B}
    {\Gamma\pipe \Delta \vdash \unroll M \ltdyn \unroll M' : \u B[\nu \u Y. \u B/\u Y]}\\
    \framebox{Recursive Type Axioms}
    \medskip
  \end{mathpar}
    
  \begin{tabular}{c|c|c}
    Type & $\beta$ & $\eta$\\
    \hline
    $\mu$
    &
    ${\pmmuXtoYinZ{\roll V}{x}{E} \equidyn E[V/x]}$
    &
    $\begin{array}{l}
      E \equidyn \pmmuXtoYinZ x {y} E[\roll y/x] \\
      \text{where } {x : \mu X. A \vdash E : T}
    \end{array}$\\
    \hline
    $\nu$
    &
    ${\unroll\roll M \equidyn M}$
    &
    ${\bullet : \nu \u Y. \u B \vdash \bullet \equidyn \roll\unroll \bullet : \nu \u Y. \u B}$\\
  \end{tabular}
  \end{small}
  \caption{\cbpvstar\  types, terms, recursive types (diff from GTT),
    full rules in the extended version}
  \label{fig:cbpv-star}
\end{figure}

\end{longonly}

\subsection{Interpreting the Dynamic Types}
\label{sec:dynamic-type-interp}

As shown in Theorems~\ref{thm:decomposition},
\ref{thm:functorial-casts}, \ref{thm:monadic-comonadic-casts}, almost
all of the contract translation is uniquely determined already.
However, the interpretation of the dynamic types and the casts between
the dynamic types and ground types $G$ and $\u G$ are not determined
(they were still postulated in Lemma~\ref{lem:casts-admissible}).  
For this reason, our translation is \emph{parameterized} by an
interpretation of the dynamic types and the ground casts.
By Theorems~\ref{thm:cast-adjunction}, \ref{thm:retract-general}, we know
that these must be \emph{embedding-projection pairs} (ep pairs), which
we now define in \cbpvstar.
\begin{longonly}
There are two kinds of ep pairs we consider: those between value types
(where the embedding models an upcast) and those between computation
types (where the projection models a downcast).
\end{longonly}

\begin{definition}[Value and Computation Embedding-Projection Pairs] ~~ \label{def:cbpvstar-eppairs}
  \begin{enumerate}
  \item 
  A \emph{value ep pair} from $A$ to $A'$ consists of
  an \emph{embedding} value $x:A\vdash V_e : A'$
  and \emph{projection} stack $\bullet : \u F A' \vdash S_p : \u F A$,
  satisfying the \emph{retraction} and \emph{projection} properties:
    \[
    x : A \vdash \ret x \equidyn S_p[\ret V_e] : \u F A
    \qquad
    \bullet : \u F A' \vdash \bindXtoYinZ {S_p} x \ret V_e \ltdyn \bullet : \u F A'
    \]
  \item 
  A \emph{computation ep pair} from $\u B$ to $\u B'$ consists of
  an \emph{embedding} value $z : U \u B \vdash V_e : U \u B'$
  and a \emph{projection} stack $\bullet : \u B' \vdash S_p : \u B$
  satisfying \emph{retraction} and \emph{projection} properties:
    \[
    z : U \u B \vdash \force z \equidyn S_p[\force V_e] : \u B
    \qquad
    w : U \u B' \vdash V_e[\thunk {S_p[\force w]}] \ltdyn w : U \u B'
    \]
  \end{enumerate}
\end{definition}

\begin{longonly}
While this formulation is very convenient in that both kinds of ep
pairs are pairs of a value and a stack, the projection properties are
often occur more naturally in the following forms:
\begin{lemma}[Alternative Projection]
  If $(V_e,S_p)$ is a value ep pair from $A$ to $A'$ and $\Gamma,
  y:A'\pipe\Delta \vdash M : \u B$, then
  \[ \Gamma , x' : A' \vdash \bindXtoYinZ {S_p[\ret x']} x M[V_e/y] \ltdyn M[x'/y] \]

  Similarly, if $(V_e,S_p)$ is a computation ep pair from $\u B$ to
  $\u B'$, and $\Gamma \vdash M : \u B'$then
  \[ \Gamma \vdash V_e[\thunk S_p[M]] \ltdyn \thunk M : U \u B' \]
\end{lemma}
\begin{longproof}
  For the first,
  \begin{align*}
    \bindXtoYinZ {S_p[\ret x']} x M[V_e/y]
    & \equidyn
    \bindXtoYinZ {(\bindXtoYinZ {S_p[\ret x']} x \ret V_e)} y M\tag{comm conv, $\u F \beta$}\\
    &\bindXtoYinZ {\ret x'} y M\tag{projection}\\
    &M[x'/y]\tag{$\u F\beta$}
  \end{align*}
  For the second,
  \begin{align*}
    V_e[\thunk S_p[M]]
    &\equidyn V_e[\thunk S_p[\force\thunk M]] \tag{$U\beta$}\\
    &\ltdyn \thunk M\tag{projection}
  \end{align*}
\end{longproof}
\end{longonly}

Using this, and using the notion of ground type from
Section~\ref{sec:upcasts-necessarily-values} \emph{with $0$ and $\top$ removed}, we define

\begin{definition}[Dynamic Type Interpretation]
  A $\dynv,\dync$ interpretation $\rho$ consists of (1) a
  \cbpvtxt\ value type $\rho(\dynv)$, (2) a \cbpvtxt\ computation
  type $\rho(\dync)$, 
  (3)
  for each value ground type $G$,
  a value ep pair $(x.\rho_{e}(G), \rho_{p}(G))$ from $\srho G$ to
  $\rho(\dynv)$, and (4) for each computation ground type $\u G$, a
  computation ep pair $(z.\rho_{e}(\u G), \rho_{p}(\u G))$ from
  $\srho{\u G}$ to $\rho(\dync)$.  We write 
  $\srho G$ and $\srho {\u G}$ for the interpretation of a ground type,
  replacing $\dynv$ with $\rho(\dynv)$, $\dync$ with $\rho(\dync)$, and
  compositionally otherwise.
\end{definition}

Next, we show several possible interpretations of the dynamic type
that will all give, by construction, implementations that satisfy the
gradual guarantee.
Our interpretations of the value dynamic type are not surprising.
They are the usual construction of the dynamic type using type tags:
i.e., a recursive sum of basic value types.
On the other hand, our interpretations of the computation dynamic type
are less familiar.
In duality with the interpretation of $\dynv$, we interpret $\dync$ as
a recursive \emph{product} of basic computation types.
This interpretation has some analogues in previous work on the duality
of computation \citep{girard01locussolum,zeilberger09thesis}, but the
most direct interpretation (definition \ref{def:natural-type-interp})
does not correspond to any known work on dynamic/gradual typing.
Then we show that a particular choice of which computation types is
basic and which are derived produces an interpretation of the dynamic
computation type as a type of variable-arity functions whose arguments
are passed on the stack, producing a model similar to Scheme without
accounting for control effects (definition
\ref{def:scheme-like-type-interp}).

\subsubsection{Natural Dynamic Type Interpretation}

Our first dynamic type interpretation is to make the value and
computation dynamic types sums and products of the ground value and
computation types, respectively.
This forms a model of GTT for the following reasons.
For the value dynamic type $\dynv$, we need a value embedding (the
upcast) from each ground value type $G$ with a corresponding projection.
The easiest way to do this would be if for each $G$, we could rewrite
$\dynv$ as a sum of the values that fit $G$ and those that don't:
$\dynv \cong G + \dynv_{-G}$ because of the following lemma.

\begin{lemma}[Sum Injections are Value Embeddings]\label{lem:injections-are-embeddings}
  For any $A, A'$, there are value ep pairs from $A$ and $A'$ to
  $A+A'$ where the embeddings are $\inl$ and $\inr$.
\end{lemma}
\begin{proof}
  Define the embedding of $A$ to just be $x. \inl x$ and the
  projection to be $\bindXtoYinZ \bullet y \caseofXthenYelseZ y {\inl x. \ret
    x}{\inr _. \err}$.
  \begin{longonly}
    This satisfies retraction (using $\u F(+)$ induction (lemma \ref{lem:f-induction}), $\inr$ case is the same):
    \begin{align*}
      \bindXtoYinZ {\inl x} y \caseofXthenYelseZ y {\inl x. \ret x}{\inr _. \err}
      &\equidyn \caseofXthenYelseZ {\inl x} {\inl x. \ret x}{\inr _. \err}\tag{$\u F\beta$}\\
      &\equidyn \ret x\tag{$+\beta$}
    \end{align*}
    and projection (similarly using $\u F(+)$ induction):
    \begin{align*}
      x': A+A'
      &\vdash \bindXtoYinZ x {(\bindXtoYinZ {\ret x'} y \caseofXthenYelseZ y {\inl x. \ret x}{\inr _. \err})} \ret \inl x\\
      &\equidyn \bindXtoYinZ x {(\caseofXthenYelseZ {x'} {\inl x. \ret x}{\inr _. \err})}\ret \inl x\tag{$\u F\beta$}\\
      &\equidyn {(\caseofXthenYelseZ {x'} {\inl x. \bindXtoYinZ {\ret x} x \ret\inl x}{\inr _. \bindXtoYinZ \err x \ret\inl x})}\tag{commuting conversion}\\
      &\equidyn {(\caseofXthenYelseZ {x'} {\inl x. \ret\inl x}{\inr _. \err})}\tag{$\u F\beta,\err$ strictness}\\
      &\ltdyn {(\caseofXthenYelseZ {x'} {\inl x. \ret\inl x}{\inr y. \ret\inl y})}\tag{$\err$ bottom}\\
      &\equidyn \ret x' \tag{$+\eta$}
    \end{align*}
  \end{longonly}
\end{proof}
\begin{longonly}
  Whose proof relies on the following induction principle for the
  returner type:
  \begin{lemma}[$\u F(+)$ Induction Principle]
    \label{lem:f-induction}
  $\Gamma\pipe \cdot : \u F (A_1 + A_2) \vdash M_1 \ltdyn M_2 : \u B$
  holds if and only if
  $\Gamma, V_1: A_1 \vdash M_1[\ret \inl V_1] \ltdyn M_2[\ret \inl V_2] : \u B$ and
  $\Gamma, V_2: A_2 \vdash M_2[\ret \inr V_2] \ltdyn M_2[\ret \inr V_2] : \u B$ 
\end{lemma}
\end{longonly}

This shows why the type tag interpretation works: it makes the dynamic
type in some sense the minimal type with injections from each $G$:
the sum of all value ground types $? \cong \Sigma_{G} G$.

The dynamic computation type $\dync$ can be naturally defined by a
dual construction, by the following dual argument.
First, we want a computation ep pair from $\u G$ to $\dync$ for each
ground computation type $\u G$.
Specifically, this means we want a stack from $\dync$ to $\u G$ (the
downcast) with an embedding.
The easiest way to get this is if, for each ground computation type
$\u G$, $\dync$ is equivalent to a lazy product of $\u G$ and ``the
other behaviors'', i.e., $\dync \cong \u G \with \dync_{-\u G}$.
Then the embedding on $\pi$ performs the embedded computation, but on
$\pi'$ raises a type error.
The following lemma, dual to lemma \ref{lem:injections-are-embeddings}
shows this forms a computation ep pair:

\begin{lemma}[Lazy Product Projections are Computation Projections]\label{lem:projections-are-projections}
  For any $\u B, \u B'$, there are computation ep pairs from $\u B$
  and $\u B'$ to $\u B \with \u B'$ where the projections are $\pi$
  and $\pi'$.
\end{lemma}
\begin{proof}
  Define the projection for $\u B$ to be $\pi$. Define the embedding
  by $z. \pair{\force z}{\err}$. Similarly define the projection for
  $\u B'$.
  \begin{longonly}
    This satisfies retraction:
    \begin{align*}
      \pi\force\thunk\pair{\force z}{\err}
      &\equidyn \pi\pair{\force z}{\err}\tag{$U\beta$}\\
      &\equidyn \force z\tag{$\with\beta$}
    \end{align*}
    and projection:
    \begin{align*}
      &\thunk\pair{\force\thunk\pi\force w}{\err}\\
      &\equidyn \thunk\pair{\pi\force w}{\err} \tag{$U\beta$}\\
      &\ltdyn \thunk\pair{\pi\force w}{\pi'\force w}\tag{$\err$ bottom}\\
      &\equidyn \thunk\force w\tag{$\with\eta$}\\
      &\equidyn w \tag{$U\eta$}
    \end{align*}
  \end{longonly}
\end{proof}

From this, we see that the easiest way to construct an interpretation
of the dynamic computation type is to make it a lazy product of all
the ground types $\u G$: $\dync \cong \With_{\u G} \u G$.
Using recursive types, we can easily make this a definition of the
interpretations:

\begin{definition}[Natural Dynamic Type Interpretation]
  \label{def:natural-type-interp}
  The following defines a dynamic type interpretation.
  We define the types to satisfy the isomorphisms
  \[
    \dynv \cong 1 + (\dynv \times \dynv) + (\dynv + \dynv) + U\dync \qquad
    \dync \cong (\dync \with \dync) \with (\dynv \to \dync) \with \u F \dynv
  \]
  with the ep pairs defined as in
  Lemma~\ref{lem:injections-are-embeddings} and
  \ref{lem:projections-are-projections}.  
\end{definition}
\begin{longproof}
  We can construct $\dynv, \dync$ explicitly using recursive and
  corecursive types. Specifically, we make the recursion explicit by
  defining open versions of the types:
  \begin{align*}
    X,\u Y \vdash \dynv_o &= 1 + (X \times X) + (X + X) + U\u Y\vtype \\
    X,\u Y \vdash \dync_o &= (\u Y \with \u Y) \with (X \to \u Y) \with \u F X \ctype
  \end{align*}
  Then we define the types $\dynv, \dync$ using a standard encoding:
  \begin{align*}
    \dynv &= \mu X. \dynv_o[\nu \u Y. \dync_o/\u Y]\\
    \dync &= \nu \u Y. \dync_o[\mu X. \dynv_o/X]
  \end{align*}
  Then clearly by the roll/unroll isomorphism we get the desired
  isomorphisms:
  \begin{align*}
    \dynv &\cong \dynv_o[\dync/\u Y,\dynv/X] = 1 + (\dynv \times \dynv) + (\dynv + \dynv) + U\dync \\
    \dync &\cong\dynv_c[\dynv/X,\dync/\u Y] = (\dync \with \dync) \with (\dynv \to \dync) \with \u F \dynv
  \end{align*}
\end{longproof}

This dynamic type interpretation is a natural fit for CBPV because the
introduction forms for $\dynv$ are exactly the introduction forms for
all of the value types (unit, pairing,$\texttt{inl}$, $\texttt{inr}$, $\texttt{force}$), while
elimination forms are all of the elimination forms for computation types
($\pi$, $\pi'$, application and binding); such ``bityped'' languages
are related to \citet{girard01locussolum,zeilberger09thesis}.
\begin{shortonly}
  In the extended version, we give an extension of GTT axiomatizing this
  implementation of the dynamic types.
\end{shortonly}
\begin{longonly}
Based on this dynamic type interpretation, we can extend GTT to support
a truly dynamically typed style of programming, where one can perform
case-analysis on the dynamic types at runtime, in addition to the type
assertions provided by upcasts and downcasts.  
\begin{figure}
\begin{small}
  \begin{mathpar}
    \inferrule*[right=$\dyn$E]
    {\Gamma\pipe \Delta \vdash V : \dynv\\
      \Gamma,x_1 : 1\pipe \Delta \vdash E_1 : T\\
      \Gamma,x_\times : \dynv\times\dynv\pipe \Delta \vdash E_\times : T\\
      \Gamma,x_+ : \dynv+\dynv\pipe \Delta \vdash E_+ : T\\
      \Gamma,x_U : U\dync \pipe \Delta \vdash E_U : T\\
    }
    {\Gamma\pipe \Delta \vdash \dyncaseofXthenOnePairSumU {V} {x_{1}. E_1}{x_{\times}. E_{\times}}{x_{+}. E_{+}}{x_{U}. E_U} : T}\and

    \dyncaseofXthenOnePairSumU {(\upcast{G}{\dynv}V)} {x_{1}. E_1}{x_{\times}. E_{\times}}{x_{+}. E_{+}}{x_{U}. E_U} \equidyn E_{G}[V/x_G]\qquad(\dynv\beta)\and

    \inferrule*[right=$\dynv\eta$]
    {\Gamma , x : \dynv \pipe \Delta \vdash E : \u B}
    {E \equidyn \dyncaseofXthenOnePairSumU x
      {x_1. E[\upcast{1}{\dynv}/x_1]}
      {x_{\times}. E[\upcast{{\times}}{\dynv}/x_{\times}]}
      {x_+. E[\upcast{+}{\dynv}/x_+]}
      {x_U. E[\upcast{U}{\dynv}/x_U]}}\and
        
    \inferrule*[right=$\dync$]
    {\Gamma \pipe \Delta \vdash M_{\to} : \dynv \to \dync\\
      \Gamma \pipe \Delta \vdash M_{\with} : \dync \with \dync\\
      \Gamma \pipe \Delta \vdash M_{\u F} : \u F}
    {\Gamma \pipe \Delta \vdash \dyncocaseWithFunF{M_{\with}}{M_{\to}}{M_{\u F}} : \dync}\and

    \dncast{\u G}{\dync}\dyncocaseWithFunF{M_{\with}}{M_{\to}}{M_{\u F}} \equidyn M_{\u G}\quad(\dync\beta)\and

    {\bullet : \dync \vdash \bullet
      \equidyn
      \dyncocaseWithFunF
          {\dncast{\dync\with\dync}{\dync}\bullet}
          {\dncast{\dynv\to\dync}{\dync}\bullet}
          {\dncast{\u F\dynv}{\dync}\bullet}}\quad(\dync\eta)
  \end{mathpar}
  \end{small}
  \caption{Natural Dynamic Type Extension of GTT}
\end{figure}

The axioms we choose might seem to under-specify the dynamic type, but
because of the uniqueness of adjoints, the following are derivable.
\begin{lemma}[Natural Dynamic Type Extension Theorems]
  The following are derivable in GTT with the natural dynamic type extension
  \begin{mathpar}
    {\dncast{\u F 1}{\u F \dynv}\ret V \equidyn \dyncaseofXthenYelseZ V {x_1. \ret x_1}{\els \err}}\\
    {\dncast{\u F(\dynv\times\dynv)}{\u F \dynv}\ret V \equidyn \dyncaseofXthenYelseZ V {x_\times. \ret x_\times}{\els \err}}\\
    {\dncast{\u F(\dynv + \dynv)}{\u F \dynv}\ret V \equidyn \dyncaseofXthenYelseZ V {x_+. \ret x_+}{\els \err}}\\  
    {\dncast{\u F U\dync}{\u F\dynv}\ret V \equidyn \dyncaseofXthenYelseZ V {x_U. \ret x_U}{\els \err}}\\
    \force\upcast{U(\dync\with\dync)}{U\dync}V \equidyn \dyncocaseWithFunF{\force V}{\err}{\err}\\
    \force\upcast{U(\dynv \to \dync)}{U\dync}V \equidyn \dyncocaseWithFunF{\err}{\force V}{\err}\\
    \force\upcast{U\u F\dynv}{U\dync}V \equidyn \dyncocaseWithFunF{\err}{\err}{\force V}\\
  \end{mathpar}
\end{lemma}
We explore this in more detail with the next dynamic type
interpretation.
\end{longonly}

\begin{longonly}
Next, we easily see that if we want to limit GTT to just the CBV types
(i.e. the only computation types are $A \to \u F A'$), then we can
restrict the dynamic types as follows:
\begin{definition}[CBV Dynamic Type Interpretation]
  The following is a dynamic type interpretation for the ground types of
  GTT with only function computation types:
  \[
    \dynv \cong 1 + (\dynv + \dynv) + (\dynv \times \dynv) + U(\dync) \qquad
    \dync \cong \dynv \to \u F \dynv
  \]
\end{definition}

And finally if we restrict GTT to only CBN types (i.e., the only value
type is booleans $1+1$), we can restrict the dynamic types as follows:
\begin{definition}[CBN Dynamic Type Interpretation]
  The following is a dynamic type interpretation for the ground types of
  GTT with only boolean value types:
  \[
    \dynv = (1 + 1) \qquad
    \dync \cong (\dync \with \dync) \with (U\dync \to \dync)
    \with \u F \dynv
  \]
\end{definition}
\end{longonly}

\subsubsection{Scheme-like Dynamic Type Interpretation}

The above dynamic type interpretation does not correspond to any
dynamically typed language used in practice, in part because it
includes explicit cases for the ``additives'', the sum type $+$ and
lazy product type $\with$.
Normally, these are not included in this way, but rather sums are
encoded by making each case use a fresh constructor (using nominal
techniques like opaque structs in Racket) and then making the sum the
union of the constructors, as argued in \citet{siekth16recursiveunion}.
We leave modeling this nominal structure to future work, but in
minimalist languages, such as simple dialects of Scheme and Lisp, sum
types are often encoded \emph{structurally} rather than nominally by
using some fixed sum type of \emph{symbols}, also called \emph{atoms}.
Then a value of a sum type is modeled by a pair of a symbol (to indicate
the case) and a payload with the actual value.
We can model this by using the canonical isomorphisms
\[ \dynv + \dynv \cong ((1+1) \times \dynv) \qquad \dync \with \dync \cong (1+1) \to \dync \]
and representing sums as pairs, and lazy products as functions.
\begin{longonly}
The fact that isomorphisms are ep pairs is useful for constructing the
ep pairs needed in the dynamic type interpretation.  
\begin{lemma}[Isomorphisms are EP Pairs]
  \label{lem:isos-are-ep}
  If $x:A \vdash V' : A'$ and $x':A' \vdash V : A$ are an isomorphism in
  that $V[V'/x'] \equidyn x$ and $V[V/x]\equidyn x'$, then $(x.V',
  \bindXtoYinZ \bullet {x'} \ret V')$ are a value ep pair from $A$ to
  $A'$.  Similarly if $\bullet : \u B \vdash S' : \u B'$ and $\bullet :
  \u B' \vdash S : \u B$ are an isomorphism in that $S[S']\equiv
  \bullet$ and $S'[S] \equiv \bullet$ then $(z. S'[\force z], S)$ is an
  ep pair from $\u B$ to $\u B'$.
\end{lemma}
\end{longonly}

With this in mind, we remove the cases for sums and lazy pairs from the
natural dynamic types, and include some atomic type as a case of
$\dynv$---for simplicity we will just use booleans.
%
We also do not need a case for $1$, because we can identify it with one
of the booleans, say $\texttt{true}$.
This leads to the following definition:

\begin{definition}[Scheme-like Dynamic Type Interpretation] \label{def:scheme-like-type-interp}
  We can define a dynamic type interpretation with the following type
  isomorphisms:
  \begin{mathpar}
    \dynv \cong (1+1) + U\dync + (\dynv \times \dynv)\and
    \dync \cong (\dynv \to \dync) \with \u F \dynv
  \end{mathpar}
\end{definition}
\begin{proof}
  \begin{shortonly}
    The details of constructing the two mutually recursive types from
    our recursive type mechanism are in the extended version. 
  \end{shortonly}
  \begin{longonly}
  We construct $\dynv, \dync$ explicitly as follows.

  First define $X : \vtype \vdash \texttt{Tree}[X] \vtype$ to be the
  type of binary trees:
  \[ \texttt{Tree} = \mu X'. X + (X' \times X') \]
  Next, define $X:\vtype, \u Y: ctype \vdash \texttt{VarArg}[X,\u Y]
  \ctype$ to be the type of variable-arity functions from $X$ to $\u
  Y$:
  \[ \texttt{VarArg} = \nu \u Y'. \u Y \with (X \to \u Y') \]

  Then we define an open version of $\dynv, \dync$ with respect to a
  variable representing the occurrences of $\dynv$ in $\dync$:
  \begin{align*}
    X \vtype \vdash \dynv_o &= \texttt{Tree}[(1+1) + U \dync_o] \ctype\\
    X \vtype \vdash \dync_o &= \texttt{VarArg}[\u F X/\u Y] \ctype\\
  \end{align*}

  Then we can define the closed versions using a recursive type:
  \begin{mathpar}
    \dynv = \mu X. \dynv_o\and \dync = \dync_o[\dynv]
  \end{mathpar}
  \end{longonly}
  \ The ep pairs for $\times, U,\u F, \to$ are clear.  To define the
  rest, first note that there is an ep pair from $1+1$ to $\dynv$ by
  Lemma~\ref{lem:injections-are-embeddings}.  Next, we can define $1$ to
  be the ep pair to $1+1$ defined by the left case and
  Lemma~\ref{lem:injections-are-embeddings}, composed with this.  The ep
  pair for $\dynv + \dynv$ is defined by composing the isomorphism
  (which is always an ep pair)
  $(\dynv + \dynv) \cong ((1+1) \times \dynv)$ with the ep pair for
  $1+1$ using the action of product types on ep pairs (proven as part of
  Theorem \ref{thm:axiomatic-graduality}): $(\dynv + \dynv) \cong
  ((1+1)\times \dynv) \,\triangleleft\, (\dynv \times \dynv) \,\triangleleft\,
  \dynv$ (where we write $A \triangleleft A'$ to mean there is an ep
  pair from $A$ to $A'$).  Similarly, for $\dync \with \dync$, we use
  action of the function type on ep pairs (also proven as part of
  Theorem \ref{thm:axiomatic-graduality}): $\dync \with \dync \cong
  ((1+1) \to \dync) \,\triangleleft\, (\dynv \to \dync) \,\triangleleft\, \dync$
\end{proof}

\begin{shortonly}
  Intuitively, the above definition of $\dynv$ says that it is a binary
  tree whose leaves are either booleans or closures---a simple type of
  S-expressions.  On the other hand, the above definition of $\dync$
  models a \emph{variable-arity function} (as in Scheme), which is
  called with any number of dynamically typed value arguments $\dynv$
  and returns a dynamically typed result $\u F \dynv$.  To see why a
  $\dync$ can be called with any number of arguments, observe that its
  infinite unrolling is $\u F \dynv \with (\dynv \to \u F \dynv) \with
  (\dynv \to \dynv \to \u F \dynv) \with \ldots$.  This type is
  isomorphic to a function that takes a list of $\dynv$ as input ($(\mu
  X. 1 + (\dynv \times X)) \to \u F \dynv$), but operationally $\dync$
  is a more faithful model of Scheme implementations, because all of the
  arguments are passed individually on the stack, not as a
  heap-allocated single list argument.  These two are distinguished in
  Scheme and the ``dot args'' notation witnesses the isomorphism.
\end{shortonly}

\begin{longonly}
If we factor out some of the recursion to use inductive and
coinductive types, we get the following isomorphisms:
\begin{mathpar}
  \dynv \cong \texttt{Tree}[(1+1) + U\dync]\and
  \dync \cong \texttt{VarArg}[\dynv][\u F \dynv]
\end{mathpar}

That is a dynamically typed value is a binary tree whose leaves are
either booleans or closures.
We think of this as a simple type of S-expressions.
A dynamically typed computation is a variable-arity function that is
called with some number of dynamically typed value arguments $\dynv$
and returns a dynamically typed result $\u F \dynv$.
This captures precisely the function type of Scheme, which allows for
variable arity functions!

What's least clear is \emph{why} the type
\[
\texttt{VarArg}[X][\u Y] = \nu \u Y'. (X \to \u Y') \with \u Y
\]
Should be thought of as a type of variable arity functions.
First consider the infinite unrolling of this type:
\[
\texttt{VarArg}[X][\u Y] \simeq \u Y \with (X \to \u Y) \with (X \to X \to \u Y) \with \cdots
\]
this says that a term of type $\texttt{VarArg}[X][Y]$ offers an
infinite number of possible behaviors: it can act as a function from
$X^n \to \u Y$ for any $n$.
Similarly in Scheme, a function can be called with any number of
arguments.
Finally note that this type is isomorphic to a function that takes a
\emph{cons-list} of arguments:
\begin{align*}
  &\u Y \with (X \to \u Y) \with (X \to X \to \u Y) \with \cdots\\
  &\cong(1 \to \u Y) \with ((X \times 1) \to \u Y) \with ((X \times X \times 1) \to \u Y) \with \cdots\\
  &\cong(1 + (X \times 1) + (X \times X \times 1) + \cdots) \to \u Y\\
  &\cong(\mu X'. 1 + (X\times X')) \to \u Y
\end{align*}

But operationally the type $\texttt{VarArg}[\dynv][\u F\dynv]$ is a
more faithful model of Scheme implementations because all of the
arguments are passed individually on the stack, whereas the type $(\mu
X. 1 + (\dynv \times X)) \to \u F X$ is a function that takes a single
argument that is a list.
These two are distinguished in Scheme and the ``dot args'' notation
witnesses the isomorphism.
\end{longonly}

Based on this dynamic type interpretation we can make a ``Scheme-like''
extension to GTT in Figure~\ref{fig:scheme}.
First, we add a boolean type $\bool$ with $\tru$, $\fls$ and
if-then-else.
Next, we add in the elimination form for $\dynv$ and the introduction
form for $\dync$.
The elimination form for $\dynv$ is a typed version of Scheme's
\emph{match} macro.
The introduction form for $\dync$ is a typed, CBPV version of Scheme's
\emph{case-lambda} construct.
Finally, we add type dynamism rules expressing the representations of
$1$, $A + A$, and $A \times A$ in terms of booleans that were explicit
in the ep pairs used in Definition~\ref{def:scheme-like-type-interp}.
\begin{shortonly}
  In the extended version of the paper, we include the appropriate term
  dynamism axioms, which are straightforward syntactifications of the
  properties of the dynamic type interpretation, and prove a unique
  implementation theorem for the new casts.
\end{shortonly}

\begin{figure}
\begin{small}
\begin{mathpar}
  1 \ltdyn \bool\and
  A + A \equidyn \bool \times A\and
  \u B \with \u B \equidyn \bool \to \u B
  
  \begin{longonly}
    \\
  \inferrule*[right=$\bool$I]
  { }
  {\Gamma \vdash \tru, \fls : \bool}

  \inferrule*[right=$\bool$E]
  {\Gamma \vdash V : \bool\\
    \Gamma \vdash E_t : T\\
    \Gamma \vdash E_f : T}
  {\Gamma \pipe \Delta \vdash \ifXthenYelseZ V {E_t} {E_f} : T}

  \\
  \ifXthenYelseZ \tru {E_t} {E_f} \equidyn E_t\and
  \ifXthenYelseZ \fls {E_t} {E_f} \equidyn E_f\\
  x : \bool \vdash E \equidyn \ifXthenYelseZ x {E[\tru/x]} {E[\fls/x]}\\

  \upcast{1}{\bool}V \equidyn \tru\and
  \upcast{A+A}{\bool \times A}\inl V \equidyn (\tru, V)\and
  \upcast{A+A}{\bool \times A}\inr V \equidyn (\fls, V)\\

  \pi\dncast{\u B\with\u B}{\bool \to \u B}M \equidyn M\,\tru\and
  \pi'\dncast{\u B\with\u B}{\bool \to \u B}M \equidyn M\,\fls\\
  \end{longonly}
  
  \inferrule*[right=$\dync$I]
  {\Gamma \pipe \Delta \vdash M_{\to} : \dynv \to \dync\\
    \Gamma \pipe \Delta \vdash M_{\u F} : \u F \dynv}
  {\Gamma \pipe \Delta \vdash \dyncocaseFunF{M_{\to}}{M_{\u F}} : \dync}\\

  \begin{longonly}
  \dncast{\u G}{\dync}\dyncocaseFunF{M_{\to}}{M_{\u F}} \equidyn M_{\u G}\quad(\dync\beta)

  {\bullet : \dync \vdash \bullet
    \equidyn
    \dyncocaseFunF
        {\dncast{\dynv\to\dync}{\dync}\bullet}
        {\dncast{\u F\dynv}{\dync}\bullet}}\quad(\dync\eta)\\

  \end{longonly}

  \inferrule*[right=$\dynv$E]
  {\Gamma\pipe \Delta \vdash V : \dynv \and
    \Gamma, x_{\bool}:\bool\pipe \Delta  \vdash E_\bool : T\and
    \Gamma,x_U : U\dync \pipe \Delta \vdash E_U : T\and
    \Gamma,x_\times : \dynv\times\dynv\pipe \Delta \vdash E_\times : T\and
  }
  {\Gamma\pipe \Delta \vdash \dyncaseofXthenBoolUPair {V} {x_{\bool}. E_{\bool}}{x_{U}. E_U}{x_{\times}. E_{\times}} : T}\\

  \begin{longonly}
  \inferrule
  {G \in \{ \bool, \times, U\}}
  {\dyncaseofXthenBoolUPair {(\upcast{G}{\dynv}V)} {x_{\bool}. E_{\bool}}{x_{U}. E_U}{x_{\times}. E_{\times}} \equidyn E_{G}[V/x_G]} \qquad(\dynv\beta)\\

  \inferrule*[right=$\dynv\eta$]
  {\Gamma , x : \dynv \pipe \Delta \vdash E : \u B}
  {E \equidyn \dyncaseofXthenBoolUPair x
    {x_\bool. E[\upcast{\bool}{\dynv}/x_\bool]}
    {x_{\times}. E[\upcast{{\times}}{\dynv}/x_{\times}]}
    {x_U. E[\upcast{U}{\dynv}/x_U]}}\\
  \end{longonly}

\end{mathpar}
\end{small}
\vspace{-0.4in}
\caption{Scheme-like Extension to GTT}
\label{fig:scheme}
\end{figure}

\begin{longonly}
  The reader may be surprised by how \emph{few} axioms we need to add
  to GTT for this extension: for instance we only define the upcast
  from $1$ to $\bool$ and not vice-versa, and similarly the sum/lazy
  pair type isomorphisms only have one cast defined when a priori
  there are $4$ to be defined.
  Finally for the dynamic types we define $\beta$ and $\eta$ laws
  that use the ground casts as injections and projections
  respectively, but we don't define the corresponding dual casts (the
  ones that possibly error).

  In fact all of these expected axioms can be \emph{proven} from those
  we have shown.
  Again we see the surprising rigidity of GTT: because an $\u F$
  downcast is determined by its dual value upcast (and vice-versa for
  $U$ upcasts), we only need to define the upcast as long as the
  downcast \emph{could} be implemented already.
  Because we give the dynamic types the universal property of a
  sum/lazy product type respectively, we can derive the
  implementations of the ``checking'' casts.
  All of the proofs are direct from the uniqueness of adjoints
  lemma.

  \begin{theorem}[Boolean to Unit Downcast]
    In Scheme-like GTT, we can prove
    \[
    \dncast{\u F1}{\u F\bool}\bullet
    \equidyn
    \bindXtoYinZ \bullet x \ifXthenYelseZ x {\ret()}{\err}
    \]
  \end{theorem}

  \begin{theorem}[Tagged Value to Sum]
    In Scheme-like GTT, we can prove
    \[
    \upcast{\bool \times A}{A+A}V \equidyn \pmpairWtoXYinZ V {x}{y} \ifXthenYelseZ x {\inl y}{\inr y}
    \]
    and the downcasts are given by lemma \ref{lem:isos-are-ep}.
  \end{theorem}
  \begin{theorem}[Lazy Product to Tag Checking Function]
    In Scheme-like GTT, we can prove
    \[
    \dncast{\bool\to \u B}{\u B\with\u B}\bullet
    \equidyn
    \lambda x:\bool. \ifXthenYelseZ x {\pi \bullet}{\pi' \bullet}
    \]
    and the upcasts are given by lemma \ref{lem:isos-are-ep}.
  \end{theorem}

  \begin{theorem}[Ground Mismatches are Errors]
    In Scheme-like GTT we can prove
    \begin{mathpar}
      {\dncast{\u F \bool}{\u F \dynv}\ret V \equidyn \dyncaseofXthenYelseZ V {x_\bool. \ret x_\bool}{\els \err}}\\
      {\dncast{\u F(\dynv\times\dynv)}{\u F \dynv}\ret V \equidyn \dyncaseofXthenYelseZ V {x_\times. \ret x_\times}{\els \err}}\\
      {\dncast{\u F U\dync}{\u F\dynv}\ret V \equidyn \dyncaseofXthenYelseZ V {x_U. \ret x_U}{\els \err}}\\

      \force\upcast{U(\dynv \to \dync)}{U\dync}V \equidyn \dyncocaseFunF{\force V}{\err}\\
      \force\upcast{U\u F\dynv}{U\dync}V \equidyn \dyncocaseFunF{\err}{\force V}\\
    \end{mathpar}
  \end{theorem}

  Finally, we note now that all of these axioms are satisfied when
  using the Scheme-like dynamic type interpretation and extending the
  translation of GTT into \cbpvstar\  with the following, tediously
  explicit definition:
  \begin{align*}
    &\sem{\bool} = 1+1\\
    &\sem{\tru} =\inl()\\
    &\sem{\fls} =\inr()\\
    &\sem{\ifXthenYelseZ V {E_t} {E_f}} = \caseofXthenYelseZ {\sem V}{x. E_t}{x.E_f}\\
    &\sem{\dyncaseofXthenBoolUPair x {x_\bool. E_\bool}{x_U. E_U}{x_\times. E_\times}}
    =\\
    &\quad\pmmuXtoYinZ {(x:\dynv)}  {x'} \pmmuXtoYinZ {x' : \texttt{Tree}[(1+1)+U\dync]}t \caseofX t\\
    &\qquad\thenY{l. \caseofXthenYelseZ l {x_\bool. \sem{E_\bool}}{x_U. \sem{E_U}}}\\
    &\qquad\elseZ{x_\times. \sem{E_\times}}\\
    &\sem{\dyncocaseFunF{M_\to}{M_{\u F}}}
    = \rollty{\nu \u Y. (\dynv \to \u Y)\with \u F\dynv}\pair{\sem{M_\to}}{\sem{M_{\u F}}}
  \end{align*}
\end{longonly}

\subsection{Contract Translation}

Having defined the data parameterizing the translation, we now consider
the translation of GTT into \cbpvstar\ itself.
For the remainder of the paper, we assume that we have a fixed dynamic
type interpretation $\rho$, and all proofs and definitions work for any
interpretation.

\begin{figure}
\begin{small}
  \begin{mathpar}
      x:\sem{A} \vdash \sem{\upcast{A}{A'}} : \sem{A'}\and
      \bullet:\sem{\u B'} \vdash \sem{\dncast{\u B}{\u B'}} : \sem{\u B}\\
      \begin{array}{rcl}
\iflong
       x : 0 \vdash \supcast{0}{A} & = & \absurd x\\
      \bullet : A \vdash \sdncast{\u F0}{\u F A} &=& \bindXtoYinZ \bullet x \err\\
\fi
      x : \sem{\dynv} \vdash \sem{\upcast{\dynv}{\dynv}} & = & x\\
      \bullet : \u F \dynv \vdash \sdncast{\u F \dynv}{\u F\dynv} &=& \bullet\\
      x : \sem{G} \vdash \sem{\upcast{G}{\dynv}} & = & \rho_{up}(G)\\
      \bullet : \u F \dynv \vdash \sdncast{\u F G}{\u F\dynv} &=& \rho_{dn}(G)\\
      x : \sem{A} \vdash \sem{\upcast{A}{\dynv}} & = & \sem{\upcast{\lfloor A \rfloor}{\dynv}}[{\sem{\upcast{A}{\lfloor A \rfloor}}}/x]\\
      \bullet: \u F\dynv \vdash \sdncast{A}{\dynv} &=& \sdncast{A}{\floor A}[{\sdncast{\floor A}{\dynv}}]\\
\iflong
      x : \sem{A_1} + \sem{A_2} \vdash \sem{\upcast{A_1 + A_2}{A_1' + A_2'}}
      & = & \caseofX x \\
      && \thenY{x_1. \sem{\upcast{A_1}{A_1'}}[x_1/x]}\\
      && \elseZ{x_2. \sem{\upcast{A_2}{A_2'}}[x_2/x]}\\
      \bullet : \sem{A_1} + \sem{A_2} \vdash
      \sem{\dncast{\u F(A_1 + A_2)}{\u F(A_1' + A_2')}}
      &=&
      \bindXtoYinZ \bullet {x'} \caseofX {x'}\\
      &&\thenY{x_1'. \bindXtoYinZ {(\sdncast{\u FA_1}{\u F A_1'}\ret x_1')} {x_1} \ret x_1}\\
      &&\elseZ{x_2'. \bindXtoYinZ {(\sdncast{\u FA_2}{\u F A_2'}\ret x_2')} {x_2} \ret x_2}\\
      x : 1 \vdash \supcast{1}{1} &=& x\\
      \bullet : \u F 1 \vdash \sdncast{\u F1}{\u F1} &=& x\\
\fi
      x : \sem{A_1}\times\sem{A_2} \vdash \sem{\upcast{A_1\times A_2}{A_1'\times A_2'}} &=& \pmpairWtoXYinZ x {x_1}{x_2}\\
      &&(\supcast{A_1}{A_1'}[x_1], \supcast{A_2}{A_2'}[x_2])\\
      \bullet \vdash \sdncast{\u F(A_1 \times A_2)}{\u F(A_1' \times A_2')}
      &=&
      \bindXtoYinZ \bullet {x'} \pmpairWtoXYinZ {x'} {x_1'}{x_2'}\\
      &&\bindXtoYinZ {\sdncast{\u FA_1}{\u FA_1'}\ret x_1'} {x_1}\\
      && \bindXtoYinZ {\sdncast{\u FA_2}{\u FA_2'}\ret x_2'} {x_2} \ret(x_1,x_2)\\
      x : U\u F \sem{A} \vdash \sem{\upcast{U\u F A}{U \u F A'}} &=&
      \thunk (\bindXtoYinZ {\force x} y \ret \sem{\upcast{A}{A'}}[y/x])\\\\
\iflong
      \bullet : \u B \vdash \sdncast{\top}{\u B} &=& \{ \}\\
      x:U\top \vdash \supcast{U\top}{U\u B} &=& \thunk \err\\
      \bullet : \dync \vdash \sdncast{\dync}{\dync} &=& \bullet\\
      x:U\dync \vdash \supcast{U\dync}{U\dync} &=& x\\
      \bullet : \dync \vdash \sdncast{\u G}{\dync} &=& \rho_{dn}(\u G)\\
      x:U\u G \vdash \supcast{U\u G}{U\dync} &=& \rho_{up}(\u G)\\
      \bullet : \dync \vdash \sdncast{\u B}{\dync} &=& \sdncast{\u B}{\floor {\u B}}[\sdncast{\floor{\u B}}{\dync}]\\
      x:U\dync \vdash \supcast{U\u B}{U\dync}&=& \supcast{U\floor{\u B}}{U\dync}[\supcast{U\u B}{U\floor{\u B}}]\\
      \bullet : \sem{\u B_1'}\with \sem{\u B_2'}\vdash \sdncast{\u B_1\with\u B_2}{\u B_1'\with\u B_2'} &=& \pairone{\sdncast{\u B_1}{\u B_1'}\pi\bullet}\\
              &&\pairtwo{\sdncast{\u B_2}{\u B_2'}\pi'\bullet}\\
      x : U(\sem{\u B_1}\with \sem{\u B_2}) \vdash
      {\supcast{U(\u B_1 \with \u B_2)}{U(\u B_1'\with\u B_2')}}
      &=&
      \thunk\\
      &&\pairone{\force \supcast{\u B_1}{\u B_1'}{(\thunk \pi\force x)}}\\
      &&\pairtwo{\force \supcast{\u B_2}{\u B_2'}{(\thunk \pi'\force x)}}\\
      \bullet \vdash \sdncast{A \to \u B}{A' \to \u B'} &=& \lambda x:A. \sdncast{\u B}{\u B'}{(\bullet\, (\supcast{A}{A'}{x}))}\\
      f : U(\sem{A} \to \sem{\u B}) \vdash
      \supcast{U(A \to \u B)}{U(A' \to \u B')}
      &=&
      \thunk\lambda x':A'. \\
      &&\bindXtoYinZ {\sdncast{\u F A}{\u FA'}\ret x'} x\\
      &&\force\supcast{U\u B}{U\u B'}\thunk {(\force f)\, x'}\\
      \bullet : \u FU\u B' \vdash \sdncast{\u FU\u B}{\u FU\u B'}
      &=&
      \bindXtoYinZ \bullet {x'} \sdncast{\u B}{\u B'}\force x'
\fi
      \end{array}
  \end{mathpar}
  \vspace{-0.15in}
  \caption{Cast to Contract Translation }
  \label{fig:cast-to-contract}
\end{small}
\end{figure}

\begin{longonly}
\subsubsection{Interpreting Casts as Contracts} ~
\end{longonly}
The main idea of the translation is an extension of the dynamic type
interpretation to an interpretation of \emph{all} casts in GTT
(Figure~\ref{fig:cast-to-contract}) as contracts in \cbpvstar, following
the definitions in Lemma~\ref{lem:casts-admissible}.  
\begin{shortonly}
  To verify the totality and coherence of this definition, we define (in
  the extended version) a normalized version of the type dynamism rules
  from Figure~\ref{fig:gtt-type-dynamism}, which is interderivable but
  has at most one derivation of $T \ltdyn T'$ for a given $T$ and $T'$.
  The main idea is to restrict reflexivity to base types, and restrict
  transitivity to $A \ltdyn \floor{A} \ltdyn \dynv$, where $\floor{A}$
  is the ground type with the same outer connective as $A$.
\end{shortonly}
\begin{longonly}
Some clauses of the translation are overlapping, which we resolve by
considering them as ordered (though we will ultimately show they are
equivalent).
The definition is also not obviously total: we need to verify that it
covers every possible case where $A \ltdyn A'$ and $\u B \ltdyn \u
B'$.
To prove totality and coherence, we could try induction on the type
dynamism relation of Figure~\ref{fig:gtt-type-dynamism}, but it is
convenient to first give an alternative, normalized set of rules for
type dynamism that proves the same relations, which we do in
Figure~\ref{fig:normalized}.

\begin{figure}
\begin{small}
  \begin{mathpar}
  \inferrule
  {A \in \{\dynv, 1\}}
  {A \ltdyn A}

  \inferrule
  {A \in \{\dynv, 0\}}
  {0 \ltdyn A}

  \inferrule
  {A \ltdyn \floor A\and
    A \not\in\{0,\dynv \}}
  {A \ltdyn \dynv}\\

  \inferrule
  {\u B \ltdyn \u B'}
  {U B \ltdyn U B'}

  \inferrule
  {A_1 \ltdyn A_1' \and A_2 \ltdyn A_2' }
  {A_1 + A_2 \ltdyn A_1' + A_2'}

  \inferrule
  {A_1 \ltdyn A_1' \and A_2 \ltdyn A_2' }
  {A_1 \times A_2 \ltdyn A_1' \times A_2'}\\

  \inferrule
  {}
  {\dync \ltdyn \dync}

  \inferrule
  {\u B \in \{ \dync, \top \}}
  {\top \ltdyn \u B}

  \inferrule
  {\u B \ltdyn \floor {\u B} \and \u B \not\in \{ \top, \dync \}}
  {\u B \ltdyn \dync}\\

  \inferrule
  {A \ltdyn A'}
  {\u F A \ltdyn \u F A'}

  \inferrule
  {\u B_1 \ltdyn \u B_1' \and \u B_2 \ltdyn \u B_2'}
  {\u B_1 \with \u B_2 \ltdyn \u B_1' \with \u B_2'}

  \inferrule
  {A \ltdyn A' \and \u B \ltdyn \u B'}
  {A \to \u B \ltdyn A' \to \u B'}
  \end{mathpar}
  \end{small}
  \caption{Normalized Type Dynamism Relation}
  \label{fig:normalized}
\end{figure}

\begin{lemma}[Normalized Type Dynamism is Equivalent to Original]
  \label{lem:norm-type-dyn}
  $T \ltdyn T'$ is provable in the normalized typed dynamism
  definition iff it is provable in the original typed
  dynamism definition.
\end{lemma}
\begin{longproof}
It is clear that the normalized system is a subset of the original:
every normalized rule corresponds directly to a rule of the original
system, except the normalized $A \ltdyn \dynv$ and $\u B \ltdyn \dync$
rules have a subderivation that was not present originally.  

For the converse, first we show by induction that reflexivity is
admissible:
  \begin{enumerate}
  \item If $A \in \{\dynv, 1, 0\}$, we use a normalized rule.
  \item If $A \not\in\{\dynv, 1, 0\}$, we use the inductive hypothesis
    and the monotonicity rule.
  \item If $\u B\in \{\dync, \top\}$ use the normalized rule.
  \item If $\u B \not\in\{\dync, \top\}$ use the inductive hypothesis
    and monotonicity rule.
  \end{enumerate}
  Next, we show that transitivity is admissible:
  \begin{enumerate}
  \item Assume we have $A \ltdyn A' \ltdyn A''$
    \begin{enumerate}
    \item If the left rule is $0 \ltdyn A'$, then either $A' = \dynv$
      or $A' = 0$. If $A' = 0$ the right rule is $0 \ltdyn A''$ and we
      can use that proof. Otherwise, $A' = \dynv$ then the right rule
      is $\dynv \ltdyn \dynv$ and we can use $0 \ltdyn \dynv$.
    \item If the left rule is $A \ltdyn A$ where $A \in \{ \dynv, 1\}$
      then either $A = \dynv$, in which case $A'' = \dynv$ and we're
      done.  Otherwise the right rule is either $1 \ltdyn 1$ (done) or
      $1 \ltdyn \dynv$ (also done).
    \item If the left rule is $A \ltdyn \dynv$ with
      $A\not\in\{0,\dynv\}$ then the right rule must be $\dynv \ltdyn
      \dynv$ and we're done.
    \item Otherwise the left rule is a monotonicity rule for one of
      $U, +, \times$ and the right rule is either monotonicity (use
      the inductive hypothesis) or the right rule is $A' \ltdyn \dynv$
      with a sub-proof of $A' \ltdyn \floor{A'}$. Since the left rule
      is monotonicity, $\floor{A} = \floor{A'}$, so we inductively use
      transitivity of the proof of $A \ltdyn A'$ with the proof of $A'
      \ltdyn \floor{A'}$ to get a proof $A \ltdyn \floor{A}$ and thus
      $A \ltdyn \dynv$.
    \end{enumerate}
  \item Assume we have $\u B \ltdyn \u B' \ltdyn \u B''$.
    \begin{enumerate}
    \item If the left rule is $\top \ltdyn \u B'$ then $\u B'' \in
      \{\dync, \top\}$ so we apply that rule.
    \item If the left rule is $\dync\ltdyn \dync$, the right rule must
      be as well.
    \item If the left rule is $\u B \ltdyn \dync$ the right rule must
      be reflexivity.
    \item If the left rule is a monotonicity rule for $\with, \to, \u
      F$ then the right rule is either also monotonicity (use the
      inductive hypothesis) or it's a $\u B \ltdyn \dync$ rule and we
      proceed with $\dynv$ above
    \end{enumerate}
  \end{enumerate}
  Finally we show $A \ltdyn \dynv$, $\u B \ltdyn \dync$ are admissible
  by induction on $A$, $\u B$.  
  \begin{enumerate}
  \item If $A \in \{ \dynv, 0\}$ we use the primitive rule.
  \item If $A \not\in \{ \dynv, 0 \}$ we use the $A \ltdyn \dynv$ rule
    and we need to show $A \ltdyn \floor A$. If $A = 1$, we use the
    $1\ltdyn 1$ rule, otherwise we use the inductive hypothesis and
    monotonicity.
  \item If $\u B \in \{ \dync, \top\}$ we use the primitive rule.
  \item If $\u B \not\in \{ \dync, \top \}$ we use the $\u B \ltdyn
    \dync$ rule and we need to show $\u B \ltdyn \floor {\u B}$, which
    follows by inductive hypothesis and monotonicity.
  \end{enumerate}
  Every other rule in Figure~\ref{fig:gtt-type-dynamism} is a rule of
  the normalized system in Figure~\ref{fig:normalized}.
\end{longproof}

Based on normalized type dynamism, we show
\begin{theorem}
If $A \ltdyn A'$ according to Figure~\ref{fig:normalized}, then there is
a unique complex value $x : A \vdash \supcast{A}{A'}{x} : A'$
and
if $\u B \ltdyn \u B'$ according to Figure~\ref{fig:normalized}, then there is
a unique complex stack $x : \u B \vdash \supcast{\u B}{\u B'}{x} : \u B'$
\end{theorem}

\smallskip
\subsubsection{Interpretation of Terms}~
\end{longonly}
\ Next, we extend the translation of casts to a translation of all terms
by congruence, since all terms in GTT besides casts are
in \cbpvstar.  This satisfies:
\begin{lemma}[Contract Translation Type Preservation]
  If $\Gamma\pipe\Delta \vdash E : T$ in GTT, then $\sem{\Gamma}
  \pipe\sem\Delta\vdash \sem E : \sem T$ in \cbpvstar.
\end{lemma}

\iflong
\subsubsection{Interpretation of Term Dynamism}
\fi
We have now given an interpretation of the types, terms, and
type dynamism proofs of GTT in \cbpvstar.
To complete this to form a \emph{model} of GTT, we need to give an
interpretation of the \emph{term dynamism} proofs, which is 
established by the
following ``axiomatic graduality'' theorem.  
GTT has \emph{heterogeneous} term dynamism
rules indexed by type dynamism, but \cbpvstar\  has only \emph{homogeneous}
inequalities between terms, i.e., if $E \ltdyn E'$, then $E,E'$ have
the \emph{same} context and types.
Since every type dynamism judgement has an associated contract, we can
translate a heterogeneous term dynamism to a homogeneous inequality
\emph{up to contract}.  Our next overall goal is to prove
\begin{theorem}[Axiomatic Graduality] \label{thm:axiomatic-graduality}
  For any dynamic type interpretation,
  \begin{small}
  \[
    \inferrule
    {\Phi : \Gamma \ltdyn \Gamma'\\
      \Psi : \Delta \ltdyn \Delta'\\
      \Phi \pipe \Psi \vdash M \ltdyn M' : \u B \ltdyn \u B'}
    {\sem\Gamma \pipe \sem{\Delta'} \vdash \sem M[\sem{\Psi}] \ltdyn \sdncast{\u B}{\u B'}[\sem{M'}[\sem{\Phi}]] : \sem{\u B}}
    \quad
    \inferrule
    {\Phi : \Gamma \ltdyn \Gamma' \\
      \Phi \vdash V \ltdyn V' : A \ltdyn A'}
    {\sem{\Gamma} \vdash \supcast{A}{A'}[\sem{V}] \ltdyn\sem{V'}[\sem\Phi] : \sem {A'}}
    \]
  \end{small}
    where we define $\sem{\Phi}$ to upcast each variable, and
    $\sem{\Delta}$ to downcast $\bullet$ if it is nonempty, and if
    $\Delta = \cdot$, then $M[\sem{\Delta}] = M$.
  \begin{longonly}
    More explicitly,
    \begin{enumerate}
    \item If $\Phi : \Gamma \ltdyn \Gamma'$, then there exists $n$
      such that $\Gamma = x_1:A_1,\ldots,x_n:A_n$ and $\Gamma' =
      x_1':A_1',\ldots,x_n':A_n'$ where $A_i \ltdyn A_i'$ for each
      $i\leq n$.
      Then $\sem{\Phi}$ is a substitution from $\sem{\Gamma}$ to $\sem{\Gamma'}$
      defined as
      \[ \sem{\Phi} = \supcast{A_1}{A_1'}x_1/x_1',\ldots\supcast{A_n}{A_n'}x_n/x_n' \]
    \item If $\Psi : \Delta \ltdyn \Delta'$, then we similarly define
      $\sem{\Psi}$ as a ``linear substitution''. That is, if $\Delta =
      \Delta' = \cdot$, then $\sem{\Psi}$ is an empty substitution and
      $M[\sem{\Psi}] = M$, otherwise $\sem{\Psi}$ is a linear
      substitution from $\Delta' = \bullet : \u B'$ to $\Delta =
      \bullet : \u B$ where $\u B \ltdyn \u B'$ defined as
      \[ \sem\Psi = \sdncast{\u B}{\u B'}\bullet/\bullet \]
    \end{enumerate}
  \end{longonly}
\end{theorem}

\begin{longonly}
  Relative to previous work on graduality \citep{newahmed18},
the distinction between complex value upcasts and complex stack
downcasts guides the formulation of the theorem; e.g. using upcasts in
the left-hand theorem would require more thunks/forces.  
\end{longonly}

\begin{shortonly}
  The full proof can be found in the extended version, and uses a
  sequence of lemmas.  The first lemma shows that the translations of
  casts in Figure~\ref{fig:cast-to-contract} do form ep pairs in the
  sense of Definition~\ref{def:cbpvstar-eppairs}.  One of the biggest
  advantages of using an explicit syntax for complex values and complex
  stacks is that the ``shifted'' casts (the downcast between $\u F$
  types for $A \ltdyn A'$ and the upcast between $U$ types for $\u B
  \ltdyn \u B'$) are the only effectful terms, and this lemma is the
  only place where we need to reason about their definitions
  explicitly---afterwards, we can simply use the fact that they are ep
  pairs with the ``pure'' value upcasts and stack downcasts, which
  compose much more nicely than effectful terms.  This is justified by two
  additional lemmas, which show that a projection is determined
  by its embedding and vice versa, and that embedding-projections
  satisfy an adjunction/Galois connection property.  The final lemmas
  show that, according to Figure~\ref{fig:cast-to-contract},
  $\supcast{A}{A'}$ is equivalent to the identity and
  $\supcast{A'}{A''}\supcast{A}{A'}$ is $\supcast{A}{A''}$, and
  similarly for downcasts.  All of these properties are theorems in GTT
  (Section~\ref{sec:theorems-in-gtt}), and in the extended version it
  takes quite a bit of work to prove them true under translation, which
  illustrates that the axiomatic theory of GTT encodes a lot of
  information with relatively few rules.
\end{shortonly}

\begin{longonly}
  We now develop some lemmas on the way towards proving this result.  
  First, to keep proofs high-level, we establish the following cast
  reductions that follow easily from $\beta,\eta$ principles.
\begin{lemma}[Cast Reductions]
  The following are all provable
  \begin{align*}
    &\sem{\upcast{A_1+A_2}{A_1'+A_2'}}[\inl V] \equidyn \inl \sem{\upcast{A_1}{A_1'}}[V]\\
    &\sem{\upcast{A_1+A_2}{A_1'+A_2'}}[\inr V] \equidyn \inr \sem{\upcast{A_2}{A_2'}}[V]\\
    &\sem{\dncast{\u F(A_1+A_2)}{\u F(A_1'+A_2')}}[\ret \inl V] \equidyn
    \bindXtoYinZ {\sem{\dncast{A_1}{A_1'}}[\ret V]} {x_1} \ret \inl x_1\\
    &\sem{\dncast{\u F(A_1+A_2)}{\u F(A_1'+A_2')}}[\ret \inr V] \equidyn
    \bindXtoYinZ {\sem{\dncast{A_2}{A_2'}}[\ret V]} {x_2} \ret \inr x_2\\
    &\sem{\dncast{\u F 1}{\u F1}} \equidyn \bullet\\
    &\sem{\upcast{ 1}{1}}[x] \equidyn x\\
    &\sem{\dncast{\u F(A_1\times A_2)}{\u F(A_1'\times A_2')}}[\ret (V_1,V_2)]\\
    &\quad\equidyn
    \bindXtoYinZ {\sdncast{\u FA_1}{\u F A_1'}[\ret V_1]} {x_1} \bindXtoYinZ {\sdncast{\u FA_2}{\u F A_2'}[\ret V_2]} {x_2} \ret (x_1,x_2)\\
    &\supcast{A_1\times A_2}{A_1'\times A_2'}[(V_1,V_2)] \equidyn (\supcast{A_1}{A_1'}[V_1], \supcast{A_2}{A_2'}[V_2])\\
    &(\sdncast{A \to \u B}{A' \to \u B'} M)\, V \equidyn
    (\sdncast{\u B}{\u B'} M)\, (\supcast{A}{A'}{V})\\
    &(\force (\supcast{U(A\to\u B)}{U(A'\to\u B')} V))\,V'\\
    &\quad\equidyn
    \bindXtoYinZ {\dncast{\u FA}{\u FA'}[\ret V']} x {\force (\supcast{U\u B}{U\u B'}{(\thunk (\force V\, x))})}\\
    &\pi \sdncast{\u B_1 \with \u B_2}{\u B_1' \with \u B_2'} M \equidyn
    \sdncast{\u B_1}{\u B_1'} \pi M\\
    &\pi' \sdncast{\u B_1 \with \u B_2}{\u B_1' \with \u B_2'} M \equidyn
    \sdncast{\u B_2}{\u B_2'} \pi' M\\
    &\pi \force (\supcast{U(\u B_1 \with \u B_2)}{U(\u B_1' \with \u B_2')} V)
    \equidyn
    \force \supcast{U\u B_1}{U\u B_1'}{\thunk (\pi \force V)}\\
    &\pi' \force (\supcast{U(\u B_1 \with \u B_2)}{U(\u B_1' \with \u B_2')} V)
    \equidyn
    \force \supcast{U\u B_2}{U\u B_2'}{\thunk (\pi' \force V)}\\
    &\sdncast{\u F U \u B}{\u F U \u B'}[\ret V] \equidyn \ret\thunk \sdncast{\u B}{\u B'}\force V\\
    &\force \supcast{U\u FA}{U \u F A'}[V]
    \equidyn
    \bindXtoYinZ {\force V} x \thunk\ret\upcast{A}{A'} x
  \end{align*}
\end{lemma}

Our next goal is to show that from the basic casts being ep pairs, we
can prove that all casts as defined in Figure~\ref{fig:cast-to-contract}
are ep pairs.  Before doing so, we prove the following lemma, which is
used for transitivity (e.g. in the $A \ltdyn \dynv$ rule, which uses a
composition $A \ltdyn \floor{A} \ltdyn \dynv$):
\begin{lemma}[EP Pairs Compose]\hfill
  \label{lem:ep-pairs-compose}
  \begin{enumerate}
  \item If $(V_1, S_1)$ is a value ep pair from $A_1$ to $A_2$ and
    $(V_2,S_2)$ is a value ep pair from $A_2$ to $A_3$, then
    $(V_2[V_1], S_1[S_2])$ is a value ep pair from $A_1$ to $A_3$.
  \item If $(V_1, S_1)$ is a computation ep pair from $\u B_1$ to $\u B_2$ and
    $(V_2,S_2)$ is a computation ep pair from $\u B_2$ to $\u B_3$, then
    $(V_2[V_1], S_1[S_2])$ is a computation ep pair from $\u B_1$ to $\u B_3$.
  \end{enumerate}
\end{lemma}
\begin{longproof}
  \begin{enumerate}
  \item First, retraction follows from retraction twice:
    \[ S_1[S_2[\ret V_2[V_1[x]]]] \equidyn S_1[\ret [V_1[x]]] \equidyn x \]
    and projection follows from projection twice:
    \begin{align*}
      \bindXtoYinZ {S_1[S_2[\bullet]]} x \ret V_2[V_1[x]]
      &\equidyn
      {\bindXtoYinZ {S_1[S_2[\bullet]]} x \bindXtoYinZ {\ret [V_1[x]]} y \ret V_2[y]}\tag{$\u F\beta$}\\
      &\equidyn
      \bindXtoYinZ {(\bindXtoYinZ {S_1[S_2[\bullet]]} x {\ret [V_1[x]]})} y \ret V_2[y]\tag{Commuting conversion}\\
      &\ltdyn
      \bindXtoYinZ {S_2[\bullet]} y \ret V_2[y]\tag{Projection}\\
      &\ltdyn \bullet \tag{Projection}
    \end{align*}
  \item Again retraction follows from retraction twice:
    \[ S_1[S_2[\force V_2[V_1[z]]]] \equidyn S_1[\force V_1[z]] \equidyn \force z \]
    and projection from projection twice:
    \begin{align*}
      V_2[V_1[\thunk S_1[S_2[\force w]]]]
      &\equidyn V_2[V_1[\thunk S_1[\force \thunk S_2[\force w]]]]\tag{$U\beta$}\\
      &\ltdyn V_2[\thunk S_2[\force w]] \tag{Projection}\\
      &\ltdyn w \tag{Projection}
    \end{align*}
  \end{enumerate}
\end{longproof}
\begin{longonly}
\begin{lemma}[Identity EP Pair]
  \label{ep-pair-id}
  $(x. x, \bullet)$ is an ep pair (value or computation).
\end{lemma}  
\end{longonly}

Now, we show that all casts are ep pairs.
The proof is a somewhat tedious, but straightforward calculation.

\begin{lemma}[Casts are EP Pairs]\hfill
  \label{lem:casts-are-ep-pairs}
  \begin{enumerate}
  \item For any $A \ltdyn A'$, the casts $(x.\sem{\upcast{A}{A'}x},
    \sem{\dncast{\u F A}{\u F A'}})$ are a value ep pair from
    $\sem{A}$ to $\sem{A'}$
  \item For any $\u B \ltdyn \u B'$, the casts $(z. \sem{\upcast{U \u
      B}{U \u B'}z}, \sem{\dncast{\u B}{\u B'}})$ are a computation ep
    pair from $\sem{\u B}$ to $\sem{\u B'}$.
  \end{enumerate}
\end{lemma}
\begin{longproof}
  By induction on normalized type dynamism derivations.
  \begin{enumerate}
  \item $A \ltdyn A$ ($A \in \{\dyn, 1\}$), because identity is an ep pair.
  \item $0 \ltdyn A$ (that $A \in \{ \dyn, 0 \}$ is not important):
    \begin{enumerate}
    \item Retraction is
      \[ x : 0 \vdash \ret x \equidyn \bindXtoYinZ {\ret\absurd x} y \err : \u F A \]
      which holds by $0\eta$
    \item Projection is
      \[ \bullet : \u F A \vdash \bindXtoYinZ {(\bindXtoYinZ \bullet y \err)} x {\ret\absurd x} \ltdyn \bullet : \u F A \]
      Which we calculate:
      \begin{align*}
        &\bindXtoYinZ {(\bindXtoYinZ \bullet y \err)} x {\ret\absurd x}\\
        &\equidyn \bindXtoYinZ \bullet y \bindXtoYinZ \err x {\ret\absurd x}\tag{comm conv}\\
        &\equidyn \bindXtoYinZ \bullet y \err \tag{Strictness of Stacks}\\
        &\ltdyn \bindXtoYinZ \bullet y \ret y \tag{$\err$ is $\bot$}\\
        &\equidyn \bullet \tag{$\u F\eta$}
      \end{align*}
    \end{enumerate}
  \item $+$:
    \begin{enumerate}
    \item Retraction is
      \begin{align*}
        &x : A_1 + A_2 \vdash\\
        &\sem{\dncast{\u F(A_1+A_2)}{\u F(A_1'+A_2')}}[\ret \sem{\upcast{A_1+A_2}{A_1'+A_2'}}[x]]\\
        &=\sdncast{\u F(A_1+A_2)}{\u F(A_1'+A_2')}[\ret\caseofXthenYelseZ x {x_1. \inl\supcast{A_1}{A_1'}[x_1]}{x_1. \inr\supcast{A_2}{A_2'}[x_2]}]\\
        &\equidyn
        \caseofX x\tag{commuting conversion}\\
        &\quad\thenY {x_1. \sdncast{\u F(A_1+A_2)}{\u F(A_1'+A_2')}[\ret\inl\supcast{A_1}{A_1'}[x_1]]}\\
        &\quad\elseZ {x_2. \sdncast{\u F(A_1+A_2)}{\u F(A_1'+A_2')}[\ret\inr\supcast{A_2}{A_2'}[x_2]]}\\
        &\equidyn
        \caseofX x\tag{cast computation}\\
        &\quad\thenY{x_1. \bindXtoYinZ {\sdncast{\u F A_1}{\u F A_1'}[\ret \supcast{A_1}{A_1'}x_1]} {x_1} \ret \inl x_1}\\
        &\quad\elseZ{x_2. \bindXtoYinZ {\sdncast{\u F A_2}{\u F A_2'}[\ret \supcast{A_2}{A_2'}x_2]} {x_2} \ret \inr x_2}\\
        &\equidyn \caseofXthenYelseZ x {x_1. \ret \inl x_1} {x_2. \ret \inr x_2}\tag{IH retraction}\\
        &\equidyn \ret x\tag{$+\eta$}
      \end{align*}
    \item For Projection:
      \begin{align*}
        &\bullet : A_1' + A_2' \vdash\\
        &\bindXtoYinZ {\sdncast{\u F(A_1+A_2)}{\u F(A_1'+A_2')}} x \supcast{A_1+A_2}{A_1'+A_2'}[x]\\
        &=
        \bindXtoYinZ {(\bindXtoYinZ \bullet {x'} \caseofXthenYelseZ {x'} {x_1'. \bindXtoYinZ {\sem{\dncast{\u FA_1}{\u FA_1'}}[\ret x_1']} {x_1} \ret\inl x_1}{x_2'. \cdots})} x\\
        &\quad\supcast{A_1+A_2}{A_1'+A_2'}\\
        &\equidyn
        \bindXtoYinZ \bullet x' \caseofX {x'} \tag{Commuting Conversion}\\
        &\qquad \thenY {x_1'. \bindXtoYinZ {\sem{\dncast{\u FA_1}{\u FA_1'}}[\ret x_1']} {x_1} \supcast{A_1+A_2}{A_1'+A_2'}{\ret\inl x_1}}\\
        &\qquad \elseZ {x_2'. \bindXtoYinZ {\sem{\dncast{\u FA_2}{\u FA_2'}}[\ret x_2']} {x_2} \supcast{A_1+A_2}{A_1'+A_2'}{\ret\inr x_2}}\\
        &\equidyn
        \bindXtoYinZ \bullet x' \caseofX {x'}\tag{Cast Computation}\\
        &\qquad \thenY {x_1'. \bindXtoYinZ {\sem{\dncast{\u FA_1}{\u FA_1'}}[\ret x_1']} {x_1} {\ret\inl \supcast{A_1}{A_1'}x_1}}\\
        &\qquad \elseZ {x_2'. \bindXtoYinZ {\sem{\dncast{\u FA_2}{\u FA_2'}}[\ret x_2']} {x_2} {\ret\inr \supcast{A_2}{A_2'}x_2}}\\
        &\ltdyn
        \bindXtoYinZ \bullet x' \caseofXthenYelseZ {x'} {x_1'. \ret\inl x_1'} {x_2'. \ret \inr x_2'}\tag{IH projection}\\
        &\equidyn \bindXtoYinZ \bullet x' \ret x'\tag{$+\eta$}\\
        &\equidyn \bullet \tag{$\u F\eta$}\\
      \end{align*}\
    \end{enumerate}
  \item $\times$:
    \begin{enumerate}
    \item First, Retraction:
      \begin{align*}
        &x : A_1\times A_2 \vdash\\
        &\sdncast{\u F(A_1\times A_2)}{\u F(A_1'\times A_2')}[\ret \supcast{A_1\times A_2}{A_1' \times A_2'}[x]]\\
        &=\sdncast{\u F(A_1\times A_2)}{\u F(A_1'\times A_2')}[\ret\pmpairWtoXYinZ x {x_1}{x_2} (\supcast{A_1}{A_1'}[x_1], \supcast{A_2}{A_2'}[x_2])]\\
        &\equidyn
        \pmpairWtoXYinZ x {x_1} {x_2} \sdncast{\u F(A_1\times A_2)}{\u F(A_1'\times A_2')}[\ret(\supcast{A_1}{A_1'}[x_1], \supcast{A_2}{A_2'}[x_2])]\tag{commuting conversion}\\
        &\equidyn
        \pmpairWtoXYinZ x {x_1} {x_2} \tag{cast reduction}\\
        &\quad\bindXtoYinZ {\sdncast{\u F A_1}{\u F A_1'}[\ret\supcast{A_1}{A_1'}[x_1]]} {y_1}\\
        &\quad\bindXtoYinZ {\sdncast{\u F A_2}{\u F A_2'}[\ret\supcast{A_2}{A_2'}[x_2]]} {y_2}\\
        &\quad\ret(y_1, y_2)\\
        &\equidyn
        \pmpairWtoXYinZ x {x_1} {x_2} \bindXtoYinZ {\ret x_1} {y_1} \bindXtoYinZ {\ret x_2} {y_2} \ret(y_1, y_2)\tag{IH retraction}\\
        &\equidyn
        \pmpairWtoXYinZ x {x_1} {x_2} \ret(x_1,x_2) \tag{$\u F\beta$}\\
        &\equidyn \ret x\tag{$\times\eta$}
      \end{align*}
    \item Next, Projection:
      \begin{align*}
        &\bullet : \u F A'\vdash\\
        &\bindXtoYinZ {\sdncast{\u F(A_1\times A_2)}{\u F(A_1'\times A_2')}[\bullet]} x \ret\supcast{A_1\times A_2}{A_1' \times A_2'}[x]\\
        &\equidyn\bindXtoYinZ \bullet {x'} \pmpairWtoXYinZ {{x'}} {x_1'}{x_2'} \tag{$\u F\eta, \times\eta$}\\
        &\quad \bindXtoYinZ {\sdncast{\u F(A_1\times A_2)}{\u F(A_1'\times A_2')}[\ret (x_1',x_2')]} x\\
        &\quad \ret\supcast{A_1\times A_2}{A_1' \times A_2'}[x]\\
        &\equidyn\bindXtoYinZ \bullet {x'} \pmpairWtoXYinZ {{x'}} {x_1'}{x_2'}\tag{cast reduction}\\
        &\quad \bindXtoYinZ {\sdncast{\u F A_1}{\u F A_1'}[\ret x_1']} {x_1}\\
        &\quad \bindXtoYinZ {\sdncast{\u F A_2}{\u F A_2'}[\ret x_2']} {x_2}\\
        &\quad \ret\supcast{A_1\times A_2}{A_1' \times A_2'}[(x_1,x_2)]\\
        &\equidyn\bindXtoYinZ \bullet {x'} \pmpairWtoXYinZ {{x'}} {x_1'}{x_2'}\tag{cast reduction}\\
        &\quad \bindXtoYinZ {\sdncast{\u F A_1}{\u F A_1'}[\ret x_1']} {x_1}\\
        &\quad \bindXtoYinZ {\sdncast{\u F A_2}{\u F A_2'}[\ret x_2']} {x_2}\\
        &\quad \ret(\supcast{A_1}{A_1'}[x_1], \supcast{A_2}{A_2'}[x_2])\\
        &\equidyn\bindXtoYinZ \bullet {x'} \pmpairWtoXYinZ {{x'}} {x_1'}{x_2'}\tag{$\u F \beta$, twice}\\
        &\quad \bindXtoYinZ {\sdncast{\u F A_1}{\u F A_1'}[\ret x_1']} {x_1}\\
        &\quad \bindXtoYinZ {\sdncast{\u F A_2}{\u F A_2'}[\ret x_2']} {x_2}\\
        &\quad \bindXtoYinZ {\ret \supcast{A_2}{A_2'}[x_2]}{y_2'}\\
        &\quad \bindXtoYinZ {\ret \supcast{A_1}{A_1'}[x_1]}{y_1'}\\
        &\quad \ret(y_1',y_2')\\
        &\ltdyn\bindXtoYinZ \bullet {x'} \pmpairWtoXYinZ {{x'}} {x_1'}{x_2'}\tag{IH Projection}\\
        &\quad \bindXtoYinZ {\sdncast{\u F A_1}{\u F A_1'}[\ret x_1']} {x_1}\\
        &\quad \bindXtoYinZ {\ret x_2'} {y_2'}\\
        &\quad \bindXtoYinZ {\ret \supcast{A_1}{A_1'}[x_1]} {y_1'}\\
        &\quad \ret(y_1',y_2')\\
        &\equidyn\bindXtoYinZ \bullet {x'} \pmpairWtoXYinZ {{x'}} {x_1'}{x_2'}\tag{$\u F\beta$}\\
        &\quad \bindXtoYinZ {\sdncast{\u F A_1}{\u F A_1'}[\ret x_1']} {x_1}\\
        &\quad \bindXtoYinZ {\ret \supcast{A_1}{A_1'}[x_1]} {y_1'}\\
        &\quad \ret(x_1',y_2')\\
        &\ltdyn\bindXtoYinZ \bullet {x'} \pmpairWtoXYinZ {{x'}} {x_1'}{x_2'}\tag{IH Projection}\\
        &\quad \bindXtoYinZ {\ret x_1'} {y_1'}\\
        &\quad \ret(x_1',y_2')\\
        &\equidyn\bindXtoYinZ \bullet {x'} \pmpairWtoXYinZ {{x'}} {x_1'}{x_2'}
        \ret(x_1',x_2')\tag{$\u F\beta$}\\
        &\equidyn\bindXtoYinZ \bullet {x'} \ret {x'}\tag{$\times\eta$}\\
        &\equidyn\bullet \tag{$\u F \eta$}
      \end{align*}
    \end{enumerate}
  \item $U$: By inductive hypothesis, $(x.\sem{\upcast{U\u B}{U\u
      B'}}, \dncast{\u B}{\u B'})$ is a computation ep pair
    \begin{enumerate}
    \item To show retraction we need to prove:
      \[
      x : U \u B \vdash \ret x \equidyn \bindXtoYinZ {(\ret \thunk {\sem{\upcast{U\u B}{U \u B'}}})} y {\ret \thunk \sem{\dncast{\u B}{\u B'}}[\force y]} : \u F U \u B'
      \]
      Which we calculate as follows:
      \begin{align*}
        &x : U\u B \vdash \\
        &\sdncast{\u FU\u B}{\u FU\u B'}[{(\ret {\sem{\upcast{U\u B}{U \u B'}}[x]})}]\\
        &\equidyn
        \ret\thunk(\sdncast{\u B}{\u B'}[\force {\sem{\upcast{U\u B}{U \u B'}}}[x]])\tag{Cast Reduction}\\
        &\equidyn \ret\thunk \force x \tag{IH Retraction}\\
        &\equidyn \ret x \tag{$U\eta$}
      \end{align*}
    \item To show projection we calculate:
      \begin{align*}
        &\bindXtoYinZ {\sdncast{\u FU\u B}{\u FU\u B'}[\bullet]} x \supcast{U\u B}{U\u B'}[x]\\
        &\equidyn
        \bindXtoYinZ \bullet {x'} \bindXtoYinZ {\sdncast{\u FU\u B}{\u FU\u B'}[\ret x']} x \supcast{U\u B}{U\u B'}[x]\tag{$\u F\eta$}\\
        &\equidyn
        \bindXtoYinZ \bullet {x'} \bindXtoYinZ {\ret\thunk(\sdncast{\u B}{\u B'}[\force x'])} x \supcast{U\u B}{U\u B'}[x]\tag{Cast Reduction}\\
        &\equidyn
        \bindXtoYinZ \bullet {x'} \supcast{U\u B}{U\u B'}[\thunk(\sdncast{\u B}{\u B'}[\force x'])] \tag{$\u F\beta$}\\
        &\ltdyn \bindXtoYinZ \bullet {x'} x'\tag{IH Projection}\\
        &\equidyn \bullet \tag{$\u F\eta$}
      \end{align*}
    \end{enumerate}
  \end{enumerate}

  \begin{enumerate}
  \item There's a few base cases about the dynamic computation type, then
  \item $\top$:
    \begin{enumerate}
    \item Retraction is by $\top\eta$:
      \begin{align*}
        z : U \top \vdash \force z \equidyn \{ \} : \top
      \end{align*}
    \item Projection is
      \begin{align*}
        \thunk \err
        &\ltdyn \thunk \force w \tag{$\err$ is $\bot$}\\
        &\equidyn w \tag{$U\eta$}
      \end{align*}
    \end{enumerate}
  \item $\with$:
    \begin{enumerate}
    \item Retraction
      \begin{align*}
        &z : U (\u B_1 \with \u B_2)\vdash \\
        &\sdncast{\u B_1 \with \u B_2}{\u B_1' \with \u B_2'}[\force \supcast{U(\u B_1 \with \u B_2)}{U(\u B_1' \with \u B_2')}[z]]\\
        &\equidyn
        \pairone{\pi\sdncast{\u B_1 \with \u B_2}{\u B_1' \with \u B_2'}[\force \supcast{U(\u B_1 \with \u B_2)}{U(\u B_1' \with \u B_2')}[z]]} \tag{$\with\eta$}\\
        &\qquad\pairtwo{\pi' \sdncast{\u B_1 \with \u B_2}{\u B_1' \with \u B_2'}[\force \supcast{U(\u B_1 \with \u B_2)}{U(\u B_1' \with \u B_2')}[z]]}\\
        &\equidyn
        \pairone{\sdncast{\u B_1}{\u B_1'}[\pi\force \supcast{U(\u B_1 \with \u B_2)}{U(\u B_1' \with \u B_2')}[z]]} \tag{Cast reduction}\\
        &\qquad\pairtwo{\sdncast{\u B_2}{\u B_2'}[\pi'\force \supcast{U(\u B_1 \with \u B_2)}{U(\u B_1' \with \u B_2')}[z]]}\\
        &\equidyn
        \pairone{\sdncast{\u B_1}{\u B_1'}[\force\supcast{U\u B_1}{U\u B_1'}[\thunk \pi \force z]]} \tag{Cast reduction}\\
        & \qquad\pairtwo{\sdncast{\u B_2}{\u B_2'}[\force\supcast{U\u B_2}{U\u B_2'}[\thunk \pi' \force z]]}\\
        &\equidyn
        \pair{\force \thunk \pi \force z}{\force \thunk \pi' \force z} \tag{IH retraction}\\
        &\equidyn \pair{\pi \force z}{\pi' \force z}\tag{$U\beta$}\\
        &\equidyn \force z \tag{$\with\eta$}
      \end{align*}
    \item Projection
      \begin{align*}
        &w : U {\u B_1' \with \u B_2'} \vdash\\
        & \supcast{U(\u B_1 \with \u B_2)}{U(\u B_1' \with \u B_2')}[\thunk \sdncast{\u B_1 \with \u B_2}{\u B_1' \with \u B_2'}[\force w]]\\
        &\equidyn
        \thunk\force\supcast{U(\u B_1 \with \u B_2)}{U(\u B_1' \with \u B_2')}[\thunk \sdncast{\u B_1 \with \u B_2}{\u B_1' \with \u B_2'}[\force w]]\tag{$U\eta$}\\
        &\equidyn
        \thunk\pairone{\pi\force\supcast{U(\u B_1 \with \u B_2)}{U(\u B_1' \with \u B_2')}[\thunk \sdncast{\u B_1 \with \u B_2}{\u B_1' \with \u B_2'}[\force w]]}\\
        &\qquad\qquad\pairtwo{\pi'\force\supcast{U(\u B_1 \with \u B_2)}{U(\u B_1' \with \u B_2')}[\thunk \sdncast{\u B_1 \with \u B_2}{\u B_1' \with \u B_2'}[\force w]]}\tag{$\with\eta$}\\
        &\equidyn
        \thunk\pairone{\force\supcast{U\u B_1}{U\u B_1'}[\thunk\pi\force\thunk \sdncast{\u B_1 \with \u B_2}{\u B_1' \with \u B_2'}[\force w]]}\\
        &\qquad\qquad\pairtwo{\force\supcast{U\u B_2}{U\u B_2'}[\thunk\pi'\force\thunk \sdncast{\u B_1 \with \u B_2}{\u B_1' \with \u B_2'}[\force w]]}\tag{cast reduction}\\
        &\equidyn
        \thunk\pairone{\force\supcast{U\u B_1}{U\u B_1'}[\thunk\pi \sdncast{\u B_1 \with \u B_2}{\u B_1' \with \u B_2'}[\force w]]} \tag{$U\beta$}\\
        &\qquad\qquad\pairtwo{\force\supcast{U\u B_2}{U\u B_2'}[\thunk\pi' \sdncast{\u B_1 \with \u B_2}{\u B_1' \with \u B_2'}[\force w]]}\\
        &\equidyn
        \thunk\pairone{\force\supcast{U\u B_1}{U\u B_1'}[\thunk\sdncast{\u B_1}{\u B_1'}[\pi\force w]]} \tag{cast reduction}\\
        &\qquad\qquad\pairtwo{\force\supcast{U\u B_2}{U\u B_2'}[\thunk\sdncast{\u B_2}{\u B_2'}[\pi'\force w]]}\\
        &\equidyn
        \thunk\pairone{\force\supcast{U\u B_1}{U\u B_1'}[\thunk\sdncast{\u B_1}{\u B_1'}[\force\thunk\pi\force w]]} \tag{$U\beta$}\\
        &\qquad\qquad\pairtwo{\force\supcast{U\u B_2}{U\u B_2'}[\thunk\sdncast{\u B_2}{\u B_2'}[\force\thunk\pi'\force w]]}\\
        &\ltdyn
        \thunk\pair{\force\thunk\pi\force w}{\force\thunk\pi'\force w} \tag{IH projection}\\
        &\equidyn
        \thunk\pair{\pi\force w}{\pi'\force w} \tag{$U\beta$}\\
        &\equidyn \thunk\force w \tag{$\with\eta$}\\
        &\equidyn w \tag{$U\eta$}\\
      \end{align*}
    \end{enumerate}
  \item $\to$:
    \begin{enumerate}
    \item Retraction
      \begin{align*}
        &z : U (A \to \u B)\vdash \\
        &\sdncast{A \to \u B}{A' \to \u B'}[\force \supcast{U(A \to \u B)}{U(A' \to \u B')}[z]]\\
        &\equidyn
        \lambda x:A. (\sdncast{A \to \u B}{A' \to \u B'}[\force \supcast{U(A \to \u B)}{U(A' \to \u B')}[z]])\,x\tag{$\to\eta$}\\
        &\equidyn
        \lambda x:A.
        \sdncast{\u B}{\u B'}[(\force \supcast{U(A \to \u B)}{U(A' \to \u B')}[z])(\supcast{A}{A'}[x])] \tag{cast reduction}\\
        &\equidyn
        \lambda x:A.\tag{cast reduction}\\
        &\quad
        \sdncast{\u B}{\u B'}[\bindXtoYinZ{\sdncast{\u F A}{\u F A'}[\ret \upcast{A}{A'}[x]]} {y} \force \upcast{U \u B}{U\u B'}[\thunk((\force z)\, y)]]\\
        &\equidyn
        \lambda x:A. \sdncast{\u B}{\u B'}[\bindXtoYinZ{\ret x} {y} \force \upcast{U \u B}{U\u B'}[\thunk((\force z)\, y)]]\tag{IH Retraction}\\
        &\equidyn
        \lambda x:A. \sdncast{\u B}{\u B'}[\force \upcast{U \u B}{U\u B'}[\thunk((\force z)\, x)]]\tag{$\u F\beta$}\\
        &\equidyn \lambda x:A. \force \thunk((\force z)\, x) \tag{IH retraction}\\
        &\equidyn \lambda x:A. (\force z)\, x \tag{$U\beta$}\\
        &\equidyn \force z \tag{$\to\eta$}\\
      \end{align*}
    \item Projection
      \begin{align*}
        &w : U (A' \to \u B') \vdash\\
        & \supcast{U(A \to \u B)}{U(A' \to \u B')}[\thunk \sdncast{A \to \u B}{A' \to \u B'}[\force w]]\\
        &\equidyn
        \thunk\force\supcast{U(A \to \u B)}{U(A' \to \u B')}[\thunk \sdncast{A \to \u B}{A' \to \u B'}[\force w]]\tag{$U\eta$}\\
        &\equidyn
        \thunk\lambda x':A'.\\
        &\quad(\force\supcast{U(A \to \u B)}{U(A' \to \u B')}[\thunk \sdncast{A \to \u B}{A' \to \u B'}[\force w]])\,x'\tag{$\to\eta$}\\
        &\equidyn \thunk\lambda x':A'.\\
        &\qquad \bindXtoYinZ {\sdncast{\u F A}{\u F A'}[\ret x']} x\tag{cast reduction}\\
        &\qquad
        \force\supcast{U\u B}{U \u B'}[\thunk((\force \thunk \sdncast{A \to \u B}{A' \to \u B'}[\force w])\, x)]\\
        &\equidyn \thunk\lambda x':A'.\\
        &\qquad \bindXtoYinZ {\sdncast{\u F A}{\u F A'}[\ret x']} x\tag{$U\beta$}\\
        &\qquad
        \force\supcast{U\u B}{U \u B'}[\thunk((\sdncast{A \to \u B}{A' \to \u B'}[\force w])\, x)]\\
        &\equidyn
        \thunk\lambda x':A'.\\
        &\qquad \bindXtoYinZ {\sdncast{\u F A}{\u F A'}[\ret x']} x\tag{cast reduction}\\
        &\qquad
        \force\supcast{U\u B}{U \u B'}[\thunk\sdncast{\u B}{\u B'}[(\force w)\,(\upcast{A}{A'}[x])]]\\
        &\equidyn
        \thunk\lambda x':A'.\\
        &\qquad\bindXtoYinZ {\sdncast{\u F A}{\u F A'}[\ret x']} x\tag{$\u F\beta$}\\
        &\qquad\bindXtoYinZ {\ret{\upcast{A}{A'}[x]}} {x'}\\
        &\qquad\force\supcast{U\u B}{U \u B'}[\thunk\sdncast{\u B}{\u B'}[(\force w)\,x']]\\
        &\ltdyn
        \thunk\lambda x':A'.\tag{IH projection}\\
        &\qquad\bindXtoYinZ {\ret x'} {x'}\\
        &\qquad\force\supcast{U\u B}{U \u B'}[\thunk\sdncast{\u B}{\u B'}[(\force w)\,x']]\\
        &\equidyn
        \thunk\lambda x':A'.
        \force\supcast{U\u B}{U \u B'}[\thunk\sdncast{\u B}{\u B'}[(\force w)\,x']]\tag{$\u F\beta$}\\
        &\equidyn
        \thunk\lambda x':A'.
        \force\supcast{U\u B}{U \u B'}[\thunk\sdncast{\u B}{\u B'}[\force\thunk((\force w)\,x')]]\tag{$\u F\beta$}\\
        &\ltdyn
        \thunk\lambda x':A'. \force\thunk((\force w)\,x')\tag{IH projection}\\
        &\equidyn \thunk\lambda x':A'. ((\force w)\,x')\tag{$U\beta$}\\
        &\equidyn \thunk\force w\tag{$\to\eta$}\\
        &\equidyn w \tag{$U\eta$}\\
      \end{align*}
    \end{enumerate}
  \item $\u F$:
    \begin{enumerate}
    \item To show retraction we need to show
      \[
      z : U \u F A \vdash
      \force z \equidyn
      \sem{\dncast{\u F A}{\u F A'}}[\force \thunk (\bindXtoYinZ {\force z} x \ret \sem{\upcast{A}{A'}})]
      \]
      We calculate:
      \begin{align*}
        &\sem{\dncast{\u F A}{\u F A'}}[\force \thunk (\bindXtoYinZ {\force z} x \ret \sem{\upcast{A}{A'}})]\\
        &\equidyn
        \sem{\dncast{\u F A}{\u F A'}}[(\bindXtoYinZ {\force z} x \ret \sem{\upcast{A}{A'}})]\tag{$U\beta$}\\
        &\equidyn
        \bindXtoYinZ {\force z} x \sem{\dncast{\u F A}{\u F A'}}[\ret \sem{\upcast{A}{A'}}] \tag{comm conv}\\
        &\equidyn
        \bindXtoYinZ {\force z} x \ret x \tag{IH value retraction}\\
        &\equidyn \force z \tag{$\u F\eta$}        
      \end{align*}
    \item To show projection we need to show
      \[
      w : U \u F A' \vdash
      \thunk {(\bindXtoYinZ {\force {\thunk \sem{\dncast{\u F A}{\u F A'}}[\force w]}} x \ret \sem{\upcast{A}{A'}})}
      \ltdyn w : U \u B'
      \]
      We calculate as follows
      \begin{align*}
        &\thunk {(\bindXtoYinZ {\force {\thunk \sem{\dncast{\u F A}{\u F A'}}[\force w]}} x \ret \sem{\upcast{A}{A'}})}\\
        &\equidyn\thunk {(\bindXtoYinZ {{\sem{\dncast{\u F A}{\u F A'}}[\force w]}} x \ret \sem{\upcast{A}{A'}})} \tag{$U\beta$}\\
        &\ltdyn \thunk {\force w} \tag{IH value projection}\\
        & \equidyn w \tag{$U\eta$}
      \end{align*}
    \end{enumerate}
  \end{enumerate}
\end{longproof}

While the above was tedious, this pays off greatly in later proofs: this
is the \emph{only} proof in the entire development that needs to inspect
the definition of a ``shifted'' cast (a downcast between $\u F$ types or
an upcast between $U$ types).
All later lemmas
have cases for these shifted casts, but \emph{only} use the property
that they are part of an ep pair.
This is one of the biggest advantages of using an explicit syntax for
complex values and complex stacks: the shifted casts are the only ones
that non-trivially use effectful terms, so after this lemma is
established we only have to manipulate values and stacks, which
compose much more nicely than effectful terms.
Conceptually, the main reason we can avoid reasoning about the
definitions of the shifted casts directly is that any two shifted casts
that form an ep pair with the same value embedding/stack projection are
equal:
\begin{lemma}[Value Embedding determines Projection, Computation Projection determines Embedding]
  \label{lem:adjoints-unique-cbpvstar}
  For any value $x : A \vdash V_e : A'$ and stacks $\bullet : \u F A'
  \vdash S_1 : \u F A$ and $\bullet : \u F A' \vdash S_2 : \u F A$, if
  $(V_e, S_1)$ and $(V_e, S_2)$ are both value ep pairs, then
  \[ S_1 \equidyn S_2 \]

  Similarly for any values $x : U\u B \vdash V_1 : U \u B'$ and $x :
  U\u B \vdash V_2 : U \u B'$ and stack $\bullet : \u B' \vdash S_p :
  \u B$, if $(V_1, S_p)$ and $(V_2, S_p)$ are both computation ep pairs then
  \[ V_1 \equidyn V_2 \]
\end{lemma}
\begin{longproof}
  By symmetry it is sufficient to show $S_1 \ltdyn S_2$.

  \begin{mathpar}
    \inferrule
    {\inferrule
    {\inferrule
    {\inferrule
    {S_1 \ltdyn S_1}
    {\bindXtoYinZ {S_1} x \ret x \ltdyn \bindXtoYinZ \bullet x S_1[\ret x]}}
    {\bindXtoYinZ {S_1} x \ret V_e \ltdyn \bindXtoYinZ \bullet x \ret x}}
    {\bindXtoYinZ {S_1} x \ret x \ltdyn \bindXtoYinZ \bullet x S_2[\ret x]}}
    {\bullet : \u F A' \vdash S_1 \ltdyn S_2 : \u F A}
  \end{mathpar}

  similarly to show $V_1 \ltdyn V_2$:
  \begin{mathpar}
    \inferrule%
    {\inferrule
    {\inferrule
    {x : U \u B \vdash \thunk\force V_2 \ltdyn \thunk \force V_2 : U \u B'}
    {x : U \u B \vdash \thunk \force x \ltdyn \thunk S_p[\force V_2]}}
    {x : U \u B \vdash \thunk\force V_1 \ltdyn \thunk \force V_2 : U \u B'}}
    {x : U \u B \vdash V_1 \ltdyn V_2 : U \u B'}
  \end{mathpar}
\end{longproof}

The next two lemmas on the way to axiomatic graduality show that
Figure~\ref{fig:cast-to-contract} translates $\upcast{A}{A}$ to the
identity and $\upcast{A'}{A''}{\upcast{A}{A'}}$ to the same contract as
$\upcast{A}{A''}$, and similarly for downcasts.
Intuitively, for all connectives except $\u F, U$, this is because of
functoriality of the type constructors on values and stacks.
For the $\u F, U$ cases, we will use the corresponding fact about the
dual cast, i.e., to prove the $\u F A$ to $\u F A$ downcast is the
identity stack, we know by inductive hypothesis that the $A$ to $A$
upcast is the identity, and that the identity stack is a projection
for the identity.
Therefore Lemma~\ref{lem:adjoints-unique-cbpvstar} implies that the $\u
FA$ downcast must be equivalent to the identity.
We now discuss these two lemmas and their proofs in detail.  

First, we show that the casts from a type to itself are equivalent to
the identity.
Below, we will use this lemma to prove the reflexivity case of the
axiomatic graduality theorem, and to prove a conservativity result,
which says that a GTT homogeneous term dynamism is the same as a
\cbpvstar\/ inequality between their translations.
\begin{lemma}[Identity Expansion]
  \label{lem:ident-expansion}
  For any $A$ and $\u B$,
  \begin{mathpar}
    x:A \vdash \sem{\upcast{A}{A}} \equidyn x : A\and
    \bullet : \u B \vdash \sem{\dncast{\u B}{\u B}} \equidyn \bullet : \u B
  \end{mathpar}
\end{lemma}
\begin{proof}
  We proceed by induction on $A, \u B$, following the proof that
  reflexivity is admissible given in Lemma \ref{lem:norm-type-dyn}.
  \begin{enumerate}
  \item If $A \in \{1, \dynv \}$, then $\supcast{A}{A}[x] = x$.
  \item If $A = 0$, then $\absurd x \equidyn x$ by $0\eta$.
  \item If $A = U \u B$, then by inductive hypothesis $\sdncast{\u
    B}{\u B} \equidyn \bullet$. By Lemma \ref{ep-pair-id},
    $(x. x, \bullet)$ is a computation ep pair from $\u B$ to
    itself. But by Lemma \ref{lem:casts-are-ep-pairs}, $(\supcast{U\u
      B}{U\u B}[x], \bullet)$ is also a computation ep pair so the
    result follows by uniqueness of embeddings from computation
    projections Lemma \ref{lem:adjoints-unique-cbpvstar}.
  \item If $A = A_1\times A_2$ or $A = A_1+A_2$, the result follows by
    the $\eta$ principle and inductive hypothesis.
  \item If $\u B = \dync$, $\sdncast{\dync}{\dync} = \bullet$.
  \item For $\u B = \top$, the result follows by $\top\eta$.
  \item For $\u B = \u B_1 \with \u B_2$ or $\u B = A \to \u B'$, the
    result follows by inductive hypothesis and $\eta$.
  \item For $\u B = \u FA$, by inductive hypothesis, the downcast is a
    projection for the value embedding $x.x$, so the result follows by
    identity ep pair and uniqueness of projections from value
    embeddings.
  \end{enumerate}
\end{proof}

Second, we show that a composition of upcasts is translated to the same
thing as a direct upcast, and similarly for downcasts.  Below, we will
use this lemma to translate \emph{transitivity} of term dynamism in GTT.
\begin{lemma}[Cast Decomposition]
  For any dynamic type interpretation $\rho$,
  \begin{small}
  \begin{mathpar}
    \inferrule
    {A \ltdyn A' \ltdyn A''}
    {x : A \vdash \srho{\upcast A {A''}} \equidyn \srho{\upcast {A'} {A''}}[\srho{\upcast {A} {A'}}] : A''}

    \inferrule
    {\u B \ltdyn \u B' \ltdyn \u B''}
    {\bullet : \u B'' \vdash \srho{\dncast{\u B}{\u B''}} \equidyn
      \srho{\dncast{\u B}{\u B'}}[\srho{\dncast{\u B'}{\u B''}}]}
  \end{mathpar}
  \end{small}
\end{lemma}
\begin{longproof}
  By mutual induction on $A, \u B$.
  \begin{enumerate}
  \item $A \ltdyn A' \ltdyn A''$
    \begin{enumerate}
    \item If $A = 0$, we need to show $x : 0 \vdash
      \supcast{0}{A''}[x] \equidyn
      \supcast{A'}{A''}[\supcast{0}{A'}[x]] : A''$ which follows by
      $0\eta$.
    \item If $A = \dynv$, then $A' = A'' = \dynv$, and both casts are
      the identity.
    \item If $A \not\in \{\dynv, 0 \}$ and $A' = \dynv$, then $A'' =
      \dynv$ and $\supcast{\dynv}{\dynv}[\supcast{A}{\dynv}] =
      \supcast{A}{\dynv}$ by definition.
    \item If $A, A' \not\in \{\dynv, 0 \}$ and $A'' = \dynv$, then
      $\floor A = \floor {A'}$, which we call $G$ and
      \[ \supcast{A}{\dynv} = \supcast{G}{\dynv}[\supcast{A}{G}] \]
      and
      \[ \supcast{A'}{\dynv}[\supcast{A}{A'}] = \supcast{G}{\dynv}[\supcast{A'}{G}[\supcast{A}{A'}]] \]
      so this reduces to the case for $A \ltdyn A' \ltdyn G$, below.
    \item If $A,A',A'' \not\in \{\dynv, 0 \}$, then they all have the same
      top-level constructor:
      \begin{enumerate}
      \item $+$: We need to show for $A_1 \ltdyn A_1' \ltdyn A_1''$
        and $A_2 \ltdyn A_2' \ltdyn A_2''$:
        \[
        x : \sem{A_1} + \sem{A_2} \vdash
        \supcast{A_1'+A_2'}{A_1''+A_2''}[\supcast{A_1+A_2}{A_1'+A_2'}[x]]\equidyn
        \supcast{A_1+A_2}{A_1''+A_2''}[x]
        : \sem{A_1''}+\sem{A_2''}.
        \]
        We proceed as follows:
        \begin{align*}
          &\supcast{A_1'+A_2'}{A_1''+A_2''}[\supcast{A_1+A_2}{A_1'+A_2'}[x]]\\
          &\equidyn \caseofX {x}\tag{$+\eta$}\\
          &\qquad\thenY {x_1. \supcast{A_1'+A_2'}{A_1''+A_2''}[\supcast{A_1+A_2}{A_1'+A_2'}[\inl x_1]]}\\
          &\qquad\elseZ {x_2. \supcast{A_1'+A_2'}{A_1''+A_2''}[\supcast{A_1+A_2}{A_1'+A_2'}[\inr x_2]]}\\
          &\equidyn \caseofX {x}\tag{cast reduction}\\
          &\qquad\thenY {x_1. \supcast{A_1'+A_2'}{A_1''+A_2''}[\inl\supcast{A_1}{A_1'}[x_1]]}\\
          &\qquad\elseZ {x_2. \supcast{A_1'+A_2'}{A_1''+A_2''}[\inr\supcast{A_2}{A_2'}[x_2]]}\\
          &\equidyn \caseofX {x}\tag{cast reduction}\\
          &\qquad\thenY {x_1. \inl\supcast{A_1'}{A_1''}[\supcast{A_1}{A_1'}[x_1]]}\\
          &\qquad\elseZ {x_2. \inr\supcast{A_2'}{A_2''}[\supcast{A_2}{A_2'}[x_2]]}\\
          &\equidyn \caseofX {x}\tag{IH}\\
          &\qquad\thenY {x_1. \inl\supcast{A_1}{A_1''}[x_1]}\\
          &\qquad\elseZ {x_2. \inr\supcast{A_2}{A_2''}[x_2]}\\
          &= \supcast{A_1+A_2}{A_1''+A_2''}[x] \tag{definition}
        \end{align*}
      \item $1$: By definition both sides are the identity.
      \item $\times$: We need to show for $A_1 \ltdyn A_1' \ltdyn A_1''$
        and $A_2 \ltdyn A_2' \ltdyn A_2''$:
        \[
        x : \sem{A_1} \times \sem{A_2} \vdash
        \supcast{A_1'\times A_2'}{A_1''\times A_2''}[\supcast{A_1\times A_2}{A_1'\times A_2'}[x]]\equidyn
        \supcast{A_1\times A_2}{A_1''\times A_2''}[x]
        : \sem{A_1''}\times \sem{A_2''}.
        \]
        We proceed as follows:
        \begin{align*}
          &\supcast{A_1'\times A_2'}{A_1''\times A_2''}[\supcast{A_1\times A_2}{A_1'\times A_2'}[x]]\\
          &\equidyn\pmpairWtoXYinZ x y z \supcast{A_1'\times A_2'}{A_1''\times A_2''}[\supcast{A_1\times A_2}{A_1'\times A_2'}[(y,z)]]\tag{$\times\eta$}\\
          &\equidyn\pmpairWtoXYinZ x y z \supcast{A_1'\times A_2'}{A_1''\times A_2''}[(\supcast{A_1}{A_1'}[y], \supcast{A_2}{A_2'}[z])]\tag{cast reduction}\\
          &\equidyn\pmpairWtoXYinZ x y z (\supcast{A_1'}{A_1''}[\supcast{A_1}{A_1'}[y]], \supcast{A_2'}{A_2''}[\supcast{A_2}{A_2'}[z]])\tag{cast reduction}\\
          &\equidyn\pmpairWtoXYinZ x y z (\supcast{A_1}{A_1''}[y], \supcast{A_2}{A_2''}[z])\tag{IH}\\
          &=\supcast{A_1\times A_2}{A_1'' \times A_2''}[x]\tag{definition}
        \end{align*}
      \item $U \u B \ltdyn U \u B' \ltdyn U \u B''$.
        We need to show
        \[
        x : U \u B \vdash \supcast{U\u B'}{U\u B''}[\supcast{U\u B}{U\u B'}[x]] \equidyn
        \supcast{U\u B}{U\u B''}[x] : U\u B''
        \]
        By composition of ep pairs, we know $(x.\supcast{U\u B'}{U\u
          B''}[\supcast{U\u B}{U\u B'}[x]], \sdncast{\u B}{\u
          B'}[\sdncast{\u B'}{\u B''}])$ is a computation ep pair.
        Furthermore, by inductive hypothesis, we know
        \[  \sdncast{\u B}{\u B'}[\sdncast{\u B'}{\u B''}] \equidyn \sdncast{\u B}{\u B''}\]
        so then both sides form ep pairs paired with $\sdncast{\u
          B}{\u B''}$, so it follows because computation projections
        determine embeddings \ref{lem:adjoints-unique-cbpvstar}.
      \end{enumerate}
    \end{enumerate}
  \item $\u B \ltdyn \u B' \ltdyn \u B''$
    \begin{enumerate}
    \item If $\u B = \top$, then the result is immediate by $\eta\top$.
    \item If $\u B = \dync$, then $\u B' = \u B'' = \dync$ then both
      sides are just $\bullet$.
    \item If $\u B \not\in \{\dync, \top\}$, and $\u B' = \dync$, then
      $\u B'' = \dync$
      \[ \sdncast{\u B}{\dync}[\sdncast{\dync}{\dync}] = \sdncast{\u B}{\dync} \]

    \item If $\u B,\u B' \not\in \{\dync,\top\}$, and $\u B'' = \dync$ , and $\floor {\u B} = \floor {\u B'}$, which we
      call $\u G$. Then we need to show
      \[ \sdncast{\u B}{\u B'}[\sdncast{\u B'}{\u G}[\sdncast{\u G}{\dync}]]
      \equidyn
      \sdncast{\u B}{\u G}[\sdncast{\u G}[\dync]]
      \]
      so the result follows from the case $\u B \ltdyn \u B' \ltdyn \u
      G$, which is handled below.
    \item If $\u B,\u B',\u B'' \not\in \{\dync, \top\}$, then they all have the
      same top-level constructor:
      \begin{enumerate}
      \item $\with$ We are given $\u B_1 \ltdyn \u B_1' \ltdyn \u
        B_1''$ and $\u B_2 \ltdyn \u B_2' \ltdyn \u B_2''$ and we need to show
        \[
        \bullet : \u B_1'' \with \u B_2''
        \vdash
        \sdncast{\u B_1 \with \u B_2}{\u B_1' \with \u B_2'}[\sdncast{\u B_1' \with \u B_2'}{\u B_1'' \with \u B_2''}]
        : \u B_1 \with \u B_2
        \]
        We proceed as follows:
        \begin{align*}
          &\sdncast{\u B_1 \with \u B_2}{\u B_1' \with \u B_2'}[\sdncast{\u B_1' \with \u B_2'}{\u B_1'' \with \u B_2''}]\\
          &\equidyn\pairone{\pi\sdncast{\u B_1 \with \u B_2}{\u B_1' \with \u B_2'}[\sdncast{\u B_1' \with \u B_2'}{\u B_1'' \with \u B_2''}]}\tag{$\with\eta$}\\
          &\quad\pairtwo{\pi'\sdncast{\u B_1 \with \u B_2}{\u B_1' \with \u B_2'}[\sdncast{\u B_1' \with \u B_2'}{\u B_1'' \with \u B_2''}]}\\
          &\equidyn\pairone{\sdncast{\u B_1}{\u B_1'}[\pi\sdncast{\u B_1' \with \u B_2'}{\u B_1'' \with \u B_2''}]}\tag{cast reduction}\\
          &\quad\pairtwo{\sdncast{\u B_2}{\u B_2'}[\pi'\sdncast{\u B_1' \with \u B_2'}{\u B_1'' \with \u B_2''}]}\\
          &\equidyn\pairone{\sdncast{\u B_1}{\u B_1'}[\sdncast{\u B_1'}{\u B_1''}[\pi\bullet]]}\tag{cast reduction}\\
          &\quad\pairtwo{\sdncast{\u B_2}{\u B_2'}\sdncast{\u B_2'}{\u B_2''}[\pi'\bullet]}\\
          &\equidyn\pair{\sdncast{\u B_1}{\u B_1''}[\pi\bullet]}{\sdncast{\u B_2}{\u B_2''}[\pi'\bullet]}\tag{IH}\\
          &= \sdncast{\u B_1 \with \u B_2}{\u B_1'' \with \u B_2''} \tag{definition}
        \end{align*}
      \item $\to$, assume we are given $A \ltdyn A' \ltdyn A''$ and
        $\u B \ltdyn \u B' \ltdyn \u B''$, then we proceed:
        \begin{align*}
          &\sdncast{A \to \u B}{A' \to \u B'}[\sdncast{A' \to \u B'}{A'' \to \u B''}]\\
          &\equidyn \lambda x:A. (\sdncast{A \to \u B}{A' \to \u B'}[\sdncast{A' \to \u B'}{A'' \to \u B''}][\bullet])\,x\tag{$\to\eta$}\\
          &\equidyn \lambda x:A. \sdncast{\u B}{\u B'}[(\sdncast{A' \to \u B'}{A'' \to \u B''}[\bullet])\, \supcast{A}{A'}[x]] \tag{cast reduction}\\
          &\equidyn \lambda x:A. \sdncast{\u B}{\u B'}[\sdncast{\u B'}{\u B''}[\bullet\, \supcast{A'}{A''}[\supcast{A}{A'}[x]]]]\tag{cast reduction}\\
          &\equidyn \lambda x:A. \sdncast{\u B}{\u B''}[\bullet\,\supcast{A}{A''}[x]]\\
          &= \sdncast{A \to \u B}{A \to \u B''}[\bullet]\tag{definition}
        \end{align*}
      \item $\u F A \ltdyn \u F A' \ltdyn \u F A''$. First, by
        composition of ep pairs, we know
        \[ (x. \supcast{A'}{A''}[\supcast{A}{A'}[x]], \sdncast{\u F
          A}{\u F A'})[\sdncast{\u F A'}{\u F A''}]\]
        form a value ep pair.  
        Furthermore, by inductive hypothesis, we know
        \[ x : A \vdash \supcast{A'}{A''}[\supcast{A}{A'}[x]] \equidyn \supcast{A}{A''}[x] \]
        so the two sides of our equation are both projections with the
        same value embedding, so the equation follows from uniqueness
        of projections from value embeddings.        
      \end{enumerate}
    \end{enumerate}
  \end{enumerate}
\end{longproof}

The final lemma before the graduality theorem lets us ``move a cast''
from left to right or vice-versa, via the adjunction property for ep
pairs.  
These arise in the proof cases for $\kw{return}$ and $\kw{thunk}$, because in those
cases the inductive hypothesis is in terms of an upcast (downcast) and
the conclusion is in terms of a a downcast (upcast).
\begin{lemma}[Hom-set formulation of Adjunction]
  \label{lem:hom-set-adj}
  For any value embedding-projection pair $V_e,S_p$ from $A$ to $A'$,
  the following are equivalent:
  \begin{small}
  \begin{mathpar}
    \mprset{fraction={===}}
    \inferrule
    {\Gamma \vdash \ret V_e[V] \ltdyn M : \u F A'}
    {\Gamma \vdash \ret V \ltdyn S_p[M] : \u F A}
  \end{mathpar}
  \end{small}

  For any computation ep pair $(V_e,S_p)$ from $\u B$ to $\u B'$, the
  following are equivalent:
  \begin{small}
  \begin{mathpar}
    \mprset{fraction={===}}
    \inferrule
    {\Gamma, z' : U \u B' \vdash M \ltdyn S[S_p[\force z']] : \u C}
    {\Gamma, z : U \u B \vdash M[V_e/z'] \ltdyn S[\force z] : \u C}
  \end{mathpar}
  \end{small}
\end{lemma}
\begin{longproof}
  \begin{enumerate}
  \item Assume $\ret V_e[V] \ltdyn M : \u F A'$. Then by retraction,
    $\ret V \ltdyn S_p[\ret V_e[V]]$ so by transitivity, the result
    follows by substitution:
    \begin{mathpar}
      \inferrule
      {S_p \ltdyn S_p \and \ret V_e[V] \ltdyn M}
      {S_p[\ret V_e[V]] \ltdyn M}
    \end{mathpar}
  \item Assume $\ret V \ltdyn S_p[M] : \u F A$. Then by projection,    
    $\bindXtoYinZ {S_p[M]} x \ret V_e[x] \ltdyn M$, so it is sufficient to show
    \[ \ret V_e[V] \ltdyn \bindXtoYinZ {S_p[M]} x \ret V_e[x] \]
    but again by substitution we have
    \[ \bindXtoYinZ {\ret V} x \ret V_e[x] \ltdyn \bindXtoYinZ {S_p[M]} x \ret V_e[x]\]
    and by $\u F\beta$, the LHS is equivalent to $\ret V_e[V]$.
  \item Assume $z' : U\u {B'} \vdash M \ltdyn S[S_p[\force z']]$, then
    by projection, $S[S_p[\force V_e]] \ltdyn S[\force z]$ 
    and by substitution:
    \begin{mathpar}
      \inferrule
      {M \ltdyn S[S_p[\force z']]\and V_e \ltdyn V_e \and S[S_p[\force V_e]] = (S[S_p[\force z']])[V_e/z']}
      {M[V_e/z'] \ltdyn S[S_p[\force V_e]]}
    \end{mathpar}
  \item Assume $z : U \u B \vdash M[V_e/z'] \ltdyn S[\force z]$. Then
    by retraction, $M \ltdyn M[V_e[\thunk{S_p[\force z]}]]$ and by
    substitution:
    \[ M[V_e[\thunk{S_p[\force z]}]] \ltdyn S[\force \thunk{S_p[\force z]}] \]
    and the right is equivalent to $S[S_p[\force z]]$ by $U\beta$.
  \end{enumerate}
\end{longproof}

Finally, we prove the axiomatic graduality theorem.
In addition to the lemmas above, the main task is to prove the
``compatibility'' cases which are the congruence cases for introduction
and elimination rules.
These come down to proving that the casts ``commute'' with
introduction/elimination forms, and are all simple calculations.
\begin{nonnum-theorem}[Axiomatic Graduality]
  For any dynamic type interpretation, the following are true:
  \begin{small}
  \begin{mathpar}
    \inferrule
    {\Phi : \Gamma \ltdyn \Gamma'\\
      \Psi : \Delta \ltdyn \Delta'\\
      \Phi \pipe \Psi \vdash M \ltdyn M' : \u B \ltdyn \u B'}
    {\sem\Gamma \pipe \sem{\Delta'} \vdash \sem M[\sem{\Psi}] \ltdyn \sdncast{\u B}{\u B'}[\sem{M'}[\sem{\Phi}]] : \sem{\u B}}

    \inferrule
    {\Phi : \Gamma \ltdyn \Gamma' \\
      \Phi \vdash V \ltdyn V' : A \ltdyn A'}
    {\sem{\Gamma} \vdash \supcast{A}{A'}[\sem{V}] \ltdyn\sem{V'}[\sem\Phi] : \sem {A'}}
  \end{mathpar}
  \end{small}
\end{nonnum-theorem}
\begin{longproof}
  By mutual induction over term dynamism derivations. For the $\beta,
  \eta$ and reflexivity rules, we use the identity expansion lemma and
  the corresponding $\beta, \eta$ rule of
  \cbpvstar\ref{lem:ident-expansion}.

  For compatibility rules a pattern emerges.  Universal rules
  (positive intro, negative elim) are easy, we don't need to reason about
  casts at all. For ``(co)-pattern matching rules'' (positive elim,
  negative intro), we need to invoke the $\eta$ principle (or
  commuting conversion, which is derived from the $\eta$ principle).
  In all compatibility cases, the cast reduction lemma keeps the
  proof straightforward.

  Fortunately, all reasoning about ``shifted'' casts is handled in
  lemmas, and here we only deal with the ``nice'' value upcasts/stack
  downcasts.
  \begin{enumerate}
  \item Transitivity for values: The GTT rule is
    \[
    \inferrule{
    \Phi : \Gamma \ltdyn \Gamma' \and \Phi' : \Gamma' \ltdyn \Gamma'' \and
    \Phi'' : \Gamma \ltdyn \Gamma''
    \\
    \Phi \vdash V \ltdyn V' : A \ltdyn A'\\
    \Phi' \vdash V' \ltdyn V'' : A' \ltdyn A''\\
    }
    { \Phi'' \vdash V \ltdyn V'' : A \ltdyn A''}
    \]
    Which under translation (and the same assumptions about the contexts) is
    \[
    \inferrule
    {\sem{\Gamma} \vdash \supcast{A}{A'}[\sem{V}] \ltdyn \sem{V'}[\sem{\Phi}] : \sem{A'}\\
     \sem{\Gamma'} \vdash \supcast{A'}{A'}[\sem{V'}] \ltdyn \sem{V''}[\sem{\Phi'}] : \sem{A''}
    }
    {\sem{\Gamma} \vdash \supcast{A}{A''}[\sem{V}] \ltdyn \sem{V''}[\sem{\Phi''}] : \sem{A''}}
    \]
    We proceed as follows, the key lemma here is the cast decomposition lemma:
    \begin{align*}
      \supcast{A}{A''}[\sem{V}]
      &\equidyn
      \supcast{A'}{A''}[\supcast{A}{A'}[\sem{V}]] \tag{cast decomposition}\\
      &\ltdyn \supcast{A'}{A''}[\sem{V'}[\sem{\Phi}]] \tag{IH}\\
      &\ltdyn \sem{V''}[\sem{\Phi'}][\sem{\Phi}] \tag{IH}\\
      &\equidyn \sem{V''}[\sem{\Phi''}] \tag{cast decomposition}
    \end{align*}
  \item Transitivity for terms:
    The GTT rule is
    \[
    \inferrule{
    \Phi : \Gamma \ltdyn \Gamma' \and \Phi' : \Gamma' \ltdyn \Gamma'' \and
    \Phi'' : \Gamma \ltdyn \Gamma''
    \and \Psi : \Delta \ltdyn \Delta' \and \Psi : \Delta' \ltdyn \Delta''
    \and \Psi'' : \Delta\ltdyn \Delta''
    \\
    \Phi \pipe \Psi \vdash M \ltdyn M' : \u B \ltdyn \u B'\\
    \Phi' \pipe \Psi' \vdash M' \ltdyn M'' : \u B' \ltdyn \u B''\\
    }
    { \Phi'' \pipe \Psi'' \vdash M \ltdyn M'' : \u B \ltdyn \u B''}
    \]
    Which under translation (and the same assumptions about the contexts) is
    \[
    \inferrule
    {\sem{\Gamma} \pipe \sem{\Delta'} \vdash \sem{M}[\sem{\Psi}] \ltdyn \sdncast{\u B}{\u B'}[\sem{M'}[\sem{\Phi}]] : \sem{\u B}\\
    \sem{\Gamma'} \pipe \sem{\Delta''} \vdash \sem{M'}[\sem{\Psi'}] \ltdyn \sdncast{\u B'}{\u B''}[\sem{M''}[\sem{\Phi'}]] : \sem{\u B'}}
    {\sem{\Gamma} \pipe \sem{\Delta''} \vdash \sem{M}[\sem{\Psi''}] \ltdyn \sdncast{\u B}{\u B''}[\sem{M''}[\sem{\Phi''}]] : \sem{\u B}}
    \]
    We proceed as follows, the key lemma here is the cast decomposition lemma:
    \begin{align*}
      \sem{M}[\sem{\Psi''}]
      &\equidyn
      \sem{M}[\sem{\Psi}][\sem{\Psi'}] \tag{Cast decomposition}\\
      &\ltdyn \sdncast{\u B}{\u B'}[\sem{M'}[\sem{\Psi'}][\sem{\Phi}]]\tag{IH}\\
      &\ltdyn \sdncast{\u B}{\u B'}[\sdncast{\u B'}{\u B''}[\sem{M''}[\sem{\Phi'}][\sem{\Phi}]]]\tag{IH}\\
      &\equidyn \sdncast{\u B}{\u B''}[\sem{M''}[\sem{\Phi''}]] \tag{Cast decomposition}
    \end{align*}
  \item Substitution of a value in a value:
    The GTT rule is
    \[
    \inferrule
    {\Phi, x \ltdyn x' : A_1 \ltdyn A_1' \vdash V_2 \ltdyn V_2' : A_2 \ltdyn A_2'\\
    \Phi \vdash V_1 \ltdyn V_1' : A_1 \ltdyn A_1'}
    {\Phi \vdash V_2[V_1/x]\ltdyn V_2'[V_1'/x'] : A_2 \ltdyn A_2'}
    \]
    Where $\Phi : \Gamma \ltdyn \Gamma'$. Under translation, we need to show
    \[
    \inferrule
    {\sem\Gamma, x : \sem{A_1} \vdash \supcast{A_2}{A_2'}[\sem{V_2}] \ltdyn \sem{V_2'}[\sem\Phi][\supcast{A_1}{A_1'}[x]/x'] : \sem{A_2'}\\
     \sem\Gamma \vdash \supcast{A_1}{A_1'}[\sem{V_1}] \ltdyn \sem{V_1'}[\sem\Phi] : \sem{A_1'}}
    {\sem\Gamma \vdash \supcast{A_2}{A_2'}[\sem{V_2[V_1/x]}] \ltdyn \sem{V_2'[V_1'/x']}[\sem\Phi] : \sem{A_2'}}
    \]
    Which follows by compositionality:
    \begin{align*}
      \supcast{A_2}{A_2'}[\sem{V_2[V_1/x]}]
      &= (\supcast{A_2}{A_2'}[\sem{V_2}])[\sem{V_1}/x] \tag{Compositionality}\\
      &\ltdyn \sem{V_2'}[\sem\Phi][\supcast{A_1}{A_1'}[x]/x'][\sem{V_1}/x]\tag{IH}\\
      &= \sem{V_2'}[\sem\Phi][\supcast{A_1}{A_1'}[\sem{V_1}]/x']\\
      &\ltdyn \sem{V_2'}[\sem\Phi][\sem{V_1'}[\sem\Phi]/x']\tag{IH}\\
      &= \sem{V_2'[V_1'/x']}[\sem\Phi]
    \end{align*}
  \item Substitution of a value in a term:
    The GTT rule is
    \[
    \inferrule
    {\Phi, x \ltdyn x' : A \ltdyn A' \pipe \Psi \vdash M \ltdyn M' : \u B \ltdyn \u B'\\
      \Phi \vdash V \ltdyn V' : A \ltdyn A'
    }
    {\Phi \vdash M[V/x] \ltdyn M'[V'/x'] : \u B \ltdyn \u B'}
    \]
    Where $\Phi : \Gamma \ltdyn \Gamma'$ and $\Psi : \Delta \ltdyn \Delta'$.
    Under translation this is:
    \[
    \inferrule
    {\sem\Gamma, x : \sem{A} \pipe \sem\Delta \vdash \sem M \ltdyn \sdncast{\u B}{\u B'}[\sem {M'}[\sem\Phi][\supcast{A}{A'}[x]/x']] : \sem{\u B}\\
    \sem\Gamma \vdash \supcast{A}{A'}[{\sem V}] \ltdyn \sem{V'}[\sem\Phi] : \sem{A'}}
    {\sem\Gamma \pipe \sem\Delta \vdash \sem {M[V/x]} \ltdyn \sdncast{\u B}{\u B'}[\sem{M'[V'/x']}[\sem\Phi]] : \sem{\u B}}
    \]
    Which follows from compositionality of the translation:
    \begin{align*}
      \sem {M[V/x]}
      &= \sem{M}[\sem{V}/x] \tag{Compositionality}\\
      &\ltdyn \sdncast{\u B}{\u B'}[\sem {M'}[\sem\Phi][\supcast{A}{A'}[x]/x']][\sem{V}/x] \tag{IH}\\
      &= \sdncast{\u B}{\u B'}[\sem {M'}[\sem\Phi][\supcast{A}{A'}[\sem{V}]/x']]\\
      &\ltdyn \sdncast{\u B}{\u B'}[\sem {M'}[\sem\Phi][\sem{V'}[\sem\Phi]/x']]\tag{IH}\\
      &= \sdncast{\u B}{\u B'}[\sem{M'[V'/x']}[\sem\Phi]] \tag{Compositionality}
    \end{align*}
  \item Substitution of a term in a stack:
    The GTT rule is
    \[
    \inferrule
    {\Phi \pipe \bullet \ltdyn \bullet : \u B \ltdyn \u B' \vdash S \ltdyn S' : \u C \ltdyn \u C'\\
      \Phi \pipe \cdot \vdash M \ltdyn M' : \u B \ltdyn \u B'}
    {\Phi \pipe \cdot \vdash S[M]\ltdyn S'[M'] : \u C \ltdyn \u C'}
    \]
    Where $\Phi : \Gamma \ltdyn \Gamma'$.
    Under translation this is
    \[
    \inferrule
    {\sem\Gamma \pipe \bullet : \sem{\u B'} \vdash \sem{S}[\sdncast{\u B}{\u B'}[\bullet]] \ltdyn \sdncast{\u C}{\u C'}[\sem{S'}[\sem\Phi]] : \sem{\u C}\\
      \sem\Gamma \pipe \cdot \vdash \sem{M} \ltdyn \sdncast{\u B}{\u B'}[\sem{M'}[\sem\Phi]] : \sem{\u B}}
    {\sem\Gamma \pipe \cdot \vdash \sem{S[M]} \ltdyn \sdncast{\u C}{\u C'}[\sem{S'[M']}[\sem\Phi]] : \sem{\u C}}
    \]
    We follows easily using compositionality of the translation:
    \begin{align*}
      \sem{S[M]}
      &= \sem{S}[\sem{M}] \tag{Compositionality}\\
      &\ltdyn   \sem{S}[\sdncast{\u B}{\u B'}[\sem{M'}[\sem\Phi]]] \tag{IH}\\
      &\ltdyn   \sdncast{\u C}{\u C'}[\sem{S'}[\sem\Phi][\sem{M'}[\sem\Phi]]]\tag{IH}\\
      &= \sdncast{\u C}{\u C'}[\sem{S'[M']}[\sem\Phi]] \tag{Compositionality}
    \end{align*}
  \item Variables: The GTT rule is
    \[ \Gamma_1 \ltdyn \Gamma_1' ,x \ltdyn x' : A \ltdyn A', \Gamma_2 \ltdyn \Gamma_2' \vdash x \ltdyn x' : A \ltdyn A' \]
    which under translation is
    \[ \sem{\Gamma_1}, x : \sem A, \sem{\Gamma_2} \vdash \supcast{A}{A'}[x] \ltdyn \supcast{A}{A'}[x] : \sem{A'} \]
    which is an instance of reflexivity.
  \item Hole: The GTT rule is
    \[ \Phi \pipe \bullet \ltdyn \bullet : \u B \ltdyn \u B' \vdash \bullet \ltdyn \bullet : \u B \ltdyn \u B' \]
    which under translation is
    \[ \sem\Gamma \pipe \bullet : \u B' \vdash \sdncast{\u B}{\u B'}[\bullet] \ltdyn \sdncast{\u B}{\u B'}[\bullet] : \u B \]
    which is an instance of reflexivity.
  \item Error is bottom: The GTT axiom is
    \[ \Phi \vdash \err \ltdyn M : \u B \]
    where $\Phi : \Gamma \ltdyn \Gamma'$, so we need to show
    \[ \sem\Gamma \vdash \err \ltdyn \sdncast{\u B}{\u B}[\sem{M}[\sem{\Phi}]] : \sem{\u B} \]
    which is an instance of the error is bottom axiom of CBPV.
  \item Error strictness: The GTT axiom is
    \[
    \Phi \vdash S[\err] \ltdyn \err : \u B
    \]
    where $\Phi : \Gamma \ltdyn \Gamma'$, which under translation is
    \[
    \sem\Gamma \vdash \sem{S}[\err] \ltdyn \sdncast{\u B}{\u B}[\err] : \sem{\u B}
    \]
    By strictness of stacks in CBPV, both sides are equivalent to
    $\err$, so it follows by reflexivity.

  \item UpCast-L: The GTT axiom is
    \[
    x \ltdyn x' : A \ltdyn A' \vdash \upcast{A}{A'}x \ltdyn x' : A'
    \]
    which under translation is
    \[
    x : \sem{A} \vdash \supcast{A'}{A'}[\supcast{A}{A'}[x]] \ltdyn \supcast{A}{A'}[x] : A'
    \]
    Which follows by identity expansion and reflexivity.
  \item UpCast-R: The GTT axiom is
    \[
    x : A \vdash x \ltdyn \upcast{A}{A'}x : A \ltdyn A'
    \]
    which under translation is
    \[
    x : \sem{A} \vdash \supcast{A}{A'}[x] \ltdyn \supcast{A}{A'}[\supcast{A}{A}[x]] : \sem{A'}
    \]
    which follows by identity expansion and reflexivity.
  \item DnCast-R: The GTT axiom is
    \[
    \bullet \ltdyn \bullet : \u B \ltdyn \u B' \vdash \bullet \ltdyn \dncast{\u B}{\u B'} : \u B
    \]
    Which under translation is
    \[
    \bullet : \sem{\u B'} \vdash
    \sdncast{\u B}{\u B'}[\bullet]
    \ltdyn
    \sdncast{\u B}{\u B}[\sdncast{\u B}{\u B'}[\bullet]]
    : \sem{\u B}
    \]
    Which follows by identity expansion and reflexivity.
  \item DnCast-L: The GTT axiom is
    \[
    \bullet : \u B' \vdash \dncast{\u B}{\u B'} \bullet \ltdyn \bullet : \u B \ltdyn \u B'
    \]
    So under translation we need to show
    \[
    \bullet : \sem{\u B'} \vdash
    \sdncast{\u B}{\u B'}[\sdncast{\u B'}{\u B'}[\bullet]]
    \ltdyn
    \sdncast{\u B}{\u B'}\bullet : \sem{\u B}
    \]
    Which follows immediately by reflexivity and the lemma that
    identity casts are identities.

  \item $0$ elim, we do the term case, the value case is similar
    \[
    \inferrule
    {\upcast{0}{0}[\sem{V}] \ltdyn \sem{V'}[\sem\Phi]}
    {\absurd \sem{V} \ltdyn \dncast{\u B}{\u B'}\absurd\sem{V'}[\sem\Phi]}
    \]
    Immediate by $0\eta$.
  \item $+$ intro, we do the $\inl$ case, the $\inr$ case is the same:
    \[
    \inferrule
    {\supcast{A_1}{A_1'}[\sem{V}]\ltdyn \sem{V'}[\sem\Phi]}
    {\supcast{A_1+A_2}{A_1'+A_2'}[\inl\sem{V}]\ltdyn \inl\sem{V'}[\sem\Phi]}
    \]
    Which follows easily:
    \begin{align*}
      \supcast{A_1+A_2}{A_1'+A_2'}[\inl\sem{V}]
      &\equidyn \inl \supcast{A_1}{A_1'}\sem{V}\tag{cast reduction}\\
      &\ltdyn \inl \sem{V'}[\sem\Phi]\tag{IH}
    \end{align*}
  \item $+$ elim, we do just the cases where the continuations are terms:
    \[
    \inferrule
    {\supcast{A_1 + A_2}{A_1' + A_2'}[\sem{V}] \ltdyn \sem{V'}[\sem\Phi]\\
    \sem{M_1}[\sem\Psi] \ltdyn \sem{M_1'}[\sem\Phi][\supcast{A_1}{A_1'}[x_1]/x_1']\\
    \sem{M_2}[\sem\Psi] \ltdyn \sem{M_2'}[\sem\Phi][\supcast{A_2}{A_2'}[x_2]/x_2']}
    {\caseofXthenYelseZ {\sem V} {x_1. \sem{M_1}[\sem\Psi]}{x_2. \sem{M_2}[\sem\Psi]} \ltdyn \sdncast{\u B}{\u B'}[\caseofXthenYelseZ {\sem V'[\sem\Phi]} {x_1'. \sem{M_1'}[\sem\Phi]}{x_2'. \sem{M_2'}[\sem\Phi]}]}
    \]
    \begin{align*}
      & \caseofXthenYelseZ {\sem V} {x_1. \sem{M_1}[\sem\Psi]}{x_2. \sem{M_2}[\sem\Psi]}\\
      &\ltdyn
      \sdncast{\u B}{\u B'}[\caseofXthenYelseZ {\sem V} {x_1. \sem{M_1'}[\sem\Phi][\supcast{A_1}{A_1'}[x_1]/x_1']}{x_2. \sem{M_2'}[\sem\Phi][\supcast{A_2}{A_2'}[x_2]/x_2']}]\tag{IH}\\
      &\equidyn
      \caseofX {\sem V}\tag{comm conv}\\
      &\qquad\thenY{x_1. \sdncast{\u B}{\u B'}[\sem{M_1'}[\sem\Phi][\supcast{A_1}{A_1'}[x_1]/x_1']]}\\
      &\qquad\elseZ{x_2. \sdncast{\u B}{\u B'}[\sem{M_2'}[\sem\Phi][\supcast{A_2}{A_2'}[x_2]/x_2']]}\\
      &\equidyn
      \caseofX {\sem V}\tag{$+\beta$}\\
      &\qquad\thenY{x_1. \sdncast{\u B}{\u B'}[\caseofXthenYelseZ {\inl \supcast{A_1}{A_1'}x_1} {x_1'. \sem{M_1'}[\sem\Phi]}{x_2'. \sem{M_2'}[\sem\Phi]}]}\\
      &\qquad\elseZ{x_2. \sdncast{\u B}{\u B'}[\caseofXthenYelseZ {\inr \supcast{A_2}{A_2'}x_2} {x_1'. \sem{M_1'}[\sem\Phi]}{x_2'. \sem{M_2'}[\sem\Phi]}]}\\
      &\equidyn
      \caseofX {\sem V}\tag{cast reduction}\\
      &\qquad\thenY{x_1. \sdncast{\u B}{\u B'}[\caseofXthenYelseZ {\supcast{A_1+A_2}{A_1'+A_2'}\inl x_1} {x_1'. \sem{M_1'}[\sem\Phi]}{x_2'. \sem{M_2'}[\sem\Phi]}]}\\
      &\qquad\elseZ{x_2. \sdncast{\u B}{\u B'}[\caseofXthenYelseZ {\supcast{A_1+A_2}{A_1'+A_2'}\inr x_2} {x_1'. \sem{M_1'}[\sem\Phi]}{x_2'. \sem{M_2'}[\sem\Phi]}]}\\
      &\equidyn
      \sdncast{\u B}{\u B'}[\caseofXthenYelseZ {\supcast{A_1+A_2}{A_1'+A_2'}[\sem V]} {x_1'. \sem{M_1'}[\sem\Phi]}{x_2'. \sem{M_2'}[\sem\Phi]}]\\
      &\ltdyn
      \sdncast{\u B}{\u B'}[\caseofXthenYelseZ {\sem{V'}[\sem\Phi]} {x_1'. \sem{M_1'}[\sem\Phi]}{x_2'. \sem{M_2'}[\sem\Phi]}]\tag{IH}\\
    \end{align*}
  \item $1$ intro:
    \[
    \inferrule
    {}
    {\supcast{1}{1}[()]\ltdyn ()}
    \]
    Immediate by cast reduction.
  \item $1$ elim (continuations are terms case):
    \[
    \inferrule
    {\supcast{1}{1}[\sem{V}] \ltdyn \sem{V'}[\sem\Phi]\\
      \sem{M}[\sem\Psi] \ltdyn \sdncast{\u B}{\u B'}[\sem{M'}[\sem\Phi]]
    }
    {\pmpairWtoinZ {\sem V} {\sem{M}[\sem{\Psi}]}
      \ltdyn
      \dncast{\u B}{\u B'}[\pmpairWtoinZ {\sem V'[\sem\Phi]} {\sem{M'}[\sem{\Phi}]}]}
    \]
    which follows by identity expansion \ref{lem:ident-expansion}.
  \item $\times$ intro:
    \[
    \inferrule
    {\supcast{A_1}{A_1'}{\sem{V_1}} \ltdyn \sem{V_1'[\sem\Phi]}\\
     \supcast{A_2}{A_2'}{\sem{V_2}} \ltdyn \sem{V_2'[\sem\Phi]}}
    {\supcast{A_1 \times A_2}{A_1' \times A_2'}[(\sem{V_1},\sem{V_2})]
      \ltdyn
      (\sem{V_1'[\sem\Phi]}, \sem{V_2'[\sem\Phi]})}
    \]
    We proceed:
    \begin{align*}
      \supcast{A_1 \times A_2}{A_1' \times A_2'}[(\sem{V_1},\sem{V_2})]
      &\equidyn
      (\supcast{A_1}{A_1'}{\sem{V_1}},\supcast{A_2}{A_2'}{\sem{V_2}})\tag{cast reduction}\\
      &\ltdyn (\sem{V_1'[\sem\Phi]}, \sem{V_2'[\sem\Phi]}) \tag{IH}
    \end{align*}
  \item $\times$ elim: We show the case where the continuations are
    terms, the value continuations are no different:
    \[
    \inferrule
    {\supcast{A_1\times A_2}{A_1' \times A_2'}[\sem{V}] \ltdyn \sem{V'}[\sem\Phi]\\
      \sem{M}[\sem\Psi] \ltdyn \sdncast{\u B}{\u B'}[\sem{M'}[\sem\Phi][\supcast{A_1}{A_1'}[x]/x'][\supcast{A_2}{A_2'}[y]/y']]
    }
    {\pmpairWtoXYinZ {\sem V} x y {\sem{M}[\sem{\Psi}]}
      \ltdyn
      \dncast{\u B}{\u B'}[\pmpairWtoXYinZ {\sem V'[\sem\Phi]} {x'} {y'} {\sem{M'}[\sem{\Phi}]}]
    }
    \]
    We proceed as follows:
    \begin{align*}
      &\pmpairWtoXYinZ {\sem V} x y {\sem{M}[\sem{\Psi}]}\\
      &\ltdyn\pmpairWtoXYinZ {\sem V} x y \sdncast{\u B}{\u B'}[\sem{M'}[\sem\Phi][\supcast{A_1}{A_1'}[x]/x'][\supcast{A_2}{A_2'}[y]/y']]\tag{IH}\\
      &\equidyn
      \pmpairWtoXYinZ {\sem V} x y\tag{$\times\beta$}\\
      &\qquad \pmpairWtoXYinZ {(\supcast{A_1}{A_1'}[x],\supcast{A_2}{A_2'}[y])} {x'} {y'} \sdncast{\u B}{\u B'}[\sem{M'}[\sem\Phi]]\\
      &\equidyn
      \pmpairWtoXYinZ {\sem V} x y\tag{cast reduction}\\
      &\qquad \pmpairWtoXYinZ {\supcast{A_1\times A_2'}{A_1'\times A_2'}[(x,y)]} {x'} {y'} \sdncast{\u B}{\u B'}[\sem{M'}[\sem\Phi]]\\
      &\equidyn
      \pmpairWtoXYinZ {\supcast{A_1\times A_2}{A_1'\times A_2'}[{\sem V}]} {x'} {y'} \sdncast{\u B}{\u B'}[\sem{M'}[\sem\Phi]]\tag{$\times\eta$}\\
      &\ltdyn \pmpairWtoXYinZ {\sem{V'}[\sem\Phi]} {x'}{y'} \sdncast{\u B}{\u B'}[\sem{M'}[\sem\Phi]]\tag{IH}\\
      &\equidyn \sdncast{\u B}{\u B'}[\pmpairWtoXYinZ {\sem{V'}[\sem\Phi]} {x'}{y'}\sem{M'}[\sem\Phi]]\tag{commuting conversion}
    \end{align*}
  \item $U$ intro:
    \[
    \inferrule
    {\sem{M} \ltdyn \sdncast{\u B}{\u B'}[\sem{M'}[\sem\Phi]]}
    {\supcast{U\u B}{U \u B'}[\thunk\sem{M}] \ltdyn \thunk\sem{M'}[\sem\Phi]}
    \]
    We proceed as follows:
    \begin{align*}
      \supcast{U\u B}{U \u B'}[\thunk\sem{M}]
      &\ltdyn \supcast{U\u B}{U \u B'}[\thunk\sdncast{\u B}{\u B'}[\sem{M'}[\sem\Phi]]]\tag{IH}\\
      &\ltdyn \thunk \sem{M'}[\sem\Phi]\tag{alt projection}
    \end{align*}
  \item $U$ elim:
    \[
    \inferrule
    {\supcast{U \u B}{U \u B'}[\sem{V}] \ltdyn \sem{V'}[\sem\Phi]}
    {\force \sem V \ltdyn \sdncast{\u B}{\u B'}\force \sem {V'}[\sem\Phi]}
    \]
    By hom-set formulation of adjunction \ref{lem:hom-set-adj}.
  \item $\top$ intro:
    \[
    \inferrule{}{\{\} \ltdyn \sdncast{\top}{\top}[\{\}]}
    \]
    Immediate by $\top\eta$
  \item $\with$ intro:
    \[
    \inferrule
    {\sem{M_1}[\sem{\Psi}]\ltdyn \sdncast{\u B_1}{\u B_1'}[\sem{M_1'}[\sem{\Phi}]]\\
     \sem{M_2}[\sem{\Psi}]\ltdyn \sdncast{\u B_2}{\u B_2'}[\sem{M_2'}[\sem{\Phi}]]}
    {\pair{\sem{M_1}[\sem{\Psi}]}{\sem{M_2}[\sem{\Psi}]}
    \ltdyn
    \sdncast{\u B_1 \with \u B_2}{\u B_1' \with \u B_2'}[\pair{\sem{M_1'}[\sem{\Phi}]}{\sem{M_2'}[\sem{\Phi}]}]}
    \]
    We proceed as follows:
    \begin{align*}
      &\pair{\sem{M_1}[\sem{\Psi}]}{\sem{M_2}[\sem{\Psi}]}\\
      &\ltdyn
      \pair{\sdncast{\u B_1}{\u B_1'}[\sem{M_1'}[\sem{\Phi}]]}{\sdncast{\u B_2}{\u B_2'}[\sem{M_2'}[\sem{\Phi}]]}\tag{IH}\\
      &\equidyn
      \pairone{\pi\sdncast{\u B_1 \with \u B_2}{\u B_1' \with \u B_2'}[\pair{\sem{M_1'}[\sem{\Phi}]}{\sem{M_2'}[\sem{\Phi}]}]}\tag{cast reduction}\\
      &\quad \pairtwo{\pi'\sdncast{\u B_1 \with \u B_2}{\u B_1' \with \u B_2'}[\pair{\sem{M_1'}[\sem{\Phi}]}{\sem{M_2'}[\sem{\Phi}]}]}\\
      &\equidyn
      \sdncast{\u B_1 \with \u B_2}{\u B_1' \with \u B_2'}[\pair{\sem{M_1'}[\sem{\Phi}]}{\sem{M_2'}[\sem{\Phi}]}]\tag{$\with\eta$}
    \end{align*}
  \item $\with$ elim, we show the $\pi$ case, $\pi'$ is symmetric:
    \[
    \inferrule
    {\sem{M}[\sem{\Psi}] \ltdyn \sdncast{\u B_1 \with \u B_2}{\u B_1' \with \u B_2'}[\sem{M'}[\sem{\Phi}]]}
    {\pi\sem{M}[\sem{\Psi}] \ltdyn \sdncast{\u B_1}{\u B_1'}[\pi\sem{M'}[\sem{\Phi}]]}
    \]
    We proceed as follows:
    \begin{align*}
      \pi\sem{M}[\sem{\Psi}]
      &\ltdyn \pi \sdncast{\u B_1 \with \u B_2}{\u B_1' \with \u B_2'}[\sem{M'}[\sem{\Phi}]]\tag{IH}\\
      &\equidyn
      \sdncast{\u B_1}{\u B_1'}[\pi\sem{M'}[\sem{\Phi}]]\tag{cast reduction}
    \end{align*}
  \item
    \[
    \inferrule
    {\sem{M}[\sem{\Psi}] \ltdyn \sdncast{\u B}{\u B'}[\sem{M'}[\sem{\Phi}][\supcast{A}{A'}{x}/x']]}
    {\lambda x:A. \sem{M}[\sem{\Psi}] \ltdyn \sdncast{A \to \u B}{A'\to\u B'}[\lambda x':A'. \sem{M'}[\sem{\Phi}]]}
    \]
    We proceed as follows:
    \begin{align*}
      &\lambda x:A. \sem{M}[\sem{\Psi}]\\
      &\ltdyn
      \lambda x:A. \sdncast{\u B}{\u B'}[\sem{M'}[\sem{\Phi}][\supcast{A}{A'}{x}/x']]\tag{IH}\\
      &\equidyn
      \lambda x:A. (\sdncast{A \to \u B}{A' \to \u B'}[\lambda x'. \sem{M'}[\sem{\Phi}]])\, x\tag{cast reduction}\\
      &\equidyn
      \sdncast{A \to \u B}{A' \to \u B'}[\lambda x'. \sem{M'}[\sem{\Phi}]]\tag{$\to\eta$}
    \end{align*}
  \item We need to show
    \[
    \inferrule
    {\sem{M}[\sem{\Psi}] \ltdyn \sdncast{A \to \u B}{A' \to \u B'}[\sem{M'}[\sem{\Phi}]]\\
     \supcast{A}{A'}[\sem{V}] \ltdyn \sem{V'}[\sem{\Phi}]}
    {\sem{M}[\sem{\Psi}]\,\sem{V} \ltdyn \sdncast{\u B}{\u B'}[\sem{M'}[\sem{\Phi}]\, \sem{V'}[\sem{\Phi}]]}
    \]
    We proceed:
    \begin{align*}
      &\sem{M}[\sem{\Psi}]\,\sem{V}\\
      &\ltdyn
      (\sdncast{A \to \u B}{A' \to \u B'}[\sem{M'}[\sem{\Phi}]])\,\sem{V}\tag{IH}\\
      &\equidyn
      \sdncast{\u B}{\u B'}[\sem{M'}[\sem{\Phi}]\,(\supcast{A}{A'}{\sem{V}})]\tag{cast reduction}\\
      &\ltdyn
      \sdncast{\u B}{\u B'}[\sem{M'}[\sem{\Phi}]\,\sem{V'}[\sem{\Phi}]] \tag{IH}
    \end{align*}
  \item We need to show
    \[
    \inferrule
    {\supcast{A}{A'}[\sem{V}] \ltdyn \sem{V'}[\sem{\Phi}]}
    {\ret\sem{V}\ltdyn \sdncast{\u F A}{\u FA'}[\ret\sem{V'}[\sem{\Phi}]]}
    \]
    By hom-set definition of adjunction \ref{lem:hom-set-adj}
  \item We need to show
    \[
    \inferrule
    {\sem{M}[\sem{\Psi}] \ltdyn \sdncast{\u F A}{\u F A'}[\sem{M'}[\Phi]]\\
      \sem{N} \ltdyn \sdncast{\u B}{\u B'}[\sem{N}[\Phi][\supcast{A}{A'} x/x']]}
    {\bindXtoYinZ {\sem{M}[\sem{\Psi}]} x {\sem{N}}
    \ltdyn
    \sdncast{\u B}{\u B'}[{\bindXtoYinZ {\sem{M'}[\sem{\Phi}]} {x'} {\sem{N'}[\sem{\Phi}]}}]}
    \]
    We proceed:
    \begin{align*}
      &\bindXtoYinZ {\sem{M}[\sem{\Psi}]} x {\sem{N}}\\
      &\ltdyn \bindXtoYinZ {\sdncast{\u F A}{\u F A'}[\sem{M'}[\Phi]]} x \sdncast{\u B}{\u B'}[\sem{N}[\Phi][\supcast{A}{A'} x/x']] \tag{IH, congruence}\\
      &\equidyn
      \bindXtoYinZ {\sdncast{\u F A}{\u F A'}[\sem{M'}[\Phi]]} x\\
      &\qquad \bindXtoYinZ {\ret\supcast{A}{A'}[x]} {x'}
      \sdncast{\u B}{\u B'}[\sem{N}[\Phi]] \tag{$\u F\beta$}\\
      & \ltdyn \bindXtoYinZ {\sem{M'}[\Phi]} {x'} \sdncast{\u B}{\u B'}[\sem{N}[\Phi]] \tag{Projection}\\
      & \equidyn  \sdncast{\u B}{\u B'}[\bindXtoYinZ {\sem{M'}[\Phi]} {x'} \sem{N}[\Phi]] \tag{commuting conversion}
    \end{align*}
  \end{enumerate}
\end{longproof}
\end{longonly}

As a corollary, we have the following conservativity result, which says
that the homogeneous term dynamisms in GTT are sound and complete for
inequalities in \cbpvstar.
\begin{corollary}[Conservativity] \label{thm:gtt-cbpvstar-conservativity}
  If $\Gamma \mid \Delta \vdash E, E' : T$ are two terms of the same
  type in the intersection of GTT and \cbpvstar, then $\Gamma \mid
  \Delta \vdash E \ltdyn E' : T$ is provable in GTT iff it is
  provable in \cbpvstar.
\end{corollary}
\begin{proof}
  The reverse direction holds because \cbpvstar\ is a syntactic subset of
  GTT. The forward direction holds by axiomatic graduality and the
  fact that identity casts are identities.
\end{proof}

\section{Complex Value/Stack Elimination}
\label{sec:complex}

Next, to bridge the gap between the semantic notion of complex value
and stack with the more rigid operational notion, we perform a
complexity-elimination pass.
This translates a computation with complex values in it to an equivalent
computation without complex values: i.e., all pattern matches take place
in computations, rather than in values, and translates a term dynamism
derivation that uses complex stacks to one that uses only ``simple''
stacks without pattern-matching and computation introduction forms.
\begin{longonly}
  Stacks do not appear anywhere in the grammar of terms, but they are
used in the equational theory (computation $\eta$ rules and error
strictness).
\end{longonly}
\ This translation clarifies the behavioral meaning of complex values and
stacks, following \citet{munchmaccagnoni14nonassociative,
  fuhrmann1999direct}, and therefore of upcasts and downcasts.
\begin{longonly}
This is related to completeness of focusing: it moves inversion rules
outside of focus phases.
\end{longonly}

\begin{longonly}
The syntax of operational CBPV is as in
Figure~\ref{fig:gtt-syntax-and-terms} (unshaded), but with recursive
types added as in Section~\ref{sec:cbpvstar}, and with values and stacks
restricted
as in Figure~\ref{fig:operation-cbpv-syntax}.
  
\begin{figure}
\begin{small}
  \begin{mathpar}
  \begin{array}{lcl}
    A & \bnfdef & X \mid \mu X.A \mid U \u B \mid 0 \mid A_1 + A_2 \mid 1 \mid A_1 \times A_2 \\
    \u B  & ::= & \u Y\mid \nu \u Y. \u B \mid \u F A \mid \top \mid \u B_1 \with \u B_2 \mid A \to \u B\\
    \Gamma & ::= & \cdot \mid \Gamma, x : A \\
    \Delta  & ::= & \cdot \mid \bullet : \u B \\
    V  & ::= & x \mid \rollty{\mu X.A}V \mid \inl{V} \mid \inr{V} \mid () \mid (V_1,V_2)\mid \thunk{M}
    \\
    M & ::= & \err_{\u B} \mid \letXbeYinZ V x M \mid \pmmuXtoYinZ V x M \mid \rollty{\nu \u Y.\u B} M \mid \unroll M \mid \abort{V} \mid \\
    & & \caseofXthenYelseZ V {x_1. M_1}{x_2.M_2} \mid \pmpairWtoinZ V M \mid \pmpairWtoXYinZ V x y M
    \mid \force{V} \mid \\
    & & \ret{V} \mid \bindXtoYinZ{M}{x}{N} \mid \lambda x:A.M \mid M\,V \mid \emptypair \mid \pair{M_1}{M_2} \mid \pi M \mid \pi' M
    \\
    S & ::= & \bullet \mid \bindXtoYinZ S x M \mid S\, V \mid \pi S \mid \pi' S \mid \unrollty{\nu \u Y.\u B}{S}
  \end{array}
  \end{mathpar}
  \end{small}
\caption{Operational CBPV Syntax}
\label{fig:operation-cbpv-syntax}
\end{figure}
In \cbpv, values include only introduction forms, as usual for values in
operational semantics, and \cbpv\/ stacks consist only of elimination
forms for computation types
(the syntax of \cbpv\/ enforces an A-normal
form, where only values can be pattern-matched on, so $\kw{case}$ and
$\kw{split}$ are not evaluation contexts in the operational semantics).

\begin{figure}
\begin{small}
  \begin{mathpar}
    \inferrule
    {}
    {\Gamma,x : A,\Gamma' \vdash x \ltdyn x : A}

    \inferrule
    {}
    {\Gamma\pipe \bullet : \u B \vdash \bullet \ltdyn \bullet : \u B}

    \inferrule
    {}
    {\Gamma \vdash \err \ltdyn \err : \u B}

    \inferrule
    {\Gamma \vdash V \ltdyn V' : A \and
      \Gamma, x : A \vdash M \ltdyn M' : \u B
    }
    {\Gamma \vdash \letXbeYinZ V x M \ltdyn \letXbeYinZ {V'} {x} {M'} : \u B}

    \inferrule
    {\Gamma \vdash V \ltdyn V' : 0}
    {\Gamma \vdash \abort V \ltdyn \abort V' : \u B}

    \inferrule
    {\Gamma \vdash V \ltdyn V' : A_1}
    {\Gamma \vdash \inl V \ltdyn \inl V' : A_1 + A_2}

    \inferrule
    {\Gamma \vdash V \ltdyn V' : A_2}
    {\Gamma \vdash \inr V \ltdyn \inr V' : A_1 + A_2}

    \inferrule
    {\Gamma \vdash V \ltdyn V' : A_1 + A_2\and
      \Gamma, x_1 : A_1 \vdash M_1 \ltdyn M_1' : \u B\and
      \Gamma, x_2 : A_2 \vdash M_2 \ltdyn M_2' : \u B
    }
    {\Gamma \vdash \caseofXthenYelseZ V {x_1. M_1}{x_2.M_2} \ltdyn \caseofXthenYelseZ {V'} {x_1. M_1'}{x_2.M_2'} : \u B}

    \inferrule
    {}
    {\Gamma \vdash () \ltdyn () : 1}

    \inferrule
    {\Gamma \vdash V_1 \ltdyn V_1' : A_1\and
      \Gamma\vdash V_2 \ltdyn V_2' : A_2}
    {\Gamma \vdash (V_1,V_2) \ltdyn (V_1',V_2') : A_1 \times A_2}

    \inferrule
    {\Gamma \vdash V \ltdyn V' : A_1 \times A_2\and
      \Gamma, x : A_1,y : A_2 \vdash M \ltdyn M' : \u B
    }
    {\Gamma \vdash \pmpairWtoXYinZ V x y M \ltdyn \pmpairWtoXYinZ {V'} {x} {y} {M'} : \u B}

    \inferrule
    {\Gamma \vdash V \ltdyn V' : A[\mu X.A/X]}
    {\Gamma \vdash \rollty{\mu X.A} V \ltdyn \rollty{\mu X.A} V' : \mu X.A }
    
    \inferrule
    {\Gamma \vdash V \ltdyn V' : \mu X. A\and
      \Gamma, x : A[\mu X. A/X] \vdash M \ltdyn M' : \u B}
    {\Gamma \vdash \pmmuXtoYinZ V x M \ltdyn \pmmuXtoYinZ {V'} {x} {M'} : \u B}

    \inferrule
    {\Gamma \vdash M \ltdyn M' : \u B}
    {\Gamma \vdash \thunk M \ltdyn \thunk M' : U \u B}

    \inferrule
    {\Gamma \vdash V \ltdyn V' : U \u B}
    {\Gamma \vdash \force V \ltdyn \force V' : \u B}

    \inferrule
    {\Gamma \vdash V \ltdyn V' : A}
    {\Gamma \vdash \ret V \ltdyn \ret V' : \u F A}

    \inferrule
    {\Gamma \vdash M \ltdyn M' : \u F A\and
      \Gamma, x: A \vdash N \ltdyn N' : \u B}
    {\Gamma \vdash \bindXtoYinZ M x N \ltdyn \bindXtoYinZ {M'} {x} {N'} : \u B}

    \inferrule
    {\Gamma, x: A \vdash M \ltdyn M' : \u B}
    {\Gamma \vdash \lambda x : A . M \ltdyn \lambda x:A. M' : A \to \u B}

    \inferrule
    {\Gamma \vdash M \ltdyn M' : A \to \u B\and
      \Gamma \vdash V \ltdyn V' : A}
    {\Gamma \vdash M\,V \ltdyn M'\,V' : \u B }

    \inferrule
    {\Gamma \vdash M_1 \ltdyn M_1' : \u B_1\and
      \Gamma \vdash M_2 \ltdyn M_2' : \u B_2}
    {\Gamma \vdash \pair {M_1} {M_2} \ltdyn \pair {M_1'} {M_2'} : \u B_1 \with \u B_2}

    \inferrule
    {\Gamma \vdash M \ltdyn M' : \u B_1 \with \u B_2}
    {\Gamma \vdash \pi M \ltdyn \pi M' : \u B_1}

    \inferrule
    {\Gamma \vdash M \ltdyn M' : \u B_1 \with \u B_2}
    {\Gamma \vdash \pi' M \ltdyn \pi' M' : \u B_2}

    \inferrule
    {\Gamma \vdash M \ltdyn M' : \u B[{\nu \u Y. \u B}/\u Y]}
    {\Gamma \vdash \rollty{\nu \u Y. \u B} M \ltdyn \rollty{\nu \u Y. \u B} M' : {\nu \u Y. \u B}}

    \inferrule
    {\Gamma \vdash M \ltdyn M' : {\nu \u Y. \u B}}
    {\Gamma \vdash \unroll M \ltdyn \unroll M' : \u B[{\nu \u Y. \u B}/\u Y]}
  \end{mathpar}
  \end{small}
  \caption{CBPV Inequational Theory (Congruence Rules)}
\end{figure}

\begin{figure}
\begin{small}
  \begin{mathpar}
    \inferrule
    {}
    {\caseofXthenYelseZ{\inl V}{x_1. M_1}{x_2. M_2} \equidyn M_1[V/x_1]}

    \inferrule
    {}
    {\caseofXthenYelseZ{\inr V}{x_1. M_1}{x_2. M_2} \equidyn M_2[V/x_2]}

    \inferrule
    {\Gamma, x : A_1 + A_2 \vdash M : \u B}
    {\Gamma, x : A_1 + A_2 \vdash M \equidyn \caseofXthenYelseZ x {x_1. M[\inl x_1/x]}{x_2. M[\inr x_2/x]} : \u B}

    \inferrule
    {}
    {\pmpairWtoXYinZ{(V_1,V_2)}{x_1}{x_2}{M} \equidyn M[V_1/x_1,V_2/x_2]}

    \inferrule
    {\Gamma, x : A_1 \times A_2 \vdash M : \u B}
    {\Gamma, x : A_1 \times A_2 \vdash M \equidyn \pmpairWtoXYinZ x {x_1}{x_2} M[(x_1,x_2)/x] : \u B}

    \inferrule
    {\Gamma, x : 1 \vdash M : \u B}
    {\Gamma, x : 1 \vdash M \equidyn M[()/x] : \u B}

    \inferrule
    {}
    {\pmmuXtoYinZ{\rollty A V}{x}{M} \equidyn M[V/x]}

    \inferrule
    {\Gamma, x : \mu X. A \vdash M :\u B}
    {\Gamma, x : \mu X. A \vdash M \equidyn \pmmuXtoYinZ{x}{y}{M[\rollty{\mu X.A} y/x]} : \u B}

    \inferrule
    {}
    {\force\thunk M \equidyn M}

    \inferrule
    {\Gamma \vdash V : U \u B}
    {\Gamma \vdash V \equidyn \thunk\force V : U \u B}

    \inferrule
    {}
    {\letXbeYinZ V x M \equidyn M[V/x]}

    \inferrule
    {}
    {\bindXtoYinZ {\ret V} x M \equidyn M[V/x]}

    \inferrule
    {}
    {\Gamma \pipe \bullet : \u F A \vdash \bullet \equidyn \bindXtoYinZ \bullet x \ret x : \u F A}

    \inferrule
    {}
    {(\lambda x:A. M)\,V \equidyn M[V/x]}

    \inferrule
    {\Gamma \vdash M : A \to \u B}
    {\Gamma \vdash M \equidyn \lambda x:A. M\,x : A \to \u B}

    \inferrule
    {}
    {\pi \pair{M}{M'} \equidyn M}

    \inferrule
    {}
    {\pi' \pair{M}{M'} \equidyn M'}

    \inferrule
    {\Gamma \vdash M : \u B_1 \with \u B_2}
    {\Gamma \vdash M \equidyn\pair{\pi M}{\pi' M} : \u B_1 \with \u B_2}

    \inferrule
    {\Gamma \vdash M : \top}
    {\Gamma \vdash M \equidyn \{\} : \top}

    \inferrule
    {}
    {\unroll \rollty{\u B} M \equidyn M}

    \inferrule
    {\Gamma \vdash M : \nu \u Y. \u B}
    {\Gamma \vdash M \equidyn \rollty{\nu \u Y.\u B}\unroll M : \nu \u Y. \u B}
  \end{mathpar}
  \end{small}
  \caption{CBPV $\beta, \eta$ rules}
\end{figure}

\begin{figure}
\begin{small}
  \begin{mathpar}
    \inferrule
    {}
    {\Gamma \vdash \err \ltdyn M : \u B}

    \inferrule
    {}
    {\Gamma \vdash S[\err] \equidyn \err : \u B}

    \inferrule
    {}
    {\Gamma \vdash M \ltdyn M : \u B}

    \inferrule
    {}
    {\Gamma \vdash V \ltdyn V : A}

    \inferrule
    {}
    {\Gamma \pipe \u B \vdash S \ltdyn S : \u B'}

    \inferrule
    {\Gamma \vdash M_1 \ltdyn M_2 : \u B \and \Gamma \vdash M_2 \ltdyn M_3 : \u B}
    {\Gamma \vdash M_1 \ltdyn M_3 : \u B}

    \inferrule
    {\Gamma \vdash V_1 \ltdyn V_2 : A \and \Gamma \vdash V_2 \ltdyn V_3 : A}
    {\Gamma \vdash V_1 \ltdyn V_3 : A}

    \inferrule
    {\Gamma \pipe \u B \vdash S_1 \ltdyn S_2 : \u B' \and \Gamma \pipe \u B \vdash S_2 \ltdyn S_3 : \u B'}
    {\Gamma \pipe \u B \vdash S_1 \ltdyn S_3 : \u B'}

    \inferrule
    {\Gamma, x : A \vdash M_1 \ltdyn M_2 : \u B \and
      \Gamma \vdash V_1 \ltdyn V_2 : A}
    {\Gamma \vdash M_1[V_1/x] \ltdyn M_2[V_2/x] : \u B}

    \inferrule
    {\Gamma, x : A \vdash V_1' \ltdyn V_2' : A' \and
      \Gamma \vdash V_1 \ltdyn V_2 : A}
    {\Gamma \vdash V_1'[V_1/x] \ltdyn V_2'[V_2/x] : A'}

    \inferrule
    {\Gamma, x : A \pipe \u B \vdash S_1 \ltdyn S_2 : \u B' \and
      \Gamma \vdash V_1 \ltdyn V_2 : A}
    {\Gamma \pipe \u B \vdash S_1[V_1/x] \ltdyn S_2[V_2/x] : \u B'}

    \inferrule
    {\Gamma \pipe \u B \vdash S_1 \ltdyn S_2 : \u B' \and
      \Gamma \vdash M_1 \ltdyn M_2 : \u B}
    {\Gamma \vdash S_1[M_1] \ltdyn S_2[M_2] : \u B'}

    \inferrule
    {\Gamma \pipe \u B' \vdash S_1' \ltdyn S_2' : \u B'' \and
      \Gamma \pipe \u B \vdash S_1 \ltdyn S_2 : \u B'}
    {\Gamma \pipe \u B \vdash S_1'[S_1] \ltdyn S_2'[S_2] : \u B''}
  \end{mathpar}
  \end{small}
  \caption{CBPV logical and error rules}
\end{figure}

\end{longonly}

\citet{levy03cbpvbook} translates \cbpvstar\/ to \cbpv, but not does not prove
the inequality preservation that we require here, so we give
an
alternative translation for which this property is easy to
verify .
We translate both complex values and complex
stacks to fully general computations, so that computation
pattern-matching can replace the pattern-matching in complex values/stacks.  
\begin{longonly}
For example, for a closed value, we could ``evaluate away''
the complexity and get a closed simple value (if we don't use $U$), but
for open terms, evaluation will get ``stuck'' if we pattern match on
a variable---so not every complex value can be translated to a value in
\cbpv.  
\end{longonly}
More formally, we translate a \cbpvstar\/ complex value $V : A$ to a
\cbpv\/ computation $\simp{V} : \u F A$ that in \cbpvstar\ is equivalent
to $\ret V$.
Similarly, we translate a \cbpvstar\/ complex stack $S$ with hole
$\bullet : \u B$ to a \cbpv\  computation $\simp{S}$ with a free
  variable $z : U \u B$ such that in \cbpvstar, $\simp S \equidyn
S[\force z]$.
Computations $M : \u B$ are translated to computations $\simp{M}$ with
the same type.

\begin{longonly}
The \emph{de-complexification} procedure is defined as follows.
We note that this translation is not the one presented in
\citet{levy03cbpvbook}, but rather a more inefficient version that, in CPS
terminology, introduces many administrative redices.
Since we are only proving results up to observational equivalence
anyway, the difference doesn't change any of our theorems, and makes
some of the proofs simpler.
\begin{definition}[De-complexification]
  We define 
  \begin{small}
  \begin{mathpar}
    \begin{array}{rcl}
      \simp \bullet &=& \force z\\
      \simp x &=& \ret x\\\\
      
      \simpp {\ret V} &= & \bindXtoYinZ {\simp V} x \ret x\\
      \simpp {M\, V}  &=& \bindXtoYinZ {\simp V} x \simp M\, x\\\\

      \simpp{\force V} &=& \bindXtoYinZ {\simp V} x \force x\\
      \simpp{\absurd V} &=& \bindXtoYinZ {\simp V} x \absurd x\\
      \simpp{\caseofXthenYelseZ V {x_1. E_1}{x_2. E_2}} &=&
      \bindXtoYinZ {\simp V} x \caseofXthenYelseZ x {x_1. \simp {E_1}}{x_2. \simp {E_2}}\\
      \simpp{\pmpairWtoinZ V {E}} &=&
      \bindXtoYinZ V w {\pmpairWtoinZ w \simp {E}}\\
      \simpp{\pmpairWtoXYinZ V x y {E}} &=&
      \bindXtoYinZ V w {\pmpairWtoXYinZ w x y \simp {E}}\\
      \simpp{\pmmuXtoYinZ V x E} &=& \bindXtoYinZ {\simp V} y \pmmuXtoYinZ y x \simp{E}\\\\

      \simpp{\inl V} &=& \bindXtoYinZ {\simp V} x \ret\inl x\\
      \simpp{\inr V} &=& \bindXtoYinZ {\simp V} x \ret\inr x\\
      \simp{()} &=& \ret ()\\
      \simp{(V_1,V_2)} &=& \bindXtoYinZ {\simp {V_1}}{x_1} \bindXtoYinZ {\simp {V_2}} {x_2} \ret (x_1,x_2)\\
      \simpp{\thunk M} &=& \ret \thunk \simp M\\
      \simpp{\roll V} &=& \bindXtoYinZ {\simp V} x \roll x\\
    \end{array}
  \end{mathpar}
  \end{small}
\end{definition}

The translation is type-preserving and the identity from \cbpvstar's point of view
\begin{lemma}[De-complexification De-complexifies]
  For any \cbpvstar\/ term $\Gamma \pipe \Delta \vdash E : T$, $\simp E$
  is a term of \cbpv\/ satisfying $\Gamma, \simp\Delta \vdash \simp E :
  \simp T$ where
  $\simp{\cdot} = \cdot$ $\simpp{\bullet:\u B} = z:U\u B$,
  $\simp{\u B} = \u B$, $\simp A = \u F A$.
\end{lemma}

\begin{lemma}[De-complexification is Identity in \cbpvstar]
  Considering CBPV as a subset of \cbpvstar\, we have
  \begin{enumerate}
  \item If $\Gamma \pipe \cdot \vdash M : \u B$  then $M \equidyn \simp M$.
  \item If $\Gamma \pipe \Delta \vdash S : \u B$ then $S[\force z] \equidyn \simp S$.
  \item If $\Gamma \vdash V : A$ then $\ret V \equidyn \simp V$.
  \end{enumerate}
  Furthermore, if $M, V, S$ are in \cbpv, the proof holds in \cbpv.
\end{lemma}
\end{longonly}

Finally, we need to show that the translation preserves inequalities
($\simp{E} \ltdyn \simp{E'}$ if $E \ltdyn E'$), but because complex
values and stacks satisfy more equations than arbitrary computations in
the types of their translations do, we need to isolate the special
``purity'' property that their translations have.
We show that complex values are translated to computations that satisfy
\emph{thunkability}~\cite{munchmaccagnoni14nonassociative}, which
%
intuitively means $M$ should have no observable effects, and so
can be freely duplicated or discarded like a value.
In the inequational theory of \cbpv\/, this is defined by saying that
running $M$ to a value and then duplicating its value is the same as
running $M$ every time we need its value:
\iflong{
  \begin{definition}[Thunkable Computation]
    A computation $\Gamma \vdash M : \u FA$ is \emph{thunkable} if \\
\fi
  \[\Gamma \vdash \ret{( \thunk M)} \equidyn \bindXtoYinZ M x \ret{(\thunk (\ret x))} : \u FU\u F A\]
\iflong
  \end{definition}
\fi
Dually, we show that complex stacks are translated to computations that
satisfy (semantic) \emph{linearity}~\cite{munchmaccagnoni14nonassociative}, where intuitively a computation $M$
with a free variable $x : U \u B$ is linear in $x$ if $M$ behaves as if
when it is forced, the first thing it does is forces $x$, and that is the only time
it uses $x$.  This is described in the CBPV inequational theory as
follows:
\iflong
if we have a thunk $z : U\u F U \u B$, then either we can force
it now and pass the result to $M$ as $x$, or we can just run $M$ with a
thunk that will force $z$ each time $M$ is forced---but if $M$ forces
$x$ exactly once, first, these two are the same.
\begin{definition}[Linear Term]
  A term $\Gamma, x : U\u B \vdash M : \u C$ is \emph{linear in $x$}
  if\\
\fi
 \[ \Gamma, z : U\u FU\u B \vdash
  \bindXtoYinZ {\force z} x M
  \equidyn M[\thunk{(\bindXtoYinZ {(\force z)} x \force x)}]
  \]
\iflong
\end{definition}
\fi
\begin{longonly}
  Thunkability/linearity of the translations of complex values/stacks are
used to prove the preservation of the $\eta$ principles for positive
types and the strictness of complex stacks with respect to errors under
decomplexification.
\end{longonly}

\begin{shortonly}
\noindent Composing this with the translation from GTT to \cbpvstar\/
shows that \emph{GTT value upcasts are thunkable and computation
  downcasts are linear}, which justifies a number of program transformations.
\end{shortonly}

\begin{longonly}
We need a few lemmas about thunkables and linears to prove that complex
values become thunkable and complex stacks become linear.

First, the following lemma is useful for optimizing programs with
thunkable subterms.  Intuitively, since a thunkable has ``no effects''
it can be reordered past any other effectful binding.  Furhmann
\citep{fuhrmann1999direct} calls a morphism that has this property
\emph{central} (after the center of a group, which is those elements
that commute with every element of the whole group).
\begin{lemma}[Thunkable are Central]
  If $\Gamma \vdash M : \u F A$ is thunkable and $\Gamma \vdash N : \u
  F A'$ and $\Gamma , x:A, y:A' \vdash N' : \u B$, then
  \[
  \bindXtoYinZ M x \bindXtoYinZ N y N'
  \equidyn
 \bindXtoYinZ N y \bindXtoYinZ M x N'
  \]
\end{lemma}
\begin{proof}
  \begin{align*}
    &\bindXtoYinZ M x \bindXtoYinZ N y N'\\
    &\equidyn
    \bindXtoYinZ M x \bindXtoYinZ N y \bindXtoYinZ {\force \thunk \ret x} x N' \tag{$U\beta,\u F\beta$}\\
    &\equidyn\bindXtoYinZ M x \bindXtoYinZ {\ret\thunk\ret x} w \bindXtoYinZ N y \bindXtoYinZ {\force w} x N' \tag{$\u F\beta$}\\
    &\equidyn\bindXtoYinZ {(\bindXtoYinZ M x {\ret\thunk\ret x})} w \bindXtoYinZ N y \bindXtoYinZ {\force w} x N' \tag{$\u F\eta$}\\
    &\equidyn\bindXtoYinZ {\ret \thunk M} w \bindXtoYinZ N y \bindXtoYinZ {\force w} x N' \tag{$M$ thunkable}\\    
    &\equidyn\bindXtoYinZ N y \bindXtoYinZ {\force \thunk M} x N' \tag{$\u F\beta$}\\
    &\equidyn\bindXtoYinZ N y \bindXtoYinZ M x N' \tag{$U\beta$}\\    
  \end{align*}
\end{proof}

Next, we show thunkables are closed under composition and that return
of a value is always thunkable.  This allows us to easily build up
bigger thunkables from smaller ones.
\begin{lemma}[Thunkables compose]
  If $\Gamma \vdash M : \u F A$ and $\Gamma, x : A \vdash N : \u F A'$
  are thunkable, then
  \[ \bindXtoYinZ M x N \]
  is thunkable.
\end{lemma}
\begin{proof}
  \begin{align*}
    &\bindXtoYinZ {(\bindXtoYinZ M x N)} y \ret\thunk\ret y\\
    &\equidyn \bindXtoYinZ M x \bindXtoYinZ N y \ret\thunk\ret y\tag{$\u F\eta$}\\
    &\equidyn \bindXtoYinZ M x \ret \thunk N \tag{$N$ thunkable}\\
    &\equidyn \bindXtoYinZ M x \ret \thunk (\bindXtoYinZ {\ret x} x N)\tag{$\u F\beta$}\\
    &\equidyn \bindXtoYinZ M x \bindXtoYinZ {\ret\thunk\ret x} w \ret \thunk (\bindXtoYinZ {\force w} x N)\tag{$\u F\beta,U\beta$}\\
    &\equidyn \bindXtoYinZ {(\bindXtoYinZ M x \ret\thunk\ret x)} w \ret \thunk (\bindXtoYinZ {\force w} x N)\tag{$\u F\eta$}\\
    &\equidyn \bindXtoYinZ {\ret\thunk M} w \ret \thunk (\bindXtoYinZ {\force w} x N)\tag{$M$ thunkable}\\
    &\equidyn  \ret \thunk (\bindXtoYinZ {\force \thunk M} x N)\tag{$\u F\beta$}\\
    &\equidyn  \ret \thunk (\bindXtoYinZ {M} x N)\tag{$U\beta$}\\
  \end{align*}
\end{proof}

\begin{lemma}[Return is Thunkable]
  If $\Gamma \vdash V : A$ then $\ret V$ is thunkable.
\end{lemma}
\begin{proof}
  By $\u F\beta$:
  \[ \bindXtoYinZ {\ret V} x \ret\thunk\ret x \equidyn \ret\thunk\ret V \]
\end{proof}

\begin{lemma}[Complex Values Simplify to Thunkable Terms]
  If $\Gamma \vdash V : A$ is a (possibly) complex value, then $\Gamma
  \vdash \simp V : \u F A$ is thunkable.
\end{lemma}
\begin{longproof}
  Introduction forms follow from return is thunkable and thunkables
  compose. For elimination forms it is sufficient to show that when
  the branches of pattern matching are thunkable, the pattern match
  is thunkable.
  \begin{enumerate}
  \item $x$: We need to show $\simp x = \ret x$ is thunkable, which we
    proved as a lemma above.
  \item{} $0$ elim, we need to show 
    \[ \bindXtoYinZ {\absurd V} y \ret\thunk\ret y\equidyn \ret\thunk {\absurd V}\]
    but by $\eta0$ both sides are equivalent to $\absurd V$.
  \item{} $+$ elim, we need to show
    \[
    \ret\thunk (\caseofXthenYelseZ V {x_1. M_1} {x_2. M_2})
    \equidyn 
    \bindXtoYinZ {(\caseofXthenYelseZ V {x_1. M_1} {x_2. M_2})} y \ret\thunk \ret y
    \]
    \begin{align*}
      &\ret\thunk (\caseofXthenYelseZ V {x_1. M_1} {x_2. M_2})\\
      &\equidyn
      \caseofX V \tag{$+\eta$}\\
      &\qquad\thenY {x_1. \ret\thunk (\caseofXthenYelseZ {\inl x_1} {x_1. M_1} {x_2. M_2})}\\
      &\qquad\elseZ {x_2. \ret\thunk (\caseofXthenYelseZ {\inr x_2} {x_1. M_1} {x_2. M_2})}\\
      &\equidyn\caseofX V \tag{$+\beta$}\\
      &\qquad\thenY {x_1. \ret\thunk M_1}\\
      &\qquad\elseZ {x_2. \ret\thunk M_2}\\
      &\equidyn\caseofX V \tag{$M_1,M_2$ thunkable}\\
      &\qquad\thenY {x_1. \bindXtoYinZ {M_1} y \ret\thunk\ret y}\\
      &\qquad\elseZ {x_2. \bindXtoYinZ {M_2} y \ret\thunk\ret y}\\
      &\equidyn \bindXtoYinZ {(\caseofXthenYelseZ V {x_1. M_1}{x_2. M_2})} y \ret\thunk\ret y\tag{commuting conversion}\\
    \end{align*}
  \item{} $\times$ elim
    \begin{align*}
      &\ret\thunk (\pmpairWtoXYinZ V x y M)\\
      &\equidyn \pmpairWtoXYinZ V x y \ret\thunk \pmpairWtoXYinZ {(x,y)} x y M\tag{$\times\eta$}\\
      &\equidyn \pmpairWtoXYinZ V x y \ret\thunk M\tag{$\times\beta$}\\
      &\equidyn \pmpairWtoXYinZ V x y \bindXtoYinZ M z \ret\thunk\ret z\tag{$M$ thunkable}\\
      &\equidyn \bindXtoYinZ {(\pmpairWtoXYinZ V x y M)} z \ret\thunk\ret z\tag{commuting conversion}
    \end{align*}
  \item $1$ elim
    \begin{align*}
      &\ret\thunk (\pmpairWtoinZ V x y M)\\
      &\equidyn \pmpairWtoinZ V \ret\thunk \pmpairWtoinZ {()} M\tag{$1\eta$}\\
      &\equidyn \pmpairWtoinZ V \ret\thunk M\tag{$1\beta$}\\
      &\equidyn \pmpairWtoinZ V \bindXtoYinZ M z \ret\thunk\ret z\tag{$M$ thunkable}\\
      &\equidyn \bindXtoYinZ {(\pmpairWtoinZ V M)} z \ret\thunk\ret z\tag{commuting conversion}
    \end{align*}  \item $\mu$ elim
    \begin{align*}
      &\ret\thunk (\pmmuXtoYinZ V x M)\\
      &\equidyn \pmmuXtoYinZ V x \ret\thunk \pmmuXtoYinZ {\roll x} x M\tag{$\mu\eta$}\\
      &\equidyn \pmmuXtoYinZ V x \ret\thunk M\tag{$\mu\beta$}\\
      &\equidyn \pmmuXtoYinZ V x \bindXtoYinZ M y \ret\thunk\ret y\tag{$M$ thunkable}\\
      &\equidyn \bindXtoYinZ {(\pmmuXtoYinZ V x M)} y \ret\thunk\ret y\tag{commuting conversion}
    \end{align*}
  \end{enumerate}
\end{longproof}

Dually, we have that a stack out of a force is linear and that linears
are closed under composition, so we can easily build up bigger linear
morphisms from smaller ones.
\begin{lemma}[Force to a stack is Linear]
  If $\Gamma \pipe \bullet : \u B \vdash S : \u C$, then
  $\Gamma , x : U\u B\vdash S[\force x] : \u B$ is linear in $x$.
\end{lemma}
\begin{proof}
  \begin{align*}
    S[\force \thunk {(\bindXtoYinZ {\force z} x \force x)}]
    &\equidyn
    S[{(\bindXtoYinZ {\force z} x \force x)}]\tag{$U\beta$}\\
    &\equidyn \bindXtoYinZ {\force z} x S[\force x] \tag{$\u F\eta$}
  \end{align*}
\end{proof}

\begin{lemma}[Linear Terms Compose]
  If $\Gamma , x : U \u B \vdash M : \u B'$ is linear in $x$ and
  $\Gamma , y : \u B' \vdash N : \u B''$ is linear in $y$, then
  $\Gamma , x : U \u B \vdash N[\thunk M/y] : $
\end{lemma}
\begin{proof}
  \begin{align*}
    &N[\thunk M/y][\thunk{(\bindXtoYinZ {\force z} x \force x)}/x]\\
    &= N[\thunk {(M[\thunk{(\bindXtoYinZ {\force z} x \force x)}])}/y]\\
    &\equidyn N[\thunk {(\bindXtoYinZ {\force z} x M)}/y]\tag{$M$ linear}\\
    &\equidyn N[\thunk{(\bindXtoYinZ {\force z} x \force \thunk M)}/y] \tag{$U\beta$}\\
    &\equidyn
    N[\thunk{(\bindXtoYinZ {\force z} x \bindXtoYinZ {\ret\thunk M} y \force y)}/y] \tag{$\u F\beta$}\\
    &\equidyn
    N[\thunk{(\bindXtoYinZ {(\bindXtoYinZ {\force z} x \ret\thunk M)} y \force y)}/y] \tag{$\u F\eta$}\\
    &\equidyn
    N[\thunk{(\bindXtoYinZ {\force w} y \force y)}/y][\thunk(\bindXtoYinZ {\force z} x \ret\thunk M)/w] \tag{$U\beta$}\\
    &\equidyn (\bindXtoYinZ {\force w} y N)[\thunk (\bindXtoYinZ {\force z} x \ret \thunk M)/w] \tag{$N$ linear}\\
    &\equidyn (\bindXtoYinZ {(\bindXtoYinZ {\force z} x \ret \thunk M)} y N) \tag{$U\beta$}\\
    &\equidyn (\bindXtoYinZ {\force z} x \bindXtoYinZ {\ret\thunk M} y N \tag{$\u F\eta$}\\
    &\equidyn \bindXtoYinZ {\force z} x N[\thunk M/y]
  \end{align*}
\end{proof}

\begin{lemma}[Complex Stacks Simplify to Linear Terms]
  If $\Gamma\pipe \bullet : \u B \vdash S : \u C$ is a (possibly)
  complex stack, then $\Gamma, z : U\u B \vdash \simpp{S} : \u C$ is linear in $z$.
\end{lemma}
\begin{longproof}
  There are $4$ classes of rules for complex stacks: those that are
  rules for simple stacks ($\bullet$, computation type elimination
  forms), introduction rules for negative computation types where the
  subterms are complex stacks, elimination of positive value types
  where the continuations are complex stacks and finally application
  to a complex value.

  The rules for simple stacks are easy: they follow immediately from
  the fact that forcing to a stack is linear and that complex stacks
  compose.  For the negative introduction forms, we have to show that
  binding commutes with introduction forms. For pattern matching
  forms, we just need commuting conversions. For function application,
  we use the lemma that binding a thunkable in a linear term is
  linear.
  \begin{enumerate}
  \item $\bullet$: This is just saying that $\force z$ is linear,
    which we showed above.
  \item $\to$ elim We need to show, assuming that $\Gamma, x : \u B
    \vdash M : \u C$ is linear in $x$ and $\Gamma \vdash N : \u F A$
    is thunkable, that
    \[
    \bindXtoYinZ N y M\,y
    \]
    is linear in $x$.
    \begin{align*}
      &\bindXtoYinZ N y (M[\thunk{(\bindXtoYinZ {\force z} x \force x)}/x])\,y\\
      &\equidyn \bindXtoYinZ N y (\bindXtoYinZ {\force z} x M)\,y \tag{$M$ linear in $x$}\\
      &\equidyn \bindXtoYinZ N y \bindXtoYinZ {\force z} x M\,y \tag{$\u F\eta$}\\
      &\equidyn \bindXtoYinZ {\force z} x \bindXtoYinZ N y M\,y\tag{thunkables are central}
    \end{align*}
  \item $\to$ intro
    \begin{align*}
      & \lambda y:A. M[\thunk{(\bindXtoYinZ {\force z} x \force x)}/x]\\
      &\equidyn \lambda y:A. \bindXtoYinZ {\force z} x M \tag{$M$ is linear}\\
      &\equidyn \lambda y:A. \bindXtoYinZ {\force z} x (\lambda y:A. M)\, y \tag{$\to\beta$}\\
      &\equidyn \lambda y:A. (\bindXtoYinZ {\force z} x (\lambda y:A. M))\, y \tag{$\u F\eta$}\\
      &\equidyn \bindXtoYinZ {\force z} x (\lambda y:A. M) \tag{$\to\eta$}
    \end{align*}
  \item $\top$ intro
    We need to show
    \[ \bindXtoYinZ {\force z} w \{\} \equidyn \{\} \]
    Which is immediate by $\top\eta$
  \item $\with$ intro
    \begin{align*}
      &     \pairone{M[\thunk {(\bindXtoYinZ {\force z} x \force x)}]/x}\\
      &\pairtwo{N[\thunk {(\bindXtoYinZ {\force z} x \force x)}/x]}\\
      &\equidyn \pairone{\bindXtoYinZ {\force z} x M}\tag{$M, N$ linear}\\
      &\qquad    \pairtwo{\bindXtoYinZ {\force z} x N}\\
      &\equidyn \pairone{\bindXtoYinZ {\force z} x {\pi \pair M N}}\tag{$\with\beta$}\\
      &\qquad    \pairtwo{\bindXtoYinZ {\force z} x {\pi' \pair M N}}\\
      &\equidyn \pairone{\pi({\bindXtoYinZ {\force z} x \pair M N})}\tag{$\u F\eta$}\\
      &\qquad    \pairtwo{\pi'({\bindXtoYinZ {\force z} x \pair M N})}\\
      &\equidyn \bindXtoYinZ {\force z} x \pair M N\tag{$\with\eta$}
    \end{align*}
  \item $\nu$ intro
    \begin{align*}
      & \roll M[\thunk{(\bindXtoYinZ {\force z} x \force x)}/x]\\
      &\equidyn \roll (\bindXtoYinZ {\force z} x M) \tag{$M$ is linear} \\
      &\equidyn \roll (\bindXtoYinZ {\force z} x \unroll \roll M) \tag{$\nu\beta$}\\
      &\equidyn \roll \unroll (\bindXtoYinZ {\force z} x \roll M) \tag{$\u F\eta$}\\
      &\equidyn \bindXtoYinZ {\force z} x (\roll M) \tag{$\nu\eta$}
    \end{align*}
  \item $\u F$ elim: Assume $\Gamma, x : A \vdash M : \u F A'$ and
    $\Gamma, y : A' \vdash N : \u B$, then we need to show
    \[ \bindXtoYinZ M y N \]
    is linear in $M$.
    \begin{align*}
      & \bindXtoYinZ {M[\thunk{(\bindXtoYinZ {\force z} x \force x)}/x]} y N\\
      & \equidyn
      \bindXtoYinZ {(\bindXtoYinZ {\force z} x M)} y N\tag{$M$ is linear}\\
      &\equidyn
      \bindXtoYinZ {\force z} x \bindXtoYinZ M y N\tag{$\u F\eta$}
    \end{align*}
  \item $0$ elim: We want to show $\Gamma, x:U\u B \vdash \absurd V :
    \u C$ is linear in $x$, which means showing:
    \[ \absurd V \equidyn \bindXtoYinZ {\force z} x \absurd V
    \]
    which follows from $0\eta$
  \item $+$ elim: Assuming $\Gamma, x : U\u B, y_1 : A_1 \vdash M_1 :
    \u C$ and $\Gamma, x : U\u B, y_2: A_2\vdash M_2 : \u C$ are
    linear in $x$, and $\Gamma \vdash V : A_1 + A_2$, we need to show
    \[ \caseofXthenYelseZ V {y_1. M_1} {y_2. M_2} \]
    is linear in $x$.
    \begin{align*}
      & \caseofX V\\
      & \,\,\thenY {y_1. M_1[\thunk{(\bindXtoYinZ {\force z} x \force x)}/x]}\\
      & \,\,\elseZ {y_2. M_2[\thunk{(\bindXtoYinZ {\force z} x \force x)}/x]}\\
      &\equidyn \caseofXthenYelseZ V {y_1. \bindXtoYinZ {\force z} x M_1}{y_2. \bindXtoYinZ {\force z} x M_2}\tag{$M_1,M_2$ linear}\\
      &\equidyn
       \bindXtoYinZ {\force z} x \caseofXthenYelseZ V {y_1. M_1}{y_2. M_2}
    \end{align*}
  \item $\times$ elim: Assuming $\Gamma, x:U\u B, y_1 : A_1, y_2 : A_2
    \vdash M : \u B$ is linear in $x$ and $\Gamma \vdash V : A_1
    \times A_2$, we need to show
    \[ \pmpairWtoXYinZ V {y_1}{y_2} M \]
    is linear in $x$.
    \begin{align*}
      &\pmpairWtoXYinZ V {y_1}{y_2} M[[\thunk{(\bindXtoYinZ {\force z} x \force x)}/x]]\\
      &\equidyn \pmpairWtoXYinZ V {y_1}{y_2} \bindXtoYinZ {\force z} x M\tag{$M$ linear}\\
      &\equidyn  \bindXtoYinZ {\force z} x\pmpairWtoXYinZ V {y_1}{y_2} M\tag{comm. conv}\\
    \end{align*}
  \item $\mu$ elim: Assuming $\Gamma , x : U \u B, y : A[\mu X.A/X]
    \vdash M : \u C$ is linear in $x$ and $\Gamma \vdash V : \mu X.A$,
    we need to show
    \[ \pmmuXtoYinZ V y M \]
    is linear in $x$.
    \begin{align*}
      & \pmmuXtoYinZ V y M[\thunk{(\bindXtoYinZ {\force z} x \force x)}/x]\\
      & \equidyn \pmmuXtoYinZ V y \bindXtoYinZ {\force z} x M\tag{$M$ linear}\\
      &\equidyn  \bindXtoYinZ {\force z} x\pmmuXtoYinZ V y M \tag{commuting conversion}
    \end{align*}
  \end{enumerate}
\end{longproof}

Composing this with the previous translation from GTT to \cbpvstar\/
shows that \emph{GTT value type upcasts are thunkable and computation
  type downcasts are linear}.

Since the translation takes values and stacks to terms, it cannot
preserve substitution up to equality.
Rather, we get the following, weaker notion that says that the
translation of a syntactic substitution is equivalent to an effectful
composition.
\begin{lemma}[Compositionality of De-complexification]
  \begin{enumerate}
  \item If $\Gamma, x : A\pipe \Delta\vdash E : T$ and $\Gamma \vdash V : A$
    are complex terms, then
    \[ 
    \simpp{E[V/x]} \equidyn \bindXtoYinZ {\simp V} x {\simp E}
    \]
  \item If $\Gamma \pipe \bullet : \u B \vdash S : \u C$ and $\Gamma
    \pipe \Delta \vdash M : \u B$, then
    \[
    \simpp{S[M]} \equidyn \simp{S}[\thunk\simp{M}/z]
    \]
  \end{enumerate}
\end{lemma}
\begin{longproof}
  \begin{enumerate}
  \item   First, note that every occurrence of a variable in $\simp E$ is of
  the form $\ret x$ for some variable $x$. This means we can define
  substitution of a \emph{term} for a variable in a simplified term by
  defining $\simp{E}[N/\ret x]$ to replace every $\ret x : \u F A$
  with $N : \u F A$. Then it is an easy observation that
  simplification is compositional on the nose with respect to this
  notion of substitution:
  \[ \simpp{E[V/x]} = \simp{E}[\simp V / \ret x] \]

  Next by repeated invocation of $U\beta$,
  \[ \simp{E}[\simp V/\ret x] \equidyn \simp{E}[\force\thunk\simp V/\ret x] \]

  Then we can lift the definition of the thunk to the top-level by $\u F\beta$:
  \[ \simp{E}[\force\thunk\simp V/\ret x] \equidyn
  \bindXtoYinZ \ret\thunk \simp V w \simp{E}[\force w/\ret x]
  \]
  Then because $\simp V$ is thunkable, we can bind it at the top-level
  and reduce an administrative redex away to get our desired result:
  \begin{align*}
    &\bindXtoYinZ \ret\thunk \simp V w \simp{E}[\force w/\ret x]\\
    &\equidyn \bindXtoYinZ {\simp V} x \bindXtoYinZ {\ret\thunk\ret x} w \simp{E}[\force w/\ret x]\tag{$V$ thunkable}\\
    &\equidyn \bindXtoYinZ {\simp V} x \simp{E}[\force \thunk\ret x/\ret x]\tag{$\u F\beta$}\\
    &\equidyn \bindXtoYinZ {\simp V} x \simp{E}[\ret x/\ret x]\tag{$U\beta$}\\
    &\equidyn \bindXtoYinZ {\simp V} x \simp{E}\\
  \end{align*}
  \item Note that every occurrence of $z$ in $\simp S$ is of the form
    $\force z$. This means we can define substitution of a \emph{term}
    $M : \u B$ for $\force z$ in $\simp S$ by replacing $\force z$
    with $M$.  It is an easy observation that simplification is
    compositional on the nose with respect to this notion of
    substitution:
    \[ \simpp{S[M/\bullet]} = \simp S[\simp M/\force z] \]
    Then by repeated $U\beta$, we can replace $\simp M$ with a forced thunk:
    \[ \simp S[\simp M/\force z] \equidyn \simp S[\force\thunk \simp M/\force z] \]
    which since we are now substituting a force for a force is the
    same as substituting the thunk for the variable:
    \[ \simp S[\force\thunk \simp M/\force z]
    \equidyn
    \simp S[\thunk \simp M / z]
    \]
  \end{enumerate}
\end{longproof}

\begin{theorem}[De-complexification preserves Dynamism]
  If $\Gamma \pipe \Delta \vdash E \ltdyn E' : T$ then ${\Gamma, \simp
    \Delta \vdash \simp E \ltdyn \simp{E'} : \simp T}$
\end{theorem}
\begin{longproof}
  \begin{enumerate}
  \item Reflexivity is translated to reflexivity.
  \item Transitivity is translated to transitivity.
  \item Compatibility rules are translated to compatibility rules.
  \item Substitution of a Value
    \[
    \inferrule
    {\Gamma, x : A, \simp\Delta \vdash \simp E \ltdyn \simp {E'} : \simp T \and \Gamma \vdash \simp V \ltdyn \simp {V'} : \u F A}
    {\Gamma, \simp\Delta \vdash \simp{E[V/x]} \ltdyn \simp{E'[V'/x]} : \simp T}
    \]
    By the compositionality lemma, it is sufficient to show:
    \[ \bindXtoYinZ {\simp V} x {\simp E} \ltdyn \bindXtoYinZ {\simp {V'}} {x} E' \]
    which follows by bind compatibility.
  \item Plugging a term into a hole:
    \[
    \inferrule
    {\Gamma, z : U{\u C} \vdash \simp {S} \ltdyn \simp{S'} : \u B\and
      \Gamma,\simp\Delta \vdash \simp{M} \ltdyn \simp{M'} : \u C}
    {\Gamma, \simp\Delta \vdash \simp{S[M]} \ltdyn \simp{S'[M']} : \u B}
    \]
    By compositionality, it is sufficient to show
    \[ \simp{S}[\thunk{\simp M}/z] \ltdyn \simp{S'}[\thunk{\simp{M'}}/z] \]
    which follows by thunk compatibility and the simple substitution rule.
  \item Stack strictness
    We need to show for $S$ a complex stack,
    that
    \[ \simpp{S[\err]} \equidyn \err \]
    By stack compositionality we know
    \[ \simpp{S[\err]} \equidyn \simp{S}[{\thunk \err/z}] \]
    \begin{align*}
      \sem{S}[{\thunk \err/z}]
      &\equidyn \simp{S}[\thunk {(\bindXtoYinZ \err y \err)}/z]\tag{Stacks preserve $\err$}\\
      &\equidyn
      \bindXtoYinZ \err y \simp{S}[{\thunk \err/z}] \tag{$\simp S$ is linear in $z$}\\
      &\equidyn \err \tag{Stacks preserve $\err$}
    \end{align*}
  \item $1\beta$ By compositionality it is sufficient to show
    \[\bindXtoYinZ {\ret ()} x \pmpairWtoinZ x {\simp E} \equidyn \bindXtoYinZ {\ret ()} x \simp E \]
    which follows by $\u F\beta, 1\beta$.
  \item $1\eta$ We need to show for $\Gamma, x : 1 \pipe \Delta \vdash E : T$
    \[ \simp{E} \equidyn \bindXtoYinZ {\ret x} x \pmpairWtoinZ x {\simpp{E[()/x]}}\]
    after a $\u F\beta$, it is sufficient using $1\eta$ to prove:
    \[ {\simpp{E[()/x]}} \equidyn \simp E[()/x] \]
    which follows by compositionality and $\u F\beta$:
    \[ {\simpp{E[()/x]}} \equidyn \bindXtoYinZ {\ret ()} x {\simp E} \equidyn \simp{E}[()/x] \]
  \item $\times\beta$ By compositionality it is sufficient to show
    \begin{align*}
    &\bindXtoYinZ {(\bindXtoYinZ {\simp{V_1}} {x_1} \bindXtoYinZ {\simp{V_2}} {x_2} {\ret (x_1,x_2)})} x \pmpairWtoXYinZ {x} {x_1}{x_2} {\simp E}\\
      &\equidyn \bindXtoYinZ {\simp{V_1}} {x_1} \bindXtoYinZ {\simp{V_2}} {x_2} \simp E
    \end{align*}
    which follows by $\u F\eta, \u F\beta, \times\beta$.
  \item $\times\eta$ We need to show for $\Gamma, x : A_1\times A_2
    \pipe\Delta \vdash E : T$ that
    \[ \simp{E} \equidyn \bindXtoYinZ {\ret x} x \pmpairWtoXYinZ x {x_1}{x_2} \simpp{E[(x_1,x_2)/x]} \]
    by $\u F\beta,\times\eta$ it is sufficient to show
    \[ \simp{E[(x_1,x_2)/x]} \equidyn \simp{E}[(x_1,x_2)/x] \]
    Which follows by compositionality:
    \begin{align*}
      &\simp{E[(x_1,x_2)/x]}\\
      &\equidyn \bindXtoYinZ {x_1}{x_1}\bindXtoYinZ {x_2}{x_2} \bindXtoYinZ {\ret (x_1,x_2)} x \simp E\tag{compositionality}\\
      &\equidyn \bindXtoYinZ {\ret (x_1,x_2)} x \simp E\tag{$\u F\beta$}\\
      &\equidyn \simp E[(x_1,x_2)/x]
    \end{align*}
  \item $0\eta$
    We need to show for any $\Gamma, x : 0 \pipe \Delta \vdash E : T$
    that
    \[ \simp E \equidyn \bindXtoYinZ {\ret x} x \absurd x \]
    which follows by $0\eta$
  \item $+\beta$ Without loss of generality, we do the $\inl$ case
    By compositionality it is sufficient to show
    \[
    \bindXtoYinZ {(\bindXtoYinZ {\simp V} x {\inl x})} x \caseofXthenYelseZ x {x_1. \simp E_1}{x_2. \simp E_2}
    \equidyn \simp{E_1[V/x_1]}    
    \]
    which holds by $\u F\eta,\u F\beta, +\beta$
  \item $+\eta$ We need to show for any $\Gamma, x:A_1+A_2
    \pipe\Delta\vdash E :T$ that
    \[ \simp E \equidyn
    \bindXtoYinZ {\ret x} x \caseofXthenYelseZ x {x_1. \simpp{E[\inl x_1/x]}}{x_2. \simpp{E[\inl x_2/x]}} \]
    \begin{align*}
      &\simp E\\
      &\equidyn \caseofXthenYelseZ x {x_1. \simp{E}[\inl x_1/x]}{x_2. \simp{E}[\inl x_2/x]} \tag{$+\eta$}\\
      &\equidyn
      \caseofXthenYelseZ x {x_1. \bindXtoYinZ {\ret \inl x_1} x \simp{E}}{x_2. \bindXtoYinZ {\ret \inl x_2} x \simp{E}}\tag{$\u F\beta$}\\
      &\equidyn
      \caseofXthenYelseZ x {x_1. \simp{E[\inl x_1]/x}}{x_2. \simp{E[\inl x_2]/x}}\tag{compositionality}\\
      &\equidyn 
      \bindXtoYinZ {\ret x} x \caseofXthenYelseZ x {x_1. \simp{E[\inl x_1]/x}}{x_2. \simp{E[\inl x_2]/x}}\tag{$\u F\beta$}
    \end{align*}
  \item $\mu\beta$ By compositionality it is sufficient to show
    \begin{align*}
    &\bindXtoYinZ {(\bindXtoYinZ {\simp V} y {\ret \roll y})} x \pmmuXtoYinZ {x} y E\\
      &\equidyn \bindXtoYinZ {\simp V} y \simp E
    \end{align*}
    which follows by $\u F\eta, \u F\beta, \mu\beta$.
  \item $\mu\eta$ We need to show for $\Gamma, x : \mu X. A \pipe\Delta \vdash E : T$ that
    \[ \simp{E} \equidyn \bindXtoYinZ {\ret x} x \pmmuXtoYinZ x y \simpp{E[\roll y/x]} \]
    by $\u F\beta,\times\eta$ it is sufficient to show
    \[ \simp{E[\roll y/x]} \equidyn \simp{E}[\roll y/x] \]
    Which follows by compositionality:
    \begin{align*}
      &\simp{E[\roll y/x]}\\
      &\equidyn \bindXtoYinZ {\ret y} y \bindXtoYinZ {\ret \roll y} x \simp E\tag{compositionality}\\
      &\equidyn \bindXtoYinZ {{\ret\roll y}} x \simp E\tag{$\u F\beta$}\\
      &\equidyn \simp E[\roll y/x] \tag{$\u F\beta$}
    \end{align*}
  \item $U\beta$ We need to show
    \[ \bindXtoYinZ {\ret \simp M} x {\force x} \equidyn \simp M \]
    which follows by $\u F\beta, U\beta$
  \item $U\eta$ We need to show for any $\Gamma \vdash V : U\u B$ that
    \[ \simp V \equidyn \ret \thunk{(\bindXtoYinZ {\simp V} x \force x)}\]
    By compositionality it is sufficient to show
    \[ \simp V \equidyn \bindXtoYinZ {\simp V} x \ret\thunk{(\bindXtoYinZ {\ret x} x \force x)}\]
    which follows by $U\eta$ and some simple reductions:
    \begin{align*}
      \bindXtoYinZ {\simp V} x \ret\thunk{(\bindXtoYinZ {\ret x} x \force x)}\\
      &\equidyn \bindXtoYinZ {\simp V} x\ret\thunk{\force x}\tag{$\u F\beta$}\\
      &\equidyn \bindXtoYinZ {\simp V} x\ret x\tag{$U\eta$}\\
      &\equidyn \simp V\tag{$\u F\eta$}
    \end{align*}
  \item $\to\beta$
    By compositionality it is sufficient to show
    \[ \bindXtoYinZ {\simp V} x (\lambda x:A. \simp M)\,x
    \equidyn \bindXtoYinZ {\simp V} x \simp M \]
    which follows by $\to\beta$
  \item $\to\eta$ We need to show
    \[ z : U(A \to \u B) \vdash
    \force z \equidyn
    \lambda x:A. \bindXtoYinZ {\ret x} x (\force z)\,x
    \]
    which follows by $\u F\beta, \to\eta$
  \item $\top\eta$ We need to show
    \[ z : U\top \vdash \force z \equidyn \{\} \]
    which is exactly $\top\eta$.
  \item $\with\beta$ Immediate by simple $\with\beta$.
  \item $\with\eta$ We need to show
    \[ z : U(\u B_1\with\u B_2) \vdash \force z \equidyn \pair{\pi\force z}{\pi'\force z}\]
    which is exactly $\with\eta$
  \item $\nu\beta$ Immediate by simple $\nu\beta$
  \item $\nu\eta$ We need to show
    \[ z : U(\nu \u Y. \u B) \vdash \force z \equidyn \roll\unroll z \]
    which is exactly $\nu\eta$
  \item $\u F\beta$
    We need to show
    \[ \bindXtoYinZ {\simp V} x \simp M \equidyn \simp{M[V/x]}\]
    which is exactly the compositionality lemma.
  \item $\u F\eta$ We need to show
    \[ z : U(\u F A)\force z \vdash \bindXtoYinZ {\force z} x \bindXtoYinZ {\ret x} x \ret x \]
    which follows by $\u F\beta,\u F\eta$
  \end{enumerate}
\end{longproof}

\begin{theorem}[Complex CBPV is Conservative over CBPV]
  If $M, M'$ are terms in CBPV and $M \ltdyn M'$ is provable in \cbpvstar\ 
  then $M \ltdyn M'$ is provable in CBPV.
\end{theorem}
\begin{longproof}
  Because de-complexification preserves dynamism, $\simp M \ltdyn
  \simp{M'}$ in simple CBPV. Then it follows because
  de-complexification is equivalent to identity (in CBPV):
  \[ M \equidyn \simp M \ltdyn \simp {M'} \equidyn M' \]
\end{longproof}
\end{longonly}

\section{Operational Model of GTT}
\label{sec:operational}

In this section, we establish a model of our CBPV inequational theory
using a notion of observational approximation based on the CBPV
operational semantics.
By composition with the axiomatic graduality theorem, this establishes
the \emph{operational graduality} theorem, i.e., a theorem analogous
to the \emph{dynamic gradual guarantee}~\cite{refined}.  

\subsection{Call-by-push-value operational semantics}

We use a small-step operational semantics for CBPV
\iflong
in figure
\ref{fig:cbpv-operational-semantics}.
\begin{figure}
\fi
\begin{small}
\begin{minipage}[t]{0.65\textwidth}
\[
  \begin{array}{rcl}
    S[\err] &\stepzero& \err\\
\iflong
    S[\caseofXthenYelseZ{\inl V}{x_1. M_1}{x_2. M_2}] &\stepzero & S[M_1[V/x_1]]\\
    S[\caseofXthenYelseZ{\inr V}{x_1. M_1}{x_2. M_2}] &\stepzero & S[M_2[V/x_2]]\\
\fi
    S[\pmpairWtoXYinZ{(V_1,V_2)}{x_1}{x_2}{M}] &\stepzero & S[M[V_1/x_1,V_2/x_2]]\\
    S[\pmmuXtoYinZ{\rollty A V}{x}{M}] &\stepone & S[M[V/x]]\\
    S[\force\thunk M] &\stepzero & S[M]\\
\iflong
    S[\letXbeYinZ V x M] &\stepzero & S[M[V/x]]\\
\fi
    S[\bindXtoYinZ {\ret V} x M] &\stepzero & S[M[V/x]]\\
    S[(\lambda x:A. M)\,V] &\stepzero & S[M[V/x]]\\
\iflong
    S[\pi \pair{M}{M'}] &\stepzero & S[M]\\
    S[\pi' \pair{M}{M'}] &\stepzero & S[M']\\
\fi
    S[\unroll \rollty{\u B} M] &\stepone & S[M]\\
\end{array}
\]
\end{minipage}%
\begin{minipage}[t]{0.27\textwidth}
\begin{mathpar}
    \inferrule
    { }
    {M \bigstepsin 0 M}

\vspace{1.3em}

    \inferrule
    {M_1 \stepsin{i} M_2 \and M_2 \bigstepsin j M_3}
    {M_1 \bigstepsin {i+j} M_3}
  \end{mathpar}
\end{minipage}
\end{small}
\iflong
  \caption{CBPV Operational Semantics}
  \label{fig:cbpv-operational-semantics}
\end{figure}
\fi

This is morally the same as in \citet{levy03cbpvbook}, but we present
stacks in a manner similar to Hieb-Felleisen style evaluation
contexts\iflong(rather than as an explicit stack machine with stack frames)\fi.
We also make the step relation count unrollings of a recursive or
corecursive type, for the step-indexed logical relation later.
The operational semantics is only defined for terms of type
$\cdot \vdash M : \u F (1+1)$, which we take as the type of whole
programs.

\iflong
We can then observe the following standard operational properties. (We
write $M \step N$ with no index when the index is irrelevant.)
\begin{lemma}[Reduction is Deterministic]
  If $M \step M_1$ and $M \step M_2$, then $M_1 = M_2$.
\end{lemma}

\begin{lemma}[Subject Reduction]
  If $\cdot \vdash M : \u F A$ and $M \step M'$ then
  $\cdot \vdash M' : \u F A$.
\end{lemma}

\begin{lemma}[Progress]
  If $\cdot \vdash M : \u F A$ then one of the following holds:
  \begin{mathpar}
    M = \err \and M = \ret V \text{with} V:A \and \exists M'.~ M \step M'
  \end{mathpar}
\end{lemma}
\fi

\begin{shortonly}
  It is easy to see that the operational semantics is deterministic
  and progress and type preservation theorems hold, which allows us to
  define the ``final result'' of a computation as follows:
\end{shortonly}
\begin{longonly}
The standard progress-and-preservation properties allow us to define
 the ``final result'' of a computation as follows:
\end{longonly}
\begin{corollary}[Possible Results of Computation]
  For any $\cdot \vdash M : \u F 2$,
  \begin{longonly}
    one of the following is true:
  \begin{mathpar}
    M \Uparrow \and M \Downarrow \err\and M \Downarrow \ret \tru \and
    M \Downarrow \ret \fls
  \end{mathpar}
  \end{longonly}
  \begin{shortonly}
    either $M \Uparrow$ or $M \Downarrow \err$ or $M \Downarrow \ret
    \tru$ or  $M \Downarrow \ret \fls$.
  \end{shortonly}
\end{corollary}
\begin{longproof}
  We define $M \Uparrow$ to hold when if $M \bigstepsin{i} N$ then
  there exists $N'$ with $N \step N'$. For the terminating results, we
  define $M \Downarrow R$ to hold if there exists some $i$ with $M
  \bigstepsin{i} R$. Then we prove the result by coinduction on
  execution traces. If $M \in \{ \err, \ret\tru, \ret\fls \}$ then we
  are done, otherwise by progress, $M \step M'$, so we need only
  observe that each of the cases above is preserved by $\step$.
\end{longproof}
\begin{definition}[Results]
   The possible results of a computation are $ \diverge, \err,
   \ret \tru$ and $\ret \fls$. We denote a result by $R$, and define a
   function $\result$ which takes a program $\cdot \vdash M : \u F 2$,
   and returns its end-behavior, i.e., $\result(M)= \diverge$ if $M
   \Uparrow$ and otherwise $M \Downarrow \result(M)$.
\end{definition}

\subsection{Observational Equivalence and Approximation}
\label{sec:obs-equiv-approx}

Next, we define observational equivalence and approximation in CBPV.
\begin{longonly}
The (standard) definition of observational equivalence is that we
consider two terms (or values) to be equivalent when replacing one
with the other in any program text produces the same overall resulting
computation.
\end{longonly}
\ Define a context $C$ to be a term/value/stack with a single $[\cdot]$ as
some subterm/value/stack, and define a typing $C : (\Gamma \vdash \u B)
\Rightarrow (\Gamma' \vdash \u B')$ to hold when for any $\Gamma \vdash
M : \u B$, $\Gamma' \vdash C[M] : \u B'$ (and similarly for
values/stacks).  Using contexts, we can lift any relation on
\emph{results} to relations on open terms, values and stacks.
\begin{definition}[Contextual Lifting]
  Given any relation ${\sim} \subseteq \text{Result}^2$, we can define
  its \emph{observational lift} $\ctxize\sim$ to be the typed relation
  defined by
  \[ \Gamma \pipe \Delta \vDash E \ctxize\sim E' \in T = \forall C : (\Gamma\pipe\Delta \vdash T) \Rightarrow (\cdot \vdash \u F2).~ \result(C[E]) \sim \result(C[E'])\]
\end{definition}
\begin{longfigure}
\begin{small}
  \begin{mathpar}
    \begin{array}{rcl}
    C_V  & ::= [\cdot] & \rollty{\mu X.A}C_V \mid \inl{C_V} \mid \inr{C_V} \mid (C_V,V)\mid(V,C_V)\mid \thunk{C_M}\\
    \\
    C_M & ::= & [\cdot] \mid \letXbeYinZ {C_V} x M \mid \letXbeYinZ V x
    C_M \mid \pmmuXtoYinZ {C_V} x M \mid\pmmuXtoYinZ V x C_M \\
    & & \mid \rollty{\nu \u Y.\u B} C_M \mid \unroll C_M \mid \abort{C_V} \mid \caseofXthenYelseZ {C_V} {x_1. M_1}{x_2.M_2} \\
    & & 
    \mid\caseofXthenYelseZ V {x_1. C_M}{x_2.M_2} \mid\caseofXthenYelseZ
    V {x_1. M_1}{x_2.C_M} \mid \pmpairWtoinZ {C_V} M\\
    & & \mid \pmpairWtoinZ V C_M \mid \pmpairWtoXYinZ {C_V} x y M\mid \pmpairWtoXYinZ V x y C_M
    \mid \force{C_V} \\
    & & \mid \ret{C_V} \mid \bindXtoYinZ{C_M}{x}{N}
    \mid\bindXtoYinZ{M}{x}{C_M} \mid \lambda x:A.C_M \mid C_M\,V \mid M\,C_V \\
    & & \mid \pair{C_M}{M_2}\mid \pair{M_1}{C_M} \mid \pi C_M \mid \pi' C_M
    \\
    C_S &=& \pi C_S \mid \pi' C_S \mid S\,C_V\mid C_S\,V\mid \bindXtoYinZ {C_S} x M \mid \bindXtoYinZ S x C_M
    \end{array}
  \end{mathpar}
  \end{small}
  \caption{CBPV Contexts}
\end{longfigure}

\begin{shortonly}
The contextual lifting $\ctxize\sim$ is a preorder or equivalence
relation whenever the original relation $\sim$ is, and all $\sim$'s we
use will be at least preorders, so we write $\apreorder$ instead of
$\sim$ for a relation on results.
\end{shortonly}

\begin{longonly}
The contextual lifting $\ctxize\sim$ inherits much structure of the
original relation $\sim$ as the following lemma shows.
This justifies calling $\ctxize\sim$ a contextual preorder when $\sim$
is a preorder (reflexive and transitive) and similarly a contextual
equivalence when $\sim$ is an equivalence (preorder and symmetric).
\begin{definition}[Contextual Preorder, Equivalence]
  If $\sim$ is reflexive, symmetric or transitive, then for each
  typing, $\ctxize\sim$ is reflexive, symmetric or transitive as well,
  respectively.
\end{definition}
In the remainder of the paper we work only with relations that are at
least preorders so we write $\apreorder$ rather than $\sim$.
\end{longonly}

\begin{shortonly}
\noindent Three important relations arise as liftings: Equality of results lifts to observational equivalence
($\ctxize=$).  The preorder generated by $\err \ltdyn R$ (i.e. the other
three results are unrelated maximal elements) lifts to the notion of
\emph{error approximation} used in \citet{newahmed18} to prove the
graduality property ($\ctxize\ltdyn$).  The preorder generated by
$\diverge \preceq R$ lifts to the standard notion of \emph{divergence
  approximation} ($\ctxize\preceq$).
\end{shortonly}

\begin{longonly}
The most famous use of lifting is for observational equivalence,
which is the lifting of equality of results ($\ctxize=$), and we will show that
$\equidyn$ proofs in GTT imply observational equivalences.
However, as shown in \citet{newahmed18}, the graduality property is
defined in terms of an observational \emph{approximation} relation
$\ltdyn$ that places $\err$ as the least element, and every other
element as a maximal element.
Note that this is \emph{not} the standard notion of observational
approximation, which we write $\preceq$, which makes $\diverge$ a least
element and every other element a maximal element.
To distinguish these, we call $\ltdyn$ \emph{error} approximation and
$\preceq$ \emph{divergence} approximation.
We present these graphically (with two more) in Figure
\ref{fig:result-orders}.
\end{longonly}

\iflong
\begin{figure}
  \begin{small}
    \begin{minipage}{0.45\textwidth}
      \begin{center}
        \textbf{Diverge Approx. $\preceq$}\\
      \end{center}
      \begin{tikzcd}
        \ret\fls \arrow[rd, no head] & \ret \tru \arrow[d, no head] & \err \arrow[ld, no head] \\
        & \diverge & 
      \end{tikzcd}
    \end{minipage}
    \begin{minipage}{0.45\textwidth}
      \begin{center}
        \textbf{
          Error Approx. $\ltdyn$}
      \end{center}
      \begin{tikzcd}
        \ret\fls \arrow[rd, no head] & \ret \tru \arrow[d, no head] & \diverge \arrow[ld, no head] \\
        & \err & 
      \end{tikzcd}
    \end{minipage}
    \\\vspace{1em}
    \begin{minipage}{0.45\textwidth}
      \begin{center}
        \textbf{Error Approx. up to left-divergence
          $\errordivergeleft$}\\
      \end{center}
      \begin{tikzcd}
        \ret\fls \arrow[rd, no head] &  & \ret \tru \arrow[ld, no head] \\
        & \err , \diverge & 
      \end{tikzcd}
    \end{minipage}
    \begin{minipage}{0.45\textwidth}
      \vspace{1em}
      \begin{center}
        \textbf{Error Approx. up to right-divergence}
        $\errordivergeright$\\
      \end{center}
      \begin{tikzcd}
        & \diverge \arrow[ld, no head] \arrow[rd, no head] &  \\
        \ret\fls \arrow[rd, no head] &  & \ret \tru \arrow[ld, no head] \\
        & \err & 
      \end{tikzcd}
    \end{minipage}
    \\\vspace{1em}
    \begin{minipage}{0.45\textwidth}
      \vspace{1em}
      \begin{center}
        \textbf{Error Approx. up to right-divergence Op}
        $\errordivergerightop$\\
      \end{center}
      \begin{tikzcd}
        & \err \arrow[ld, no head] \arrow[rd, no head] &  \\
        \ret\fls \arrow[rd, no head] &  & \ret \tru \arrow[ld, no head] \\
        & \diverge & 
      \end{tikzcd}
    \end{minipage}
  \end{small}
  \caption{Result Orderings}
  \label{fig:result-orders}
\end{figure}
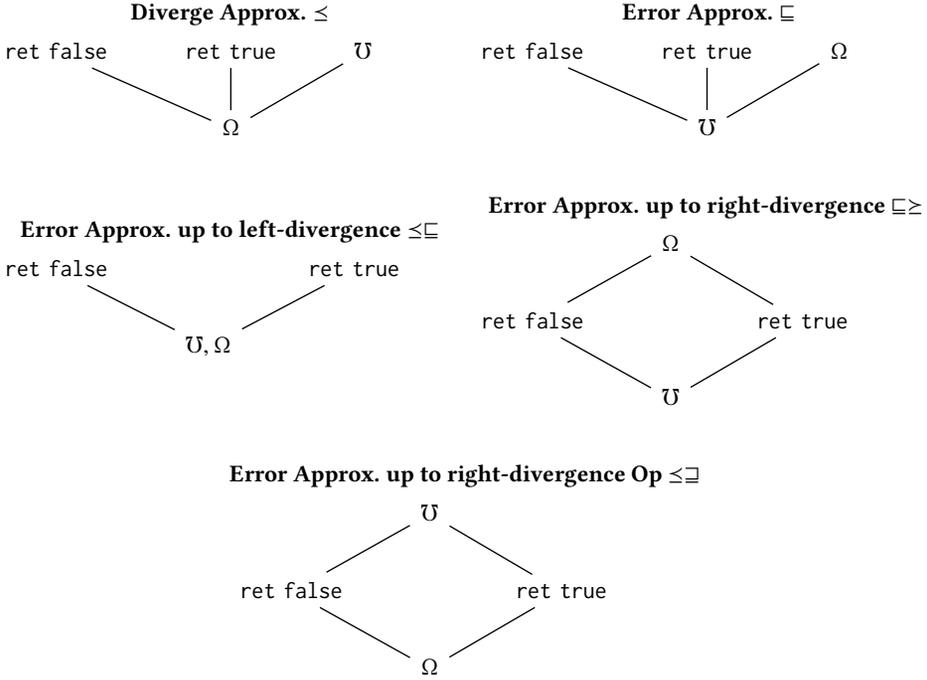
\fi

The goal of this section is to prove that a symmetric equality $E \equidyn
E'$ in CBPV (i.e. $E \ltdyn E'$ and $E' \ltdyn E$) implies contextual
equivalence $E \ctxize= E'$ and that inequality in CBPV $E \ltdyn E'$
implies error approximation $E \ctxize\ltdyn E'$, proving graduality of the operational model\ifshort\else :\fi
\begin{longonly}
\begin{small}
\begin{mathpar}
   \inferrule{\Gamma \pipe \Delta \vdash E \equidyn E' : T}{\Gamma \pipe \Delta \vDash E \ctxize= E' \in T}\and
  \inferrule{\Gamma \pipe \Delta \vdash E \ltdyn E' : T}{\Gamma \pipe \Delta \vDash E \ctxize\ltdyn E' \in T}
\end{mathpar}  
\end{small}
\end{longonly}
Because we have non-well-founded $\mu/\nu$ types, we use a
\emph{step-indexed logical relation} to prove properties about the
contextual lifting of certain preorders $\apreorder$ on results.
In step-indexing, the \emph{infinitary} relation given by
$\ctxize\apreorder$ is related to the set of all of its \emph{finitary
  approximations} $\ix\apreorder i$, which ``time out'' after observing
$i$ steps of evaluation and declare that the
terms \emph{are} related.
\begin{shortonly}
  A preorder $\apreorder$ is only recoverable from its finite
  approximations if $\diverge$ is a \emph{least} element, $\diverge
  \apreorder R$, because a diverging term will cause a time out for
  any finite index. We call a preorder with $\diverge \apreorder R$ a
  \emph{divergence preorder}.~
\end{shortonly}
\begin{longonly}
This means that the original relation is only recoverable from the
finite approximations if $\diverge$ is always related to another
element: if the relation is a preorder, we require that $\diverge$ is
a \emph{least} element.

We call such a preorder a \emph{divergence preorder}.
\begin{definition}[Divergence Preorder]
  A preorder on results $\apreorder$ is a divergence preorder if
  $\diverge \apreorder R$ for all results $R$.
\end{definition}
\end{longonly}
But this presents a problem, because \emph{neither} of our intended
relations ($=$ and $\ltdyn$) is a divergence preorder; rather both have
$\diverge$ as a \emph{maximal} element.
\begin{shortonly}
  For observational equivalence, because contextual equivalence is
  symmetric divergence approximation ($M \ctxize= N$ iff $M
  \ctxize\preceq N$ and $N \ctxize\preceq M$), we can use a step-indexed
  logical relation to characterize $\preceq$, and then obtain results
  about observational equivalence from that~\cite{ahmed06:lr}.
  A similar move works for error
  approximation~\cite{newahmed18}, but since $R \ltdyn R'$ is \emph{not} symmetric, it is decomposed as the conjunction of two
  orderings: error approximation up to divergence on the left
  $\errordivergeleft$ (the preorder where $\err$ and $\diverge$ are both
  minimal: $\err \preceq\ltdyn R$ and $\diverge \preceq\ltdyn R$) and
  error approximation up to divergence on the right $\errordivergeright$
  (the diamond preorder where $\err$ is minimal and $\diverge$ is
  maximal, with $\tru$/$\fls$ in between).  Then $\preceq\ltdyn$ and the
  \emph{opposite} of $\ltdyn\succeq$ (written $\errordivergerightop$)
  are divergence preorders, so we can use a step-indexed logical
  relation to characterize them.  Overall, because $=$ is symmetric
  $\ltdyn$, and $\ltdyn$ is the conjunction of $\errordivergeleft$ and
  $\errordivergeright$, and contextual lifting commutes with conjunction
  and opposites, it will suffice to develop logical relations for
  divergence preorders.
\end{shortonly}

\begin{longonly}
However, there is a standard ``trick'' for subverting this obstacle in
the case of contextual equivalence~\cite{ahmed06:lr}: we notice
that we can define equivalence as the symmetrization of divergence
approximation, i.e., $M \ctxize= N$ if and only if $M \ctxize\preceq
N$ and $N \ctxize\preceq M$, and since $\preceq$ has $\diverge$ as
a least element, we can use a step-indexed relation to prove it.
As shown in \citet{newahmed18}, a similar trick works for error
approximation, but since $\ltdyn$ is \emph{not} an equivalence
relation, we decompose it rather into two \emph{different} orderings:
error approximation up to divergence on the left $\errordivergeleft$ and
error approximation up to divergence on the right $\errordivergeright$,
also shown in figure \ref{fig:result-orders}.
Note that $\errordivergeleft$ is a preorder, but not a poset because
$\err, \diverge$ are order-equivalent but not equal.
Then clearly $\errordivergeleft$ is a divergence preorder and the
\emph{opposite} of $\errordivergeright$, written $\errordivergerightop$
is a divergence preorder.

Then we can completely reduce the problem of proving $\ctxize=$ and
$\ctxize\ltdyn$ results to proving results about divergence preorders
by the following observations.
\newcommand{\ctxsimi}[1]{\mathrel{\sim_{#1}^{\text{ctx}}}}
\begin{lemma}[Decomposing Result Preorders] \label{lem:decomposing-result}
  Let $R, S$ be results.
  \begin{enumerate}
  \item $R = S$ if and only if $R \ltdyn S$ and $S \ltdyn R$.
  \item $R = S$ if and only if $R \preceq S$ and $S \preceq R$.
  \item $R \errordivergeleft S$ iff $R \ltdyn S$ or $R \preceq S$.
  \item $R \errordivergeright S$ iff $R \ltdyn S$ or $R \succeq S$.
  \end{enumerate}
\end{lemma}

In the following, we write $\sim^\circ$ for the opposite of a relation
($x \sim^\circ y$ iff $y \sim x$), $\Rightarrow$ for
containment/implication ($\sim \Rightarrow \sim'$ iff $x \sim y$ implies
$x \sim' y$), $\Leftrightarrow$ for bicontainment/equality, $\vee$ for
union ($x (\sim \vee \sim') y$ iff $x \sim y$ or $x \sim' y$), and
$\wedge$ for intersection ($x (\sim \wedge \sim') y$ iff $x \sim y$ and $x \sim' y$).

\begin{lemma}[Contextual Lift commutes with Conjunction] \label{lem:ctx-commutes-conjunction}
  \[
  \ctxize{(\simsub 1 \wedge \simsub 2)} \Leftrightarrow \ctxize{\simsub 1} \wedge \ctxize{\simsub 2}
  \]
\end{lemma}
 
\begin{lemma}[Contextual Lift commutes with Dualization] \label{lem:ctx-commutes-dual}
  \[
  \ctxize{\sim^\circ} \Leftrightarrow \ctxize{\sim}^\circ
  \]
\end{lemma}

\begin{lemma}[Contextual Decomposition Lemma] \label{lem:contextual-decomposition}
Let $\sim$ be a reflexive relation $(= \Rightarrow \sim)$, and $\leqslant$
be a reflexive, antisymmetric relation (${=} \Rightarrow {\leqslant}$ and
$(\leqslant \wedge {\leqslant^\circ}) \Leftrightarrow {=}$).  Then
\[
\ctxize\sim \Leftrightarrow \ctxize{(\sim \vee \leqslant)} \wedge (\ctxize{(\sim^\circ \vee \leqslant)})^\circ
\]
\end{lemma}

\begin{proof}
Note that despite the notation, $\leqslant$ need not be assumed to be
transitive.  
Reflexive relations form a lattice with $\wedge$ and $\vee$ with $=$ as
$\bot$ and the total relation as $\top$ (e.g. $(= \vee \sim)
\Leftrightarrow \sim$ because $\sim$ is reflexive, and $(= \wedge \sim)
\Leftrightarrow =$).  So we have
\[
\sim \Leftrightarrow (\sim \vee \leqslant) \wedge (\sim \vee \leqslant^\circ)
\]
because FOILing the right-hand side gives
\[
(\sim \wedge \sim) \vee (\leqslant \wedge \sim) \vee (\sim \wedge \leqslant^\circ) \vee (\leqslant \wedge \leqslant^\circ)
\]
By antisymmetry, $(\leqslant \wedge \leqslant^\circ)$ is $=$, which is the
unit of $\vee$, so it cancels.  By idempotence, $(\sim \wedge \sim)$ is $\sim$.
Then by absorption, the whole thing is $\sim$.

Opposite is \emph{not} de Morgan: $(P \vee Q)^\circ = P^\circ \vee
Q^\circ$, and similarly for $\wedge$.  But it is involutive:
$(P^\circ)^\circ \Leftrightarrow P$.  

So using Lemmas~\ref{lem:ctx-commutes-conjunction}, \ref{lem:ctx-commutes-dual} we can calculate as follows:
\[
\begin{array}{rcl}
\ctxize\sim & \Leftrightarrow &\ctxize{((\sim \vee \leqslant) \wedge (\sim \vee \leqslant^\circ))} \\
            & \Leftrightarrow &\ctxize{(\sim \vee \leqslant)} \wedge \ctxize{(\sim \vee \leqslant^\circ)}\\
            & \Leftrightarrow &\ctxize{(\sim \vee \leqslant)} \wedge \ctxize{((\sim \vee \leqslant^\circ)^\circ)^\circ}\\
            & \Leftrightarrow &\ctxize{(\sim \vee \leqslant)} \wedge \ctxize{((\sim^\circ \vee (\leqslant^\circ)^\circ)^\circ)}\\
            & \Leftrightarrow &\ctxize{(\sim \vee \leqslant)} \wedge \ctxize{(\sim^\circ \vee \leqslant)^\circ}\\
            & \Leftrightarrow &\ctxize{(\sim \vee \leqslant)} \wedge \ctxize{(\sim^\circ \vee \leqslant)}^\circ
\end{array}
\]
\end{proof}

As a corollary, the decomposition of contextual equivalence into diverge
approximation in \citet{ahmed06:lr} and the decomposition of dynamism in
\citet{newahmed18} are really the same trick:
\begin{corollary}[Contextual Decomposition] ~~~ \label{cor:contextual-decomposition}
  \begin{enumerate}
  \item $\ctxize= \mathbin{\Leftrightarrow} \ctxize{\preceq} \wedge
    (\ctxize{(\preceq)})^\circ$
  \item $\ctxize= \mathbin{\Leftrightarrow} \ctxize{\ltdyn} \wedge (\ctxize{(\ltdyn)})^\circ$
  \item $\ctxize\ltdyn \mathbin{\Leftrightarrow} \ctxize{\errordivergeleft} \wedge (\ctxize{(\errordivergerightop)})^\circ$
  \end{enumerate}
\end{corollary}
\begin{proof}

For part 1 (though we will not use this below), applying
Lemma~\ref{lem:contextual-decomposition} with $\sim$ taken to be $=$
(which is reflexive) and $\leqslant$ taken to be $\preceq$ (which is
reflexive and antisymmetric) gives that contextual equivalence is
symmetric contextual divergence approximation:
\[
\ctxize= \Leftrightarrow \ctxize{(= \vee \preceq)} \wedge (\ctxize{(=^\circ \vee \preceq)})^\circ
         \Leftrightarrow \ctxize{\preceq} \wedge (\ctxize{(\preceq)})^\circ
\]

For part (2), the same argument with $\sim$ taken to be $=$ and
$\leqslant$ taken to be $\ltdyn$ (which is also antisymmetric) gives that
contextual equivalence is symmetric contextual dynamism:
\[
\ctxize= \Leftrightarrow \ctxize{\ltdyn} \wedge (\ctxize{(\ltdyn)})^\circ
\]

For part (3), applying Lemma~\ref{lem:contextual-decomposition} with $\sim$
taken to be $\ltdyn$ and $\leqslant$ taken to be $\preceq$ gives that
dynamism decomposes as
\[
\ctxize\ltdyn \Leftrightarrow \ctxize{(\ltdyn \vee \preceq)} \wedge (\ctxize{(\ltdyn^\circ \vee \preceq)})^\circ
              \Leftrightarrow \ctxize{\errordivergeleft} \wedge (\ctxize{(\errordivergerightop)})^\circ
\]
Since both ${\errordivergeleft}$ and $\errordivergerightop$ are of the
form $- \vee \preceq$, both are divergence preorders.  Thus, it suffices
to develop logical relations for divergence preorders below.
\end{proof}
\end{longonly}

\subsection{CBPV Step Indexed Logical Relation}
\label{sec:lr}

\begin{shortonly}
We use a logical relation to prove results about $E \ctxize\apreorder
E'$ where $\apreorder$ is a divergence preorder.  The
``finitization'' of a divergence preorder is a relation between
\emph{programs} and \emph{results}: a program approximates a result $R$
at index $i$ if it reduces to $R$ in $< i$ steps or it ``times out'' by reducing at least $i$ times.
\end{shortonly}

\begin{longonly}
Next, we turn to the problem of proving results about $E
\ctxize\apreorder E'$ where $\apreorder$ is a divergence preorder.
Dealing directly with a contextual preorder is practically impossible,
so instead we develop an alternative formulation as a logical relation
that is much easier to use.
Fortunately, we can apply standard logical relations techniques to
provide an alternate definition \emph{inductively} on types.
However, since we have non-well-founded type definitions using
$\mu$ and $\nu$, our logical relation will also be defined inductively on a
\emph{step index} that times out when we've exhausted our step budget.
To bridge the gap between the indexed logical relation and the
divergence preorder we care about, we define the ``finitization'' of a
divergence preorder to be a relation between \emph{programs} and
\emph{results}: the idea is that a program approximates a result $R$
at index $i$ if it reduces to $R$ in less than $i$ steps or it reduces
at least $i$ times.
\end{longonly}

\begin{definition}[Finitized Preorder]
  Given a divergence preorder $\apreorder$, we define the
  \emph{finitization} of $\apreorder$ to be, for each natural number
  $i$, a relation between programs and results
\iflong
  \[ {\ix\apreorder i} \subseteq \{ M \pipe \cdot\vdash M : \u F 2\} \times \text{Results} \]
\fi
  defined by
  \[
  M \ix \apreorder i R = (\exists M'.~ M \bigstepsin{i} M') \vee (\exists (j< i). \exists R_M.~ M \bigstepsin{j} R_M \wedge R_M \apreorder R)
  \]
\end{definition}

\begin{longonly}
Note that in this definition, unlike in the definition of divergence,
we only count non-well-founded steps.
This makes it slightly harder to establish the intended equivalence $M
\ix \apreorder \omega R$ if and only if $\result(M) \apreorder R$, but
makes the logical relation theorem stronger: it proves that diverging
terms must use recursive types of some sort and so any term that does
not use them terminates.
This issue would be alleviated if we had proved type safety by a
logical relation rather than by progress and preservation.

However, the following properties of the indexed relation can easily
be established.
First, a kind of ``transitivity'' of the indexed relation with respect
to the original preorder, which is key to proving transitivity of the
logical relation.
\begin{lemma}[Indexed Relation is a Module of the Preorder]
\label{lem:module}
  If $M \ix\apreorder i R$ and $R \apreorder R'$ then $M \ix\apreorder i R'$
\end{lemma}
\begin{longproof}
  If $M \bigstepsin{i} M'$ then there's nothing to show, otherwise
  $M \bigstepsin{j< i} \result(M)$ so it follows by transitivity of the
  preorder: $\result(M) \apreorder R \apreorder R'$.
\end{longproof}

Then we establish a few basic properties of the finitized preorder.
\begin{lemma}[Downward Closure of Finitized Preorder]
  If $M \ix\apreorder i R$ and $j\leq i$ then $M \ix \apreorder j R$.
\end{lemma}
\begin{longproof} \hfill
  \begin{enumerate}
  \item If $M \bigstepsin{i} M_i$ then $M \bigstepsin{j} M_j$ and otherwise 
  \item If $M \bigstepsin{j \leq k i} \result(M)$ then $M \bigstepsin{j} M_j$
  \item if $M \bigstepsin{k < j \leq i} \result(M)$ then $\result(M) \apreorder R$.
  \end{enumerate}
\end{longproof}
\begin{lemma}[Triviality at $0$]
  For any $\cdot \vdash M : \u F 2$, $M \ix\apreorder 0 R$
\end{lemma}
\begin{longproof}
  Because $M \bigstepsin{0} M$
\end{longproof}
\begin{lemma}[Result (Anti-)reduction]
  If $M \bigstepsin{i} N$ then $\result(M) = \result(N)$.
\end{lemma}
\begin{lemma}[Anti-reduction]
  If $M \ix\apreorder i R$ and $N \bigstepsin{j} M$, then $N \ix\apreorder {{i+j}} R$
\end{lemma}
\begin{longproof}
  \begin{enumerate}
  \item If $M \bigstepsin{i} M'$ then $N \bigstepsin{i+j} M'$
  \item If $M \bigstepsin{k < i} \result(M)$ then $N \bigstepsin{k+j}
    \result(M)$ and $\result(M) = \result(N)$ and $k+j < i+j$.
  \end{enumerate}
\end{longproof}
\end{longonly}

\begin{figure}
\begin{small}
  \begin{mathpar}
\iflong
    {\itylrof\apreorder{i}{A}} \subseteq \{ \cdot \vdash V : A \}^2
    \qquad\qquad\qquad{\itylrof\apreorder{i}{\u B}}\subseteq \{ \cdot \pipe \u B \vdash S
    : \u F (1 + 1) \}^2\\
\fi
    \begin{array}{rcl}
\iflong      
      \cdot \itylrof\apreorder i {\cdot} \cdot &=& \top\\
      \gamma_1,V_1/x \itylrof\apreorder i {\Gamma,x:A} \gamma_2,V_2/x &=& \gamma_1 \itylrof\apreorder i \Gamma \gamma_2 \wedge V_1 \itylrof\apreorder i A V_2\\
\fi
      V_1 \itylr i 0 V_2 &=& \bot\\
\iflong
      \inl V_1 \itylr i {A + A'} \inl V_2 &= & V_1 \itylr i A V_2\\
      \inr V_1 \itylr i {A + A'} \inr V_2 &= & V_1 \itylr i {A'} V_2 \\
      () \itylr i 1 () &=& \top\\
\fi
      (V_1,V_1') \itylr i {A \times A'} (V_2, V_2') &=& V_1 \itylr i A V_2 \wedge V_1' \itylr i {A'} V_2'\\
      \rollty {\mu X. A} V_1 \itylr i {\mu X. A} \rollty {\mu X. A} V_2 &=& i = 0 \vee V_1 \itylr {i-1} {A[\mu X.A/X]} V_2\\
      V_1 \itylr i {U \u B} V_2 &=& \forall j \leq i, S_1 \itylr j {\u B} S_2.~ S_1[\force V_1] \ix\apreorder j \result(S_2[\force V_2]) \\\\

      S_1[\bullet V_1] \itylr i {A \to \u B} S_1[\bullet V_2] & = & V_1 \itylr i A V_2 \wedge S_1 \itylr {i}{\u B} S_2\\
\iflong
      S_1[\pi_1 \bullet] \itylr i {\u B \with \u B'} S_2[\pi_1 \bullet] &=& S_1 \itylr i {\u B} S_2\\
      S_1[\pi_2 \bullet] \itylr i {\u B \with \u B'} S_2[\pi_2 \bullet] &=& S_1 \itylr i {\u B'} S_2\\
      S_1 \itylr i {\top} S_2 &=& \bot\\
\fi
      S_1[\unroll \bullet] \itylr i {\nu \u Y. \u B} S_2[\unroll \bullet] &=& i = 0 \vee S_1 \itylr {i-1} {\u B[\nu \u Y. \u B/\u Y]} S_2\\
      S_1 \itylr i {\u F A} S_2 & = & \forall j\leq i, V_1 \itylr j A V_2.~ S_1[\ret V_1] \ix\apreorder j \result(S_2[\ret V_2])
    \end{array}
  \end{mathpar}
  \end{small}
  \vspace{-0.1in}
  \caption{Logical Relation from a Preorder $\apreorder$ }
  \label{fig:lr}
\end{figure}

\begin{shortonly}
The (closed) \emph{logical} preorder (for closed values/stacks) is in Figure
\ref{fig:lr}.  For every $i$ and value type $A$, we define a relation
$\itylrof \apreorder i A$ between two closed values of type $A$, and for
every $i$ and $\u B$, we define a relation for two ``closed'' stacks $\u
B \vdash \u F 2$ outputting the observation type $\u F 2$---the
definition is by mutual lexicographic induction on $i$ and $A/\u B$.
Two values or stacks are related if they have the same structure, where
for $\mu,\nu$ we decrement $i$ and succeed if $i = 0$.  The shifts $\u
F/U$ take the \emph{orthogonal} of the relation: the set of all
stacks/values that when composed with those values/stacks are related by
$\apreorder^{j \le i}$; the quantifier over $j \leq i$ is needed to make the
relation downward closed.
\end{shortonly}

\begin{longonly}
Next, we define the (closed) \emph{logical} preorder (for closed values/stacks) by induction on types and
the index $i$ in figure \ref{fig:lr}.
Specifically, for every $i$ and value type $A$ we define a relation
$\itylrof \apreorder i A$ between closed values of type $A$ because
these are the only ones that will be pattern-matched against at
runtime.
The relation is defined in a type-directed fashion, the intuition being
that we relate two positive values when they are built up in the same
way: i.e., they have the same introduction form and their subterms are
related.
For $\mu$, this definition would not be well-founded, so we decrement
the step index, giving up and relating the terms if $i = 0$.
Finally $U$ is the only negative value type, and so it is treated
differently.
A thunk $V : U\u B$ cannot be inspected by pattern matching, rather
the only way to interact with it is to force its evaluation.
By the definition of the operational semantics, this only ever occurs
in the step $S[\force V]$, so (ignoring indices for a moment), we
should define $V_1 \apreorder V_2$ to hold in this case when, given
$S_1 \apreorder S_2$, the result of $S_2[\force V_2]$ is approximated
by $S_1[\force V_1]$.
To incorporate the indices, we have to quantify over $j \leq i$ in
this definition because we need to know that the values are related in
all futures, including ones where some other part of the term has been
reduced (consuming some steps).
Technically, this is crucial for making sure the relation is
downward-closed.
This is known as the \emph{orthogonal} of the relation, and one
advantage of the CBPV language is that it makes the use of
orthogonality \emph{explicit} in the type structure, analogous to the
benefits of using Nakano's \emph{later} modality \cite{nakano} for step indexing
(which we ironically do not do).

Next, we define when two \emph{stacks} are related.
First, we define the relation only for two ``closed'' stacks, which 
both have the same type of their hole $\u B$ and both have
``output'' the observation type $\u F 2$.
The reason is that in evaluating a program $M$, steps always occur as
$S[N] \bigstepany S[N']$ where $S$ is a stack of this form.
An intuition is that for negative types, two stacks are related when
they start with the same elimination form and the remainder of the
stacks are related.
For $\nu$, we handle the step indices in the same way as for $\mu$.
For $\u F A$, a stack $S[\bullet : \u F A]$ is strict in its input and
waits for its input to evaluate down to a value $\ret V$, so two
stacks with $\u F A$ holes are related when in any future world, they
produce related behavior when given related values.

We note that in the CBV restriction of CBPV, the function type is
given by $U(A \to \u F A')$ and the logical relation we have presented
reconstructs the usual definition that involves a double orthogonal.

Note that the definition is well-founded using the lexicographic
ordering on $(i, A)$ and $(i, \u B)$: either the type reduces and the
index stays the same or the index reduces.
We extend the definition to contexts to \emph{closing substitutions}
pointwise: two closing substitutions for $\Gamma$ are related at $i$
if they are related at $i$ for each $x:A \in \Gamma$.
\end{longonly}

The logical preorder for open terms is defined as usual by quantifying
over all related closing substitutions, but also over all stacks to the
observation type $\u F (1+1)$:
%
\begin{definition}[Logical Preorder]
  For a divergence preorder $\apreorder$, its step-indexed logical
  preorder is
  \begin{shortonly}
    for terms (open stack, value cases are defined in the extended version):
    $\Gamma \vDash M_1 \ilrof\apreorder{i} M_2 \in \u B$ iff for every $\gamma_1 \itylrof\apreorder i {\Gamma} \gamma_2$ and $S_1
    \itylrof\apreorder i {\u B} S_2$, $S_1[M_1[\gamma_1]] \ix\apreorder
    i \result(S_2[M_2[\gamma_2]])$.
  \end{shortonly}
  \begin{longonly}
  \begin{enumerate}
  \item $\Gamma \vDash M_1 \ilrof\apreorder{i} M_2 \in \u B$ iff for every $\gamma_1 \itylrof\apreorder i {\Gamma} \gamma_2$ and $S_1
    \itylrof\apreorder i {\u B} S_2$, $S_1[M_1[\gamma_1]] \ix\apreorder
    i \result(S_2[M_2[\gamma_2]])$.
  \item $\Gamma \vDash V_1 \ilrof\apreorder{i} V_2 \in A$ iff
    for every $\gamma_1 \itylrof\apreorder i {\Gamma} \gamma_2$, $V_1[\gamma_1] \itylrof\apreorder i A V_2[\gamma_2]$
  \item $\Gamma \pipe \u B \vDash S_1 \ilrof\apreorder{i} S_2 \in \u B'$ 
    iff for every $\gamma_1 \itylrof\apreorder i {\Gamma} \gamma_2$ and
    $S_1' \itylrof\apreorder i {\u B'} S_2'$, $S_1'[S_1[\gamma_1]] \itylrof \apreorder
    i {\u B} S_2'[S_2[\gamma_2]])$.
  \end{enumerate}    
  \end{longonly}
\end{definition}

\begin{longonly}
We next want to prove that the logical preorder is a congruence
relation, i.e., the fundamental lemma of the logical relation.
This requires the easy lemma, that the relation on closed terms and
stacks is downward closed.
\begin{lemma}[Logical Relation Downward Closure]
  For any type $T$, if $j \leq i$ then $\itylrof\apreorder i T
  \subseteq \itylrof\apreorder j T$
\end{lemma}
\end{longonly}

Next, we show the fundamental theorem:
\begin{theorem}[Logical Preorder is a Congruence]
  For any divergence preorder, the logical preorder $E \ilrof\apreorder
  i E'$ is \iflong a congruence relation, i.e., it is \fi closed under
  applying any value/term/stack constructors to both sides.
\end{theorem}
\begin{longproof}
  For each congruence rule
  \[
  \inferrule
  {\Gamma \pipe \Delta \vdash E_1 \ltdyn E_1' : T_1 \cdots}
  {\Gamma' \pipe \Delta' \vdash E_c \ltdyn E_c' : T_c}
  \]
  we prove for every $i \in \mathbb{N}$ the validity of the rule
  \[
  \inferrule
  {\Gamma \pipe \Delta \vDash E_1 \ilr i E_1' \in T_1\cdots }
  {\Gamma \pipe \Delta \vDash E_c \ilr i E_c' \in T_c}
  \]
  \begin{enumerate}
  \item $\inferrule {} {\Gamma,x : A,\Gamma' \vDash x \ilr i x \in
    A}$. Given $\gamma_1 \itylr i {\Gamma,x:A,\Gamma'} \gamma_2$,
    then by definition $\gamma_1(x) \itylr i A \gamma_2(x)$.

  \item $\inferrule{}{\Gamma \vDash \err \itylr \err \in \u B}$ We
    need to show $S_1[\err] \ix\apreorder i \result(S_2[\err])$. By
    anti-reduction and strictness of stacks, it is sufficient to show
    $\err \ilr i \err$. If $i = 0$ there is nothing to show,
    otherwise, it follows by reflexivity of $\apreorder$.

  \item $\inferrule
    {\Gamma \vDash V \ilr i V' \in A \and
      \Gamma, x : A \vDash M \ilr i M' \in \u B
    }
    {\Gamma \vDash \letXbeYinZ V x M \ilr i \letXbeYinZ {V'} {x} {M'} \in \u B}$
    
    Each side takes a $0$-cost step, so by anti-reduction, this reduces to
    \[ S_1[M[\gamma_1,V/x]] \ix\apreorder i \result(S_2[M'[\gamma_2,V'/x]]) \] which follows by the assumption $\Gamma, x : A \vDash M \ilr i M' \in \u B$

  \item $\inferrule
    {\Gamma \vDash V \ilr i V' \in 0}
    {\Gamma \vDash \abort V \ilr i \abort V' \in \u B}$.
    By assumption, we get $V[\gamma_1] \logty i {0} V'[\gamma_2]$, but this is a contradiction.

  \item $\inferrule
    {\Gamma \vDash V \ilr i V' \in A_1}
    {\Gamma \vDash \inl V \ilr i \inl V' \in A_1 + A_2}$.
    Direct from assumption, rule for sums.

  \item $\inferrule
    {\Gamma \vDash V \ilr i V' \in A_2}
    {\Gamma \vDash \inr V \ilr i \inr V' \in A_1 + A_2}$
    Direct from assumption, rule for sums.

  \item $\inferrule
    {\Gamma \vDash V \ilr i V' \in A_1 + A_2\and
      \Gamma, x_1 : A_1 \vDash M_1 \ilr i M_1' \in \u B\and
      \Gamma, x_2 : A_2 \vDash M_2 \ilr i M_2' \in \u B
    }
    {\Gamma \vDash \caseofXthenYelseZ V {x_1. M_1}{x_2.M_2} \ilr i \caseofXthenYelseZ {V'} {x_1. M_1'}{x_2.M_2'} \in \u B}$\\
    By case analysis of $V[\gamma_1] \ilr i V'[\gamma_2]$.
    \begin{enumerate}
    \item If $V[\gamma_1]=\inl V_1, V'[\gamma_2] = \inl V_1'$ with
      $V_1 \itylr i {A_1} V_1'$, then taking $0$ steps, by anti-reduction
      the problem reduces to
      \[ S_1[M_1[\gamma_1,V_1/x_1]] \ix\apreorder i \result(S_1[M_1[\gamma_1,V_1/x_1]]) \]
      which follows by assumption.
    \item For $\inr{}$, the same argument.
    \end{enumerate}

  \item $\inferrule
    {}
    {\Gamma \vDash () \ilr i () \in 1}$ Immediate by unit rule.

  \item $\inferrule
    {\Gamma \vDash V_1 \ilr i V_1' \in A_1\and
      \Gamma\vDash V_2 \ilr i V_2' \in A_2}
    {\Gamma \vDash (V_1,V_2) \ilr i (V_1',V_2') \in A_1 \times A_2}$
    Immediate by pair rule.

  \item $\inferrule
    {\Gamma \vDash V \ilr i V' \in A_1 \times A_2\and
      \Gamma, x : A_1,y : A_2 \vDash M \ilr i M' \in \u B
    }
    {\Gamma \vDash \pmpairWtoXYinZ V x y M \ilr i \pmpairWtoXYinZ {V'} {x} {y} {M'} \in \u B}$
    By $V \itylr i {A_1 \times A_2} V'$, we know $V[\gamma_1] =
    (V_1,V_2)$ and $V'[\gamma_2] = (V_1', V_2')$ with $V_1 \itylr i
    {A_1} V_1'$ and $V_2 \itylr i {A_2} V_2'$.
    Then by anti-reduction, the problem reduces to
    \[ S_1[M[\gamma_1,V_1/x,V_2/y]] \ix\apreorder i \result(S_1[M'[\gamma_1,V_1'/x,V_2'/y]]) \]
    which follows by assumption.

  \item $\inferrule
    {\Gamma \vDash V \ilr i V' \in A[\mu X.A/X]}
    {\Gamma \vDash \rollty{\mu X.A} V \ilr i \rollty{\mu X.A} V' \in \mu X.A }$
    If $i = 0$, we're done. Otherwise $i=j+1$, and our assumption is
    that $V[\gamma_1] \itylr {j+1} {A[\mu X.A/X]} V'[\gamma_2]$ and we need to show
    that $\roll V[\gamma_1] \itylr {j+1} {\mu X. A}\roll
    V'[\gamma_2]$. By definition, we need to show $V[\gamma_1] \itylr
    j {A[\mu X.A/X]} V'[\gamma_2]$, which follows by downward-closure.
    
  \item $\inferrule
    {\Gamma \vDash V \ilr i V' \in \mu X. A\and
      \Gamma, x : A[\mu X. A/X] \vDash M \ilr i M' \in \u B}
    {\Gamma \vDash \pmmuXtoYinZ V x M \ilr i \pmmuXtoYinZ {V'} {x} {M'} \in \u B}$
    If $i = 0$, then by triviality at $0$, we're done.
    Otherwise, $V[\gamma_1] \itylr {j+1} {\mu X. A} V'[\gamma_2]$ so
    $V[\gamma_1] = \roll V_\mu, V'[\gamma_2] = \roll V_\mu'$ with
    $V_\mu \itylr j {A[\mu X.A/X]} V_\mu'$. Then each side takes $1$ step, so by anti-reduction it is sufficient to show
    \[ S_1[M[\gamma_1,V_\mu/x]] \ix\apreorder j \result(S_2[M'[\gamma_2,V_\mu'/x]]) \] which follows by assumption and downward closure of the stack, value relations.

  \item $\inferrule {\Gamma \vDash M \ilr i M' \in \u B} {\Gamma
    \vDash \thunk M \ilr i \thunk M' \in U \u B}$.  We need to show
    $\thunk M[\gamma_1] \itylr i {U \u B} \thunk M'[\gamma_2]$, so let
    $S_1 \itylr j {\u B} S_2$ for some $j \leq i$, and we need to show
    \[ S_1[\force \thunk M_1[\gamma_1]] \ix\apreorder j \result(S_2[\force \thunk M_2[\gamma_2]]) \]
    Then each side reduces in a $0$-cost step and it is sufficient to show
    \[ S_1[M_1[\gamma_1]] \ix\apreorder j \result(S_2[M_2[\gamma_2]]) \]
    Which follows by downward-closure for terms and substitutions.

  \item $\inferrule {\Gamma \vDash V \ilr i V' \in U \u B} {\Gamma
    \vDash \force V \ilr i \force V' \in \u B}$. \\ We need to show
    $S_1[\force V[\gamma_1]] \ix\apreorder i \result(S_2[\force
    V'[\gamma_2]])$, which follows by the definition of $V[\gamma_1]
    \itylr i {U \u B} V'[\gamma_2]$.

  \item $\inferrule
    {\Gamma \vDash V \ilr i V' \in A}
    {\Gamma \vDash \ret V \ilr i \ret V' \in \u F A}$\\
    We need to show $S_1[\ret V[\gamma_1]] \ix\apreorder i \result(S_2[\ret
      V'[\gamma_2]])$, which follows by the orthogonality definition of
    $S_1 \itylr i {\u F A} S_2$.

  \item $\inferrule
    {\Gamma \vDash M \ilr i M' \in \u F A\and
      \Gamma, x: A \vDash N \ilr i N' \in \u B}
    {\Gamma \vDash \bindXtoYinZ M x N \ilr i \bindXtoYinZ {M'} {x} {N'} \in \u B}$.

    We need to show $\bindXtoYinZ {M[\gamma_1]} x {N[\gamma_2]} \ix\apreorder i \result(\bindXtoYinZ {M'[\gamma_2]} {x} {N'[\gamma_2]})$.
    By $M \ilr i M' \in \u F A$, it is sufficient to show that
    \[ \bindXtoYinZ \bullet x {N[\gamma_1]} \itylr i {\u F A} \bindXtoYinZ \bullet {x} {N'[\gamma_2]}\]
    So let $j \leq i$ and $V \itylr j A V'$, then we need to show
    \[ \bindXtoYinZ {\ret V} x {N[\gamma_1]} \itylr j {\u F A} \bindXtoYinZ {\ret V'} {x} {N'[\gamma_2]} \]
    By anti-reduction, it is sufficient to show
    \[ N[\gamma_1,V/x] \ix\apreorder j \result(N'[\gamma_2,V'/x]) \]
    which follows by anti-reduction for $\gamma_1 \itylr i {\Gamma} \gamma_2$ and $N \ilr i N'$.

  \item $\inferrule
    {\Gamma, x: A \vDash M \ilr i M' \in \u B}
    {\Gamma \vDash \lambda x : A . M \ilr i \lambda x:A. M' \in A \to \u B}$
    We need to show
    \[S_1[\lambda x:A. M[\gamma_1]] \ix\apreorder i \result(S_2[\lambda x:A.M'[\gamma_2]]).\]
    By $S_1 \itylr i {A \to \u B} S_2$, we know $S_1 = S_1'[\bullet V_1]$, $S_2 = S_2'[\bullet V_2]$ with $S_1' \itylr i {\u B} S_2'$ and $V_1 \itylr i {A} V_2$.
    Then by anti-reduction it is sufficient to show
    \[
    S_1'[M[\gamma_1,V_1/x]] \ix\apreorder i \result(S_2'[M'[\gamma_2,V_2/x]])
    \]
    which follows by $M \ilr i M'$.

  \item $\inferrule
    {\Gamma \vDash M \ilr i M' \in A \to \u B\and
      \Gamma \vDash V \ilr i V' \in A}
    {\Gamma \vDash M\,V \ilr i M'\,V' \in \u B }$
    We need to show
    \[S_1[M[\gamma_1]\,V[\gamma_1]] \ix\apreorder i \result(S_2[M'[\gamma_2]\,V'[\gamma_2]])\] so by $M \ilr i M'$ it is sufficient to show $S_1[\bullet V[\gamma_1]] \itylr i {A \to \u B} S_2[\bullet V'[\gamma_2]]$ which follows by definition and assumption that $V \ilr i V'$.

  \item $\inferrule{}{\Gamma \vdash \{\} : \top}$ We assume we are
    given $S_1 \itylr i {\top} S_2$, but this is a contradiction.
    
  \item $\inferrule
    {\Gamma \vDash M_1 \ilr i M_1' \in \u B_1\and
      \Gamma \vDash M_2 \ilr i M_2' \in \u B_2}
    {\Gamma \vDash \pair {M_1} {M_2} \ilr i \pair {M_1'} {M_2'} \in \u B_1 \with \u B_2}$
    We need to show
    \[S_1[\pair{M_1[\gamma_1]}{M_2[\gamma_1]}] \ix\apreorder i \result(S_2[\pair{M_1'[\gamma_1]}{M_2'[\gamma_2]}]).\]
    We proceed by case analysis of $S_1 \itylr i {\u B_1 \with \u B_2} S_2$
    \begin{enumerate}
    \item In the first possibility $S_1 = S_{1}'[\pi \bullet], S_2 =
      S_2'[\pi \bullet]$ and $S_1' \itylr i {\u B_1} S_2'$.
      Then by anti-reduction, it is sufficient to show
      \[ S_1'[M_1[\gamma_1]] \ix\apreorder i \result(S_2'[M_1'[\gamma_2]]) \]
      which follows by $M_1 \ilr i M_1'$.
    \item Same as previous case.
    \end{enumerate}

  \item $\inferrule
    {\Gamma \vDash M \ilr i M' \in \u B_1 \with \u B_2}
    {\Gamma \vDash \pi M \ilr i \pi M' \in \u B_1}$
    We need to show $S_1[\pi M[\gamma_1]] \ix\apreorder i \result(S_2[\pi
      M'[\gamma_2]])$, which follows by $S_1[\pi \bullet] \itylr i {\u
      B_1 \with \u B_2} S_2[\pi \bullet]$ and $M \ilr i M'$.

  \item $\inferrule {\Gamma \vDash M \ilr i M' \in \u B_1 \with \u
    B_2} {\Gamma \vDash \pi' M \ilr i \pi' M' \in \u B_2}$ Similar
    to previous case.

  \item $\inferrule
    {\Gamma \vDash M \ilr i M' \in \u B[{\nu \u Y. \u B}/\u Y]}
    {\Gamma \vDash \rollty{\nu \u Y. \u B} M \ilr i \rollty{\nu \u Y. \u B} M' \in {\nu \u Y. \u B}}$
    We need to show that
    \[ S_1[ \rollty{\nu \u Y. \u B} M[\gamma_1]]
    \ix\apreorder i \result(S_2[ \rollty{\nu \u Y. \u B} M'[\gamma_2]]) \]
    If $i = 0$, we invoke triviality at $0$.
    Otherwise, $i = j + 1$ and we know by $S_1 \itylr {j+1} {\nu \u Y. \u B} S_2$ that
    $S_1 = S_1'[\unroll \bullet]$ and $S_2 = S_2'[\unroll \bullet]$ with $S_1' \itylr j {\u B[{\nu \u Y. \u B}/\u Y]} S_2'$, so by anti-reduction it is sufficient to show
    \[ S_1'[ M[\gamma_1]] \ix\apreorder i \result(S_2'[ M'[\gamma_2]]) \]
    which follows by $M \ilr i M'$ and downward-closure.

  \item $\inferrule
    {\Gamma \vDash M \ilr i M' \in {\nu \u Y. \u B}}
    {\Gamma \vDash \unroll M \ilr i \unroll M' \in \u B[{\nu \u Y. \u B}/\u Y]}$
    We need to show
    \[S_1[\unroll M] \ix\apreorder i \result(S_2[\unroll M']),\] which
    follows because $S_1[\unroll \bullet] \itylr i {\nu \u Y. \u B}
    S_2[\unroll \bullet]$ and $M \ilr i M'$.
  \end{enumerate}
\end{longproof}

\begin{longonly}
As a direct consequence we get the reflexivity of the relation
\begin{corollary}[Reflexivity]
  For any $\Gamma \vdash M : \u B$, and $i \in \mathbb{N}$,
  \(\Gamma \vDash M \ilrof\apreorder i  M \in \u B.\)
\end{corollary}
\end{longonly}

\begin{shortonly}
  This in particular implies that the relation is reflexive ($\Gamma
  \vDash M \ilrof\apreorder i M \in \u B$ for all well-typed $M$),
\end{shortonly}
so we
have the following \emph{strengthening} of the progress-and-preservation
type soundness theorem: because $\ix\apreorder i$ only counts unrolling
steps, terms that never use $\mu$ or $\nu$ types (for example) are
guaranteed to terminate.
\begin{corollary}[Unary LR]
  For every program $\cdot \vdash M : \u F 2$ and $i \in \mathbb{N}$,
  $M \ix\apreorder i \result(M)$
\end{corollary}
\begin{longproof}
  By reflexivity, $\cdot \vDash M \ix\apreorder i M \in \u F 2$ and by
  definition $\bullet \itylrof\apreorder i {\u F 2} \bullet$, so
  unrolling definitions we get $M \ix\apreorder i \result(M)$.  
\end{longproof}

\noindent Using reflexivity, we prove that the indexed relation between terms and
results recovers the original preorder in the limit as $i \to \omega$.
We write $\ix\apreorder \omega$ to mean the relation holds for every
$i$, i.e., $\ix\apreorder\omega =
\bigcap_{i\in\mathbb{N}} \ix\apreorder i$.
\begin{corollary}[Limit Lemma]
  \label{lem:limit}
  For any divergence preorder $\apreorder$, \( \result(M) \apreorder
  R\) iff \( M \ix\apreorder \omega R \).
\end{corollary}
\begin{longproof}
  Two cases
  \begin{enumerate}
  \item If $\result(M) \apreorder R$ then we need to show for every $i
    \in \mathbb{N}$, $M \ix \apreorder i R$. By the unary model lemma,
    $M \ix\apreorder i \result(M)$, so the result follows by the
    module lemma \ref{lem:module}.
  \item If $M \ix\apreorder i R$ for every $i$, then there are two
    possibilities: $M$ is always related to $R$ because it takes $i$
    steps, or at some point $M$ terminates.
    \begin{enumerate}
    \item If $M \bigstepsin{i} M_i$ for every $i \in \mathbb{N}$, then
      $\result(M) = \diverge$, so $\result(M) \apreorder R$ because
      $\apreorder$ is a divergence preorder.
    \item Otherwise there exists some $i \in \mathbb{M}$ such that $M
      \bigstepsin{i} \result(M)$, so it follows by the module lemma
      \ref{lem:module}.
    \end{enumerate}
  \end{enumerate}
\end{longproof}

\begin{corollary}[Logical implies Contextual] \label{lem:logical-implies-contextual}
  If $\Gamma \vDash E \ilrof\apreorder \omega E' \in \u B$
  then
  $\Gamma \vDash E \ctxize\apreorder E' \in \u B$.
\end{corollary}
\begin{proof}
  Let $C$ be a closing context. By congruence, $C[M] \ilrof\apreorder
  \omega C[N]$, so using empty environment and stack, $C[M]
  \ix\apreorder\omega \result(C[N])$ and by the limit lemma, we have
  $\result(C[M]) \apreorder \result(C[N])$.
\end{proof}

\begin{longonly}
In fact, we can prove the converse, that at least for the term case,
the logical preorder is \emph{complete} with respect to the contextual
preorder, though we don't use it.
\begin{lemma}[Contextual implies Logical]
  For any $\apreorder$, if $\Gamma \vDash M \ctxize \apreorder N \in
  \u B$, then $\Gamma \vDash M \ilrof\apreorder \omega N \in \u B$.
\end{lemma}
\begin{longproof}
  Let $S_1 \itylr i {\u B} S_2$ and $\gamma_1 \itylr i \Gamma \gamma_2$. We need to show that
  \[
  S_1[M[\gamma_1]] \ix\apreorder i \result(S_2[N[\gamma_2]])
  \]

  So we need to construct a \emph{context} that when $M$ or $N$ is
  plugged into the hole will reduce to the above.

  To do this, first, we deconstruct the context
  $x_1:A_1,\ldots,x_n:A_n = \Gamma$.  Then we define $\cdot \vdash M'
  : A_1\to \cdots \to A_n \to \u B$ as
  \[ \lambda x_1:A_1.\ldots\lambda x_n:A_n. M \]
  And similarly define $N'$. Then clearly
  \[ S[M' \,V_1\, \cdots V_n] \bigstepsin{0} S[M[V_1/x_1,\ldots,V_n/x_n]] \]
  so in particular
  \[ S[M'\,\gamma(x_1)\cdots\gamma(x_n)] \bigstepsin{0} S[M[\gamma]]\]
  and similarly for $N'$ if $x_1,\ldots,x_n$ are all of the variables
  in $\gamma$.

  Then the proof proceeds by the following transitivity chain:
  \begin{align*}
    S_1[M[\gamma_1]] &\ix\apreorder i \result(S_2[M[\gamma_2]])\tag{$M \ilr i M$}\\
    &=\result(S_2[M'\,\gamma_2(x_1)\,\cdots\,\gamma_2(x_n)])\tag{reduction}\\
    &\apreorder \result(S_2[N'\,\gamma_2(x_1)\,\cdots\,\gamma_2(x_n)])\tag{$M \ctxize\apreorder N$}\\
    &= \result(S_2[N[\gamma_2]])\tag{reduction}
  \end{align*}

  So $S_1[M[\gamma_1]] \ix\apreorder i \result(S_2[N[\gamma_2]])$ by
  the module lemma \ref{lem:module}.
\end{longproof}
\end{longonly}

This establishes that our logical relation can prove graduality, so it
only remains to show that our \emph{inequational theory} implies our
logical relation.
Having already validated the congruence rules and reflexivity, we
validate the remaining rules of transitivity, error, substitution, and
$\beta\eta$ for each type constructor.
Other than the $\err \ltdyn M$ rule, all of these hold for any
divergence preorder.

For transitivity, with the unary model and limiting lemmas in hand, we
can prove that all of our logical relations (open and closed) are
transitive in the limit. To do this, we first prove the following kind
of ``quantitative'' transitivity lemma, and then transitivity in the
limit is a consequence.
\begin{lemma}[Logical Relation is Quantitatively Transitive] \hfill
  
  \iflong
  \begin{enumerate}
  \item
  \fi
    If $V_1 \itylr i A V_2$ and $V_2 \itylr
    \omega A V_3$, then $V_1 \itylr i A V_3$
  \iflong
  \item If $S_1 \itylr i {\u B} S_2$ and $S_2 \itylr
    \omega {\u B} S_3$, then $S_1 \itylr i {\u B} S_3$
  \end{enumerate}
    \fi
\end{lemma}
\begin{longproof}
  Proof is by mutual lexicographic induction on the pair $(i, A)$ or
  $(i, \u B)$. All cases are straightforward uses of the inductive
  hypotheses except the shifts $U, \u F$.
  \begin{enumerate}
  \item If $V_1 \itylr i {U \u B} V_2$ and $V_2
    \itylr \omega {U \u B} V_3$, then we need to show that
    for any $S_1 \itylr j {\u B} S_2$ with $j \leq i$,
    \[ S_1[\force V_1] \ix\apreorder j \result(S_2[\force V_3]) \]
    By reflexivity, we know $S_2 \itylr \omega {\u B} S_2$, so by assumption
    \[ S_2[\force V_2] \ix\apreorder \omega \result(S_2[\force V_3])\]
    which by the limiting lemma \ref{lem:limit} is equivalent to
    \[ \result(S_2[\force V_2]) \apreorder \result(S_2[\force V_3]) \]
    so then by the module lemma \ref{lem:module}, it is sufficient to show
    \[ S_1[\force V_1] \ix\apreorder j \result(S_2[\force V_2]) \]
    which holds by assumption.
  \item If $S_1 \itylr i {\u F A} S_2$ and $S_2 \itylr \omega {\u F A}
    S_3$, then we need to show that for any $V_1 \itylr A j V_2$ with $j \leq i$ that
    \[ S_1[\ret V_1] \ix\apreorder j \result(S_3[\ret V_2])\]
    First by reflexivity, we know $V_2 \itylr \omega A V_2$, so by assumption,
    \[ S_2[\ret V_2] \ix\apreorder \omega \result(S_3[\ret V_2]) \]
    Which by the limit lemma \ref{lem:limit} is equivalent to
    \[ \result(S_2[\ret V_2]) \ix\apreorder \omega \result(S_3[\ret V_2]) \]
    So by the module lemma \ref{lem:module}, it is sufficient to show
    \[ S_1[\ret V_1] \ix\apreorder j \result(S_2[\ret V_2]) \]
    which holds by assumption.
  \end{enumerate}
\end{longproof}
\iflong
\begin{lemma}[Logical Relation is Quantitatively Transitive (Open Terms)]\hfill
  \begin{enumerate}
  \item If $\gamma_1 \itylr i \Gamma \gamma_2$ and $\gamma_2 \itylr
    \omega \Gamma \gamma_3$, then $\gamma_1 \itylr i \Gamma \gamma_3$
  \item If $\Gamma \vDash M_1 \ilr i M_2 \in \u B$ and
    $\Gamma \vDash M_2 \ilr \omega M_3 \in \u B$, then
    $\Gamma \vDash M_1 \ilr i M_3 \in \u B$.
  \item If $\Gamma \vDash V_1 \ilr i V_2 \in A$ and
    $\Gamma \vDash V_2 \ilr \omega V_3 \in A$, then
    $\Gamma \vDash V_1 \ilr i V_3 \in A$.
  \item If $\Gamma \pipe \bullet : \u B \vDash S_1 \ilr i S_2 \in \u B'$ and
    $\Gamma\pipe \bullet : \u B \vDash S_2 \ilr \omega S_3 \in \u B'$, then
    $\Gamma\pipe \bullet : \u B \vDash S_1 \ilr i S_3 \in \u B'$.
  \end{enumerate}
\end{lemma}
\begin{longproof}
  \begin{enumerate}
  \item By induction on the length of the context, follows from closed value case.
  \item Assume $\gamma_1 \itylr i \Gamma \gamma_2$ and $S_1 \itylr i {\u B} S_2$.
    We need to show
    \[ S_1[M_1[\gamma_1]] \ix\apreorder{i} \result(S_2[M_3[\gamma_2]]) \]
    by reflexivity and assumption, we know
    \[ S_2[M_2[\gamma_2]] \ix\apreorder \omega \result(S_2[M_3[\gamma_2]])\]
    and by limit lemma \ref{lem:limit}, this is equivalent to
    \[ \result(S_2[M_2[\gamma_2]]) \apreorder \result(S_2[M_3[\gamma_2]])\]
    so by the module lemma \ref{lem:module} it is sufficient to show
    \[ S_1[M_1[\gamma_1]] \ix\apreorder{i} \result(S_2[M_2[\gamma_2]]) \]
    which follows by assumption.
  \item Assume $\gamma_1 \itylr i \Gamma \gamma_2$.  Then
    $V_1[\gamma_1] \itylr i A V_2[\gamma_2]$ and by reflexivity
    $\gamma_2 \itylr \omega \Gamma \gamma_2$ so $V_2[\gamma_2] \itylr
    \omega A V_3[\gamma_2]$ so the result holds by the closed case.
  \item Stack case is essentially the same as the value case.
  \end{enumerate}
\end{longproof}
\fi
\begin{corollary}[Logical Relation is Transitive in the Limit]
  \begin{shortonly}
    $\ilrof\apreorder \omega$ is transitive.
  \end{shortonly}
  \begin{longonly}
    \hfill
  \begin{enumerate}
  \item If $\Gamma \vDash M_1 \ilrof\apreorder \omega M_2 \in \u B$ and
    $\Gamma \vDash M_2 \ilrof\apreorder \omega M_3 \in \u B$, then
    $\Gamma \vDash M_1 \ilrof\apreorder \omega M_3 \in \u B$.
  \item If $\Gamma \vDash V_1 \ilrof\apreorder \omega V_2 \in A$ and
    $\Gamma \vDash V_2 \ilrof\apreorder \omega V_3 \in A$, then
    $\Gamma \vDash V_1 \ilrof\apreorder \omega V_3 \in A$.
  \item If $\Gamma \pipe \bullet : \u B \vDash S_1 \ilrof\apreorder \omega S_2 \in \u B'$ and
    $\Gamma\pipe \bullet : \u B \vDash S_2 \ilrof\apreorder \omega S_3 \in \u B'$, then
    $\Gamma\pipe \bullet : \u B \vDash S_1 \ilrof\apreorder \omega S_3 \in \u B'$.
  \end{enumerate}
  \end{longonly}
\end{corollary}

\iflong
Next, we verify the $\beta, \eta$ equivalences hold as orderings each
way.
\begin{lemma}[$\beta, \eta$]
  For any divergence preorder, the $\beta, \eta$
  laws are valid for $\ilrof\apreorder \omega$
\end{lemma}
\begin{longproof}
  The $\beta$ rules for all cases except recursive types are direct
  from anti-reduction.
  \begin{enumerate}
  \item $\mu X.A-\beta$:
    \begin{enumerate}
    \item We need to show
      \[ S_1[\pmmuXtoYinZ {\rollty{\mu X.A} V[\gamma_1]} x M[\gamma_1]] \ilr i \result(S_2[M[\gamma_2,V[\gamma_2]/x]]) \]
      The left side takes $1$ step to $S_1[M[\gamma_1,V[\gamma_1]/x]]$ and we know
      \[ S_1[M[\gamma_1,V[\gamma_1]/x]] \ilr i \result (S_2[M[\gamma_2,V[\gamma_2]/x]])  \]
      by assumption and reflexivity, so by anti-reduction we have
      \[ S_1[\pmmuXtoYinZ {\rollty{\mu X.A} V[\gamma_1]} x M[\gamma_1]] \ilr {i+1} \result(S_2[M[\gamma_2,V[\gamma_2]/x]]) \]
      so the result follows by downward-closure.
      
    \item For the other direction we need to show
      \[ S_1[M[\gamma_1,V[\gamma_1]/x]] \ilr i \result(S_2[\pmmuXtoYinZ {\rollty{\mu X.A} V[\gamma_2]} x M[\gamma_2]]) \]
      Since results are invariant under steps, this is the same as
      \[ S_1[M[\gamma_1,V[\gamma_1]/x]] \ilr i \result(S_2[M[\gamma_2,V[\gamma_2/x]]]) \]
      which follows by reflexivity and assumptions about the stacks
      and substitutions.
    \end{enumerate}
  \item $\mu X.A-\eta$:
    \begin{enumerate}
    \item We need to show for any $\Gamma, x : \mu X. A \vdash M : \u B$,
      and appropriate substitutions and stacks,
      \[ S_1[\pmmuXtoYinZ {\rollty{\mu X.A} {\gamma_1(x)}} {y}  M[\rollty{\mu X.A}y/x][\gamma_1]] \ilr i \result(S_2[M[\gamma_2]]) \]
      By assumption, $\gamma_1(x) \itylr i {\mu X.A} \gamma_2(x)$, so we know
      \[ \gamma_1(x) = \rollty{\mu X.A} V_1 \]
      and
      \[ \gamma_2(x) = \rollty{\mu X.A} V_2 \]
      so the left side takes a step:
      \begin{align*}
        S_1[\pmmuXtoYinZ {\roll {\gamma_1(x)}} {y}  M[\roll y/x][\gamma_1]]
        &\bigstepsin{1} S_1[M[\roll y/x][\gamma_1][V_1/y]]\\
        &= S_1[M[\roll V_1/x][\gamma_1]]\\
        & = S_1[M[\gamma_1]]
      \end{align*}
      and by reflexivity and assumptions we know
      \[ S_1[M[\gamma_1]] \ilr {i} \result(S_2[M[\gamma_2]]) \]
      so by anti-reduction we know 
      \[ S_1[\pmmuXtoYinZ {\rollty{\mu X.A} {\gamma_1(x)}} {y}  M[\rollty{\mu X.A}y/x][\gamma_1]] \ilr {i+1} \result(S_2[M[\gamma_2]]) \]
      so the result follows by downward closure.
    \item Similarly, to show
      \[ S_1[M[\gamma_1]] \ilr i \result(S_2[\pmmuXtoYinZ {\rollty{\mu X.A} {\gamma_2(x)}} {y}  M[\rollty{\mu X.A}y/x][\gamma_2]]) \]
      by the same reasoning as above, $\gamma_2(x) = \rollty{\mu X.A}V_2$, so because result is invariant under reduction we need to show
      \[ S_1[M[\gamma_1]] \ilr i \result(S_2[M[\gamma_2]]) \]
      which follows by assumption and reflexivity.
    \end{enumerate}
  \item $\nu \u Y. \u B-\beta$
    \begin{enumerate}
    \item We need to show
      \[ S_1[\unroll \rollty{\nu \u Y. \u B} M[\gamma_1]] \ix\apreorder i
      \result(S_2[M[\gamma_2]]) \]
      By the operational semantics,
      \[ S_1[\unroll \rollty{\nu \u Y. \u B} M[\gamma_1]] \bigstepsin{1} S_1[M[\gamma_1]] \]
      and by reflexivity and assumptions
      \[ S_1[M[\gamma_1]] \ix\apreorder {i} S_2[M[\gamma_2]] \]
      so the result follows by anti-reduction and downward closure.
    \item We need to show
      \[ S_1[M[\gamma_1]] \ix\apreorder i \result(S_2[\unroll \rollty{\nu \u Y. \u B} M[\gamma_2]]) \]
      By the operational semantics and invariance of result under reduction this is equivalent to
      \[ S_1[M[\gamma_1]] \ix\apreorder i \result(S_2[M[\gamma_2]]) \]
      which follows by assumption.
    \end{enumerate}
  \item $\nu \u Y. \u B-\eta$
    \begin{enumerate}
    \item We need to show
      \[ S_1[\roll \unroll M[\gamma_1]] \ix\apreorder i \result(S_2[M[\gamma_2]]) \]
      by assumption, $S_1 \itylr i {\nu \u Y.\u B} S_2$, so
      \[ S_1 = S_1'[\unroll \bullet] \]
      and therefore the left side reduces:
      \begin{align*}
         S_1[\roll \unroll M[\gamma_1]]
         &= S_1'[\unroll\roll\unroll M[\gamma_1]]\\
         &\bigstepsin{1} S_1'[\unroll M[\gamma_1]]\\
         &= S_1[M[\gamma_1]]
      \end{align*}
      and by assumption and reflexivity,
      \[ S_1[M[\gamma_1]] \ix\apreorder i \result(S_2[M[\gamma_2]]) \]
      so the result holds by anti-reduction and downward-closure.
    \item Similarly, we need to show
      \[ S_1[M[\gamma_1]] \ix\apreorder i \result(S_2[\roll\unroll M[\gamma_2]])\]
      as above, $S_1 \itylr i {\nu \u Y.\u B} S_2$, so we know
      \[ S_2 = S_2'[\unroll\bullet] \]
      so
      \[ \result(S_2[\roll\unroll M[\gamma_2]]) = \result(S_2[M[\gamma_2]])\]
      and the result follows by reflexivity, anti-reduction and downward closure.
    \end{enumerate}
  \item $0\eta$ Let $\Gamma, x : 0 \vdash M : \u B$.
    \begin{enumerate}
    \item We need to show
      \[ S_1[\absurd \gamma_1(x)] \ix\apreorder i \result(S_2[M[\gamma_2]])\]
      By assumption $\gamma_1(x) \itylr i 0 \gamma_2(x)$ but this is a contradiction
    \item Other direction is the same contradiction.
    \end{enumerate}
  \item $+\eta$. Let $\Gamma , x:A_1 + A_2 \vdash M : \u B$
    \begin{enumerate}
    \item We need to show
      \[ S_1[\caseofXthenYelseZ {\gamma_1(x)} {x_1. M[\inl x_1/x][\gamma_1]}{x_2. M[\inr x_2/x][\gamma_1]}]
      \ix\apreorder i \result(S_2[M[\gamma_2]]) \] by assumption
      $\gamma_1(x) \itylr i {A_1 + A_2} \gamma_2(x)$, so either it's
      an $\inl$ or $inr$. The cases are symmetric so assume
      $\gamma_1(x) = \inl V_1$.
      Then
      \begin{align*}
         S_1[\caseofXthenYelseZ {\gamma_1(x)} {x_1. M[\inl x_1/x][\gamma_1]}{x_2. M[\inr x_2/x][\gamma_1]}]\\
         =S_1[\caseofXthenYelseZ {(\inl V_1)} {x_1. M[\inl x_1/x][\gamma_1]}{x_2. M[\inr x_2/x][\gamma_1]}]\\
         \bigstepsin{0} S_1[M[\inl V_1/x][\gamma_1]]\\
         = S_1[M[\gamma_1]]
      \end{align*}
      and so by anti-reduction it is sufficient to show
      \[ S_1[M[\gamma_1]] \ix\apreorder i S_2[M[\gamma_2]]\]
      which follows by reflexivity and assumptions.
    \item Similarly, We need to show
      \[
      \result(S_1[M[\gamma_1]])
      \ix\apreorder i
      \result(S_2[\caseofXthenYelseZ {\gamma_2(x)} {x_1. M[\inl x_1/x][\gamma_2]}{x_2. M[\inr x_2/x][\gamma_2]}])
      \]
      and by assumption $\gamma_1(x) \itylr i {A_1 + A_2}
      \gamma_2(x)$, so either it's an $\inl$ or $inr$. The cases are
      symmetric so assume $\gamma_2(x) = \inl V_2$.
      Then
      \[ S_2[\caseofXthenYelseZ {\gamma_2(x)} {x_1. M[\inl x_1/x][\gamma_2]}{x_2. M[\inr x_2/x][\gamma_2]}] \bigstepsin{0}
      S_2[M[\gamma_2]]
      \]
      So the result holds by invariance of result under reduction,
      reflexivity and assumptions.
    \end{enumerate}
  \item $1\eta$ Let $\Gamma, x : 1 \vdash M : \u B$
    \begin{enumerate}
    \item We need to show
      \[ S_1[M[()/x][\gamma_1]] \ix\apreorder i \result(S_2[M[\gamma_2]])\]
      By assumption $\gamma_1(x) \itylr i 1 \gamma_2(x)$ so $\gamma_1(x) = ()$, so this is equivalent to 
      \[ S_1[M[\gamma_1]] \ix\apreorder i \result(S_2[M[\gamma_2]])\]
      which follows by reflexivity, assumption.
    \item Opposite case is similar.
    \end{enumerate}
  \item $\times\eta$ Let $\Gamma, x : A_1\times A_2 \vdash M : \u B$
    \begin{enumerate}
    \item We need to show
      \[ S_1[\pmpairWtoXYinZ x {x_1}{y_1} M[(x_1,y_1)/x][\gamma_1]] \ix\apreorder i \result(S_2[M[\gamma_2]]) \]
      By assumption $\gamma_1(x) \itylr i {A_1\times A_2} \gamma_2(x)$, so $\gamma_1(x) = (V_1,V_2)$, so
      \begin{align*}
        S_1[\pmpairWtoXYinZ x {x_1}{y_1} M[(x_1,y_1)/x][\gamma_1]]
        &= S_1[\pmpairWtoXYinZ {(V_1,V_2)} {x_1}{y_1} M[(x_1,y_1)/x][\gamma_1]]\\
        &\bigstepsin{0} S_1[M[(V_1,V_2)/x][\gamma_1]]\\
        &= S_1[M[\gamma_1]]
      \end{align*}
      So by anti-reduction it is sufficient to show
      \[ S_1[M[\gamma_1]] \ix\apreorder i \result(S_2[M[\gamma_2]]) \]
      which follows by reflexivity, assumption.
    \item Opposite case is similar.
    \end{enumerate}
  \item $U\eta$ Let $\Gamma \vdash V : U \u B$
    \begin{enumerate}
    \item We need to show that
      \[ \thunk\force V[\gamma_1] \itylr i {U \u B} V[\gamma_2] \]
      So assume $S_1 \itylr j {\u B} S_2$ for some $j\leq i$, then we need to show
      \[ S_1[\force \thunk\force V[\gamma_1]] \ix\apreorder j \result(S_2[\force V[\gamma_2]])\]
      The left side takes a step:
      \[ S_1[\force \thunk\force V[\gamma_1]] \bigstepsin{0} S_1[\force V[\gamma_1]] \]
      so by anti-reduction it is sufficient to show
      \[ S_1[\force V[\gamma_1]] \ix\apreorder j \result(S_2[\force V[\gamma_2]]) \]
      which follows by assumption.
    \item Opposite case is similar.
    \end{enumerate}
  \item $F\eta$
    \begin{enumerate}
    \item We need to show that given $S_1 \itylr i {\u F A} S_2$,
      \[ S_1[\bindXtoYinZ \bullet x \ret x] \itylr i {\u F A} S_2 \]
      So assume $V_1 \itylr j A V_2$ for some $j\leq i$, then we need to show
      \[ S_1[\bindXtoYinZ \bullet {\ret V_1} \ret x] \ix\apreorder j \result(S_2[\ret V_2])
      \]
      The left side takes a step:
      \[ S_1[\bindXtoYinZ \bullet {\ret V_1} \ret x] \bigstepsin{0} S_1[\ret V_1]\]
      so by anti-reduction it is sufficient to show
      \[ S_1[\ret V_1] \ix\apreorder j \result(S_2[\ret V_2])\]
      which follows by assumption
    \item Opposite case is similar.
    \end{enumerate}
  \item $\to\eta$ Let $\Gamma \vdash M : A \to \u B$
    \begin{enumerate}
    \item We need to show
      \[ S_1[(\lambda x:A. M[\gamma_1]\, x)] \ix\apreorder i \result(S_2[M[\gamma_2]])
      \]
      by assumption that $S_1 \itylr i {A \to \u B} S_2$, we know
      \[ S_1 = S_1'[\bullet\, V_1]\]
      so the left side takes a step:
      \begin{align*}
         S_1[(\lambda x:A. M[\gamma_1]\, x)]
         &= S_1'[(\lambda x:A. M[\gamma_1]\, x)\, V_1]\\
         &\bigstepsin{0} S_1'[M[\gamma_1]\, V_1]\\
         &= S_1[M[\gamma_1]]
      \end{align*}
      So by anti-reduction it is sufficient to show
      \[ S_1[M[\gamma_1]] \ix\apreorder i \result(S_2[M[\gamma_2]])\]
      which follows by reflexivity, assumption.
    \item Opposite case is similar.
    \end{enumerate}
  \item $\with\eta$ Let $\Gamma \vdash M : \u B_1 \with \u B_2$
    \begin{enumerate}
    \item We need to show
      \[ S_1[\pair{\pi M[\gamma_1]}{\pi' M[\gamma_1]}] \ix\apreorder i \result(S_1[M[\gamma_2]]) \]
      by assumption, $S_1 \itylr i {\u B_1 \with \u B_2} S_2$ so
      either it starts with a $\pi$ or $\pi'$ so assume that $S_1 =
      S_1'[\pi \bullet]$ ($\pi'$ case is similar).
      Then the left side reduces
      \begin{align*}
        S_1[\pair{\pi M[\gamma_1]}{\pi' M[\gamma_1]}]
        &= S_1'[\pi\pair{\pi M[\gamma_1]}{\pi' M[\gamma_1]}]\\
        &\bigstepsin{0} S_1'[\pi M[\gamma_1]]\\
        &= S_1[M[\gamma_1]]
      \end{align*}
      So by anti-reduction it is sufficient to show
      \[ S_1[M[\gamma_1]] \ix\apreorder i \result(S_2[M[\gamma_2]]) \]
      which follows by reflexivity, assumption.
    \item Opposite case is similar.
    \end{enumerate}
  \item $\top\eta$ Let $\Gamma \vdash M : \top$
    \begin{enumerate}
    \item In either case, we assume we are given $S_1 \itylr i \top
      S_2$, but this is a contradiction.
    \end{enumerate}
  \end{enumerate}
\end{longproof}

\begin{lemma}[Substitution Principles]
  For any diverge-bottom preorder $\apreorder$, the following are
  valid
  \begin{enumerate}
  \item $\inferrule{\Gamma \vDash V_1 \ilr i V_2 \in A
    \and \Gamma, x : A \vDash V_1' \ilr V_2' \in A'}{\Gamma \vDash V_1'[V_1/x] \ilr V_2'[V_2/x] \in A'}$
  \item $\inferrule{\Gamma \vDash V_1 \ilr i V_2 \in A
    \and \Gamma, x : A \vDash M_1 \ilr M_2 \in \u B}{\Gamma \vDash M_1[V_1/x] \ilr M_2[V_2/x] \in \u B}$
  \end{enumerate}
\end{lemma}
\begin{longproof}
  We do the term case, the value case is similar.  Given $\gamma_1
  \itylr i \Gamma \gamma_2$, we have $V_1[\gamma_1] \itylr i A
  V_2[\gamma_2]$ so
  \[ \gamma_1,V_1[\gamma_1]/x \itylr i {\Gamma, x : A} \gamma_2, V_2[\gamma_2]/x \]
  and by associativity of substitution
  \[ M_1[V_1/x][\gamma_1] = M_1[\gamma_1,V_1[\gamma_1]/x] \]
  and similarly for $M_2$, so if $S_1 \itylr i {\u B} S_2$ then
  \[ S_1[M_1[\gamma_1,V_1[\gamma_1]/x]] \ix\apreorder i \result(S_2[M_2[\gamma_2,V_2[\gamma_2]/x]])\]
\end{longproof}
\fi
For errors, the strictness axioms hold for any $\apreorder$, but the axiom that
$\err$ is a least element is specific to the definitions of
$\precltdyn, \ltdyn\succeq$
\begin{lemma}[Error Rules]
  For any divergence preorder $\apreorder$ and appropriately
  typed $S, M$,
  \begin{small}
  \begin{mathpar}
    S[\err] \ilr \omega \err \and
    \err \ilr \omega S[\err] \and
    \err \ilrof\precltdyn \omega M \and
    M \ilrof{\errordivergerightop} \omega \err
  \end{mathpar}    
  \end{small}
\end{lemma}
\begin{longproof}
  \begin{enumerate}
  \item It is sufficient by the limit lemma to show $\result(S[\err])
    \apreorder \err$ which holds by reflexivity because $S[\err]
    \bigstepzero \err$.
  \item We need to show $S[\err] \ix\precltdyn i R$ for arbitrary $R$,
    so by the limit lemma it is sufficient to show $\err \precltdyn
    R$, which is true by definition.
  \item By the limit lemma it is sufficient to show $R
    \mathrel{\errordivergerightop} \err$ which is true by definition.
  \end{enumerate}
\end{longproof}

The lemmas we have proved cover all of the inequality rules of CBPV, so
applying them with $\apreorder$ chosen to be $\errordivergeleft$ and
$\errordivergerightop$ gives
\begin{lemma}[$\precltdyn$ and $\ltdynsucc$ are Models of CBPV] \label{lem:errordivergeleftrightopmodels}
  If $\Gamma \pipe \Delta \vdash E \ltdyn E' : \u B$ then 
  $\Gamma \pipe \Delta \vDash E \ix\precltdyn \omega E' \in \u B$ and
  $\Gamma \pipe \Delta \vDash E' \ix{\mathrel{\preceq\gtdyn}} \omega E \in \u B$.
\end{lemma}

Because logical implies contextual equivalence, we can
conclude with the main theorem:
\begin{theorem}[Contextual Approximation/Equivalence Model CBPV] ~~\\
  If $\Gamma \pipe \Delta \vdash E \ltdyn E' : T$ then
  $\Gamma \pipe \Delta \vDash E \ctxize\ltdyn E' \in T$;
  if
  ${\Gamma \pipe \Delta \vdash E \equidyn E' : T}$ then
  ${\Gamma \pipe \Delta \vDash E \ctxize= E' \in T}$.
\end{theorem}
\begin{longproof}

  
For the first part, from Lemma~\ref{lem:errordivergeleftrightopmodels},
we have $E \ix\precltdyn \omega E'$ and $E' \ix{\mathrel{\preceq\gtdyn}}
\omega E$.  By Lemma~\ref{lem:logical-implies-contextual}, we then have
$E \ctxize{\errordivergeleft} E'$ and $E' \ctxize{\errordivergerightop}
E$.  Finally, by Corollary~\ref{cor:contextual-decomposition}, $E
\ctxize\ltdyn E' \text{ iff } E \ctxize{\errordivergeleft} E' \text{and
} E (\ctxize{(\errordivergerightop)})^\circ E'$, so we have the result.

For the second part, applying the first part twice gives $E
\ctxize\ltdyn E'$ and $E' \ctxize\ltdyn E$, and we concluded in
Corollary~\ref{cor:contextual-decomposition} that this coincides with
contextual equivalence.
\end{longproof}

\section{Discussion and Related Work}
\label{sec:related}

In this paper, we have given a logic for reasoning about gradual
programs in a mixed call-by-value/call-by-name language, shown that
the axioms uniquely determine almost all of the contract translation
implementing runtime casts, and shown that the axiomatics is sound for
contextual equivalence/approximation in an operational model.
\iflong

\fi
In immediate future work, we believe it is straightforward to add
inductive/coinductive types and obtain similar unique cast
implementation theorems
(e.g. $\upcast{\mathtt{list}(A)}{\mathtt{list}(A')} \equidyn
\mathtt{map}\upcast{A}{A'}$).  Additionally, since more efficient cast
implementations such as optimized cast calculi (the lazy variant in
\citet{herman2010spaceefficient}) and threesome
casts~\cite{siekwadler10zigzag}, are equivalent to the lazy contract
semantics, they should also be models of GTT, and if so we could use GTT
to reason about program transformations and optimizations in them.

\iflong
\paragraph{Applicability of Cast Uniqueness Principles}
\fi

The cast uniqueness principles given in theorem
\ref{thm:functorial-casts} are theorems in the formal logic of Gradual
Type Theory, and so there is a question of to what languages the
theorem applies.
The theorem applies to any \emph{model} of gradual type theory, such
as the models we have constructed using call-by-push-value given in
Sections \ref{sec:contract}, \ref{sec:complex}, \ref{sec:operational}.
We conjecture that simple call-by-value and call-by-name gradual
languages are also models of GTT, by extending the translation of
call-by-push-value into call-by-value and call-by-name in the appendix
of Levy's monograph \cite{levy03cbpvbook}.
In order for the theorem to apply, the language must validate an
appropriate version of the $\eta$ principles for the types.
So for example, a call-by-value language that has reference equality
of functions does \emph{not} validate even the value-restricted $\eta$
law for functions, and so the case for functions does not apply.
It is a well-known issue that in the presence of pointer equality of
functions, the lazy semantics of function casts is not compatible with
the graduality property, and our uniqueness theorem provides a
different perspective on this phenomenon
\cite{findlerflattfelleisen04,chaperonesimpersonators, refined}.
However, we note that the cases of the uniqueness theorem for each
type connective are completely \emph{modular}: they rely only on the
specification of casts and the $\beta,\eta$ principles for the
particular connective, and not on the presence of any other types,
even the dynamic types.
So even if a call-by-value language may have reference equality
functions, if it has the $\eta$ principle for strict pairs, then the
pair cast must be that of Theorem \ref{thm:functorial-casts}.

Next, we consider the applicability to non-eager languages.
Analogous to call-by-value, our uniqueness principle should apply to
simple \emph{call-by-name} gradual languages, where full $\eta$
equality for functions is satisfied, but $\eta$ equality for booleans
and strict pairs requires a ``stack restriction'' dual to the value
restriction for call-by-value function $\eta$.
We are not aware of any call-by-name gradual languages, but there is
considerable work on \emph{contracts} for non-eager languages,
especially Haskell \cite{hinzeJeuringLoh06,XuPJC09}.
However, we note that Haskell is \emph{not} a call-by-name language in
our sense for two reasons.
First, Haskell uses call-by-need evaluation where results of
computations are memoized. However, when only considering Haskell's
effects (error and divergence), this difference is not observable so
this is not the main obstacle.
The bigger difference between Haskell and call-by-name is that Haskell
supports a \texttt{seq} operation that enables the programmer to force
evaluation of a term to a value.
This means Haskell violates the function $\eta$ principle because
$\Omega$ will cause divergence under $\seq$, whereas $\lambda
x. \Omega$ will not.
This is a crucial feature of Haskell and is a major source of
differences between implementations of lazy contracts, as noted in
\citet{Degen2012TheIO}.
We can understand this difference by using a different translation
into call-by-push-value: what Levy calls the ``lazy paradigm'', as
opposed to call-by-name \cite{levy03cbpvbook}.
Simply put, connectives are interpreted as in call-by-value, but with
the addition of extra thunks $UF$, so for instance the lazy function
type $A \to B$ is interpreted as $UFU(UFA \to FB)$ and the extra $UFU$
here is what causes the failure of the call-by-name $\eta$ principle.
With this embedding and the uniqueness theorem, GTT produces a
definition for lazy casts, and the definition matches the work of
\citet{XuPJC09} when restricting to non-dependent contracts.

\iflong\paragraph{Comparing Soundness Principles for Cast Semantics}\fi
\citet{greenmanfelleisen:2018} gives a spectrum of
differing syntactic type soundness theorems for different semantics of
gradual typing.
Our work here is complementary, showing that certain program
equivalences can only be achieved by certain cast semantics.

\citet{Degen2012TheIO} give an analysis of different cast semantics
for contracts in lazy languages, specifically based on Haskell, i.e., 
call-by-need with \seq.
They propose two properties ``meaning preservation'' and
``completeness'' that they show are incompatible and identify which
contract semantics for a lazy language satisfy which of the
properties.
The meaning preservation property is closely related to graduality: it
says that evaluating a term with a contract either produces blame or
has the same observable effect as running the term without the
contract.
Meaning preservation rules out overly strict contract systems that
force (possibly diverging) thunks that wouldn't be forced in a
non-contracted term.
Completeness, on the other hand, requires that when a contract is
attached to a value that it is \emph{deeply} checked.
The two properties are incompatible because, for instance, a pair of a
diverging term and a value can't be deeply checked without causing the
entire program to diverge.
Using Levy's embedding of the lazy paradigm into call-by-push-value
their incompatibility theorem should be a consequence of our main
theorem in the following sense.
We showed that any contract semantics departing from the
implementation in Theorem \ref{thm:functorial-casts} must violate
$\eta$ or graduality.
Their completeness property is inherently eager, and so must be
different from the semantics GTT would provide, so either the
restricted $\eta$ or graduality fails.
However, since they are defining contracts within the language, they
satisfy the restricted $\eta$ principle provided by the language, and
so it must be graduality, and therefore meaning preservation that
fails.

\iflong\paragraph{Axiomatic Casts}\fi

Henglein's work on dynamic typing also uses an axiomatic semantics of
casts, but axiomatizes behavior of casts at each type directly whereas
we give a uniform definition of all casts and derive implementations
for each type \cite{henglein94:dynamic-typing}.
Because of this, the theorems proven in that paper are more closely
related to our model construction in Section
\ref{sec:contract}.
More specifically, many of the properties of casts needed to prove
Theorem \ref{thm:axiomatic-graduality} have direct analogues
in Henglein's work, such as the coherence theorems.
We have not included these lemmas in the paper because they are quite
similar to lemmas proven in \citet{newahmed18}; see there for a more
detailed comparison, and the extended version of this paper for full proof details \citep{newlicataahmed19:extended}.
Finally, we note that our assumption of compositionality, i.e., that
all casts can be decomposed into an upcast followed by a downcast, is
based on Henglein's analysis, where it was proven to hold in his
coercion calculus.

\iflong
\paragraph{Gradual Typing Frameworks}
\fi

In this work we have applied a method of ``gradualizing'' axiomatic
type theories by adding in dynamism orderings and adding dynamic
types, casts and errors by axioms related to the dynamism orderings.
This is similar in spirit to two recent frameworks for designing
gradual languages: Abstracting Gradual Typing (AGT) \citep{AGT} and the
Gradualizer \citep{gradualizer16,gradualizer17}.
All of these approaches start with a typed language and construct a
related gradual language.
A major difference between our approach and those is that our work is
based on axiomatic semantics and so we take into account the equality
principles of the typed language, whereas Gradualizer is based on the
typing and operational semantics and AGT is based on the type safety
proof of the typed language.
Furthermore, our approach produces not just a single language, but
also an axiomatization of the structure of gradual typing and so we
can prove results about many languages by proving theorems in GTT.
The downside to this is that our approach doesn't directly provide an
operational semantics for the gradual language, whereas for AGT this
is a semi-mechanical process and for Gradualizer, completely
automated.
Finally, we note that AGT produces the ``eager'' semantics for
function types, and it is not clear how to modify the AGT methodology
to reproduce the lazy semantics that GTT provides.
More generally, both AGT and the Gradualizer are known to produce
violations of parametricity when applied to polymorphic languages,
with the explanation being that the parametricity property is in no
way encoded in the input to the systems: the operational semantics and
the type safety proof.
In future work, we plan to apply our axiomatic approach to gradualizing
polymorphism and state by starting with the rich \emph{relational logics
  and models} of program equivalence for these
features~\cite{plotkinabadi93, dunphyphd, ahmed08:paramseal, neis09,
  ahmed09:sdri}, which may lend insight into existing
proposals~\cite{siek15:mono,ahmed17,igarashipoly17,siek-taha06}--- for
example, whether the ``monotonic'' \citep{siek15:mono} and ``proxied''
\citep{siek-taha06} semantics of references support relational reasoning
principles of local state.

\iflong \paragraph{Blame}
\fi
We do not give a treatment of runtime blame reporting, but we argue that
the observation that upcasts are thunkable and downcasts are linear is
directly related to blame soundness~\cite{tobin-hochstadt06,wadler-findler09} in that if
an upcast were \emph{not} thunkable, it should raise positive blame and
if a downcast were \emph{not} linear, it should raise negative blame.
First, consider a potentially effectful stack upcast of the form
$\upcast{\u F A}{\u F A'}$. If it is not thunkable, then in our logical
relation this would mean there is a value $V : A$ such that $\upcast{\u
  F A}{\u F A'}(\ret V)$ performs some effect.
Since the only observable effects for casts are dynamic type errors, 
$\upcast{\u F A}{\u F A'}(\ret V) \mapsto \err$, and we must decide
whether the positive party or negative party is at fault.
However, since this is call-by-value evaluation, this error happens
unconditionally on the continuation, so the continuation never had a
chance to behave in such a way as to prevent blame, and so we must blame the
positive party.
Dually, consider a value downcast of the form $\dncast{U \u B}{U \u B'}$.
If it is not linear, that would mean it forces its $U \u B'$
input either never or more than once.
Since downcasts should refine their inputs, it is not possible for
the downcast to use the argument twice, since e.g. printing twice does not
refine printing once.
So if the cast is not linear, that means it fails without ever forcing
its input, in which case it knows nothing about the positive party and
so must blame the negative party.
In future work, we plan to investigate extensions of GTT with more than
one $\err$ with different blame labels, and an axiomatic account of
a blame-aware observational equivalence.




\begin{longonly}
\paragraph{Denotational and Category-theoretic Models}

We have presented certain concrete models of GTT using ordered CBPV
with errors, in order to efficiently arrive at a concrete operational
interpretation.
It may be of interest to develop a more general notion of model of GTT
for which we can prove soundness and completeness theorems, as in
\citet{newlicata2018-fscd}.
A model would be a strong adjunction between double categories where
one of the double categories has all ``companions'' and the other has
all ``conjoints'', corresponding to our upcasts and downcasts.
Then the contract translation should be a construction that takes a
strong adjunction between 2-categories and makes a strong adjunction
between double categories where the ep pairs are ``Kleisli'' ep pairs:
the upcast is has a right adjoint, but only in the Kleisli category and
vice-versa the downcast has a left adjoint in the co-Kleisli category.

Furthermore, the ordered CBPV with errors should also have a sound and
complete notion of model, and so our contract translation should have
a semantic analogue as well.
\end{longonly}

\begin{longonly}
  \paragraph{Gradual Session Types} ~
\end{longonly}
Gradual session types~\cite{igarashi+17gradualsession} share some
similarities to GTT, in that there are two sorts of types (values and
sessions) with a dynamic value type and a dynamic session type.
However, their language is not \emph{polarized} in the same way as CBPV,
so there is not likely an analogue between our upcasts always being
between value types and downcasts always being between computation
types.  Instead, we might reconstruct this in a polarized session
type language~\cite{pfenninggriffith15session}.
\begin{longonly}
The two dynamic types would then be the ``universal sender'' and
``universal receiver'' session types.
\end{longonly}


\begin{longonly}
\paragraph{Dynamically Typed Call-by-push-value}

Our interpretation of the dynamic types in CBPV suggests a design for
a Scheme-like language with a value and computation distinction.
This may be of interest for designing an extension of Typed Racket that
efficiently supports CBN or a Scheme-like language with codata types.
While the definition of the dynamic computation type by a lazy product
may look strange, we argue that it is no stranger than the use of its
dual, the sum type, in the definition of the dynamic value type.
That is, in a truly dynamically typed language, we would not think of
the dynamic type as being built out of some sum type construction, but
rather that it is the \emph{union} of all of the ground value types, and the
union happens to be a \emph{disjoint} union and so we can 
model it as a sum type.
In the dual, we don't think of the computation dynamic type as a
\emph{product}, but instead as the \emph{intersection} of the ground
computation types.
Thinking of the type as unfolding:
\[ \dync = \u F\dync \wedge (\dynv \to \u F \dynv) \wedge (\dynv \to \dynv \to \u F \dynv) \wedge \cdots \]
This says that a dynamically typed computation is one that can be
invoked with any finite number of arguments on the stack, a fairly
accurate model of implementations of Scheme that pass multiple
arguments on the stack.
\end{longonly}

\iflong
\paragraph{Dependent Contract Checking}
\fi

We also plan to explore using GTT's specification of casts in a
dependently typed setting, building on work using
Galois connections for casts between dependent
types~\cite{dagand+18interoperability}, and work on effectful dependent
types based a CBPV-like judgement
structure~\cite{ahman+16fiberedeffects}.




\paragraph{Acknowledgments}
We thank Ron Garcia, Kenji Maillard and Gabriel Scherer for helpful
discussions about this work. We thank the anonymous reviewers for
helpful feedback on this article. This material is based on research
sponsored by the National Science Foundation under grant CCF-1453796
and the United States Air Force Research Laboratory under agreement
number FA9550-15-1-0053 and FA9550-16-1-0292.
The views and conclusions contained herein are those of the authors and
should not be interpreted as necessarily representing the official
policies or endorsements, either expressed or implied, of the United
States Air Force Research Laboratory, the U.S. Government, or Carnegie
Mellon University.

\bibliography{max}

\end{document}